\documentclass[10pt,letterpaper]{report} % creates the book pdf version

\usepackage{fullpage}
\usepackage[mathscr]{eucal}
\usepackage[cmex10]{amsmath}
\usepackage{epsfig,epsf,psfrag}
\usepackage{amssymb,amsmath,amsthm,amsfonts,latexsym}
\usepackage{amsmath,graphicx,bm,xcolor,url,overpic}
\usepackage[caption=false]{subfig} 
\usepackage{fixltx2e}%ordering of single and double column floats
\usepackage{array}%array and tabular environments
\usepackage{verbatim}
\usepackage{bm}
\usepackage{algorithmic}
\usepackage{algorithm}
\usepackage{verbatim}
\usepackage{textcomp}
\usepackage{mathrsfs}
\usepackage{epstopdf}

\newcommand{\openone}{\leavevmode\hbox{\small1\normalsize\kern-.33em1}}

\catcode`~=11 \def\UrlSpecials{\do\~{\kern -.15em\lower .7ex\hbox{~}\kern .04em}} \catcode`~=13 

\allowdisplaybreaks[4]
 
\newcommand{\snr}{\mathsf{snr}}
\newcommand{\inr}{\mathsf{inr}}

\newcommand{\nn}{\nonumber}

\newcommand{\calA}{\mathcal{A}}
\newcommand{\calB}{\mathcal{B}}
\newcommand{\calC}{\mathcal{C}}
\newcommand{\calD}{\mathcal{D}}
\newcommand{\calE}{\mathcal{E}}
\newcommand{\calF}{\mathcal{F}}
\newcommand{\calG}{\mathcal{G}}

\newcommand{\calI}{\mathcal{I}}

\newcommand{\calK}{\mathcal{K}}
\newcommand{\calL}{\mathcal{L}}
\newcommand{\calM}{\mathcal{M}}
\newcommand{\calN}{\mathcal{N}}

\newcommand{\calP}{\mathcal{P}}

\newcommand{\calR}{\mathcal{R}}
\newcommand{\calS}{\mathcal{S}}
\newcommand{\calT}{\mathcal{T}}
\newcommand{\calU}{\mathcal{U}}

\newcommand{\calX}{\mathcal{X}}
\newcommand{\calY}{\mathcal{Y}}
\newcommand{\calZ}{\mathcal{Z}}

\newcommand{\ba}{\mathbf{a}}
\newcommand{\bA}{\mathbf{A}}

\newcommand{\bB}{\mathbf{B}}

\newcommand{\bD}{\mathbf{D}}

\newcommand{\bG}{\mathbf{G}}
\newcommand{\bh}{\mathbf{h}}
\newcommand{\bH}{\mathbf{H}}

\newcommand{\bI}{\mathbf{I}}
\newcommand{\bj}{\mathbf{j}}
\newcommand{\bJ}{\mathbf{J}}
\newcommand{\bk}{\mathbf{k}}

\newcommand{\bR}{\mathbf{R}}
\newcommand{\bs}{\mathbf{s}}

\newcommand{\bu}{\mathbf{u}}
\newcommand{\bU}{\mathbf{U}}
\newcommand{\bv}{\mathbf{v}}
\newcommand{\bV}{\mathbf{V}}
\newcommand{\bw}{\mathbf{w}}

\newcommand{\bx}{\mathbf{x}}

\newcommand{\by}{\mathbf{y}}

\newcommand{\bz}{\mathbf{z}}

\newcommand{\rmc}{\mathrm{c}}

\newcommand{\rmd}{\mathrm{d}}

\newcommand{\rme}{\mathrm{e}}

\newcommand{\rmh}{\mathrm{h}}

\newcommand{\rmp}{\mathrm{p}}
\newcommand{\rmP}{\mathrm{P}}

\newcommand{\rms}{\mathrm{s}}

\newcommand{\rmu}{\mathrm{u}}

\newcommand{\bbE}{\mathsf{E}}

\newcommand{\bbI}{\openone}

\newcommand{\bbN}{\mathbb{N}}

\newcommand{\bbR}{\mathbb{R}}

\newcommand{\scP}{\mathscr{P}}

\newcommand{\scR}{\mathscr{R}}

\newcommand{\scV}{\mathscr{V}}

\DeclareMathAlphabet{\mathbsf}{OT1}{cmss}{bx}{n}
\DeclareMathAlphabet{\mathssf}{OT1}{cmss}{m}{sl}% slanted sans serif

\newcommand{\rvC}{\mathsf{C}}

\newcommand{\rvH}{\mathsf{H}}

\newcommand{\rvV}{\mathsf{V}}

\DeclareSymbolFont{bsfletters}{OT1}{cmss}{bx}{n}  
\DeclareSymbolFont{ssfletters}{OT1}{cmss}{m}{n}
\DeclareMathSymbol{\bsfGamma}{0}{bsfletters}{'000}
\DeclareMathSymbol{\ssfGamma}{0}{ssfletters}{'000}
\DeclareMathSymbol{\bsfDelta}{0}{bsfletters}{'001}
\DeclareMathSymbol{\ssfDelta}{0}{ssfletters}{'001}
\DeclareMathSymbol{\bsfTheta}{0}{bsfletters}{'002}
\DeclareMathSymbol{\ssfTheta}{0}{ssfletters}{'002}
\DeclareMathSymbol{\bsfLambda}{0}{bsfletters}{'003}
\DeclareMathSymbol{\ssfLambda}{0}{ssfletters}{'003}
\DeclareMathSymbol{\bsfXi}{0}{bsfletters}{'004}
\DeclareMathSymbol{\ssfXi}{0}{ssfletters}{'004}
\DeclareMathSymbol{\bsfPi}{0}{bsfletters}{'005}
\DeclareMathSymbol{\ssfPi}{0}{ssfletters}{'005}
\DeclareMathSymbol{\bsfSigma}{0}{bsfletters}{'006}
\DeclareMathSymbol{\ssfSigma}{0}{ssfletters}{'006}
\DeclareMathSymbol{\bsfUpsilon}{0}{bsfletters}{'007}
\DeclareMathSymbol{\ssfUpsilon}{0}{ssfletters}{'007}
\DeclareMathSymbol{\bsfPhi}{0}{bsfletters}{'010}
\DeclareMathSymbol{\ssfPhi}{0}{ssfletters}{'010}
\DeclareMathSymbol{\bsfPsi}{0}{bsfletters}{'011}
\DeclareMathSymbol{\ssfPsi}{0}{ssfletters}{'011}
\DeclareMathSymbol{\bsfOmega}{0}{bsfletters}{'012}
\DeclareMathSymbol{\ssfOmega}{0}{ssfletters}{'012}

\newcommand{\hatH}{\hat{H}}

\newcommand{\hatI}{\hat{I}}

\newcommand{\hatl}{\hat{l}}

\newcommand{\till}{\tilde{l}}

\newcommand{\hatm}{\hat{m}}
\newcommand{\hatM}{\hat{M}}
\newcommand{\tilm}{\tilde{m}}

\newcommand{\tilP}{\tilde{P}}

\newcommand{\tilQ}{\tilde{Q}}

\newcommand{\hats}{\hat{s}}

\newcommand{\tilT}{\tilde{T}}

\newcommand{\tilW}{\tilde{W}}

\newcommand{\hatx}{\hat{x}}
\newcommand{\hatX}{\hat{X}}

\newcommand{\barx}{\bar{x}}

\newcommand{\barU}{\bar{U}}

\newcommand{\barX}{\bar{X}}

\newcommand{\bgamma}{\bm{\gamma}}

\newcommand{\bmu}{\bm{\mu}}

\newcommand{\bSigma	}{\bm{\Sigma}}

\newcommand{\hrho}{\hat{\rho}}

\newcommand{\iid}{i.i.d.\ }

\DeclareMathOperator*{\argmin}{arg\,min}

\DeclareMathOperator{\var}{\mathsf{Var}}

\DeclareMathOperator{\cov}{\mathsf{Cov}}

\newcommand{\bzero}{\mathbf{0}}
\newcommand{\bone}{\mathbf{1}}

\newtheorem{theorem}{Theorem}[chapter] 
\newtheorem{lemma}{Lemma}[chapter]

\newtheorem{proposition}{Proposition}[chapter]
\newtheorem{corollary}{Corollary}[chapter]
 
\newtheorem{example}{Example}[chapter] 
 
\newtheorem{remark}{Remark}[chapter]

\title{Asymptotic Estimates in Information Theory with Non-Vanishing Error Probabilities}
\author{
Vincent~Y.~F. Tan \\
Department of Electrical and Computer Engineering   \\
Department of Mathematics\\
National University of Singapore\\
Singapore 119077\\
Email: vtan@nus.edu.sg}

\begin{document}

% the following settings can be set or left blank at first
%\copyrightowner{V. Y. F. Tan}
%\volume{11}
%\issue{1-2}
%\pubyear{2014}
%\copyrightyear{2014}
%\isbn{978-1-60198-852-2}
%\doi{10.1561/0100000086}
%\firstpage{1}
%\lastpage{184}
%
%
%
%\frontmatter  % title page, contents, catalog information

\maketitle

\tableofcontents

%\mainmatter

\chapter*{Abstract}
\addcontentsline{toc}{chapter}{Abstract}
This  monograph  presents a unified treatment of    single- and multi-user  problems in Shannon's information theory where we depart from the requirement that the error probability decays asymptotically in the blocklength.  Instead, the error probabilities for various problems are bounded above by a  non-vanishing  constant and   the spotlight is shone on achievable coding rates  as functions of the growing  blocklengths.   
%This area of study in information theory is known as  {\em asymptotic estimates with   non-vanishing error probabilities}.% or, more succinctly as, 
This represents the study of   {\em asymptotic estimates with   non-vanishing error probabilities}.% or, more succinctly as, 
% {\em fixed error   asymptotics}.
 % in short.
% The primary focus of this monograph is to find {\em asymptotic estimates} in various information-theoretic problems each with {\em non-vanishing error probabilities}. 

In Part I, after reviewing the fundamentals of information theory, we discuss Strassen's seminal result for binary hypothesis testing where the type-I error probability is non-vanishing and the rate of  decay of the  type-II error probability  with growing number of independent observations is characterized. In Part II, we use this basic hypothesis testing result to develop second- and sometimes, even third-order asymptotic expansions for point-to-point communication. Finally in Part III, we consider network information theory problems for which  the second-order asymptotics are known. These problems include some classes of   channels with random state, the multiple-encoder distributed lossless source coding (Slepian-Wolf) problem and special cases of the Gaussian interference and multiple-access channels. Finally, we discuss avenues for further research.

\part{Fundamentals}

\newcommand{\eps}{\varepsilon}
\chapter{Introduction}
%This section serves as a guide for the reader.
%\section{Motivation}

Claude E.\ Shannon's epochal {\em ``A Mathematical Theory of Communication''}~\cite{Shannon48} marks the  dawn of the digital age. In his seminal paper, Shannon laid the theoretical and mathematical foundations for the  basis of all communication systems today. It is not an exaggeration to say that his work has   had a tremendous impact in communications  engineering and beyond, in fields as diverse as  statistics, economics, biology and cryptography, just to name a few.

It has been more than 65 years since Shannon's landmark work was published. Along with impressive research advances in the field of {\em information theory}, numerous excellent books on various aspects of the subject have been written. The author's favorites include  Cover and Thomas~\cite{Cov06}, Gallager~\cite{gallagerIT}, Csisz\'ar and K\"orner~\cite{Csi97}, Han~\cite{Han10}, Yeung~\cite{Yeung} and El Gamal and Kim~\cite{elgamal}. Is there sufficient motivation to consolidate and present another aspect of information theory systematically? It is the author's hope that the answer is in the affirmative.  

To motivate why this is so, let us recapitulate  two of Shannon's major contributions in his 1948 paper. First, Shannon showed that to {\em reliably} compress a discrete memoryless source (DMS) $X^n = (X_1, \ldots, X_n)$ where each  $X_i$ has the same distribution as a common random variable $X$, it is sufficient to use $H(X)$ bits per source symbol in the limit of large blocklengths $n$, where $H(X)$ is the Shannon entropy of the source.  By {\em reliable}, it is meant that  the   probability of incorrect  decoding of the  source sequence   tends to zero as the blocklength $n$ grows.  Second, Shannon showed that it is possible to   {\em reliably} transmit a message $M \in \{1,\ldots, 2^{nR}\}$ over a discrete memoryless channel  (DMC) $W$ as long as the message rate $R$ is smaller than the capacity of the channel $C(W)$. Similarly to the source compression scenario, by {\em reliable},  one means that the probability of  incorrectly decoding $M$   tends to zero as $n$ grows. 

There is, however, substantial motivation to revisit the criterion  of  having   error probabilities vanish asymptotically.   To state Shannon's source compression result more formally, let us define $M^*(P^n,\eps)$ to be the minimum code size for which the length-$n$ DMS $P^n$ is compressible to within an error probability $\eps \in (0,1)$. 
%By $P^n$,  we mean that the blocklength of the fixed-to-fixed length code is equal to $n$ and the source is stationary and memoryless.  
Then,  Theorem~3 of  Shannon's paper~\cite{Shannon48}, together with the strong converse for lossless source coding~\cite[Ex.~3.15]{elgamal}, states that 
\begin{equation}
\lim_{n\to\infty}\frac{1}{n}\log M^*(P^n,\eps) = H(X),\quad \mbox{bits per source symbol}. \label{eqn:intro_src}
\end{equation}
Similarly, denoting $M^*_{\mathrm{ave}}(W^n,\eps)$ as the maximum code size for which it is possible to communicate over a DMC  $W^n$ such that the average error probability is no larger than $\eps$, Theorem~11 of  Shannon's paper~\cite{Shannon48}, together with the strong converse for channel coding~\cite[Thm.~2]{Wolfowitz57},  states that 
\begin{equation}
\lim_{n\to\infty}\frac{1}{n}\log M^*_{\mathrm{ave}}(W^n,\eps) = C(W),\quad \mbox{bits per channel use}. \label{eqn:intro_ch}
\end{equation}
In many practical communication settings, one does not have the luxury of being able to design an arbitrarily long code, so one must settle for a non-vanishing,  and hence finite, error probability $\eps$.   In this {\em finite blocklength}  and {\em non-vanishing error probability} setting, how close can one hope  to get to  the asymptotic limits   $H(X)$ and $C(W)$?   This is, in general a difficult question because   exact evaluations of $\log M^*(P^n,\eps) $ and $\log M^*_{\mathrm{ave}}(W^n,\eps)$ are intractable, apart from a few special sources and channels.  

In the early years of information theory, Dobrushin~\cite{Dobrushin}, Kemperman~\cite{Kemperman} and, most prominently,   Strassen~\cite{Strassen}  studied approximations to   $\log M^*(P^n,\eps) $ and $\log M^*_{\mathrm{ave}}(W^n,\eps)$. These beautiful works were largely forgotten until recently, when interest in   so-called {\em Gaussian approximations} were revived by Hayashi~\cite{Hayashi08,Hayashi09} and Polyanskiy-Poor-Verd\'u~\cite{PPV08a,PPV10}.\footnote{Some of the results in~\cite{PPV08a,PPV10} were already announced by S.~Verd\'u in his Shannon lecture at the  2007 International Symposium on Information Theory (ISIT) in Nice, France.}  Strassen showed that the limiting statement in~\eqref{eqn:intro_src} may be refined to yield the   {\em asymptotic expansion}
\begin{equation}
\log M^*(P^n,\eps) = n H(X) -\sqrt{nV(X)} \Phi^{-1}({\eps} ) - \frac{1}{2}\log n + O(1), \label{eqn:int_src_expand}
\end{equation}
where $V(X)$ is  known as the {\em source dispersion} or the {\em varentropy},  terms introduced by Kostina-Verd\'u~\cite{kost12} and Kontoyiannis-Verd\'u~\cite{verdu14}. In~\eqref{eqn:int_src_expand}, $\Phi^{-1}$ is the inverse of the Gaussian cumulative distribution function. Observe that the first-order term in the asymptotic expansion above, namely $H(X)$, coincides with the (first-order) fundamental limit shown by Shannon.  From this expansion, one sees that  if the error probability is fixed to $\eps<\frac12$, the     extra rate above the entropy we have to pay for operating   at   finite blocklength $n$ with admissible error probability $\eps$ is approximately $\sqrt{ V(X)/n}\,  \Phi^{-1}(1-\eps)$. Thus, the quantity $V(X)$, which is a function of $P$ just like the entropy  $H(X)$, quantifies how fast the rates of optimal source codes converge to $H(X)$.  Similarly, for well-behaved DMCs, under mild conditions, Strassen showed that the limiting statement in \eqref{eqn:intro_ch} may be refined to 
\begin{equation}
\log M^*_{\mathrm{ave}}(W^n,\eps) = nC(W)  +\sqrt{nV_\eps(W)} \Phi^{-1}({\eps} ) + O(\log n)\label{eqn:int_ch_expand}
\end{equation}
and $V_\eps(W)$ is a  channel parameter known as the {\em $\eps$-channel dispersion}, a term introduced by Polyanskiy-Poor-Verd\'u~\cite{PPV10}. Thus the backoff from capacity at finite blocklengths $n$ and    average error probability $\eps$ is approximately $\sqrt{V_\eps(W)/n} \, \Phi^{-1}(1-\eps)$. 
\section{Motivation for this Monograph}
It turns out that   Gaussian approximations  (first two terms of \eqref{eqn:int_src_expand} and~\eqref{eqn:int_ch_expand}) are {\em good proxies} to the true non-asymptotic fundamental limits ($\log M^*(P^n,\eps)$ and $\log M^*_{\mathrm{ave}}(W^n,\eps)$) at moderate  blocklengths and moderate error probabilities for some channels and sources  as shown   by Polyanskiy-Poor-Verd\'u~\cite{PPV10} and Kostina-Verd\'u~\cite{kost12}. For error probabilities that are not too small (e.g., $\eps\in [10^{-6},10^{-3}]$), the Gaussian approximation is often better than that provided by traditional {\em error exponent}  or {\em reliability function} analysis~\cite{Csi97, gallagerIT}, where the   code rate is   fixed (below the first-order fundamental limit) and the exponential decay of the error probability is analyzed. Recent refinements to error exponent analysis using exact asymptotics~\cite{altug_refinement1,altug_refinement2, Sca13}  or saddlepoint approximations~\cite{Scarlett14a} are alternative proxies to the non-asymptotic fundamental limits.  The accuracy of the Gaussian approximation in {\em practical} regimes of errors and  finite blocklengths gives us motivation to  study refinements to the first-order fundamental limits of other single- and multi-user problems in Shannon theory. 

The study of {\em asymptotic estimates   with non-vanishing error probabilities}---or  more succinctly, {\em fixed error asymptotics}---also uncovers several interesting phenomena that are not observable from  studies of first-order fundamental limits in single- and multi-user information theory~\cite{Cov06,elgamal}. This analysis may give engineers deeper insight into the design of practical communication systems. A non-exhaustive list includes:
\begin{enumerate}
%\item It is known that a  random Gaussian codebook is optimal in the sense of being able to achieve the capacity of an AWGN channel \cite[Ch.~7]{Cov06}. As we see in Section~\ref{sec:awgn}, this is    not the case, as it appears that one must draw codewords uniformly at random from a power sphere \cite{Shannon59} to achieve the second- and third-order asymptotics of the AWGN channel.
\item Shannon showed that {\em separating} the tasks of source    and channel coding  is optimal rate-wise \cite{Shannon48}. As we see in Section~\ref{sec:sep} (and similarly to the case of error exponents \cite{Csi80}), this is not the case when the probability of excess distortion of the source is allowed to be  non-vanishing. 
\item Shannon showed that feedback does not increase the capacity of a DMC \cite{Sha56}. It is known, however, that variable-length feedback~\cite{PPV11b} and full output feedback \cite{AW14} improve on the fixed error asymptotics of   DMCs.% (similarly to the case of error exponents~\cite{Har77}).

\item It is known that the entropy can be achieved {\em universally}  for fixed-to-variable length  almost  lossless source coding of a DMS \cite{LZ78}, i.e., the source statistics do not have to be known.  The redundancy has also been studied  for prefix-free codes~\cite{Clarke}. In the fixed error setting (a setting complementary to \cite{Clarke}), it was shown by Kosut and Sankar~\cite{KosutSankar,KosutSankar14} that universality imposes a penalty in %neither the first-order (entropy) nor the second-order term (varentropy), but
 the {\em third-order} term of the asymptotic expansion in \eqref{eqn:int_src_expand}.
\item Han showed that the output from any source encoder at the optimal coding rate with asymptotically vanishing  error appears almost completely random  \cite{Han_folklore}. This is the so-called {\em folklore theorem}. Hayashi \cite{Hayashi08} showed that the analogue of the folklore theorem does not hold when we consider the second-order terms in asymptotic expansions (i.e., the  second-order asymptotics).
\item Slepian and Wolf showed that separate encoding of two correlated sources incurs no loss rate-wise compared to the situation where side information is also available at all encoders~\cite{sw73}. As we shall see in Chapter \ref{ch:sw}, the   fixed error asymptotics in the vicinity of a corner point of the polygonal Slepian-Wolf region suggests that side-information at the encoders may be beneficial.
\end{enumerate}
None of the  aforementioned books \cite{Cov06,Csi97, elgamal,gallagerIT, Han10, Yeung}  focus exclusively on the situation where the error probabilities of various Shannon-theoretic problems are upper bounded by $\eps\in (0,1)$ and asymptotic expansions or second-order terms  are sought. This is what this monograph attempts to do.

\section{Preview of this Monograph  }
%This paper is written for students and more experienced researchers in information theory who want a different flavor on traditional Shannon-theoretic problems. The prerequirsites are minimal: The reader is s
This monograph   is organized as follows: In the remaining parts of this chapter, we   recap some  quantities in information theory and   results in the {\em method of types}~\cite{Csis00,Csi97, Har08}, a particularly useful tool for the study of discrete memoryless systems. We also mention some probability bounds that will be used throughout the monograph. Most of these bounds are based on refinements  of the   central limit theorem, and are collectively known as {\em Berry-Esseen theorems}~\cite{Berry41, Esseen42}. In Chapter~\ref{ch:ht}, our study of asymptotic expansions of the form \eqref{eqn:int_src_expand} and \eqref{eqn:int_ch_expand} begins in earnest by revisiting Strassen's work \cite{Strassen} on    binary hypothesis testing where the probability of false alarm is constrained to not exceed a positive constant.  We find it useful to revisit the fundamentals of hypothesis testing as many information-theoretic problems such as   source  and channel coding are intimately related to hypothesis testing.

Part II of this monograph begins our study of information-theoretic   problems starting with lossless and lossy compression in Chapter~\ref{ch:src}. We emphasize, in the first part of this chapter, that (fixed-to-fixed length) lossless source coding and binary hypothesis testing are, in fact, the same problem, and so the asymptotic expansions developed in Chapter~\ref{ch:ht} may be directly employed for the purpose of lossless source coding. Lossy source coding, however, is more involved. We review the recent works in~\cite{ingber11} and~\cite{kost12}, where the authors  independently derived asymptotic expansions for the logarithm of the minimum size of a  source code that reproduces symbols up to a certain distortion, with some admissible probability of excess distortion.  Channel coding is discussed in Chapter~\ref{ch:cc}. In particular, we  study the approximation in \eqref{eqn:int_ch_expand} for both discrete memoryless and Gaussian channels. We make it a point here to be precise about the third-order $O(\log n)$ term. We  state conditions on the channel  under which the coefficient of the $O(\log n)$ term can be determined exactly. This leads to some new insights concerning optimum codes for the channel coding problem. Finally, we marry source and channel coding in the study of source-channel transmission where the probability of excess distortion in reproducing the source is non-vanishing.

\begin{center}
\begin{figure}
\centering
\begin{tabular}{cc}
 \parbox[c]{.65\columnwidth}{ 
\includegraphics[width=.65\columnwidth]{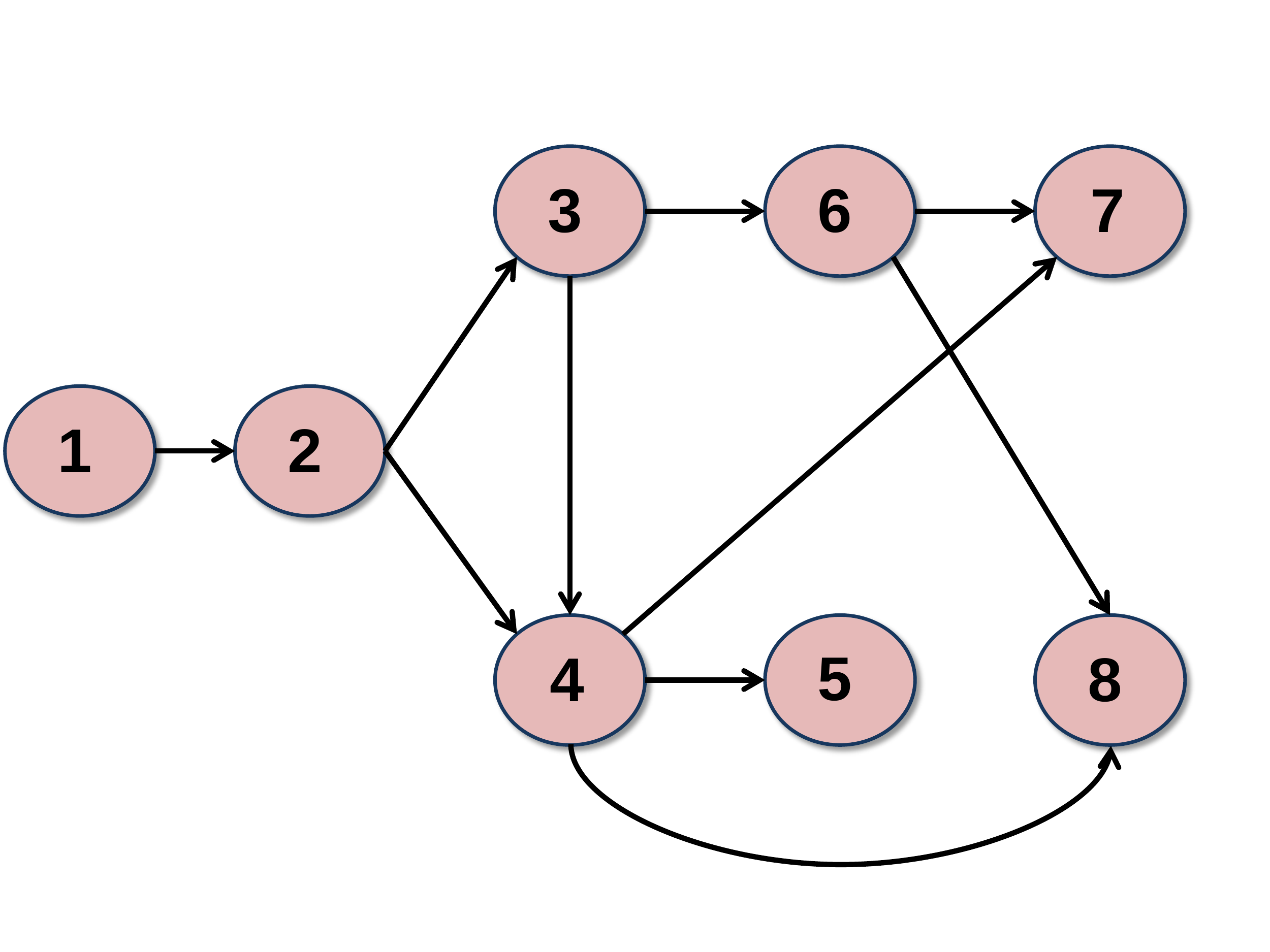} }   &    
 \parbox[c]{.35\columnwidth}{ {\footnotesize
1.\ Introduction \\
2.\ Hypothesis Testing\\
3.\ Source Coding\\
4.\  Channel Coding\\
5.\  Channels with State\\
6.\  Slepian-Wolf \\
7.\  Gaussian IC  \\
8.\  Gaussian A-MAC 
} }
\end{tabular}
\caption{Dependence graph of the chapters in this monograph. An arrow from node $s$ to $t$ means that results and techniques in Chapter $s$ are required to understand the material in  Chapter~$t$.}
\label{fig:depend}
\end{figure}
\end{center} 

Part III of this monograph contains a sparse sampling of fixed error asymptotic results in network information theory. The problems we discuss here have conclusive second-order asymptotic characterizations (analogous to the second terms in the asymptotic expansions in~\eqref{eqn:int_src_expand} and~\eqref{eqn:int_ch_expand}). They  include some channels with random state (Chapter~\ref{ch:state}), such as Costa's writing on dirty paper \cite{costa}, mixed DMCs~\cite[Sec.~3.3]{Han10}, and quasi-static single-input-multiple-output (SIMO) fading channels~\cite{Biglieri}. Under the  fixed error setup, we also consider the second-order asymptotics of the Slepian-Wolf~\cite{sw73} distributed lossless source coding problem (Chapter~\ref{ch:sw}), the Gaussian interference channel (IC) in the strictly very strong interference regime~\cite{Carleial75} (Chapter~\ref{ch:ic}), and the Gaussian multiple access channel  (MAC) with degraded message sets (Chapter~\ref{ch:mac}). The MAC with   degraded message sets   is also known as the  {\em cognitive}~\cite{dev06}   or  {\em asymmetric} \cite{Har75,vdM85,Pre84} MAC (A-MAC). Chapter~\ref{ch:con} concludes with  a  brief summary of other  results, together with open problems in this area of research. A dependence graph of the chapters in the monograph is shown in Fig.~\ref{fig:depend}.

This  area of information theory---{\em fixed error asymptotics}---is vast and, at the same time,   rapidly expanding. The results   described   herein are not  meant to be exhaustive and were somewhat dependent on the author's understanding of the subject and   his  preferences at the time of writing. However, the author has made it a point  to ensure that results herein are  {\em conclusive} in nature.   This means that the problem is {\em solved} in the information-theoretic sense in that an operational quantity is {\em equated} to an information quantity. In terms of asymptotic expansions such as \eqref{eqn:int_src_expand} and~\eqref{eqn:int_ch_expand}, by {\em solved}, we mean that either the second-order term is known  or, better still, both the second- and third-order terms are known. Having articulated this, the author confesses that there are   many relevant information-theoretic problems that can be considered  solved in the fixed error setting, but have not found their way into this monograph either due to space constraints or because  it was difficult to meld them seamlessly with the rest of the story.  %The author  humbly apologizes in advance to his esteemed colleagues for any  egregious  omissions. 
\section{Fundamentals of Information Theory} \label{sec:funda}
In this section, we   review some  basic information-theoretic quantities. 
As with every article published in the {\em Foundations and Trends in Communications and Information Theory}, the reader is expected  to have some background in information theory. Nevertheless, the only  prerequisite  required to appreciate this monograph is information theory at the level of Cover and Thomas \cite{Cov06}. We will also make extensive use of the method  of types, for which excellent expositions can be found in~\cite{Csis00, Csi97, Har08}.  The measure-theoretic foundations of probability will    not be needed  to keep the  exposition accessible to as  wide  an  audience as possible.% and because we only work with discrete alphabets and Gaussians distributions. 

\subsection{Notation} \label{sec:nota}
The notation we use is  reasonably standard and  generally follows the books by Csisz\'ar-K\"orner~\cite{Csi97} and Han~\cite{Han10}. Random variables (e.g., $X$) and their realizations (e.g., $x$) are in upper and lower case respectively. Random variables that take on finitely many values have alphabets (support) that are denoted by calligraphic font  (e.g., $\calX$).  The cardinality of the finite set $\calX$  is denoted as $| \calX |$. Let the random vector $X^n$ be the vector of random variables $(X_1, \ldots ,X_n)$. We use bold face $\bx=(x_1, \ldots, x_n)$ to denote a realization of $X^n$.  The set of all distributions (probability mass  functions) supported on alphabet $\calX$ is denoted as $\scP(\calX)$. The set of all conditional distributions (i.e., channels) with the input alphabet $\calX$ and the output alphabet $\calY$ is denoted by $\scP(\calY|\calX )$.  The joint distribution induced by a marginal distribution $P  \in \scP(\calX)$ and a channel $V  \in  \scP(\calY|\calX )$ is denoted   as $P \times V$, i.e., 
\begin{equation}
(P\times V)( x,y) := P(x)V(y|x).
\end{equation}
The marginal output distribution induced by $P$ and $V$ is denoted as   $PV$, i.e., 
\begin{equation}
PV(y): = \sum_{x\in\calX} P(x) V(y|x).
\end{equation}
If $X$ has distribution $P$, we sometimes  write this as $X\sim P$. 

Vectors are indicated in lower case bold face (e.g., $\ba$) and matrices in upper case bold face (e.g., $\bA$). If we write $\ba\ge\mathbf{b}$ for two vectors $\ba$ and $\mathbf{b}$ of the same length, we mean that $a_j \ge b_j$ for every coordinate $j$. The transpose of $\bA$ is denoted as $\bA'$.  The vector of all zeros and the identity matrix are  denoted as $\bzero$ and $\bI$ respectively. We sometimes make the lengths and sizes explicit. The $\ell_q$-norm (for $q\ge 1$) of a vector $\bv=(v_1,\ldots, v_k)$ is denoted as $\|\bv\|_q :=( \sum_{i=1}^k |v_i|^q)^{1/q}$.

We use standard asymptotic notation \cite{Cor03}: $a_n\in O(b_n)$ if and only if  (iff) $\limsup_{n\to\infty} \big|a_n/b_n\big|<\infty$; $a_n \in \Omega(b_n)$ iff $b_n \in O(a_n)$; $a_n \in \Theta(b_n)$ iff $a_n \in O(b_n)\cap\Omega(b_n)$; $a_n \in o(b_n)$ iff $\limsup_{n\to\infty} \big|a_n/b_n\big|=0$; and  $a_n \in \omega(b_n)$ iff $\liminf_{n\to\infty} \big|a_n/b_n\big|=\infty$.  Finally, $a_n\sim b_n$ iff $\lim_{n\to\infty}  a_n/b_n =1$.
\subsection{Information-Theoretic Quantities}
Information-theoretic quantities are denoted in the usual way \cite{Csi97,elgamal}. All logarithms and exponential functions are to the base $2$. The {\em entropy} of a discrete random variable $X$ with probability  distribution $P\in\scP(\calX)$ is denoted as 
\begin{equation}
H(X) = H(P) := -\sum_{x\in\calX} P(x)\log P(x).
\end{equation}
For the sake of clarity, we will sometimes make the dependence on the distribution  $P$ explicit. Similarly given a pair of random variables $(X,Y)$ with joint distribution $P\times V \in \scP(\calX\times \calY)$, the {\em conditional entropy}  of $Y$ given $X$ is written as 
\begin{equation}
H(Y|X) = H(V|P) := -\sum_{x\in\calX} P(x) \sum_{y\in\calY} V(y|x) \log V(y|x).
\end{equation}
The {\em joint entropy} is denoted as 
\begin{align}
H(X,Y)  &:= H(X)+H(Y|X) ,\quad\mbox{or} \\*
H(P\times V)  &:= H(P)+H(V|P).
\end{align}
The {\em mutual information} is  a measure of the correlation or dependence between random variables $X$ and $Y$. It is interchangeably denoted as 
\begin{align}
I(X;Y)  &:= H(Y)-H(Y|X)  ,\quad\mbox{or} \\*
I(P,V)  &:= H(PV)-H(V|P)  .
\end{align}
Given three random variables $(X,Y,Z)$ with joint distribution $P\times V\times W$ where $V\in\scP(\calY|\calX)$ and $W\in\scP(\calZ|\calX\times \calY)$, the {\em conditional mutual information} is 
\begin{align}
I(Y;Z|X)  &:=  H(Z|X)-H(Z|XY) ,\quad\mbox{or} \\*
I(V,W|P) &:= \sum_{x\in\calX} P(x) I\big(V(\cdot|x)  , W(\cdot |x,\cdot)\big).
\end{align}

A particularly important quantity is the {\em relative entropy} (or {\em Kullback-Leibler divergence}~\cite{Kul51}) between $P$ and $Q$ which are distributions on the same finite  support set $\calX$. It is defined as the expectation with respect to $P$ of the log-likelihood ratio $\log\frac{P(x)}{Q(x)}$, i.e., 
\begin{equation}
D(P \| Q) := \sum_{x\in\calX} P(x)\log\frac{P(x)}{Q(x)}.
\end{equation}
Note that if  there exists an $x\in\calX$ for which $Q(x) = 0$ while $P(x)>0$, then the relative entropy $D(P \| Q) = \infty$. If for every $x\in\calX$, if $Q(x)=0$ then $P(x)=0$, we say that $P$ is absolutely continuous with respect to $Q$ and denote this relation by $P\ll Q$. In this case, the relative entropy is finite. It is well known that $D(P \| Q) \ge 0$ and equality holds if and only if $P = Q$.  Additionally, the {\em conditional relative entropy} between $V,W\in\scP(\calY|\calX)$ given $P\in\scP(\calX)$ is defined as
\begin{equation}
D(V\|W|P) := \sum_{x\in\calX}P(x) D\big(V(\cdot|x)\|W(\cdot|x)\big).
\end{equation}

The mutual information is a special case of the relative entropy. In particular, we have 
\begin{equation}
I(P,V)=D(P\times V \| P\times PV)= D(V\|PV|P) .
\end{equation}
Furthermore, if $U_\calX$ is the uniform distribution on $\calX$, i.e.,  $U_\calX(x) = 1/|\calX|$ for all $x\in\calX$,  we have 
\begin{equation}
D(P\| U_\calX) = -H(P)+ \log|\calX|. \label{eqn:div_entropy}
\end{equation}

The definition of relative entropy $D(P\| Q)$ can   be extended to the case where $Q$ is not necessarily a probability measure. In this case non-negativity does not hold in general. An important property we exploit is the  following: If $\mu$ denotes the counting measure  (i.e., $\mu(\calA)=|\calA|$ for $\calA\subset\calX$), then similarly to \eqref{eqn:div_entropy}
\begin{equation}
D(P\|\mu)=-H(P). \label{eqn:div_ent}
\end{equation}

\section{The Method of Types}
For finite alphabets, a particularly convenient tool in information theory is the {\em method of types} \cite{Csis00,Csi97,  Har08}. For a sequence $\bx = (x_1,\ldots, x_n)\in\calX^n$ in which $|\calX|$ is finite, its {\em type} or {\em empirical distribution} is the probability mass  function
\begin{equation}
P_{\bx}(x) = \frac{1}{n}\sum_{i=1}^n \bbI\{ x_i = x\} ,\qquad\forall\, x\in\calX.
\end{equation}
Throughout, we use the notation $\bbI\{\mathrm{clause}\}$ to mean the {\em indicator function}, i.e., this function equals $1$ if ``$\mathrm{clause}$'' is true and $0$ otherwise. The set of types formed from $n$-length sequences in $\calX$  is denoted as $\scP_n(\calX)$. This is clearly a subset of $\scP(\calX)$. The {\em type class} of $P$, denoted as $\calT_P$, is the set of all sequences  of length $n$  for which their type is $P$, i.e., 
\begin{equation}
\calT_P :=\left\{ \bx \in\calX^n: P_{\bx} = P \right\}.
\end{equation}
It is customary to indicate the dependence of $\calT_P$ on the {\em blocklength} $n$ but we suppress this dependence for the sake of conciseness throughout.  For a sequence $\bx\in\calT_P$, the set of all sequences $\by\in\calY^n$ such that $(\bx,\by)$ has joint type $P\times V$ is the {\em $V$-shell}, denoted as $\calT_V(\bx)$. In other words,
\begin{equation}
\calT_V(\bx):=\left\{\by\in\calY^n: P_{\bx,\by}=P\times V\right\}.
\end{equation}
The conditional distribution $V$ is also known as the {\em conditional type} of $\by$ given $\bx$. Let $\scV_n(\calY;P)$ be the set of all   $V\in\scP(\calY|\calX)$ for which the $V$-shell of a sequence of type $P$ is non-empty. 

We will often times find it useful to consider information-theoretic quantities of empirical distributions. All such quantities are denoted using hats. So for example, the {\em empirical entropy} of a sequence $\bx\in\calX^n$ is denoted as 
\begin{equation}
\hatH(\bx) := H(P_{\bx}) .
\end{equation}
The {\em empirical conditional entropy} of $\by \in\calY^n$ given  $\bx\in\calX^n$ where $\by\in\calT_V(\bx)$ is denoted as 
\begin{equation}
 \hatH(\by|\bx) := H( V|P_{\bx}).
 \end{equation} 
The {\em empirical mutual information} of a pair of sequences $(\bx,\by) \in\calX^n\times\calY^n$ with joint type $P_{\bx,\by}=P_{\bx} \times V$ is denoted as 
\begin{equation}
\hatI(\bx\wedge\by) := I(P_{\bx}, V). 
\end{equation}
 
The following lemmas  form the basis of   the method of types. The proofs can be found in~\cite{Csis00,Csi97}.

\begin{lemma}[Type Counting] \label{lem:type_count}
The sets $\scP_n(\calX)$ and $\scV_n(\calY;P)$ for $P\in\scP_n(\calX)$ satisfy 
\begin{align}
|\scP_n(\calX)|  \le (n+1)^{|\calX|} ,\quad\mbox{and}    \quad |\scV_n(\calY;P)|   \le (n+1)^{|\calX| |\calY|}. \label{eqn:type_counting1}
\end{align}
\end{lemma}
In fact, it is easy to check that $|\scP_n(\calX) | = \binom{ n+|\calX|-1}{ |\calX|-1}$ but  \eqref{eqn:type_counting1} or its slightly stronger version 
\begin{equation}
|\scP_n(\calX)|   \le (n+1)^{|\calX|-1} \label{eqn:type_coutn2}
\end{equation}
usually suffices for our purposes in this monograph. This key property says that the number of types is  polynomial in the blocklength $n$. 

\begin{lemma}[Size of Type Class] \label{lem:size_type_class}
For a type $P\in\scP_n(\calX)$, the type class  $\calT_P \subset\calX^n$ satisfies
\begin{align}
|\scP_n(\calX)|^{-1} \exp\big(n H(P) \big) \le|\calT_P|   \le \exp\big(n H(P) \big) .
\end{align} 
For a conditional type $V\in\scV_n(\calY;P)$ and a sequence $\bx\in\calT_P$, the $V$-shell $\calT_V(\bx)\subset\calY^n$   satisfies
\begin{align}
|\scV_n(\calY;P)|^{-1} \exp\big(n H(V|P) \big) \le|\calT_V(\bx)|  \le \exp\big(n H(V|P) \big)   .
\end{align}
\end{lemma}

This lemma says that,  on the exponential scale, 
\begin{equation}
|\calT_P| \cong\exp\big(nH(P) \big),\quad\mbox{and}\quad |\calT_V(\bx)| \cong   \exp\big(nH(V|P) \big),
\end{equation}
where we used the notation $a_n\cong  b_n$ to mean equality up to a polynomial, i.e., there exists polynomials $p_n$ and $q_n$ such that $a_n/ p_n \le b_n \le q_n  a_n$. We now consider probabilities of sequences. Throughout, for a   distribution $Q\in\scP(\calX)$, we let $Q^n(\bx)$ be the product distribution, i.e., 
\begin{equation}
Q^n(\bx) = \prod_{i=1}^n Q(x_i),\qquad\forall\,\bx\in\calX^n.
\end{equation}

\begin{lemma}[Probability of Sequences] \label{lem:prob_seq}
If $\bx \in \calT_P$ and $\by\in\calT_V(\bx)$,
\begin{align}
Q^n(\bx) & = \exp\big( -n D( P\| Q) - nH(P) \big) \quad\mbox{and} \\
W^n(\by|\bx) & = \exp\big( -n D( V\| W| P) - nH(V|P) \big) .
\end{align}
\end{lemma}
This, together with Lemma~\ref{lem:size_type_class}, leads immediately to the final   lemma in this section.

\begin{lemma}[Probability of Type Classes]
For a type $P\in\scP_n(\calX)$,  
\begin{equation}
|\scP_n(\calX)|^{-1}\exp\big( -nD(P\| Q) \big) \le Q^n(\calT_P)\le \exp\big( -nD(P\| Q) \big).
\end{equation}
For a conditional type $V\in\scV_n(\calY;P)$ and a sequence $\bx\in\calT_P$, we have 
\begin{align}
|\scV_n(\calY;P)|^{-1}\exp\big( -nD( V\|W|P) \big) & \le W^n(\calT_V(\bx) |\bx) \nn\\
& \le \exp\big( -nD( V\|W|P) \big). \label{eqn:prob_v_shell}
\end{align}
\end{lemma}
The interpretation of this lemma is that   the probability that a random \iid (independently and identically distributed) sequence $X^n$ generated from $Q^n$ belongs to the type class $\calT_P$ is exponentially small with exponent $D( P \| Q)$, i.e., 
\begin{equation}
Q^n(\calT_P)\cong  \exp\big( -nD(P\| Q) \big).
\end{equation}
The bounds in \eqref{eqn:prob_v_shell} can be interpreted similarly. 
\section{Probability Bounds} \label{sec:prob}
In this section, we summarize some bounds on probabilities that we use extensively in the sequel.  For a random variable $X$, we let $\bbE[X]$ and $\var(X)$ be its expectation and variance respectively. To emphasize that the expectation is taken with respect to a random variable $X$ with distribution  $P$, we sometimes make this explicit by using a subscript, i.e., $\bbE_X$ or $\bbE_P$. 

\subsection{Basic Bounds}
We start with the familiar Markov and Chebyshev inequalities. 
\begin{proposition}[Markov's inequality]
Let $X$ be a  real-valued  non-negative random variable. Then for any $a>0$, we have 
\begin{equation}
\Pr( X\ge a)\le\frac{\bbE[X]}{a}. 
\end{equation}
\end{proposition}
If we let $X$ above be the non-negative  random variable $(X-\bbE[X])^2$, we obtain Chebyshev's inequality.
\begin{proposition}[Chebyshev's inequality]
Let $X$ be a  real-valued  random variable with mean $\mu$ and variance $\sigma^2$. Then for any $b>0$, we have 
\begin{equation}
\Pr\big( |X-\mu|\ge b\sigma\big) \le \frac{1}{b^2}.
\end{equation}
\end{proposition}
We now consider a collection of real-valued  random variables that are i.i.d. In particular, let $X^n=(X_1,\ldots, X_n)$ be a collection of independent random variables where each $X_i$ has distribution $P$ with zero mean and finite variance $\sigma^2$.  
\begin{proposition}[Weak Law of Large Numbers] \label{prop:wlln}
For every $\epsilon>0$, we have
\begin{equation}
\lim_{n\to\infty}\Pr\left( \bigg| \frac{1}{n}\sum_{i=1}^n  X_i\bigg|   > \epsilon\right)= 0.  \label{eqn:lln}
\end{equation}
Consequently,  the average $\frac{1}{n}\sum_{i=1}^n  X_i$ converges  to $0$ in probability. 
\end{proposition}
This follows by applying Chebyshev's inequality to the random variable $\frac{1}{n}\sum_{i=1}^n  X_i$.  In fact, under mild conditions, the convergence to zero in \eqref{eqn:lln} occurs exponentially fast. See, for example, Cramer's theorem in \cite[Thm.~2.2.3]{Dembo}.
 
\subsection{Central Limit-Type Bounds}
In preparation for the next result, we denote the {\em probability density function} (pdf) of a  univariate Gaussian  as 
\begin{equation}
\calN(x; \mu,\sigma^2) = \frac{1}{\sqrt{2\pi\sigma^2}}\rme^{-(x-\mu)^2 / (2\sigma^2 )}.
\end{equation}
We will also  denote this as $\calN(\mu,\sigma^2)$ if the argument $x$ is unnecessary.
A {\em standard Gaussian distribution} is one in which  the mean $\mu=0$ and the standard deviation $\sigma=1$. In the multivariate case, the pdf is
\begin{equation}
\calN(\bx; \bmu,\bSigma) =\frac{1}{ \sqrt{ (2\pi )^k |\bSigma| } }\rme^{- \frac{1}{2}(\bx-\bmu)'\bSigma^{-1} (\bx-\bmu)} \label{eqn:pdf_multi_gauss}
\end{equation}
where $\bx\in\bbR^k$. 
A {\em standard multivariate Gaussian distribution} is  one in which the mean is $\bzero_{k}$ and the covariance   is the identity matrix $\bI_{k\times k}$. %$\calN(\bx; \bzero_k,\bI_{k\times k})$, where $\bzero_k$ is the length-$k$ vector of  all  zeros and $\bI_{k\times k}$ is the identity matrix of size $k$.

For the univariate case, the {\em cumulative distribution function} (cdf) of the standard  Gaussian is denoted as 
\begin{equation}
\Phi(y) := \int_{-\infty}^y \calN(x;0,1) \,  \rmd x. 
\end{equation}
We also find it convenient to introduce the inverse of $\Phi$ as 
\begin{equation}
\Phi^{-1}(\eps):=\sup \big\{ y\in\bbR : \Phi(y)\le\eps \big\}
\end{equation}
which evaluates to the usual inverse for $\eps\in (0,1)$ and extends continuously to take values $\pm\infty$ for $\eps$ outside $(0,1)$. These  monotonically increasing functions are shown in Fig.~\ref{fig:phis}.

\begin{figure}
\centering
\includegraphics[width = 1\columnwidth]{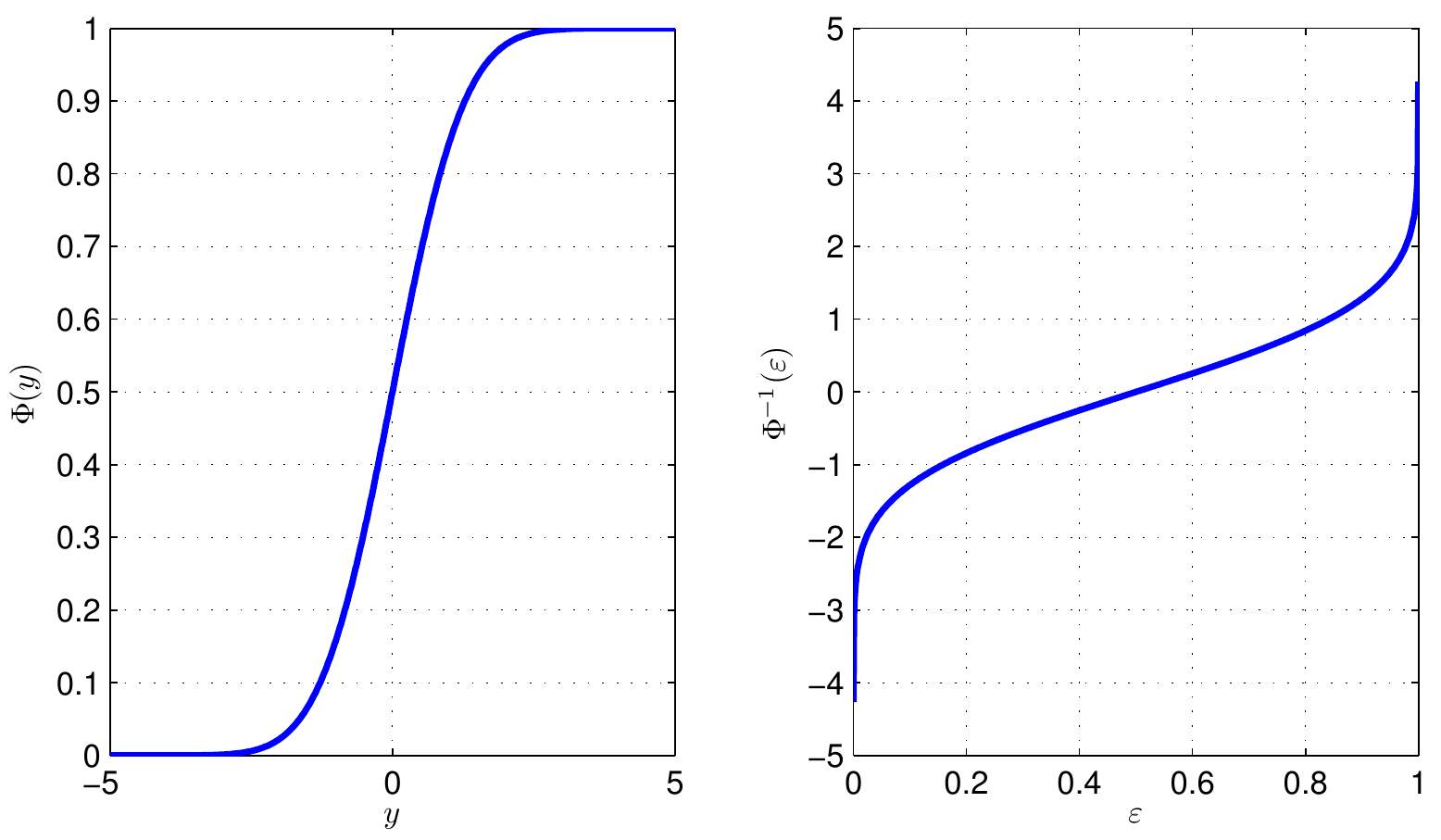}
\caption{Plots of $\Phi(y)$ and $\Phi^{-1}(\eps)$}
\label{fig:phis}
\end{figure}

If the scaling in front of the sum in the statement of the law of large numbers in~\eqref{eqn:lln} is $\frac{1}{\sqrt{n}}$ instead of $\frac{1}{n}$, the resultant random variable $\frac{1}{\sqrt{n}}\sum_{i=1}^n  X_i$ converges in distribution to a Gaussian random variable.  As in Proposition~\ref{prop:wlln},  let $X^n$ be a collection of \iid random variables where each $X_i$ has  zero mean and finite variance $\sigma^2$.

\begin{proposition}[Central Limit Theorem]
For any $a \in \bbR$, we have 
\begin{equation}
\lim_{n\to\infty} \Pr\left(  \frac{1}{\sigma\sqrt{n }}\sum_{i=1}^n  X_i  < a \right)  = \Phi (a). \label{eqn:clt}
\end{equation}
In other words, 
\begin{equation}
\frac{1}{\sigma\sqrt{n}}\sum_{i=1}^n  X_i \stackrel{\mathrm{d}}{\longrightarrow} Z
\end{equation}
where $\stackrel{\mathrm{d}}{\longrightarrow}$ means convergence in distribution and $Z$ is the standard Gaussian random variable. 
\end{proposition}

Throughout the monograph, in the evaluation of the non-asymptotic bounds, we will use a more quantitative version of the central limit theorem known as the Berry-Esseen theorem~\cite{Berry41, Esseen42}. See Feller  \cite[Sec.~XVI.5]{feller} for a proof.

\begin{theorem}[Berry-Esseen Theorem (\iid Version)] \label{thm:berry_iid}
Assume that the third absolute moment  is finite, i.e., $T:=\bbE\big[ |X_1|^3\big] <\infty$. For every $n\in\bbN$, we have
\begin{equation}
\sup_{ a\in\bbR} \left|\Pr\left(  \frac{1}{\sigma\sqrt{n}}\sum_{i=1}^n  X_i  < a \right) - \Phi (a)\right|\le \frac{  T}{ \sigma^3\sqrt{n}}.
\end{equation}
%where $c>0$ is a universal constant. In fact, $c<0.5$. 
\end{theorem}
Remarkably, the Berry-Esseen theorem says that the convergence in    the central limit theorem in \eqref{eqn:clt} is uniform in $a\in\bbR$.  Furthermore, the convergence of the distribution function of $\frac{1}{\sqrt{n}}\sum_{i=1}^n  X_i$ to the Gaussian cdf occurs at a rate of $O(\frac{1}{\sqrt{n}})$.  The constant of proportionality in the $O(\cdot)$-notation depends {\em only} on the variance and the third absolute moment and not on any other statistics of the random variables.

There are many generalizations of the Berry-Esseen theorem. One   which we will   need is the relaxation of the assumption that the random variables are identically distributed. Let $X^n = (X_1,\ldots, X_n)$ be a collection of independent random variables where each random variable has zero mean,  variance $\sigma_i^2 := \bbE[X_i^2] >0$ and third absolute moment $T_i := \bbE\big[ |X_i|^3\big] < \infty$. We respectively define  the average variance and average third absolute moment  as
\begin{equation}
  \sigma^2 := \frac{1}{n}\sum_{i=1}^n \sigma_i^2,\quad\mbox{and}\quad T:=\frac{1}{n}\sum_{i=1}^n T_i.
\end{equation}

\begin{theorem}[Berry-Esseen Theorem (General Version)] \label{thm:berry_gen}
For every $n\in\bbN$, we have
\begin{equation}
\sup_{ a\in\bbR} \left|\Pr\left(   \frac{1}{\sigma\sqrt{n}}\sum_{i=1}^n  X_i  < a \right) - \Phi (a)\right|\le \frac{6\, T}{\sigma^3\sqrt{n} }.
\end{equation}
%where $c'>0$ is a universal constant. In fact, $c'<6$. 
\end{theorem}

Observe that as with the \iid version of the   Berry-Esseen theorem, the remainder term   scales as $O(\frac{1}{\sqrt{n}})$.% because $T$ and $\sigma^2$ are both $\Theta(n)$ for independent random variables with non-zero variances and finite third absolute moments. 

The proof of the following theorem uses  the   Berry-Esseen theorem (among other techniques). This theorem is proved in Polyanskiy-Poor-Verd\'u \cite[Lem.~47]{PPV10}. Together with its variants, this theorem is useful for obtaining third-order asymptotics for binary hypothesis testing and other coding problems with  non-vanishing error probabilities. 

\begin{theorem}\label{thm:str_ld}
%Let $X^n = (X_1, \ldots, X_n)$ be a collection of independent random variables with non-zero aggregate variance $\sigma^2=\sum_{i=1}^n \var(X_i)$ and aggregate third absolute moment $T=\sum_{i=1}^n \bbE\big[ |X_i - \bbE[X_i]|^3\big]<\infty$. 
Assume the same setup as in Theorem~\ref{thm:berry_gen}. For any $\gamma \ge 0$, we have
\begin{equation}
\bbE\left[ \exp\bigg( -\sum_{i=1}^n X_i \bigg)\bbI\bigg\{ \sum_{i=1}^n X_i>\gamma\bigg\}\right]\le 2\left( \frac{\log 2}{\sqrt{2\pi}} + \frac{12T}{\sigma^2}\right) \frac{\exp(-\gamma)}{\sigma\sqrt{n}}. \label{eqn:ppv_expect}
\end{equation}
\end{theorem}
It is trivial to see that the expectation in \eqref{eqn:ppv_expect} is upper bounded by $\exp(-\gamma)$. The additional factor of $(\sigma\sqrt{n})^{-1}$ is crucial in proving coding theorems with better third-order terms.  Readers familiar with strong large deviation theorems or exact asymptotics (see, e.g.,~\cite[Thms.~3.3 and 3.5]{CS93} or~\cite[Thm.~3.7.4]{Dembo}) will notice that \eqref{eqn:ppv_expect} is in the same spirit as the theorem by Bahadur and Ranga-Rao~\cite{Bah60}. There are two advantages of \eqref{eqn:ppv_expect} compared to strong large  deviation theorems. First,   the bound is purely in terms of $\sigma^2$ and $T$, and  second, one does not have to differentiate between lattice and non-lattice random variables. The disadvantage of \eqref{eqn:ppv_expect} is that the constant is worse but this will not concern us as we focus on asymptotic results in this monograph, hence constants do not affect the main results. 

For multi-terminal problems that we encounter in the latter parts of this monograph, we will require vector (or multidimensional) versions of the Berry-Esseen theorem.   The following is due to G\"{o}tze~\cite{Got91}.
\begin{theorem}[Vector Berry-Esseen Theorem I] \label{theorem:multidimensional-berry-esseen}
Let $X_1^k,\ldots,X_n^k$ be independent $\bbR^k$-valued random vectors  with zero mean. Let 
\begin{equation}
S_n^k = \frac{1}{\sqrt{n}}\sum_{i=1}^n X_i^k . \label{eqn:sum_rvs}
\end{equation}
Assume that $S_n^k$ has the following statistics
\begin{align}
\cov(S_n^k)   =\bbE\big[ S_n^k(S_n^k)'\big]= \bI_{k\times k}, \quad\mbox{and}  \quad \xi  := \frac{1}{n} \sum_{i=1}^n \bbE\big[ \| X_i^k \|^3_2\big]. \label{eqn:xi}
\end{align}
Let $Z^k$ be a  standard Gaussian random vector, i.e., its distribution is  $\calN(0^k,\bI_{k\times k})$.
Then, for all $n \in \bbN$, we have
\begin{equation}
\sup_{\mathscr{C} \in \mathfrak{C}_k} \left| \Pr\big( S_n^k \in \mathscr{C} \big) - \Pr\big( Z^k \in \mathscr{C} \big) \right|  \le \frac{c_k \, \xi}{ \sqrt{n}},
\end{equation}
where $\mathfrak{C}_k$ is the family of all convex subsets of $\bbR^k$, and where $c_k$ is a constant
that depends only on the dimension $k$.
\end{theorem}
Theorem \ref{theorem:multidimensional-berry-esseen} can be applied for random vectors that are independent but not necessarily
identically distributed.  The constant $c_k$ can be upper bounded by $400\, k^{1/4}$ if the random vectors are i.i.d., a result by Bentkus~\cite{Ben03}. However, its precise value will not be of concern to us in this monograph. Observe that the scalar versions of the Berry-Esseen theorems  (in Theorems~\ref{thm:berry_iid} and \ref{thm:berry_gen}) are special cases (apart from the constant) of the vector version in which the family of convex subsets is restricted to   the family of semi-infinite intervals $(-\infty,a)$.

We will frequently encounter random vectors with non-identity covariance matrices.  The following  modification of Theorem~\ref{theorem:multidimensional-berry-esseen} is  due to Watanabe-Kuzuoka-Tan~\cite[Cor.~29]{WKT13}.
\begin{corollary}[Vector Berry-Esseen Theorem II] \label{corollary:multidimensional-berry-esseen}
Assume the same setup as in Theorem~\ref{theorem:multidimensional-berry-esseen}, except that $\cov(S_n^k) = \bV$, a positive definite matrix.  Then, for all $n \in \bbN$, we have
\begin{equation}
\sup_{\mathscr{C} \in \mathfrak{C}_k} \left| \Pr\big( S_n^k \in \mathscr{C} \big) - \Pr\big( Z^k \in \mathscr{C} \big) \right|  \le \frac{c_k \, \xi}{ \lambda_{\min}(\bV)^{3/2}\sqrt{n}},
\end{equation}
where   $\lambda_{\min}(\bV) >0$ is the smallest eigenvalue of $\bV$.
\end{corollary} 

The  final probability bound   is a quantitative version  of the so-called {\em multivariate delta method}~\cite[Thm.~5.15]{wasserman}.   Numerous  similar statements of varying generalities have appeared in the statistics literature (e.g., \cite{chenshao, wasserman14}). The simple version we present was shown by MolavianJazi and Laneman~\cite{Mol13} who extended    ideas in Hoeffding and Robbins' paper \cite[Thm.~4]{HR48} to provide rates of convergence to Gaussianity under appropriate technical conditions.  This result  essentially says that  a differentiable function  of a  normalized sum  of independent random vectors also satisfies a Berry-Esseen-type result. 
\begin{theorem}[Berry-Esseen Theorem for Functions of \iid Random Vectors] \label{thm:func_clt}
Assume that $X_1^k,\ldots, X_n^k$ are   $\bbR^k$-valued, zero-mean, \iid random vectors with positive definite covariance $\cov(X_1^k)$ and finite third absolute moment $\xi:=\bbE[ \|X_1^k\|_2^3]$. 
%Let $S_n^k$ denote the normalized sum  in \eqref{eqn:sum_rvs}. 
Let $\mathbf{f} (\bx)$ be a vector-valued function from $\bbR^k$ to $\bbR^l$ that is also twice continuously differentiable in a neighborhood of $\bx=\bzero$.  Let $\bJ \in\bbR^{l\times k}$ be the Jacobian matrix of $\mathbf{f}(\bx)$ evaluated at $\bx=\bzero$, i.e.,  its elements are 
\begin{equation}
J_{ij} = \frac{\partial f_i (\bx) }{\partial x_j } \bigg|_{\bx=\bzero},
\end{equation}
where $ i = 1,\ldots, l$ and $j = 1,\ldots, k$. Then, for every $n\in\bbN$, we have
\begin{equation}
\sup_{\mathscr{C} \in \mathfrak{C}_l } \left| \Pr\Bigg( \mathbf{f} \bigg( \frac{1}{n} \sum_{i=1}^n X_i^k  \bigg)  \in \mathscr{C} \Bigg) - \Pr\big( Z^l\in \mathscr{C} \big) \right|  \le \frac{c  }{ \sqrt{n}} \label{eqn:func_be} 
\end{equation}
where 	$c>0$ is a finite constant, and $Z^l$ is a     Gaussian random vector in $\bbR^l$ with mean vector and covariance matrix respectively given as 
\begin{equation}
\bbE[Z^l]= \mathbf{f}(\bzero),\quad\mbox{and}\quad \cov(Z^l)=\frac{\bJ \cov(X^k_1) \bJ'}{n} .
\end{equation}
%\bbE[Z^l]= \mathbf{f}(\bzero)$ and  covariance matrix $\cov(Z^l)=\frac{1}{n} \bJ \cov(X^k_1) \bJ'$.
%% and $\bJ \in\bbR^{l\times k}$ is the Jacobian matrix of $\mathbf{f}(\bx)$ evaluated at $\bx=\bzero$, i.e.,  its elements are 
%%\begin{equation}
%%J_{ij} = \frac{\partial f_i (\bx) }{\partial x_j } \bigg|_{\bx=\bzero},
%%\end{equation}
%where $ i = 1,\ldots, l$ and $j = 1,\ldots, k$.
\end{theorem}
In particular, the inequality in \eqref{eqn:func_be} implies that 
\begin{equation}
\sqrt{n}\left(\mathbf{f} \bigg(\frac{1}{n}\sum_{i=1}^n X_i^k \bigg)  - \mathbf{f}(\bzero) \right)\stackrel{\mathrm{d}}{\longrightarrow}   \calN \left(  \bzero,\bJ \cov(X^k_1) \bJ' \right),
\end{equation}
which is   a canonical statement in the study of the multivariate delta method~\cite[Thm.~5.15]{wasserman}.   

Finally, we remark that Ingber-Wang-Kochman \cite{ingberWK} used a result similar to that of Theorem \ref{thm:func_clt} to derive second-order asymptotic results for various Shannon-theoretic problems. However,   they analyzed the behavior of  functions of {\em distributions} instead of functions of {\em random vectors}   as in Theorem \ref{thm:func_clt}. 

\chapter{Binary Hypothesis Testing} \label{ch:ht}
In this chapter, we review asymptotic expansions in  simple (non-composite) binary hypothesis testing when one of the two error probabilities is non-vanishing. We find this useful, as many coding theorems we encounter in subsequent chapters can be stated in terms of quantities   related to binary hypothesis testing. For example, as pointed out in Csisz\'ar and K\"orner~\cite[Ch.~1]{Csi97}, fixed-to-fixed length lossless source coding and binary hypothesis testing are intimately connected through the relation   between relative entropy and entropy in~\eqref{eqn:div_entropy}. Another example is in point-to-point channel coding, where a powerful non-asymptotic converse theorem \cite[Eq.~(4.29)]{Strassen}  \cite[Sec.~III-E]{PPV10} \cite[Prop.~6]{TomTan12} can be stated in terms of the so-called  {\em $\eps$-hypothesis testing divergence} and the {\em $\eps$-information spectrum divergence} (cf.~Proposition~\ref{prop:converse}). The properties of these two quantities, as well as the relation between them are discussed. Using various probabilistic limit theorems, we also evaluate these quantities in the asymptotic setting for product distributions.  A corollary of the results  presented is the familiar Chernoff-Stein lemma~\cite[Thm.~1.2]{Csi97}, which asserts that the exponent  with growing number of observations  of the type-II error for a non-vanishing type-I error in a binary hypothesis test of $P$ against $Q$  is  the relative entropy  $D(P \| Q)$. 

The material in this chapter is based  largely on the  seminal work by Strassen \cite[Thm.~3.1]{Strassen}. The exposition is based on the more recent works by  Polyanskiy-Poor-Verd\'u~\cite[App.~C]{PPV10}, Tomamichel-Tan~\cite[Sec.~III]{TomTan12} and Tomamichel-Hayashi~\cite[Lem.~12]{Tom12}.
\section{Non-Asymptotic Quantities  and Their Properties}
Consider the simple (non-composite) binary hypothesis test:
\begin{align}
\rvH_0  &: Z\sim P \nn\\*
\rvH_1  &: Z\sim  Q \label{eqn:bin_test}
\end{align}
where $P$ and $Q$ are two probability distributions  on the same space $\calZ$. We assume that the space $\calZ$ is finite to keep the subsequent exposition simple.  The notation in \eqref{eqn:bin_test} means that under  the null hypothesis $\rvH_0$, the random variable $Z$ is distributed as $P \in \scP(\calZ)$ while under the alternative hypothesis $\rvH_1$, it is distributed according to a different distribution $Q \in \scP(\calZ)$. We would like to study the optimal performance of a hypothesis test in terms of the distributions $P$ and $Q$. 

There are several ways to measure the performance of   a hypothesis test  which, in precise terms, is a mapping $\delta$ from the observation space $\calZ$ to $[0,1]$. If the observation $z$ is such that   $\delta(z) \approx 0$, this means the test favors the null hypothesis $\rvH_0$. Conversely,  $\delta(z)\approx 1$ means that the test favors the alternative hypothesis $\rvH_1$ (or alternatively, rejects the null hypothesis $\rvH_0$). If $\delta(z)\in\{0,1\}$, the test is called {\em deterministic}, otherwise it is called {\em randomized}. Traditionally, there are three quantities that are of interest for a given test $\delta$. The first is the {\em probability of false alarm}
\begin{equation}
\rmP_{\mathrm{FA}} := \sum_{z\in\calZ}\delta(z) P(z)  = \bbE_{P}\big[\delta(Z) \big].
\end{equation}
The second is the {\em probability of missed detection}
\begin{equation}
\rmP_{\mathrm{MD}} := \sum_{z\in\calZ}  \big(1-\delta(z)\big) Q(z)  = \bbE_{Q}\big[1- \delta(Z)  \big].
\end{equation}
The third is the {\em probability of detection}, which is one minus the probability of  missed  detection, i.e., 
\begin{equation}
\rmP_{\mathrm{D}} := \sum_{z\in\calZ}\delta(z) Q(z) = \bbE_{Q}\big[\delta(Z) \big].
\end{equation}
The probability of false alarm and miss detection are traditionally called the {\em type-I} and {\em type-II errors} respectively in the statistics literature.  The probability of detection and the probability of  false alarm are also called the {\em power} and the {\em significance level} respectively. The ``holy grail'' is, of course, to design a test such that $\rmP_{\mathrm{FA}}=0$ while $\rmP_{\mathrm{D}}=1$ but this is clearly impossible unless $P$ and $Q$ are mutually singular measures. 

Since misses are usually more costly than false alarms, let us fix a number $\eps\in (0,1)$ that represents a tolerable  probability of false alarm (type-I error). Then define the smallest   type-II error in the binary hypothesis test \eqref{eqn:bin_test} with type-I error not exceeding $\eps$, i.e., 
\begin{equation}
\beta_{1-\eps}(P,Q) := \inf_{\delta : \calZ \to [0,1]}\Big\{\bbE_{Q}\big[ 1-\delta(Z)  \big]:\bbE_{P}\big[  \delta(Z)  \big]\le  \eps  \Big\}. \label{eqn:def_beta}
\end{equation}
Observe that $\bbE_{P}\big[\delta(Z) \big] \le\eps$ constrains the probability of false alarm to be no greater than $\eps$. Thus, we are searching over all tests $\delta$ satisfying  $\bbE_{P}\big[\delta(Z) \big] \le\eps$  such that the probability of missed detection  is minimized. Intuitively, $\beta_{1-\eps}(P,Q)$   quantifies, in a non-asymptotic fashion, the performance of an optimal hypothesis test between $P$ and $Q$.

A related quantity is the   {\em $\eps$-hypothesis testing divergence}
\begin{equation}
D_{\rmh}^{\eps} (P \| Q):=-\log \frac{\beta_{1-\eps}(P,Q)}{1-\eps} .\label{eqn:beta_D}
\end{equation}
This is a measure of the {\em distinguishability} of $P$ from $Q$.   As can be seen from \eqref{eqn:beta_D}, $\beta_{1-\eps}(P , Q)$ and $D_\rmh^\eps (P \| Q)$ are simple functions of each other.  We prefer to  express the results in this monograph mostly in terms of  $D_{\rmh}^{\eps} (P \| Q)$  because it shares similar properties with  the usual relative entropy $D(P\| Q)$, as is evidenced from the following lemma.

\begin{lemma}[Properties of $D_{\rmh}^{\eps}$] \label{lem:Dh1}
The $\eps$-hypothesis testing divergence satisfies the {\em positive definiteness} condition \cite[Prop.~3.2]{Dup12}, i.e., 
\begin{equation}
D_{\rmh}^{\eps} (P \| Q)\ge 0.
\end{equation}
Equality holds if and only if   $P=Q$. In addition, it also satisfies the {\em  data processing inequality} \cite[Lem.~1]{WangRenner}, i.e., for any channel $W$, 
\begin{equation}
D_{\rmh}^{\eps} (PW \| QW) \le  D_{\rmh}^{\eps} (P \| Q).
\end{equation}
\end{lemma}
%Observe that  $D_{\rmh}^{\eps} (P \| Q)$ has properties that are also shared by the usual relative entropy $D(P \| Q)$.  Hence, we use similar notation.

\begin{figure}
\centering
\setlength{\unitlength}{.4mm}
\begin{picture}(240, 80)

%\put(20, 7){\line(0,1){6}}
%\put(80, 7){\line(0,1){6}}
%\put(120, 7){\line(0,1){6}}
%\put(180, 7){\line(0,1){6}}
\put(40, 10){\circle*{4}}
%\put(90, 10){\circle*{4}}
\put(170, 10){\circle*{4}}
%\put(210, 10){\circle*{4}}
\thicklines
\put(0, 10){\vector(1, 0){220}}

\qbezier(40, 10)(50, 130)(170, 10)
%\qbezier(140, 10)(190, 100)(210, 10)
%\put(0, 60){\line(1,0){20}}

%\put(090, 15){\vector(1,0){50}}
%\put(140, 15){\vector(-1,0){50}}

%\put(110, 18){\footnotesize $C$}
%\put(234, 10){\footnotesize   $\bbR$}

%
%\put(180, 15){\vector(1,0){30}}
%\put(30, 0){\line(1, 0){30}}
%\put(30, 0){\line(0,1){30}}
%\put(60, 0){\line(0,1){30}}
%\put(30, 30){\line(1,0){30}}
%
%\put(90, 0){\line(1, 0){30}}
%\put(90, 0){\line(0,1){30}}
%\put(120, 0){\line(0,1){30}}
%\put(90, 30){\line(1,0){30}}
% 

\put(127, 50){\mbox{``Density'' of $\log\frac{P(Z)}{Q(Z)}$ when $Z\!\sim\! P$}}

\thinlines
\put(60, 10){\line(0,1){55}}

\put(60, 00){\mbox{$R^*$}}

\put(50, 35){\mbox{$\eps$}}  
\put(90, 35){\mbox{$1-\eps$}}
  \end{picture}
 \caption{Illustration of the $\eps$-information spectrum divergence  $D_{\rms}^{\eps} (P \| Q)$ which is the largest point $R^*$ for which the probability mass to the left is no larger than $\eps$.}
 \label{fig:info_spec}
\end{figure}
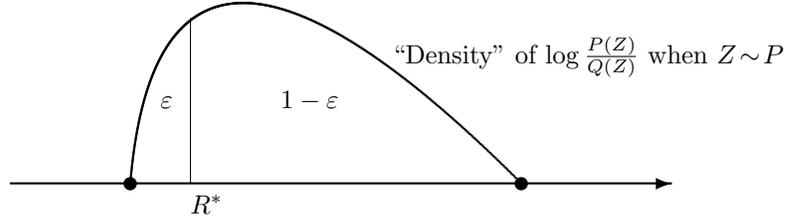

While  the $\eps$-hypothesis testing divergence occurs naturally and frequently in coding problems, it is  usually hard to analyze directly. Thus, we now introduce   an equally important  quantity.  Define the {\em $\eps$-information spectrum divergence} $D_{\rms}^{\eps} (P \| Q)$ as 
\begin{equation}
D_{\rms}^{\eps} (P \| Q) :=\sup\bigg\{ R\in\bbR :  P\bigg(\Big\{ z\in\calZ:\log\frac{P(z)}{Q(z)} \le R \Big\}\bigg)\le\eps\bigg\}. \label{eqn:def_Ds}
\end{equation}
Just as in information spectrum analysis~\cite{Han10}, this quantity places the {\em distribution} of the log-likelihood ratio $\log\frac{P(Z)}{Q(Z)}$ (where $Z\sim P$), and  not just its expectation, in the most prominent role. See Fig.~\ref{fig:info_spec} for an interpretation of the definition in \eqref{eqn:def_Ds}. 

As we will see, the  $\eps$-information spectrum divergence is intimately related to the  $\eps$-hypothesis testing divergence (cf.~Lemma~\ref{lem:relation}). The former is, however, easier to compute. Note that if $P$ and $Q$ are product measures, then by virtue of the fact that   $\log\frac{P(Z)}{Q(Z)}$  is a sum of independent random variables, one can estimate the probability in~\eqref{eqn:def_Ds} using various probability tail bounds. This we do in the following section. 

We now state two useful properties of $D_{\rms}^\eps (P\| Q)$. The proofs of these lemmas are straightforward and can be found in  \cite[Sec.~III.A]{TomTan12}.

%\begin{lemma}
%Let $Q,Q' \in\scP(\calZ)$ be two distributions such that there exists an $\alpha>0$ with the property that $Q(z)\le \alpha Q'(z)$ for all $z\in\calZ$. 
%\end{lemma}
\begin{lemma}[Sifting from a convex combination]\label{lem:Ds1}
Let $P\in \scP(\calZ)$ and $Q=\sum_k \alpha_k Q_k$ be an  at most countable  convex combination of distributions $Q_k \in\scP(\calZ)$ with non-negative weights $\alpha_k$ summing to one, i.e., $\sum_k \alpha_k=1$. Then,
\begin{equation}
 D_{\rms}^\eps (P \| Q)\le\inf_{k} \bigg\{ D_{\rms}^\eps ( P \| Q_k) +\log\frac{1}{\alpha_k}\bigg\}.
 \end{equation} \end{lemma}
 In particular, Lemma~\ref{lem:Ds1} tells us that  if  there exists some $\gamma>0$ such that  $\tilQ (z)\le\gamma Q(z)$ for all $z\in\calZ$ then,
\begin{equation}
 D_{\rms}^\eps (P \| \tilQ )\ge  D_{\rms}^\eps ( P \| Q )  - \log\gamma  . \label{eqn:pmf_dom}
\end{equation}
 
\begin{lemma}[``Symbol-wise'' relaxation of $D_\rms^\eps$]\label{lem:Ds2}
Let $W$ and $V$ be two channels from $\calX$ to $\calY$ and let $P\in\scP(\calX)$. Then,
\begin{equation}
D_{\rms}^\eps (P\times W\| P\times V)\le\sup_{x\in\calX} D_\rms^\eps(W(\cdot|x) \| V(\cdot|x)).
\end{equation}
\end{lemma}

%The $\eps$-hypothesis testing divergence and the $\eps$-information spectrum  divergence are intimately related. In fact, 
One can readily toggle between the $\eps$-hypothesis testing divergence and the $\eps$-information spectrum divergence because they satisfy the   bounds in the following lemma. The proof of this lemma mimics that of~\cite[Lem.~12]{Tom12}. 
\begin{lemma}[Relation between divergences] \label{lem:relation}
For every $\eps\in (0,1)$ and  every $\eta\in (0,1-\eps)$, we have 
\begin{align}
D_{\rms}^{\eps}( P \| Q) - \log\frac{1}{1-\eps} &\le D_{\rmh}^{\eps}( P \| Q) \label{eqn:lower_bound_dsdh} \\*
&\le  D_{\rms}^{\eps+\eta}( P \| Q) +\log\frac{1-\eps}{\eta}. \label{eqn:upper_bound_dsdh}
\end{align}
\end{lemma}
\begin{proof}
The following proof is based on that for~\cite[Lem.~12]{Tom12}. For the lower bound in \eqref{eqn:lower_bound_dsdh}, consider the likelihood ratio test 
\begin{equation}
\delta(z) := \bbI\Big\{ \log\frac{P(z)}{Q(z)}\le \gamma\Big\}
,\quad\mbox{where}
\quad 
\gamma:= D_{\rms}^{\eps}( P \| Q) -\xi
\end{equation}
for some  small $\xi>0$. This test clearly satisfies $\bbE_P\big[\delta(Z)\big]\le \eps$ by the definition of the $\eps$-information spectrum divergence. On the other hand,
\begin{align}
& \bbE_Q\big[1-\delta(Z)\big]  \nn\\
 &=\sum_{z\in\calZ} Q(z) \bbI\Big\{ \log\frac{P(z)}{Q(z)} > \gamma   \Big\} \\
&\le\sum_{z\in\calZ} P(z)\exp(-\gamma) \bbI\Big\{ \log\frac{P(z)}{Q(z)} > \gamma   \Big\} \\
 &\le\sum_{z\in\calZ} P(z)\exp(-\gamma)  \\*
&\le \exp(-\gamma)
\end{align}
As a result, by the definition of $D_{\rmh}^{\eps}( P \| Q)$, we have 
\begin{equation}
D_{\rmh}^{\eps}( P \| Q)\ge \gamma-\log\frac{1}{1-\eps} = D_{\rms}^{\eps}( P \| Q) -\xi-\log\frac{1}{1-\eps}.
\end{equation}
Finally, take $\xi\downarrow 0$ to complete the proof of \eqref{eqn:lower_bound_dsdh}.

For the upper bound in \eqref{eqn:upper_bound_dsdh},  we may assume $D_{\rmh}^{\eps}( P \| Q)$  is finite; otherwise there is nothing to prove as $P$ is not absolutely continuous with respect to $Q$ and so $D_{\rms}^{\eps+\eta}( P \| Q)$ is infinite.  According to the definition of $D_{\rmh}^{\eps}( P \| Q)$,  for any $\gamma\ge 0$, there exists a test $\delta$ satisfying $\bbE_P[\delta(Z)]\le\eps$ such that 
\begin{align}
& (1-\eps)\exp(-D_{\rmh}^{\eps}( P \| Q)) \nn\\
 &\ge \bbE_Q\big[ 1-\delta(Z)\big] \\
&\ge \sum_{z  : P(z) \le \gamma Q(z) } Q(z)\big( 1-\delta(z)\big) \\
&\ge  \frac{1}{\gamma} \sum_{z  : P(z) \le \gamma Q(z) }  P(z)\big( 1-\delta(z)\big) \\
&\ge  \frac{1}{\gamma}  \bigg[ \sum_{z  }  P(z)\big( 1-\delta(z)\big) - \sum_{z: P(z) > \gamma Q(z)}P(z)  \bigg] \label{eqn:relation_end0}\\
&\ge  \frac{1}{\gamma}  \bigg[ 1-\eps- P\bigg(  \Big\{ z:\frac{P(z)}{Q(z)} > \gamma \Big\}\bigg) \bigg] \label{eqn:relation_end}
\end{align}
where \eqref{eqn:relation_end} follows because  $\bbE_P \big[ \delta(Z)\big]\le  \eps$. 
Now fix a  small $\xi>0$ and choose 
\begin{equation}
\gamma = \exp \big(D_{\rms}^{\eps+\eta}( P \| Q)+\xi \big).
\end{equation}
Consequently, from \eqref{eqn:relation_end}, we have 
\begin{align}
D_{\rmh}^{\eps}( P \| Q)   &\le   D_{\rms}^{\eps+\eta}( P \| Q)    +\xi \nn\\
&\quad - \log \bigg( 1- \frac{P\big( \big\{z:\log \frac{P(z)}{Q(z)}  >  D_{\rms}^{\eps+\eta}( P \| Q) + \xi \big\} \big) }{1-\eps}\bigg)
\end{align}
By the definition of $D_{\rms}^{\eps+\eta}( P \| Q)$, the probability within the logarithm is upper bounded  by $1-\eps-\eta$. Taking $\xi\downarrow 0$  completes the proof of~\eqref{eqn:upper_bound_dsdh} and hence, the lemma.  
\end{proof}
\section{Asymptotic Expansions}
In this section, we consider the asymptotic expansions of $D_{\rmh}^{\eps}( P^{(n)} \| Q^{(n)})$ and $D_{\rms}^{\eps}( P^{(n)} \| Q^{(n)})$ when $P^{(n)}$ and $Q^{(n)}$ are product distributions, i.e., 
\begin{equation} \label{eqn:prod_dist}
P^{(n)}(\bz) := \prod_{i=1}^n P_i(z_i),\quad\mbox{and}\quad Q^{(n)}(\bz) = \prod_{i=1}^n Q_i(z_i),
\end{equation}
for all $\bz=(z_1,\ldots, z_n)\in\calZ^n$.  The  {\em component distributions}  $\{(P_i,Q_i)\}_{i =1}^n$  are not necessarily the same for each $i$.  However, we do assume for the sake of simplicity that for each $i$, $P_i\ll Q_i$ so $D(P_i \| Q_i)<\infty$. Let $V(P\|Q)$ be the variance of the log-likelihood ratio between $P$ and $Q$, i.e.,
\begin{equation}
V(P\|Q):=\sum_{z\in\calZ} P(z) \bigg[ \log\frac{P(z)}{Q(z)}-D(P\|Q)\bigg]^2. \label{eqn:VPQ}
\end{equation}
This is also known as the {\em relative entropy variance}. 
Let the third absolute moment of the log-likelihood ratio between $P$ and $Q$ be 
\begin{equation}
T(P\|Q):= \sum_{z\in\calZ} P(z) \bigg| \log\frac{P(z)}{Q(z)}-D(P\|Q)\bigg|^3. \label{eqn:TPQ}
\end{equation}
Also define the following quantities:
\begin{align}
D_n  &:=\frac{1}{n}\sum_{i=1}^n D(P_i \| Q_i), \\ 
 V_n  & :=\frac{1}{n}\sum_{i=1}^n V(P_i \| Q_i), \quad\mbox{and}\\
T_n  &:=\frac{1}{n}\sum_{i=1}^n T(P_i\|Q_i).
\end{align}
The first result in this section is the following:
\begin{proposition}[Berry-Esseen bounds for  $D_{\rms}^\eps$] \label{prop:ds_be}
Assume  there exists a constant $V_- >0$ such that $V_n\ge V_-$. We have 
\begin{align}
 \Phi^{-1} \bigg( \eps - \frac{6\, T_n}{\sqrt{nV_-^3}} \bigg)&  \le \frac{ D_{\rms}^\eps( P^{(n)} \|Q^{(n)})  - nD_n}{\sqrt{nV_n}} \le  \Phi^{-1}\bigg( \eps+ \frac{6\, T_n}{\sqrt{nV_-^3}} \bigg). \label{eqn:ds_be}
\end{align}
\end{proposition}
\begin{proof}
Let $Z^n$ be distributed according to $P^{(n)}$. By using the product structure of $P^{(n)}$ and $Q^{(n)}$ in \eqref{eqn:prod_dist},
\begin{align}
\Pr\bigg( \log \frac{P^{(n)}(Z^n)}{Q^{(n)}(Z^n)}\le R\bigg)   =\Pr\bigg( \sum_{i=1}^n\log \frac{P_i(Z_i)}{Q_i(Z_i)}\le R\bigg)  .
\end{align}
By the Berry-Esseen theorem in Theorem \ref{thm:berry_gen}, we have 
\begin{align}
\bigg|\Pr\bigg( \sum_{i=1}^n\log \frac{P_i(Z_i)}{Q_i(Z_i)}\le R\bigg)-\Phi\bigg( \frac{ R - nD_n} {\sqrt{nV_n}}\bigg)    \bigg|\le \frac{6\, T_n}{\sqrt{nV_n^3}} .
\end{align}
The result  immediately follows by using the bound $V_n\ge V_-$. 
\end{proof}
A special case of the bound above occurs  when $P_i = P$ and $Q_i = Q$ for all $i= 1,\ldots, n$. In this case, we write $P^n$ for $P^{(n)}$ and similarly, $Q^n$ for $Q^{(n)}$. One has:
\begin{corollary}[Asymptotics of $D_\rms^\eps$] \label{cor:ds_same}
If $V(P\| Q)>0$, then 
\begin{equation}
 D_{\rms}^\eps ( P^n \|Q^n )  = nD(P\| Q)  + \sqrt{n V(P\| Q) } \Phi^{-1}(\eps ) + O(1).
\end{equation}
\end{corollary}
\begin{proof}
Since $V(P\| Q)>0$ and $T(P\|Q)<\infty$ (because $P\ll Q$), the term $6\, T_n /\sqrt{nV_-^3}$ in \eqref{eqn:ds_be} is equal to $\frac{c}{\sqrt{n}}$ for some finite $c>0$. By Taylor expansions, 
\begin{equation}
\Phi^{-1}\bigg( \eps \pm\frac{c}{\sqrt{n}}\bigg)=\Phi^{-1}(\eps) + O\bigg( \frac{1}{\sqrt{n}}\bigg),
\end{equation}
which completes the proof.
\end{proof}
In some applications, it is not possible to guarantee that $V_n$ is uniformly bounded away from zero (per Proposition \ref{prop:ds_be}). In this case,  to obtain an upper bound on $D_\rms^\eps$, we employ Chebyshev's inequality instead of the Berry-Esseen theorem.  In the following proposition, which is usually good enough to establish {\em strong converses}, we do not assume that the component distributions are the same. 
\begin{proposition}[Chebyshev bound  for  $D_{\rms}^\eps$] \label{prop:ds_ch}
We have 
\begin{align}
 D_{\rms}^\eps( P^{(n)} \|Q^{(n)})\le nD_n + \sqrt{\frac{nV_n}{1-\eps} }. \label{eqn:D_chey} 
\end{align}
\end{proposition}
\begin{proof}
By the definition of the $\eps$-information spectrum divergence, we have 
\begin{align}
 D_{\rms}^\eps( P^{(n)} \|Q^{(n)})=\max \big\{ D^-, D^+ \big\}
\end{align}
where $D^-$ and $D^+$ are defined as 
\begin{align}
 D^- &:= \sup\bigg\{ R \le nD_n :  P\bigg(\Big\{ z\in\calZ:\log\frac{P(z)}{Q(z)} \le R \Big\}\bigg)\le\eps\bigg\}, \\
 D^+ &:= \sup\bigg\{ R  > nD_n :  P\bigg(\Big\{ z\in\calZ:\log\frac{P(z)}{Q(z)} \le R \Big\}\bigg)\le\eps\bigg\} .
\end{align}
Clearly, $D^-\le nD_n$ so it remains to upper bound $ D^+$. 
 Let $R> nD_n$ be fixed. By  Chebyshev's inequality,
\begin{equation}
 \Pr\bigg( \sum_{i=1}^n\log \frac{P_i(Z_i)}{Q_i(Z_i)}\le R\bigg) \ge 1-\frac{nV_n}{(R-nD_n)^2}.
\end{equation}
Hence, we have 
\begin{align}
 D^+ &\le \sup\left\{ R> nD_n :  1-\frac{nV_n}{(R-nD_n)^2}\le\eps\right\}  \\
&= nD_n+ \sqrt{\frac{nV_n}{1-\eps} }.\label{eqn:Dge}
\end{align}
Thus, we   see that the bound on $D^+$  dominates. This yields \eqref{eqn:D_chey} as desired. 
%Relaxing the bound on $R$ in the supremum that defines  $D_\rms^\eps$ to $R>nD_n$, we find
%\begin{align}
%D_\rms^\eps(P^{(n)}\| Q^{(n)})&\le\sup\left\{ R> nD_n :  1-\frac{nV_n}{(R-nD_n)^2}\le\eps\right\}  \\
%&= nD_n+ \sqrt{\frac{nV_n}{1-\eps} }.
%\end{align}
%This completes the proof.
\end{proof}

Now we would like an expansion for $D_\rmh^\eps$ similar to that for $D_\rms^\eps$ in  Corollary~\ref{cor:ds_same}.  The following was shown by Strassen~\cite[Thm.~3.1]{Strassen}.  %Recall that in Corollary~\ref{cor:ds_same}, we assume that $P_i=P$ and $Q_i=Q$ for all $i$, i.e., the component distributions are the same and we assume this  is true for the rest of the section.
\begin{proposition}[Asymptotics of $D_\rmh^\eps$] \label{prop:asymp_dh}
Assume the conditions in Corollary~\ref{cor:ds_same}. The following holds:
\begin{equation}
D_\rmh^\eps (P^n\| Q^n) = nD(P \| Q) + \sqrt{nV( P \|Q) }\Phi^{-1}(\eps) + \frac{1}{2} \log n + O(1). \label{eqn:asymp_dh}
\end{equation}
\end{proposition} 
As a result, in the asymptotic setting  for identical product distributions, $D_\rmh^\eps (P^n\| Q^n) $   exceeds $D_\rms^\eps (P^n\| Q^n) $ by $\frac{1}{2}\log n$ ignoring constant terms, i.e., 
\begin{equation}
D_\rmh^\eps (P^n\| Q^n) =D_\rms^\eps (P^n\| Q^n)  +\frac{1}{2}\log n +O(1).
\end{equation}
\begin{proof}
Let us first verify the upper bound. Let $\eta$ in the upper bound of Lemma~\ref{lem:relation}  be chosen to be $\frac{1}{\sqrt{n}}$. Now, for $n$ large enough (so $\frac{1}{\sqrt{n}}< 1-\eps$), combine this  upper bound  with  Corollary~\ref{cor:ds_same} to obtain that 
\begin{align}
& D_\rmh^\eps  (P^n\| Q^n) \le D_\rms^{\eps+\frac{1}{\sqrt{n}}}( P^n\|Q^n) + \frac{1}{2}\log n + \log(1-\eps)\\
&= nD(P \| Q) + \sqrt{nV( P \|Q) }\Phi^{-1}\bigg(\eps +\frac{1}{\sqrt{n}}\bigg)+ \frac{1}{2}\log n + O(1)
\end{align}
Applying a Taylor expansion to the last step and noting that $V(P \| Q)<\infty$ because $P\ll Q$ yields  the upper bound in \eqref{eqn:asymp_dh}.

The proof of the lower bound in \eqref{eqn:asymp_dh} is slightly  more involved. Observe that if we na\"{i}vely employed \eqref{eqn:lower_bound_dsdh}  to lower bound $D_\rmh^\eps (P^n\|Q^n)$ with $D_\rms^\eps(P^n\|Q^n)-\log\frac{1}{1-\eps}$, the third-order term would be $O(1)$ instead of the better $\frac{1}{2}\log n + O(1)$.  The idea is to  propose an appropriate  test for $D_\rmh^\eps$ and to use Theorem~\ref{thm:str_ld}.  Consider the   likelihood ratio test 
\begin{equation}
\delta(\bz):=\bbI\bigg\{ \log\frac{P^n(\bz)}{Q^n(\bz)}\le\gamma\bigg\} \label{eqn:lrt}
\end{equation}
Define $\sigma^2 :=   V(P\|Q)$ and $T :=   T(P\| Q)$. Also define the \iid random variables $U_i:=\log P(Z_i) -\log Q(Z_i)$, $1\le i \le n$, each having variance $\sigma^2$ and third absolute moment $T$. Consider, the expectation of $1-\delta(Z^n)$ under the distribution $Q^n$:
\begin{align}
&\bbE_{Q^n }\big[ 1-\delta(Z^n) \big]  \nn\\*
&= \sum_{\bz} Q^n(\bz) \bbI\bigg\{ \log\frac{P^n(\bz)}{Q^n(\bz)} > \gamma\bigg\}\\
&= \sum_{\bz} P^n(\bz)\exp \bigg(-\log\frac{P^n(\bz)}{Q^n(\bz)}  \bigg) \bbI\bigg\{ \log\frac{P^n(\bz)}{Q^n(\bz)} > \gamma\bigg\}\\
&=\bbE_{P^n}\bigg[ \exp\bigg( -\sum_{i=1}^n U_i \bigg)\bbI\bigg\{ \sum_{i=1}^n U_i>\gamma\bigg\}\bigg] \\*
&\le 2\bigg( \frac{\log 2}{\sqrt{2\pi}} +\frac{12\, T}{\sigma^2}\bigg)\frac{\exp(-\gamma) }{\sigma\sqrt{n}}\label{eqn:apply_str_ld}
\end{align}
where \eqref{eqn:apply_str_ld} is an application of Theorem~\ref{thm:str_ld}.   Now put 
\begin{equation}
\gamma :=nD(P\|Q) + \sqrt{nV(P\| Q)}\Phi^{-1}\bigg(\eps-\frac{6\, T(P\|Q) }{\sqrt{nV(P\|Q )^3}} \bigg) . \label{eqn:choice_gamma}
\end{equation}
An application of the Berry-Esseen theorem yields
\begin{equation}
 \bbE_P\big[ \delta(Z^n)\big]=P\bigg(\Big\{ \bz: \sum_{i=1}^n \log\frac{P (z_i)}{Q (z_i)}\le\gamma\Big\}\bigg)\le\eps. \label{eqn:fa}
 \end{equation} 
 From  \eqref{eqn:apply_str_ld}, \eqref{eqn:fa} and the definition of $D_\rmh^\eps$, we have 
\begin{align}
D_\rmh^\eps(P^n\|Q^n) & \ge \gamma+\log\big(\sigma\sqrt{n}\big) + O(1) = \gamma+\frac{1}{2}\log n + O(1) \label{eqn:lower_bd_dh} .
\end{align}
     The proof is concluded by plugging  \eqref{eqn:choice_gamma} into \eqref{eqn:lower_bd_dh} and  Taylor expanding $\Phi^{-1}(\cdot)$ around $\eps$.
\end{proof}
We remark that the lower bound in Proposition~\ref{prop:asymp_dh} can be achieved using {\em deterministic} tests, i.e., $\delta$ can be chosen to be an indicator function as in~\eqref{eqn:lrt}. Randomization is thus unnecessary. Also, one can relax the assumption that   $Q^n$ is a  product probability measure; it can be an arbitrary product measure.  These realizations are important to make the connection between   hypothesis testing and   lossless source coding which we discuss in the next chapter.

A  corollary of Proposition~\ref{prop:asymp_dh} is the   Chernoff-Stein lemma  \cite{Chernoff52}    quantifying the  error exponent of the   probability of missed detection  keeping the  probability of false alarm  bounded above by~$\eps$.
\begin{corollary}[Chernoff-Stein lemma]
Assume the conditions in Corollary~\ref{cor:ds_same} and recall the definition of $\beta_{1-\eps}$   in \eqref{eqn:def_beta}. For every $\eps\in (0,1)$,  
\begin{equation}
\lim_{n\to\infty}\frac{1}{n}\log \frac{1}{\beta_{1-\eps}(P^n , Q^n)}=\lim_{n\to\infty}\frac{D_\rmh^\eps(P^n\|Q^n)}{n}= D(P\| Q).
\end{equation}
\end{corollary}
%Thus the optimum error exponent of the probability of missed detection  given that the probability of false alarm does not exceed $\eps$ is $D(P\| Q)$.
 
 % \cite{WangRenner} 

%%%%%%%%%%%%%%%%%%%%%%%%

\part{Point-To-Point Communication}

\chapter{Source Coding} \label{ch:src}

In this chapter, we revisit the fundamental problem of fixed-to-fixed length lossless and lossy source compression. Shannon, in his original paper~\cite{Shannon48} that launched the field of information theory,  showed that the fundamental limit of compression of a discrete memoryless source (DMS) $P$ is the entropy $H(P)$. For the case of continuous sources, lossless compression is not possible and some distortion must be allowed. Shannon showed in~\cite{Shannon59b} that the corresponding fundamental limit of compression of memoryless source $P$ up to distortion $\Delta\ge 0$, assuming a separable distortion measure $d$, is the {\em rate-distortion function}
\begin{equation}
R(P,\Delta):=\min_{W \in\scP(\hat{\calX}|\calX):\bbE_{P\times W}[d(X,\hatX)]\le \Delta}I(P,W).  \label{eqn:rd_func}
\end{equation}
These first-order fundamental limits are attained as the  number of realizations of the source (i.e., the blocklength of the source) $P$ tends to infinity.  The strong converse for rate-distortion is also known and shown, for example, in \cite[Ch.~7]{Csi97}. In the following, we present known non-asymptotic bounds for lossless and lossy source coding. We then fix the permissible error probability in the lossless case and the excess distortion probability in the lossy case at some non-vanishing $\eps\in (0,1)$. The non-asymptotic  bounds    are evaluated as $n$ becomes large so as to obtain asymptotic expansions of the logarithm of the smallest achievable code size.  These refined results provide an approximation of the extra code rate (beyond $H(P)$ or $R(P,\Delta)$)  one must incur when operating in the finite blocklength regime.  Finally, for both the lossless and lossy compression problems, we provide alternative proof techniques based on the method of types that are partially universal. %discuss the price of universality of a lossless code in terms of its effect on the asymptotic expansion.

The material in this chapter concerning lossless source coding is based on the seminal work by Strassen~\cite[Thm.~1.1]{Strassen}. The material on lossy source coding is based on more recent work by Ingber-Kochman~\cite{ingber11} and Kostina-Verd\'u~\cite{kost12}.
\section{Lossless Source Coding: Non-Asymptotic Bounds}
We now set up the  almost lossless source coding problem formally. As mentioned, we only consider {\em fixed-to-fixed length} source coding  in this monograph.  A {\em source} is simply a probability mass function $P$ on some finite alphabet $\calX$ or the associated random variable $X$ with distribution $P$. See Fig.~\ref{fig:source} for an illustration of the setup.

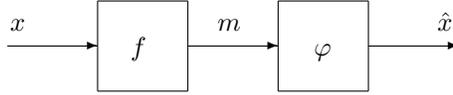
\begin{figure}[t]
\centering
\setlength{\unitlength}{.4mm}
\begin{picture}(150, 35)
%\thicklines
\put(0, 15){\vector(1, 0){30}}
\put(60, 15){\vector(1,0){30}}
\put(120, 15){\vector(1,0){30}}
%\put(180, 15){\vector(1,0){30}}
\put(30, 0){\line(1, 0){30}}
\put(30, 0){\line(0,1){30}}
\put(60, 0){\line(0,1){30}}
\put(30, 30){\line(1,0){30}}

\put(90, 0){\line(1, 0){30}}
\put(90, 0){\line(0,1){30}}
\put(120, 0){\line(0,1){30}}
\put(90, 30){\line(1,0){30}}

\put(-2, 20){  $x$}
\put(67, 20){  $m$}
%\put(51, -10){  $\bbE[\rvg(X)]\le\Gamma$}
%\put(65, 8){  $[2^{nR}]$}
\put(140, 20){  $\hatx$} 
\put(41, 12){$f$ } 
\put(102, 12){$\varphi$} 
%
%\put(150, 0){\line(1, 0){30}}
%\put(150, 0){\line(0,1){30}}
%\put(180, 0){\line(0,1){30}}
%\put(150, 30){\line(1,0){30}}
%\put(161, 12){$\varphi$} 
%\put(190, 20){  $\hatm$} 
%\put(186, 1){  $\Pr(\hatM \ne M)$} 
  \end{picture}
  \caption{Illustration of the fixed-to-fixed length   source coding problem.   }
  \label{fig:source}
\end{figure}

An {\em $(M,\eps)$-code} for the source $P\in\scP(\calX)$ consists of a pair of maps that includes an {\em encoder} $f:\calX\to\{1,\ldots, M\}$ and a {\em decoder} $\varphi: \{1,\ldots, M\}\to\calX$ such that the {\em error probability} 
\begin{equation}
P\big( \{x\in\calX: \varphi(f(x))\ne x \} \big)\le\eps.
\end{equation}
The number $M$ is called the {\em size} of the code $(f,\varphi)$. 

Given a source $P$, we define the {\em almost lossless source coding non-asymptotic fundamental limit} as 
\begin{equation}
M^*(P,\eps):= \min\big\{ M\in\bbN \,:\, \exists \mbox{ an }  (M,\eps)\mbox{-code for } P\big\}.
\end{equation}
Obviously for an arbitrary source, the exact evaluation of the minimum  code size $M^*(P,\eps)$ is   challenging. In the following, we assume that  it is a discrete memoryless source (DMS), i.e., the distribution $P^n$ consists of $n$ copies of an underlying distribution $P$.  With this assumption, we can find  an asymptotic expansion  of $\log M^*(P^n,\eps)$. 

The agenda for this and subsequent  chapters will largely be standard. We first establish ``good''  bounds on non-asymptotic quantities like $M^*(P,\eps)$. Subsequently, we replace the source or channel with $n$ independent copies of it.  Finally, we use an appropriate limit theorem (e.g., those in Section \ref{sec:prob}) to evaluate the non-asymptotic bounds in the large $n$ limit.

\subsection{An Achievability Bound}
One of the   take-home messages that we would like to convey in this section is that fixed-to-fixed length lossless source coding is nothing but binary hypothesis testing where the measures $P$ and $Q$ are chosen appropriately. In fact, a reasonable  coding scheme for the lossless source coding would simply be to encode a ``typical'' set of source symbols $\calT\subset\calX$, ignore the rest, and declare an error if the realized symbol from the source is not in $\calT$. In this way, one sees that 
\begin{equation}
M^*(P ,\eps) \le \min_{\calT\subset\calX:  P(\calX\setminus\calT ) \le\eps}    |\calT|  \label{eqn:MT}
\end{equation}
This bound can be stated in terms of $\beta_{1-\eps}(P,Q)$ or, equivalently, the $\eps$-hypothesis testing divergence $D_\rmh^\eps(P\| Q)$    if we 
restrict the tests  that define these quantities  to be deterministic, and also allow    $Q$ to be an arbitrary  measure  (not necessarily  a probability measure). Let $\mu$ be the counting measure, i.e., 
\begin{equation}
\mu(\calA) :=  |\calA|,\quad\forall\,\calA\subset\calX. 
\end{equation}
\begin{proposition}[Source coding as hypothesis testing: Achievability] \label{thm:src_ac}
Let $\eps\in (0,1)$ and $P$ be any source with countable alphabet $\calX$. We have
\begin{equation}
 M^*(P,\eps)\le \beta_{1-\eps}(P,\mu),
\end{equation}
or in terms of the $\eps$-hypothesis testing divergence (cf.~\eqref{eqn:beta_D}),
\begin{equation}
\log M^*(P,\eps) \le  - D_\rmh^\eps(P \| \mu) - \log\frac{1}{1-\eps}.
\end{equation}
\end{proposition}
\subsection{A Converse Bound}
The converse bound we evaluate is also intimately connected to a divergence we introduced in the previous chapter, namely the $\eps$-information spectrum divergence where the distribution in the alternate hypothesis $Q$ is chosen to be the counting measure. 
\begin{proposition}[Source coding as hypothesis testing: Converse] \label{thm:src_conv}
Let $\eps\in (0,1)$ and $P$ be any source with countable alphabet $\calX$. For any $\eta\in (0,1-\eps)$,   we have 
\begin{equation}
\log M^*(P,\eps)\ge - D_\rms^{\eps+\eta}(P \| \mu)   -\log \frac{1}{\eta}.
\end{equation}
\end{proposition}
This statement is exactly  Lemma~1.3.2 in Han's book~\cite{Han10}. Since the proof is short, we provide it for completeness. 
\begin{proof}
By the definition of the   $\eps$-information spectrum divergence, it is enough to  establish that every $(M,\eps)$-code for $P$ must satisfy
\begin{equation}
\eps +\eta\ge P\bigg(\Big\{x: P(x) \le \frac{\eta}{M}\Big\} \bigg) ,  \label{eqn:sc_conv}
\end{equation}
for any $\eta\in (0,1-\eps)$.   Let $\calT:= \{x: P(x) \le \frac{\eta}{M}\}$ and let $\calS:=\{x:\varphi(f(x)) =  x\}$. Clearly, $|\calS|\le M$. We have
\begin{equation}
P(\calT)\le P(\calX\setminus\calS) + P(\calT\cap\calS)\le\eps + P(\calT\cap\calS). \label{eqn:PT}
\end{equation}
Furthermore,
\begin{equation}
P(\calT\cap\calS)=\sum_{x\in\calT\cap\calS} P(x) \le \sum_{x\in\calT\cap\calS} \frac{\eta}{M}\le |\calS|  \frac{\eta}{M}\le\eta. \label{eqn:PTS}
\end{equation}
Uniting \eqref{eqn:PT} and \eqref{eqn:PTS} gives \eqref{eqn:sc_conv} as desired. 
\end{proof}
Observe the similarity of this proof to proof of the upper bound of $D_\rmh^\eps$ in terms of $D_\rms^\eps$ in Lemma~\ref{lem:relation}.
\section{Lossless Source Coding: Asymptotic Expansions} \label{sec:asymp_lossless}
Now we assume that the source $P^n$ is stationary and memoryless, i.e., a DMS.  More precisely,
\begin{equation}
P^n(\bx) =\prod_{i=1}^n P(x_i) ,\quad\, \forall\, \bx\in\calX^n. 
\end{equation}
We assume throughout that $P(x)>0$ for all $x\in\calX$.  Shannon \cite{Shannon48} showed that the minimum rate   to achieve almost lossless compression of a DMS $P$ is the entropy $H(P)$. In this  section as well as the next one, we derive finer evaluations of the fundamental compression limit by considering the asymptotic expansion of $\log M^*(P^n,\eps)$.   To do so, we need to define another important quantity related to the source $P$.

Let the {\em source dispersion}  of $P$ be the variance of the {\em self-information} random variable $-\log {P(X)}$, i.e., 
\begin{equation}
V(P):=\var\bigg[ \log\frac{1}{P(X)}\bigg] = \sum_{x\in\calX}P(x)\bigg[ \log\frac{1}{P(x)}-H(P) \bigg]^2.
\end{equation}
Note that the expectation of the  self-information is the entropy $H(P)$.  In Kontoyannis-Verd\'u~\cite{verdu14},  $V(P)$ is called the {\em varentropy}.  If $V(P)=0$ this means that the source is either deterministic or uniform.% We assume throughout that $V(P)>0$, which implies that the source is neither deterministic nor uniform. 

%Yushkevich~\cite{Yushkevich}

The two non-asymptotic theorems in the preceding section combine to give the following asymptotic expansion of the  minimum code size $M^*(P^n,\eps)$. %Note that by using the notation $P^n$, we mean that the source emits $n$ independent symbols, each from the distribution $P$. The code length is equal to $n$. 
\begin{theorem}\label{thm:asymp_src}
If the source $P\in\scP(\calX)$ satisfies $V(P)>0$, then
\begin{equation}
\log M^*(P^n,\eps)= nH(P) - \sqrt{nV(P)}\Phi^{-1}(\eps) -\frac{1}{2}\log n + O(1). \label{eqn:nonzero_VP}
\end{equation}
Otherwise, we have 
\begin{equation}
\log M^*(P^n,\eps)= nH(P)  + O(1). \label{eqn:zero_VP}
\end{equation}
\end{theorem}
\begin{proof}
For the direct part of \eqref{eqn:nonzero_VP} (upper bound), note that the term $-\log\frac{1}{1-\eps}$  in Proposition~\ref{thm:src_ac} is a constant, so we simply have to evaluate $D_{\rmh}^\eps(P^n\|\mu^n)$.\footnote{Just to be pedantic,  for any   $\calA\subset\calX^n$, the measure  $\mu^n(\calA)$ is defined as $\sum_{\bx\in\calA}\mu^n(\bx)=|\calA|$  and $\mu^n(\bx)=1$ for each $\bx\in\calX^n$. Hence, $\mu^n$ has the required product structure for the application of Corollary~\ref{cor:ds_same}, for  which the second argument of $D_\rmh^\eps(P^n \| Q^n)$ is not restricted to product  probability measures.  }  From Corollary~\ref{cor:ds_same} and its remark that the lower bound on $D_\rmh^\eps(P^n\|Q^n)$ can be achieved using deterministic tests, we have 
\begin{equation}
D_{\rmh}^\eps(P^n\|\mu^n)=nD(P\| \mu) + \sqrt{nV(P\| \mu) }\Phi^{-1}(\eps) +\frac{1}{2}\log n+O(1).
\end{equation}
It can easily be verified (cf.~\eqref{eqn:div_ent}) that 
\begin{equation}
D(P\|\mu) = -H(P) ,\quad\mbox{and}\quad V(P\|\mu) = V(P).  \label{eqn:mu_simple}
\end{equation}
This concludes the proof of the direct part in light of Proposition~\ref{thm:src_ac}. 

For the converse part of \eqref{eqn:nonzero_VP} (lower bound), choose $\eta=\frac{1}{\sqrt{n}}$ so the term $-\log\frac{1}{\eta}$  in Proposition~\ref{thm:src_conv} gives us the $-\frac{1}{2}\log n$ term. Furthermore,  by Proposition~\ref{prop:ds_be} and the simplifications in \eqref{eqn:mu_simple},
\begin{equation}
D_\rms^{\eps+\frac{1}{\sqrt{n}}}(P^n \| \mu^n) = -nH(P) + \sqrt{nV(P)}\Phi^{-1}\bigg(\eps+\frac{1}{\sqrt{n}}\bigg) + O(1). 
\end{equation}
A Taylor expansion of $\Phi^{-1}(\cdot)$ completes the proof of \eqref{eqn:nonzero_VP}.

For \eqref{eqn:zero_VP}, note that the self-information  $-\log P(X)$ takes on the value $H(P)$ with probability one. In other words, 
\begin{equation}
P^n\big( \log P^n(X^n) \le R\big)=\bbI\big\{ R\ge -nH(P)\big\}.
\end{equation}
The bounds on  $\log M^*(P^n,\eps)$ in Propositions~\ref{thm:src_ac} and~\ref{thm:src_conv} and the relaxation to the $\eps$-information spectrum divergence (Lemma~\ref{lem:relation}) yields~\eqref{eqn:zero_VP}.
\end{proof}

The expansion in~\eqref{eqn:nonzero_VP} in Theorem~\ref{thm:asymp_src}  appeared in early works by Yushkevich~\cite{Yushkevich} (with  $o(\sqrt{n})$ in place of $-\frac12\log n$ but for Markov chains) and Strassen~\cite[Thm.~1.1]{Strassen} (in the form stated). It has since appeared in various other forms and levels of generality in Kontoyannis~\cite{Kot97}, Hayashi~\cite{Hayashi08}, Kostina-Verd\'u~\cite{kost12}, Nomura-Han \cite{Nom13b} and Kontoyannis-Verd\'u~\cite{verdu14} among others.

As can be seen from the non-asymptotic bounds and the asymptotic evaluation, fixed-to-fixed length lossless source coding and binary hypothesis testing are virtually the same problem. Asymptotic expansions for $D_\rmh^\eps$ and $D_\rms^\eps$ can be used directly to estimate the minimum code size $M^*(P^n,\eps)$ for an $\eps$-reliable lossless source code. 
\section{Second-Order Asymptotics of  Lossless Source Coding via the Method of  Types}   \label{sec:universal_lossless}
Clearly, the coding scheme described  in \eqref{eqn:MT} is {\em non-universal}, i.e., the   code depends on knowledge of the source distribution. In many applications, the  exact source distribution is unknown, and hence has to be estimated {\em a priori}, or one has to design a source code that works well for any source distribution. It is a well-known application of the method of types that universal source codes achieve the lossless source coding error exponent~\cite[Thm.~2.15]{Csi97}. One then wonders whether there is any degradation in the asymptotic expansion of $\log M^*(P^n,\eps)$ if the encoder and decoder know less about the source. It turns out that the source dispersion term can be achieved only with the knowledge of $H(P)$  and $V(P)$.  However,   one has to work much harder to determine the third-order term.   For conclusive results on the third-order term for fixed-to-variable length source coding, the reader is referred to the elegant work by Kosut and Sankar~\cite{KosutSankar,KosutSankar14}. The technique outlined in this section involves the method of types.

%In this section, we show how to achieve the asymptotic expansion of $\log M^*(P^n,\eps)$ up to the second-order term using the method of types. The scheme is partially universal in the sense that the exact source statistics are not required to be known, only the first three central moments of its self-information random variable.

Let $M_\rmu^*(P,\eps)$ be the  almost lossless source coding  non-asymptotic fundamental limit  where the  source code $(f,\varphi)$ is   ignorant  of the probability distribution of the source $P$, except for the entropy  $H(P)$ and the varentropy $V(P)$. % and the third absolute moment of the self-information random variable $T(P)$.
\begin{theorem}\label{thm:src_univ}
If the source $P\in\scP(\calX)$ satisfies $V(P)>0$, then 
\begin{equation}
\log M_\rmu^*(P^n,\eps)\le  nH(P) - \sqrt{nV(P)}\Phi^{-1}(\eps) +  (|\calX|-1)\log n  + O(1).
\end{equation}
\end{theorem}
The proof   we present here results in a third-order term that is likely to be  far from optimal but we present this proof to demonstrate the similarity to the classical proof of the fixed-to-fixed length source coding error exponent using the method of types~\cite[Thm.~2.15]{Csi97}. 
\begin{proof}[Proof of Theorem~\ref{thm:src_univ}]
Let $\calX=\{a_1,\ldots, a_{d}\}$ without loss of generality.  Set $M$, the size of the code, to be the smallest integer exceeding 
\begin{align}
\exp\bigg[ n H(P) &- \sqrt{nV(P)}\Phi^{-1} \Big(\eps-\frac{c}{\sqrt{n }} \Big)  +(d-1)  \log (n+1)\bigg], \label{eqn:choiceM}
\end{align}
for some finite  constant $c>0$ (given in Theorem~\ref{thm:func_clt}). Let $\calK$ be the set of sequences in $\calX^n$ whose empirical entropy is no larger  than % $\gamma$ defined as 
\begin{equation}
\gamma :=\frac{1}{n}\log M- \frac{ (d-1)\log (n+1)}{n}. \label{eqn:src_gamma}
\end{equation}
In other words,
\begin{equation}
\calK :=  \bigcup_{Q \in\scP_n(\calX): H(Q) \le\gamma } \calT_Q  .
\end{equation}
Encode all sequences  in $\calK$ in a one-to-one way and sequences not in $\calK$ arbitrarily.  By the type counting lemma in \eqref{eqn:type_coutn2} and Lemma~\ref{lem:size_type_class} (size of type class), we have 
\begin{align}
|\calK| \! \le\! \sum_{Q  \in\scP_n(\calX): H(Q) \le\gamma } \!\exp\big(nH(Q)\big)\le (n+1)^{d-1}\exp\big( n\gamma\big)\le  M   \label{eqn:size_src_cd}
\end{align}
so the number of sequences that can be reconstructed without error is at most $M$ as required. 
An error occurs if and only if the source sequence has empirical entropy exceeding $\gamma$, i.e.,  the error probability is 
\begin{equation}
\mathfrak{p}:=\Pr( H(P_{X^n})>\gamma )\label{eqn:prob_univ_lossless0}
\end{equation}
where $P_{X^n} \in\scP_n(\calX)$ is the random type of $X^n\in\calX^n$.   This probability can be written as 
\begin{align}
\mathfrak{p}=\Pr\big( f(   P_{X^n} - P ) > \gamma\big) , \label{eqn:prob_univ_lossless} %=\Pr\left( f\big( \frac{1}{\sqrt{n}}\sum_{i=1}^n \big( \bbI\{X_i=\big)\big) \ge \gamma\right)
\end{align}
where the function $f:\bbR^d\to\bbR$ is defined as 
\begin{equation}
f(\bv) = H\left( {\bv} + P \right),\label{eqn:func_ent}
\end{equation}
In \eqref{eqn:prob_univ_lossless} and \eqref{eqn:func_ent}, we regarded the type $P_{X^n}$ and the true distribution $P$ as vectors of length $d =|\calX|$, and $H(\bw)=-\sum_j w_j\log w_j$ is the entropy.   Note that the argument of $f(\cdot)$ in \eqref{eqn:prob_univ_lossless} can be written as 
\begin{equation}
  P_{X^n} - P  =\frac{1}{ {n}}\sum_{i=1}^n\begin{bmatrix}
 \bbI\{X_i = a_1\} - P(a_1) \\    \vdots \\  \bbI\{X_i = a_{d} \} - P(a_{d})
 \end{bmatrix}=:  \frac{1}{n }\sum_{i=1}^n U_i^{d} .
\end{equation}
Since $U_i^d := [ \bbI\{X_i=a_1\}-P(a_1),\ldots, \bbI\{X_i=a_d\}-P(a_d)]'$ for $i=1,\ldots, n$ are zero-mean, \iid random vectors,   we may appeal to the Berry-Esseen theorem for functions of \iid random vectors in Theorem~\ref{thm:func_clt}.  Indeed, we note that $f(\bzero)=H(P)$, the Jacobian of $f$ evaluated at $\bv=\bzero$ is 
\begin{align}
\bJ     =   \left[
  \log \bigg( \frac{1}{\rme\, P(a_1)}\bigg)   \,\,\, \ldots \,\,\,  \log \bigg( \frac{1}{\rme\, P(a_d) }\bigg)
\right],
\end{align}
and the $(s,t) \in \{1,\ldots, d\}^2$ element of the covariance matrix of $U_1^d$ is 
\begin{align}
\left[\cov\big(U_1^{d} \big)\right]_{st}  =\left\{ \begin{array}{cc}
P(a_s)\big(1-P(a_s)\big) &  s= t\\
-P(a_s)P(a_t)& s\ne t
\end{array} \right. .
\end{align}
As such, by a routine multiplication of matrices, 
\begin{equation}
\bJ \cov\big(U_1^{d})\bJ' = V(P),
\end{equation}
 the varentropy of the source. We deduce from Theorem~\ref{thm:func_clt} that 
\begin{equation}
\mathfrak{p}\le \Phi\bigg( \frac{ H(P) -\gamma}{\sqrt{V(P)/n}}\bigg)+\frac{c}{\sqrt{n}}
\end{equation}
where $c$ is a finite positive constant (depending on $P$).  By the choice of $M$ and  $\gamma$ in \eqref{eqn:choiceM}--\eqref{eqn:src_gamma}, we see that $\mathfrak{p}$ is no larger than $\eps$.
\end{proof}

%The approximations in the  steps involving Taylor expansions in~\eqref{eqn:taylor_src} and~\eqref{eqn:taylor_src2} can be made precise. See the proof of the  rate redundancy theorem in Ingber-Kochman~\cite{ingber11} and for a more general vector version of the rate redundancy theorem, see Tan-Kosut~\cite{TK14}. 
This coding scheme is partially universal in the sense that $H(P)$ and $V(P)$  need to be known to be used to determine  and the threshold $\gamma$ in~\eqref{eqn:src_gamma}, but otherwise no other characteristic of the source $P$ is required to be known.  This achievability proof strategy is rather general and can   be applied to rate-distortion (cf.~Section~\ref{sec:lossy_types}), channel coding, joint source-channel coding~\cite{ingberWK,wang11},  as well as multi-terminal problems~\cite{TK14} (cf.~Section~\ref{sec:sw_uni}).

The point we would like to emphasize in this section is the following: In large deviations (error exponent) analysis of   almost lossless source coding, the probability of error in \eqref{eqn:prob_univ_lossless0} is evaluated using, for example, Sanov's theorem~\cite[Ex.~2.12]{Csi97}, or refined versions of it \cite[Ex.~2.7(c)]{Csi97}. In the above proof, the probability of error   is instead estimated using the Berry-Esseen theorem (Theorem \ref{thm:func_clt}) since the deviation of the code rate  from the first-order fundamental limit $H(P)$ is of the order $\Theta(\frac{1}{\sqrt{n}})$ instead of a constant. Essentially, the proof of Theorem~\ref{thm:src_univ} hinges on the fact that for a random vector $X^n$ with distribution $P^n$, the  entropy of the  type $P_{X^n}$, namely the empirical entropy $\hatH(X^n)$, satisfies the following central limit relation:
\begin{equation}
\sqrt{n}\big(    \hatH(  X^n  )  - H(P)   \big) \stackrel{\mathrm{d}}{\longrightarrow}\calN\big(0,V(P) \big). \label{eqn:clt_lossles}
\end{equation}
Finally, we note that the technique to bound the probability in \eqref{eqn:prob_univ_lossless} is similar to that suggested by Kosut and Sankar \cite[Lem.~1]{KosutSankar14}.

\section{Lossy Source Coding: Non-Asymptotic Bounds} \label{sec:lossy_src}
In the second half of this chapter, we consider the lossy source coding problem where the source $P \in\scP(\calX)$ does not have to be discrete. The setup is as in Fig.~\ref{fig:source} and the {\em reconstruction alphabet} (which need not be the same as $\calX$) is denoted as $\hat{\calX}$. For the lossy case, one considers a {\em distortion measure}  $d(x,\hatx)$ between the source $x \in \calX$ and its reconstruction $\hatx\in\hat{\calX}$.  This is simply a mapping from $\calX\times\hat{\calX}$ to the set of non-negative real numbers.% for some $d_{\max}\ge 0$. 

We make the following simplifying assumptions throughout. 
\begin{itemize}
\item[(i)] There exists a $\Delta$ such that $R(P,\Delta)$, defined in \eqref{eqn:rd_func}, is finite.
 \item[(ii)] The distortion measure is such that there exists a finite set $\calE\subset \hat{\calX}$ such that $\bbE[\min_{\hatx\in\calE}d(X,\hatx)]$ is  finite.
\item[(iii)] For every $x\in\calX$, there exists an $\hatx\in\hat{\calX}$ such that $d(x,\hatx)=0$.  
\item[(iv)] The source $P$ and the distortion $d$  are  such that the minimizing test channel $W$ in the rate-distortion function in~\eqref{eqn:rd_func} is unique and we denote it as $W^*$. 
\end{itemize}
These assumptions are not overly restrictive. Indeed, the most common distortion measures and sources, such as finite alphabet sources with the Hamming distortion $d(x,\hatx)=\bbI\{x\ne\hatx\}$ and Gaussian sources with quadratic distortion $d(x,\hatx)=(x-\hatx)^2$, satisfy these assumptions.

An {\em $(M,\Delta,d, \eps)$-code} for the source $P\in\scP(\calX)$ consists   of an {\em encoder} $f:\calX\to\{1,\ldots, M\}$ and a {\em decoder} $\varphi: \{1,\ldots, M\}\to\calX$ such that the {\em   probability of excess distortion} 
\begin{equation}
P\big( \{x\in\calX: d\big(x,\varphi(f(x)) \big)  > \Delta \} \big)\le\eps.
\end{equation}
The number $M$ is called the {\em size} of the code $(f,\varphi)$. 

Given a source $P$, define the {\em lossy source coding non-asymptotic fundamental limit} as 
\begin{equation}
M^*(P, \Delta, d,\eps):= \min\big\{ M\in\bbN \,:\, \exists \mbox{ an }  (M,\Delta,d, \eps)\mbox{-code for } P\big\}.
\end{equation}
In the following subsections, we present a non-asymptotic achievability bound and a corresponding converse bound, both of which we evaluate asymptotically in the next section.

\subsection{An Achievability Bound}
The non-asymptotic achievability bound is based on Shannon's  random coding argument, and is due to Kostina-Verd\'u~\cite[Thm.~9]{kost12}. The encoder is similar to the familiar joint typicality encoder~\cite[Ch.~2]{elgamal} with typicality  defined in terms of the distortion measure. To state the   bound  compactly, define the $\Delta$-distortion ball around $x$ as
\begin{equation}
\calB_\Delta(x) :=\big\{ \hatx\in\hat{\calX} : d(x,\hatx)\le\Delta\big\}.
\end{equation}
\begin{proposition}[Random Coding Bound] \label{prop:lossy_ach}
There exists an $(M,\Delta,d, \eps)$-code satisfying
\begin{equation}
\eps\le \inf_{P_{\hatX}}\bbE_X\Big[\rme^{-MP_{\hatX} (\calB_\Delta(X))}  \Big].
\end{equation}
\end{proposition}
%The proof is short and insightful and so we present it in full. 
\begin{proof}
We use a random coding argument. Fix $P_{\hatX}\in\scP(\hat{\calX})$. Generate $M$ codewords $\hatx(m), m=1,\ldots, M$ independently according to $P_{\hatX}$.  The encoder finds an arbitrary $\hatm$ satisfying
\begin{equation}
\hatm\in\argmin_m d \big(x,\hatx(m)\big).
\end{equation}
The excess distortion probability can then be bounded as 
\begin{align}
\Pr\big( d(X,\hatX)>\Delta\big) &=\bbE_X \big[ \Pr\big( d(X,\hatX)>\Delta \,\big|\, X \big) \big] \\
&=\bbE_X \bigg[ \prod_{m=1}^M \Pr\big( d(X,\hatX(m))>\Delta \,\big|\, X \big) \bigg] \\
&=\bbE_X \bigg[ \prod_{m=1}^M \Big( 1-P_{\hatX}\big(\calB_\Delta(X(m))\big) \Big)  \bigg] \\*
&=\bbE_X \bigg[\Big( 1-P_{\hatX}\big(\calB_\Delta(X )\big) \Big)^M  \bigg] 
\end{align}
Applying the inequality $(1-x)^k\le \rme^{-kx}$ and minimizing over all possible choices of $P_{\hatX}$ completes the proof. 
\end{proof}
\subsection{A  Converse Bound}
In order to state the converse bound, we need to introduce a quantity that is fundamental to rate-distortion theory. For discrete random variables with the Hamming distortion measure ($d(x,\hatx)=\bbI\{x\ne\hatx\}$), it coincides with the self-information random variable, which, as we have seen in  Section~\ref{sec:asymp_lossless}, plays a key role in the asymptotic expansion of $\log M^*(P^n,\eps)$. 

The {\em $\Delta$-tilted information of $x$} \cite{Kot00, kost12} for a given distortion measure $d$ (whose dependence is suppressed) is defined as 
\begin{equation}
\jmath(x;P,\Delta):= -\log \bbE_{\hatX^*} \Big[ \exp\big(\lambda^* \Delta-\lambda^*d(x,\hatX^*)\big) \Big]
\end{equation}
where $\hatX^*$ is distributed as  $PW^*$ and 
\begin{equation}
\lambda^* :=- \frac{\partial R(P,\Delta')}{\partial \Delta'}\bigg|_{\Delta'=\Delta}.
\end{equation}
The differentiability of the rate-distortion function with respect to $\Delta$ is guaranteed by the assumptions in Section~\ref{sec:lossy_src}.   The term $\Delta$-tilted information  was introduced by Kostina and Verd\'u \cite{kost12}. 

\begin{example} \label{ex:gaussian_source}
Consider the Gaussian source $X\sim P(x)=\calN(x;0,\sigma^2)$ with squared-error  distortion measure $d(x,\hatx)=(x-\hatx)^2$. For this problem, simple calculations reveal that 
\begin{equation}
\jmath(x;P,\Delta)=\frac{1}{2}\log\frac{\sigma^2}{\Delta}-\bigg( \frac{x^2}{\sigma^2}-1\bigg) \frac{\log\rme}{2}
\end{equation}
if $\Delta\le\sigma^2$, and $0$ otherwise.
\end{example}
One important property of the $\Delta$-tilted information of $x$ is that the expectation is exactly equal to the rate-distortion function, i.e.,
\begin{equation}
R(P,\Delta) = \bbE_X  \big[ \jmath(X;P,\Delta) \big].
\end{equation}
For the Gaussian source with quadratic distortion, the equality above is easy to verify from Example \ref{ex:gaussian_source}.

In view of the asymptotic expansion of lossless source coding in Theorem~\ref{thm:asymp_src}, we may expect that the variance of $\jmath(X;P,\Delta)$ characterizes the second-order asymptotics of rate-distortion. This is indeed so, as we will see in the following.   Other properties of the  $\Delta$-tilted information are summarized in \cite[Lem.~1.4]{Csis74} and \cite[Properties~1~\&~2]{kost12}.

Equipped with the definition of the $\Delta$-tilted information, we are now ready to state the non-asymptotic converse bound that will turn out to be amenable to asymptotic analyses.  This elegant  bound was proved by Kostina-Verd\'u~\cite[Thm.~7]{kost12}.
\begin{proposition}[Converse Bound for Lossy Compression] \label{prop:rd_conv}
Fix $\gamma>0$. Every $(M,\Delta,d, \eps)$-code must satisfy
\begin{equation}
\eps\ge \Pr\big( \jmath(X;P,\Delta)  \ge\log M+\gamma  \big) -\exp(-\gamma). \label{eqn:dtilted_bd}
\end{equation}
\end{proposition}
Observe that this is a generalization of Proposition~\ref{thm:src_conv}  for the lossless case. In particular, it generalizes  the bound in~\eqref{eqn:sc_conv}. It is also remarkably similar to the Verd\'u-Han  information spectrum converse bound~\cite[Lem.~4]{VH94} for channel coding (reviewed in \eqref{eqn:vh} in Section~\ref{sec:ch_conv}). This is unsurprising, as channel coding and rate-distortion are duals in many ways.  We refer the reader to~\cite[Thm.~7]{kost12} for the proof of Proposition~\ref{prop:rd_conv}. 

\section{Lossy Source Coding:  Asymptotic Expansions} 
As  mentioned in the introduction of this chapter, the first-order fundamental limit for lossy source coding of stationary and memoryless sources $P^n$ is the rate distortion function $R(P,\Delta)$. We are interested in finer approximations of the non-asymptotic fundamental limit $M^*(P^n, \Delta, d^{(n)} ,\eps)$ where $P^n$ is the distribution of a stationary, memoryless source $X$ and the distortion measure $d^{(n)}:\calX^n\to\hat{\calX}^n$ is {\em separable}, i.e., 
\begin{equation}
d^{(n)}(\bx,\hat{\bx})=\frac{1 }{n}\sum_{i=1}^n d(x_i, \hatx_i). 
\end{equation}
for any $(\bx,\hat{\bx})\in\calX^n\times\hat{\calX}^n$. 

Let the variance of the $\Delta$-tilted information of $X$ be termed the {\em rate-dispersion function}
\begin{equation}
V(P,\Delta):=\var\big( \jmath(X;P,\Delta) \big). \label{eqn:rate_disp}
\end{equation}

\begin{example}
Let us revisit the Gaussian source with quadratic distortion in Example~\ref{ex:gaussian_source}.  It is easy to verify that the variance of $\jmath(X;P,\Delta)$ is
\begin{equation}
V(P,\Delta)=\frac{\log^2\rme}{2} 
\end{equation}
if $\Delta\le\sigma^2$, and $0$ otherwise.
Hence, interestingly, the rate-dispersion function for the Gaussian source with quadratic distortion depends neither on the source variance $\sigma^2$ nor the distortion $\Delta$ if $\Delta\le\sigma^2$. This is peculiar to the Gaussian source with quadratic distortion.
\end{example}
\begin{theorem} \label{thm:disp_lossy}
If $P$ and $d$ satisfy the assumptions in Section~\ref{sec:lossy_src} and, in addition,  $V(P,\Delta)>0$ and  $\bbE_{P\times PW^*}[d(X,\hatX^*)^9]<\infty$, 
\begin{equation}
\log M^*(P^n, \Delta, d^{(n)} ,\eps)=nR(P,\Delta) - \sqrt{nV(P,\Delta)} \Phi^{-1}(\eps) + O(\log n).  \label{eqn:asymp_lossy}
\end{equation}
\end{theorem}
For the case of zero rate-dispersion function $V(P,\Delta)=0$, the reader is referred to \cite[Thm.~12]{kost12}. The condition $\bbE_{P\times PW^*}[d(X,\hatX^*)^9]<\infty$ is a technical one, made to ensure that the third absolute moment of $\jmath( X;P,\Delta)$ is finite for the applicability of the Berry-Esseen theorem.
\begin{proof}[Proof sketch]
For an \iid  source $X^n$, the $\Delta$-tilted information single-letterizes because the optimum test channel in the rate-distortion formula also has the required product structure. Hence,
\begin{equation}
\jmath(X^n;P^n,\Delta)=\sum_{i=1}^n \jmath(X_i;P,\Delta).
\end{equation}
Using the Berry-Esseen theorem, the probability in  \eqref{eqn:dtilted_bd} can be lower bounded as
\begin{equation} \label{eqn:be_lossy}
\Pr\big( \jmath(X^n;P^n,\Delta)  \ge\log M+\gamma  \big)\ge \Phi\bigg( \frac{ nR(P,\Delta) - \log M -\gamma}{\sqrt{nV(P,\Delta)}}\bigg) - \frac{\kappa}{\sqrt{n}}
\end{equation}
where $\kappa$ is a function of the third absolute moment of $\jmath(X;P,\Delta)$ which is finite by the assumption that $\bbE_{P\times PW^*}[d(X,\hatX^*)^9]<\infty$. Now set $\gamma = \frac{1}{2 }\log n$ and $M$ to the  smallest integer larger than 
\begin{equation}
\exp\bigg( nR(P,\Delta) - \sqrt{nV(P,\Delta)} \Phi^{-1}\Big(\eps' - \frac{\kappa+1/2}{\sqrt{n}} \Big) -\gamma\bigg). \label{eqn:rd_cws}
\end{equation}
By the non-asymptotic converse bound in Proposition~\ref{prop:rd_conv}, we find that $\eps\ge\eps'$. This implies that the number of codewords must not be smaller than that stated in \eqref{eqn:rd_cws}, concluding the converse proof.

For the direct part, we need a technical lemma~\cite[Lem.~2]{kost12} relating the $P_{\hatX^*}^n$-probability of a $\Delta$-distortion  ball  to the $\Delta$-tilted information. 
\begin{lemma} \label{lem:tech}
There exist constants $c,b, k>0$ such that for all sufficiently large $n$,
\begin{equation}
\Pr\left( \log\frac{1}{P_{  \hatX^{*}}^n(\calB_\Delta(X^n))}  > \sum_{i=1}^n \jmath(X_i;P,\Delta) +b\log n + c \right)\le\frac{k}{\sqrt{n}}.
\end{equation}
\end{lemma}
This lemma says that we can control the $P_{  \hatX^{*}}^n$-probability of $\Delta$-distortion balls centered at a random source sequence $X^n$ in terms of the $\Delta$-tilted information. Now define the random variable
\begin{equation}
G_n := \log M - \sum_{i=1}^n \jmath(X_i;P,\Delta)-b\log n-c.
\end{equation}
Choose the distribution $P_{\hatX}$ in the non-asymptotic achievability bound in Proposition~\ref{prop:lossy_ach} to be the product distribution  $P_{  \hatX^{*}}^n$.  Applying Lemma~\ref{lem:tech}, we find that
\begin{align}
\eps' &\le \bbE\big[ \rme^{- M P_{  \hatX^{*}}^n (\calB_\Delta(X^n)) }\big] \\*
&\le \bbE\big[ \rme^{- \exp(G_n) }\big] +\frac{k}{\sqrt{n}}\\
&\le \Pr\left(G_n\le\log\frac{\ln n}{2}\right) + \frac{1}{\sqrt{n}} \Pr\left( G_n >\log\frac{\ln n}{2}\right)+\frac{k}{\sqrt{n}}.
\end{align}
where in the final step, we split the expectation into two parts depending on whether $G_n>\log \frac{\ln n}{2}$ or otherwise. Since $G_n$ is a sum of \iid random variables, the first probability can be evaluated using the Berry-Esseen theorem similarly to \eqref{eqn:be_lossy}, and the second bounded above by $1$.% This essentially completes the achievability proof.
\end{proof}
\section{Second-Order Asymptotics of Lossy Source Coding via  the Method of Types}   \label{sec:lossy_types}
In this final section of the chapter, we briefly comment on how Theorem~\ref{thm:disp_lossy} can be obtained by means of a technique that is based on the method of types. Of course, this technique only applies to discrete  (finite alphabet) sources so it is more restrictive than the Kostina-Verd\'u~\cite{kost12} method we presented. However, as with all results proved using the method of types, the analysis technique and the form of the result may be more insightful to some readers. The  exposition in this section  is due to Ingber and Kochman~\cite{ingber11}.

We make the simplifying  assumption that the rate-distortion function $R(P,\Delta)$ is differentiable with respect to $\Delta$ (guaranteed by the assumption (iv) in Section~\ref{sec:lossy_src}) and twice differentiable with respect to the probability mass function $P$. Ingber and Kochman~\cite{ingber11} considered the fundamental quantity
\begin{equation}
R'(x ; P,\Delta) := \frac{\partial R(Q,\Delta)}{\partial Q(x)} \bigg|_{Q=P}. \label{eqn:Rprime}
\end{equation}
It can be shown  \cite[Thm.~2.2]{kostina}  that $R'(x;P,\Delta)$ and the $\Delta$-tilted information are related as follows:  
\begin{equation}
R'(x;P,\Delta) =\jmath(x;P,\Delta)-\log\rme.
\end{equation}
Hence the expectation of $R'(X;P,\Delta)$ is the rate-distortion function $R(P,\Delta)$ up to a constant and its variance is exactly  the rate-dispersion function $V(P,\Delta)$ in \eqref{eqn:rate_disp}. 

A {\em codeword} $\hat{\bx}(m) \in \hat{\calX}^n$ is simply an output of the decoder $\varphi(m)$. The collection of all $M$ codewords forms the {\em codebook}. Given  a codebook $\calC=\{\hat{\bx}(1),\ldots,\hat{\bx}(M)\}$, we say that $\bx\in\calX^n$ is {\em $\Delta$-covered} by $\calC$ if there exists a  codeword $\hat{\bx}(m) \in\calC$ such that $d^{(n)}(\bx,\hat{\bx}(m))\le\Delta$. 

The analysis technique in \cite{ingber11} is based on the following lemma. 
\begin{lemma}[Type Covering] \label{lem:type_cov}
For every type $Q\in\scP_n(\calX)$, there exists a codebook $\calC:=\{\hat{\bx}(1),\ldots,\hat{\bx}(M)\} \subset\hat{\calX}^n$ of size $M$ and a function $g_1(|\calX|,|\hat{\calX}|)$ such that  every $\bx\in\calT_P$ is $\Delta$-covered by $\calC$, and 
\begin{equation}
\frac{1}{n}\log M\le R(Q,\Delta) + g_1(|\calX|,|\hat{\calX}|)\frac{\log n}{n}  \label{eqn:type_cov1}
\end{equation}
Furthermore, let the code size  $M$ and  a type $Q\in\scP_n(\calX)$ be such that  $\log M< nR(Q,\Delta)$. Then for every codebook $\calC \subset\hat{\calX}^n$  of size $M$ the fraction of $\calT_P$ that is $\Delta$-covered by $\calC$ is at most 
\begin{equation}
\exp\left(-n  R(Q,\Delta) + \log M  -g_2(|\calX|,|\hat{\calX}| )\log n\right)  \label{eqn:type_cov2}
\end{equation}
for some function $g_2(|\calX|,|\hat{\calX}|)$.
\end{lemma}
The achievability part   of the lemma in \eqref{eqn:type_cov1} is a refined version of the type covering lemma by Berger~\cite[Sec.~6.2.1, Lem.~1]{Berger}. A slightly weaker version of the lemma is also presented in Csisz\'ar-K\"orner~\cite[Ch.~9]{Csi97} and was used by Marton~\cite{Marton74} to find the error exponent for lossy source coding. The refinement comes about in the $O(\frac{\log n}{n})$ remainder term which is required for analyzing the  setting in which the excess distortion probability   is non-vanishing. The converse part in~\eqref{eqn:type_cov2}  is a corollary of Zhang-Yang-Wei~\cite[Lem.~3]{Zhang97}.

We now provide an alternative proof of Theorem \ref{thm:disp_lossy}  using the type covering  lemma. The crux of the achievability argument is to use the type covering lemma to identify a set of sequences of size $M$ such that the sequences in $\calX^n$ that it manages to $\Delta$-cover  has   probability approximately $1-\eps$ so the excess distortion probability is roughly $\eps$.  The types of  sequences in this set is denoted as $\calK$ in the  proof below. The  $P^n$-probability of $\calK$ can be estimated using  the central limit relation similar to the analysis in the proof of Theorem~\ref{thm:src_univ}. The converse argument hinges on the fact that the codebook given the achievability part of the   type covering lemma is essentially optimal in terms of its size. 
\begin{proof}[Proof sketch of Theorem~\ref{thm:disp_lossy}]
Roughly speaking, the idea in the achievability proof  is to ``encode'' all sequences in $\calX^n$ whose empirical rate distortion function is no larger than some threshold. More specifically, encode (use codes prescribed by Lemma~\ref{lem:type_cov})  sequences  belonging to
\begin{equation}
\calK:=\bigcup_{Q \in\scP_n(\calX) : R(Q,\Delta)\le \gamma}\calT_Q,
\end{equation}
where 
\begin{equation}
\gamma:=R(P,\Delta)  -\sqrt{\frac{V(P,\Delta)}{n}} \Phi^{-1}(\eps)  .
\end{equation}
By \eqref{eqn:type_cov1} and the type counting lemma, the size of $\calK$ satisfies the requirement in Theorem~\ref{thm:disp_lossy}. The resultant probability of excess distortion is
$
\Pr \left(R(P_{X^n} , \Delta) > \gamma\right) 
$ where $P_{X^n} \in\scP_n(\calX)$ is the (random) type of $X^n\in\calX^n$.  Similarly to \eqref{eqn:clt_lossles} for the lossless case, the following central limit relation holds:
\begin{equation}
\sqrt{n}\big(R (P_{X^n}, \Delta)-R(P,\Delta) \big) \stackrel{\mathrm{d}}{\longrightarrow}\calN\big(0,V(P , \Delta) \big). \label{eqn:clt_lossy}
\end{equation}
  The above convergence can be verified by  using the Berry-Esseen theorem for functions of \iid random vectors (Theorem~\ref{thm:func_clt}) per  the proof of Theorem~\ref{thm:src_univ}.  Hence,  probability of excess distortion  is roughly $\eps$ and the achievability proof is complete.

The converse part follows from the fact that that we can  lower bound the probability of  the excess distortion event $\calE_\Delta :=\{d^{(n)}(X^n,\hatX^n)>\Delta\}$  as 
\begin{equation}
\Pr \big(\calE_\Delta\big) \ge  \Pr \big(\calE_\Delta \,\big|\, R(P_{X^n},\Delta)>R +\psi_n\big) \Pr\big(R(P_{X^n},\Delta)>R+\psi_n \big) ,
\end{equation}
where $R=\frac{1}{n}\log M$ is the code rate and $\psi_n$ is arbitrary.  Now, by \eqref{eqn:type_cov2}, if the realized type of the source is $Q\in\scP_n(\calX)$ where $R(Q,\Delta)>R+\psi_n$, then the fraction of the type class  $\calT_Q$ that is $\Delta$-covered is  at most 
\begin{equation}
\exp\left(-n  R(Q,\Delta) + nR - g_2 \log n\right)\le\exp\left( -n\psi_n-g_2\log n\right).
\end{equation}
Since all sequences in a type class are equally likely (Lemma~\ref{lem:prob_seq}), the probability of  {\em no excess distortion} conditioned on the event $\{R(P_{X^n},\Delta)>R+\psi_n\}$ is at most $\frac{1}{n}$ if $\psi_n := (-g_2 + 1) \frac{\log n}{n}$. Thus
\begin{equation}
\Pr \big(\calE_\Delta\big) \ge \bigg(1-\frac{1}{n}\bigg)\Pr\big(R(P_{X^n},\Delta)>R+\psi_n \big) . \label{eqn:final_src}
\end{equation}
For $\log M= nR $ chosen to be  as~in~\eqref{eqn:asymp_lossy} in Theorem~\ref{thm:disp_lossy}, the probability on the right  is at least  $\eps - O(\frac{1}{\sqrt{n}})$ by a quantitative version of the convergence in distribution in~\eqref{eqn:clt_lossy}. 
\end{proof}

\newcommand{\av}{\mathrm{ave}}
\chapter{Channel Coding} \label{ch:cc}
This chapter presents fixed error   asymptotic results  for point-to-point channel coding, which is perhaps {\em the} most fundamental problem in information theory. Shannon~\cite{Shannon48} showed that the maximum rate of transmission  over a memoryless channel is the {\em information capacity}
\begin{equation}
C(W) = \max_{ P\in\scP(\calX)}I(P,W).  \label{eqn:cap}
\end{equation}
This first-order fundamental limit is attained as the number of channel uses (or blocklength) tends to infinity.   Wolfowitz~\cite{Wolfowitz57} showed the strong converse for a large class of memoryless channels, which intuitively means that for codes with rates above $C(W)$,  the error probability necessarily tends to one.   The contrapositive of this statement  is that, even if we allow  the error probability to be close to one (a strange requirement in practice),  one cannot send more bits per channel use than what is prescribed by the information capacity in \eqref{eqn:cap}. 

In the rest of this chapter, we revisit the problem of channel coding from the   viewpoint of the error probability being non-vanishing. First, we define the channel coding problem as well as some important non-asymptotic fundamental limits. Next we derive bounds on these limits. Some of these bounds are intimately linked to ideas  in and quantities related to binary hypothesis testing. We then evaluate these bounds for large blocklengths while keeping the error probability (either maximum or average) bounded above by some constant $\eps\in (0,1)$.  We   only concern ourselves with two classes of channels, namely the discrete memoryless channel (DMC) and the additive white Gaussian noise (AWGN) channel. We   present  second- and even third-order asymptotic expansions for the logarithm of  the non-asymptotic fundamental limits. The chapter is concluded with a discussion of source-channel transmission and the cost of separation. 

The material in this chapter on point-to-point channel coding is based primarily on the  works by Strassen~\cite{Strassen}, Hayashi~\cite{Hayashi09}, Polyanskiy-Poor-Verd\'u~\cite{PPV10},   Altu\u{g}-Wagner~\cite{altug13}, Tomamichel-Tan~\cite{TomTan12} and  Tan-Tomamichel~\cite{TanTom13}.  The material on joint source-channel coding is based on the works by Kostina-Verd\'u~\cite{kost13} and Wang-Ingber-Kochman~\cite{wang11}. 

 % \cite{mou13b}  \cite{PPV11}   \cite{Wang09} 
\section{Definitions and Non-Asymptotic Bounds}
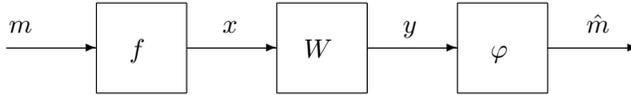
\begin{figure}[t]
\centering
\setlength{\unitlength}{.4mm}
\begin{picture}(200, 35)
%\thicklines
\put(0, 15){\vector(1, 0){30}}
\put(60, 15){\vector(1,0){30}}
\put(120, 15){\vector(1,0){30}}
\put(180, 15){\vector(1,0){30}}
\put(30, 0){\line(1, 0){30}}
\put(30, 0){\line(0,1){30}}
\put(60, 0){\line(0,1){30}}
\put(30, 30){\line(1,0){30}}

\put(90, 0){\line(1, 0){30}}
\put(90, 0){\line(0,1){30}}
\put(120, 0){\line(0,1){30}}
\put(90, 30){\line(1,0){30}}

\put(-2, 20){  $m$}
\put(69, 20){  $x$}
%\put(51, -10){  $\bbE[\rvg(X)]\le\Gamma$}
%\put(65, 8){  $[2^{nR}]$}
\put(129, 20){  $y$} 
\put(41, 12){$f$ } 
\put(99, 12){$W$} 

\put(150, 0){\line(1, 0){30}}
\put(150, 0){\line(0,1){30}}
\put(180, 0){\line(0,1){30}}
\put(150, 30){\line(1,0){30}}
\put(161, 12){$\varphi$} 
\put(190, 20){  $\hatm$} 
%\put(186, 1){  $\Pr(\hatM \ne M)$} 
  \end{picture}
  \caption{Illustration of the channel coding problem.   }
  \label{fig:cc}
\end{figure}

We now set up the channel coding problem formally. A channel is simply a stochastic map $W$ from an input alphabet $\calX$ to an output alphabet  $\calY$. For the majority of the chapter, we assume that there are no cost constraints on the codewords---the necessary changes required  for   channels with cost constraints  (such as the AWGN channel) will be pointed out.   See Fig.~\ref{fig:cc} for an illustration of the setup. 

An {\em $(M,\eps)_{\av}$-code} for the channel $W \in \scP(\calY|\calX)$ consists of a {\em message set} $\calM=\{1,\ldots,M\}$ and  pair of maps including an {\em encoder}  $f:\{1,\ldots, M\}\to\calX$ and a {\em decoder} $\varphi:\calY\to \{1,\ldots, M\}$ such that the {\em average error probability}
\begin{equation}
\frac{1}{M}\sum_{m\in\calM} W(\calY\setminus\varphi^{-1}(m) | f(m)) \le \eps. \label{eqn:av_err}
\end{equation}
An {\em $(M,\eps)_{\max}$-code}  is the same as  an   $(M,\eps)_{\av}$-code except that instead of the condition in \eqref{eqn:av_err}, the {\em maximum error probability}
\begin{equation}
\max_{m\in\calM} W(\calY\setminus\varphi^{-1}(m) | f(m)) \le \eps.
\end{equation}
The number $M$ is called the {\em size} of the code. 

We also define the following non-asymptotic fundamental limits
\begin{align}
M^*_{\av}(W,\eps)  &:=\max\big\{ M\in\bbN \,:  \exists \mbox{ an }  (M,\eps)_{\av}\mbox{-code for } W\big\} ,\,\,\mbox{and} \label{eqn:Mav}\\*
M^*_{\max}(W,\eps)  &:=\max\big\{ M\in\bbN \,: \exists \mbox{ an }  (M,\eps)_{\max}\mbox{-code for } W\big\} . \label{eqn:Mmax}
\end{align}

In the following, we will evaluate these limits when $W$ assumes some structure, for example memorylessness  and stationarity. Note that blocklength plays no role in the above definitions. In the sequel, we study the dependence of the fundamental limits on the blocklength $n$ by inserting a ``super-channel'' $W^n$ indexed by   $n$ in place of $W$ in \eqref{eqn:Mav} and \eqref{eqn:Mmax}.  Before we perform the evaluations, we   state some bounds on $M$ and $\eps$ for   arbitrary channels $W$.  
\subsection{Achievability Bounds} \label{sec:achiev}
In this section, we state  three achievability bounds. We evaluate  these bounds for memoryless channels in  the following sections. The first is Feinstein's bound \cite{Feinstein} stated in terms of the $\eps$-information spectrum divergence. 
\begin{proposition}[Feinstein's theorem] \label{prop:fein}
 Let $\eps \in (0, 1)$ and let $W$ be any channel from $\calX$ to $\calY$. Then for any $\eta\in (0,\eps)$, we have 
 \begin{align}
   \log M^*_{\max}(W,\eps)  \ge \sup_{P\in\scP(\calX)} D_\rms^{\eps-\eta}(P\times W \| P\times PW )-\log\frac{1}{\eta}. \label{eqn:fein}
 \end{align}
\end{proposition}
The proof of this bound can be found in Han's book~\cite[Lem.~3.4.1]{Han10} and uses a greedy approach for selecting codewords. The average error probability version of this bound can be proved in a more straightforward manner using threshold decoding; cf.~\cite[Thm.~1]{Han98}. The following is a slight strengthening of Feinstein's theorem. 

%Observe that Feinstein's theorem immediately leads to the direct part of Shannon's channel coding result. Indeed, replace $W$ with $n$ uses of a DMC $W^n$ and further lower bound \eqref{eqn:fein} by choosing $P\in\scP(\calX^n)$ to be $\prod_{i=1}^n P_X^*$ where $P_X^*$ is any capacity-achieving input distribution. Further choose $\eta = \exp(-\sqrt{n})$ and use the expansion for $D_\rms^{\eps-\eta}$ in Proposition \ref{prop:ds_be} to conclude that 
%\begin{equation}
%  \log M^*_{\max}(W,\eps) \ge nI(P_X^*,W)  +  O\big(\sqrt{n}\big)=nC(W)+O\big(\sqrt{n}\big).
%\end{equation}
%Normalizing by $n$ and taking the $\liminf$ proves the direct part the channel coding theorem.  In the sequel, we refine this argument by using the following   strengthening  of Feinstein's bound.
\begin{proposition} \label{prop:strengthened_fein}
There exists an $(M,\eps)_{\max}$-code for $W$  such that for any $\gamma>0$ and any input distribution $P \in\scP(\calX)$, 
\begin{equation}
\eps\le\Pr\bigg( \log\frac{W(Y|X)}{P W(Y)} \le\gamma\bigg)+ M\sup_{x\in\calX}\Pr\bigg( \log\frac{W(Y|x)}{P W(Y)}  > \gamma\bigg) ,\label{eqn:strengthed-fein}
\end{equation}
where the distribution of $(X,Y)$ is $P\times W$ in the first probability and the distribution of $Y$ is $PW$ in the second.
\end{proposition}
The proof of this bound can be found in \cite[Thm.~21]{PPV10}. It uses a sequential random coding technique where each codeword is chosen   at random based on previous choices. Feinstein's bound can be derived as  a   corollary to Proposition \ref{prop:strengthened_fein} by upper bounding the final probability in \eqref{eqn:strengthed-fein} by $\exp(-\gamma)$ and using the identification $\gamma\equiv\log\frac{1}{\eta}$. 

The previous two bounds are essentially threshold decoding bounds, i.e., we compare the likelihood ratio between the channel and the output distribution to a   threshold $\gamma$. For the average probability of error setting, one can compare the likelihood ratios of codewords directly and use  maximum likelihood decoding to obtain the following bound.  

\begin{proposition}[Random Coding Union (RCU) Bound] \label{prop:rcu}
There exists an $(M,\eps)_{\av}$-code for $W$ such that for    any input distribution $P\in\scP(\calX)$, 
\begin{equation}
\eps\le\bbE \Bigg[ \min\bigg\{ 1,M\Pr \bigg( \log\frac{W(Y|\barX)}{P W(Y)}\ge \log\frac{W(Y|X)}{P W(Y)}     \,\bigg|\,   X,Y  \bigg)\bigg\}\Bigg] 
\end{equation}
where $(X,\barX,Y)$ is distributed as $P(x)P(\barx)W(y|x)$.
\end{proposition}
The proof of this bound can be found in \cite[Thm.~16]{PPV10}. Note that the outer expectation is over $X,Y$ while the inner probability is over $\barX$. Under certain conditions on a DMC and any AWGN channel, one can use the RCU bound to prove the achievability of $\frac{1}{2}\log n+O(1)$ for the  third-order term in the asymptotic expansion of $\log M^*(W^n,\eps)$. This is what we do in the subsequent sections.
\subsection{A Converse Bound} \label{sec:ch_conv}
The only converse bound we will evaluate asymptotically appeared in different forms in the works by Verd\'u-Han~\cite[Lem.~4]{VH94}, Hayashi-Nagaoka~\cite[Lem.~4]{Hayashi03}, Polyanskiy-Poor-Verd\'u~\cite[Sec.~III-E]{PPV10}  and Tomamichel-Tan~\cite[Prop.~6]{TomTan12}. This converse bound relates channel coding to binary hypothesis testing. This relation, and its application to  asymptotic converse theorems, can be traced back to early works by     Shannon-Gallager-Berlekamp~\cite{sgb} and Wolfowitz~\cite{Wolfowitz}.    The reader is referred to Dalai's work~\cite[App.~B]{Dalai} for an excellent modern exposition on this topic.
\begin{proposition}[Symbol-Wise Converse Bound] \label{prop:converse} 
 Let $\eps \in (0, 1)$ and let $W$ be any channel from $\calX$ to $\calY$. 
  Then, for any $\eta \in (0, 1-\eps)$, we have
  \begin{align}
    \log M^*_{\av}(W,\eps)  \leq \inf_{Q \in \scP(\calY)}\, \sup_{x \in \calX}\ D_{\rms}^{\eps+\eta} ( W(\cdot|x) \| Q ) + 
    \log \frac{1}{\eta}  . \label{eqn:symbol_wise}
  \end{align}
\end{proposition}
If the codewords are constrained to belong to some set $\calA\subset\calX$ (due to cost contraints, say), the supremum above is  to be replaced by $\sup_{x\in\calA}$. 

The first part of the proof is analogous to the {\em meta-converse} in~\cite[Thm.~27]{PPV10}. See also Wang-Colbeck-Renner~\cite{Wang09} and Wang-Renner~\cite{WangRenner}, which inspired the conceptually simpler proof technique presented below. The bound in \eqref{eqn:symbol_wise} is a   ``symbol-wise''  relaxation of   the meta-converse~\cite[Thms.~28 and~31]{PPV10} and Hayashi-Nagaoka's converse~\cite[Lem.~4]{Hayashi03}. The maximization over symbols allows us to apply our converse bound on non-constant-composition codes  for DMCs directly. With an appropriate choice of $Q$, it allows us to prove a $\frac{1}{2}\log n+O(1)$ upper bound for the third-order asymptotics for positive $\eps$-dispersion DMCs (cf.~Theorem~\ref{thm:asymp_conv}). 

We remark that, in our notation, the  information spectrum converse bound in Verd\'u-Han~\cite[Lem.~4]{VH94} takes the form
\begin{equation}
 \log M^*_{\av}(W,\eps)  \leq   \sup_{P \in \scP(\calX)}\ D_{\rms}^{\eps+\eta} ( P\times W \| P\times PW ) +     \log \frac{1}{\eta}  \label{eqn:vh} 
\end{equation}
so it does not allow one to choose the output distribution $Q$. Observe the beautiful duality of the Verd\'u-Han converse with Feinstein's direct theorem.  The bound in Hayashi-Nagaoka~\cite[Lem.~4]{Hayashi03} (stated for classical-quantum channels in their context) affords this freedom and is stated as 
\begin{equation}
\log M^*_{\av}(W,\eps)  \leq   \inf_{Q\in\scP(\calY)}\sup_{P \in \scP(\calX)}\ D_{\rms}^{\eps+\eta} ( P\times W \| P\times Q ) +     \log \frac{1}{\eta}  . \label{eqn:hay_nag}
\end{equation}
Hence, we see that the bound in Proposition~\ref{prop:converse} is essentially a  ``symbol-wise''  relaxation of the Hayashi-Nagaoka converse bound~\cite[Lem.~4]{Hayashi03}  (applying Lemma~\ref{lem:Ds2}) as well as  the meta-converse theorems in~\cite[Thms.~28 and~31]{PPV10}.  

%\iffalse
%Observe that the symbol-wise converse bound immediately yields the strong converse for DMCs. Indeed, further relax the bound in  \eqref{eqn:symbol_wise} by choosing $Q\in\scP(\calY^n)$ to be $\prod_{i=1}^n (Q^*)^n$, i.e., $n$ independent copies of the unique capacity-achieving output distribution~\cite[Cor.\ 2 to Thm.\ 4.5.2]{gallagerIT}. Further choose $\eta=\exp(-\sqrt{n})$. Then the Chebyshev bound for  $D_\rms^{\eps-\eta}$ in Proposition~\ref{prop:ds_ch} yields
%\begin{align}
% \log M^*_{\av}(W,\eps) \le \sup_{\bx } D(W\| Q^*| P_{\bx} ) + O\big(\sqrt{n}\big)\le nC(W) +  O\big(\sqrt{n}\big),
%\end{align}
%where the final inequality holds because $D(W\| Q^*|P)=\sum_{x} P(x) D(W(\cdot|x) \| Q^*)\le \sum_{x} P(x) C(W)=C(W)$ for any $P$.  \fi

Since the proof of Proposition~\ref{prop:converse} is short, we provide  the details.
\begin{proof}[Proof of Proposition~\ref{prop:converse}]
Fix an $(|\calM|,\eps)_\av$-code for $W$ with message set $\calM$ and an arbitrary output distribution $Q\in\scP(\calY)$. Let $M$ and $\hatM$ be the sent message and estimated message respectively. Starting from a uniform distribution $P_M$ over $\calM$, the Markov chain 
  $M\xrightarrow{\ f\ } X \xrightarrow{\ W\ } Y \xrightarrow{\ \varphi\ } \hatM$ induces the 
  joint distribution $P_{MXY\hatM}$. 
  Due to the data-processing inequality for $D_{\rmh}^\eps$ (Lemma~\ref{lem:Dh1}),
\begin{equation}
   D_\rmh^{\eps}(P \times  W\|P \times Q) = D_\rmh^{\eps}(P_{XY} \| P_X \times Q_Y) \geq D_\rmh^{\eps}(P_{M\hatM} \| P_M \times Q_{\hatM})
  \end{equation}  
  where $P_X=P$ and $Q_{\hatM}$ is the distribution induced by $\varphi$ applied to $Q_Y = Q$.   Moreover, using the test $\delta(m,\hatm) := \bbI\{m\ne\hatm\}$, we   see that
\begin{align}
  \bbE_{P_{M\hatM}} \big[\delta(M,\hatM)\big] = \Pr(M \ne\hatM)\le\eps
\end{align}  
  where $(M,\hatM)\sim P_{M \hatM}$ above, and 
  \begin{align}
&  \bbE_{P_M \times Q_{\hatM}} \big[ \delta(M,\hatM) \big] \nn\\*
   &= \sum_{ (m,\hatm)  \in\calM \times \calM  } P_M(m) Q_{\hatM}(\hatm) \bbI\{ m \ne  \hatm\}  \\
&   = 1- \sum_{\hatm\in \calM }  Q_{\hatM} (\hatm)  \sum_{m\in\calM }  P_M(m) \bbI\{m = \hatm \}\\
&   = 1- \sum_{\hatm \in \calM }  Q_{\hatM} (\hatm)   \frac{1}{|\calM|}  \\
&   = 1-   \frac{1}{|\calM|}  .
  \end{align}
  Hence, $D_\rmh^{\eps}(P_{M\hatM} \| P_M \times Q_{\hatM}) \geq \log | \calM | + \log (1 -\eps)$ per the definition of the $\eps$-hypothesis testing divergence.
  Finally, applying Lemmas~\ref{lem:Ds1} and~\ref{lem:Ds2} yields
  \begin{align}
    \sup_{x \in \calX} D_\rms^{\eps+\eta}\big(W(\cdot|x) \big\| Q\big) &\geq 
    D_\rms^{\eps+\eta}\big(P \times W \big\| P \times Q\big) \\*
    &\geq D_{\rmh}^{\eps}\big(P \times W \big\| P \times Q\big) - \log \frac{1 - \eps}{\eta}\\*
&    \geq \log |\calM|  -
    \log \frac{1}{\eta} .
  \end{align}
 This yields the converse bound upon minimizing over $Q \in \scP(\calY)$. 
\end{proof}
\section{Asymptotic Expansions for Discrete Memoryless Channels}   
In this section, we consider asymptotic expansions for DMCs. Recall that a DMC (without feedback) for blocklength $n$ is a channel $W^n\in\scP(\calY^n|\calX^n)$ where the input and output alphabets are finite and the channel law satisfies 
\begin{equation}
W^n(\by|\bx) =\prod_{i=1}^n W(y_i|x_i),\qquad\forall\, (\bx,\by)\in\calX^n\times\calY^n.
\end{equation}
%Our aim in this section is to obtain asymptotic expansions of $\log M^*_{\av}(W^n,\eps)$ and $\log M^*_{\max}(W^n,\eps)$ for fixed $\eps\in (0,1)$. 
Thus, the channel behaves in a stationary and memoryless manner. 
Shannon~\cite{Shannon48} found the maximum  rate of reliable communication over a DMC and termed this rate the {\em capacity}  $C(W)$ given in \eqref{eqn:cap}.  In this section, we derive refinements of this fundamental limit  of communication  by characterizing the first three terms in the  asymptotic expansions of $\log M^*_\av(W^n,\eps)$ and $\log M^*_{\max}(W^n,\eps)$.   Before we do so, we recall some fundamental quantities and define a few new ones.

\subsection{Definitions for Discrete Memoryless Channels}   
Recall that  the {\em conditional relative entropy} for a fixed input and output distribution pair $(P,Q)\in \scP(\calX)\times\scP(\calY)$ is $D(W\|Q|P) :=\sum_xP(x)D(W(\cdot|x)  \| Q)$.  The {\em mutual information} is $I(P, W) := D(P\times W\| P\times PW)=D( W \| PW|P)$. Moreover, $C(W)$ is the information capacity defined in \eqref{eqn:cap} and  
\begin{align}
  \Pi(W)  := \{P \in \scP(\calX)  : I(P, W)=C(W) \}  
\end{align}
is the set of \emph{capacity-achieving input distributions} (CAIDs), respectively.\footnote{We often drop the dependence on $W$ if it is clear from context.} The set of CAIDs is convex and compact in $\scP(\calX)$. 
The unique~\cite[Cor.\ 2 to Thm.\ 4.5.2]{gallagerIT}  \emph{capacity-achieving output distribution} (CAOD) is denoted as $Q^*$ and $Q^* = P W$ for all
$P \in \Pi$. Furthermore, it satisfies $Q^*(y) > 0$ for all $y \in \calY$ \cite[Cor.\ 1 to Thm.\ 4.5.2]{gallagerIT}, where we assume that all outputs are accessible.

\subsubsection{Channel Dispersions}
Recall from \eqref{eqn:VPQ} that the variance of the log-likelihood ratio $\log\frac{P}{Q}$  under $P$ is known as the  {\em divergence variance}, i.e.,
\begin{equation}
V(P\|Q):=\sum_{x\in\calX} P(x) \bigg[ \log\frac{P(x)}{Q(x)}-D(P\|Q)\bigg]^2. \label{eqn:VPQ1}
\end{equation}
We also define the {\em conditional divergence variance} $V( W \| Q | P) := \sum_x P(x) V( W(\cdot|x) \| Q)$ and   the {\em conditional information variance}
$V(P, W) := V(W \| PW | P)$. Define the {\em unconditional  information variance} $U(P,W) := V(P\times W\| P\times PW)$. Note that 
\begin{equation}
V(P, W) =  U(P,W)\label{eqn:uequalsv}
\end{equation}
for all $P \in \Pi$ \cite[Lem.\ 62]{PPV10}. This is easy to verify because from \cite[Thm.\ 4.5.1]{gallagerIT}, we know that    all $P\in\Pi$ (i.e., CAIDs) satisfy
\begin{align}
\forall\, x: P(x) >0\qquad D\big( W(\cdot|x) \| PW \big) &=C \\*
\forall\, x: P(x) =0\qquad D\big( W(\cdot|x) \| PW \big)&\le C  .
\end{align}
  The $\eps$-\emph{channel dispersion}~\cite[Def.\ 2]{PPV10} for $\eps\in (0,1)\setminus\{\frac12\}$ is the following operational quantity.
\begin{equation}
V_\eps(W) :=\liminf_{n\to\infty} \frac{1}{n}\left( \frac{\log M_{\av}^*(W^n,\eps) -  nC(W)}{\Phi^{-1}(\eps)}\right)^2 . \label{eqn:eps_disp1}
\end{equation}
This operational quantity was shown~\cite[Eq.\ (223)]{PPV10} to be equal to\footnote{Notice that for $\eps=\frac{1}{2}$, we set $V_\eps=V_{\max}$. This is somewhat unconventional; cf.~\cite[Thm.\ 48]{PPV10}. However, doing so ensures that  subsequent theorems can be  stated compactly. Nonetheless, from the viewpoint of the normal approximation,  it is immaterial how we choose $V_{\frac{1}{2}}$  since $\Phi^{-1}   (\frac{1}{2})=0$ (cf. \cite[after Eq.\ (280)]{PPV10}). } % function dependent on the channel $W$ and the error probability $\eps \in (0,1)$:
\begin{align}
  V_{\eps}(W) := \begin{cases} V_{\min} (W) & \textrm{if } \eps < \frac{1}{2} \\ V_{\max} (W) & \textrm{if } \eps \geq \frac{1}{2} \end{cases},   \label{eqn:eps_disp2}
\end{align}
where $V_{\min} (W)\!:=\!\min_{P\in\Pi}V(P,W)$ and $V_{\max} (W)\!:=\!\max_{P\in\Pi} V(P,W)$.

\subsubsection{Singularity}
The asymptotic expansions  stated in Theorems \ref{thm:asymp_ach} and \ref{thm:asymp_conv} depend on the singularity of the channel. We say a DMC $W\in\scP(\calY|\calX)$ is {\em singular} if for all $(x,y,z) \in \calX\times\calY\times\calX$ with $W(y|x) W(y|z) >0$, one has $W(y|x) = W(y|z)$. A DMC that is not singular is called  {\em non-singular}. 

Note that if the DMC is singular, then 
\begin{equation}
\log\frac{W(y|x')}{W(y|x )} \in \{-\infty,0,\infty\}
\end{equation}
for all $(x,x', y) \in\calX\times\calX\times\calY$. Intuitively, if a DMC is singular, checking {\em feasibility} is, in fact, optimum decoding. That is, given a codebook $\calC:=\{\bx(1),\ldots, \bx(M)\}$, we decide that $m \in \{1,\ldots, M\}$ is sent if, given the channel output $\by$, it uniquely satisfies
\begin{equation}
W^n\big(\by|\bx(m)\big)=\prod_{i=1}^n W\big(y_i | x_{   i}(m) \big)>0.
\end{equation}
It is known \cite{telatar} that if $W$ is singular, the capacity of $W$ equals its zero-undetected error capacity. 
% An example of a singular DMC is the binary erasure channel (BEC). 

\begin{example}
Consider the {\em binary erasure channel}  $W$ with input alphabet $\calX=\{0,1\}$ and output alphabet $\calY=\{ 0,\rme,1\}$ where $\rme$ is the erasure symbol. The channel transition probabilities of $W$ are given by
\begin{align}
W(y| 0) = \left\{ \begin{array}{cc}
1-\delta_0 & y = 0 \\
 \delta_0 & y = \rme \\
 0 & y = 1
\end{array} \right. \,\,\mbox{and}\,\,\,
W(y| 1) = \left\{ \begin{array}{cc}
0& y = 0 \\
 \delta_1 & y = \rme \\
1-\delta_1 & y = 1
\end{array} \right.  
\end{align}
If $\delta_0=\delta_1 =\delta >0$, then $W(\rme|0)W(\rme|1)>0$ and  $W(\rme|0) = W(\rme|1)=\delta$, and so the channel is singular.  If $\delta_0\ne\delta_1$, the channel is non-singular. 
\end{example}

\subsubsection{Symmetry}
We say a DMC  is {\em symmetric} \cite[pp.~94]{gallagerIT} if the channel outputs can be partitioned
into subsets such that within each subset, the matrix of transition probabilities satisfies the following: every row (resp.\ column)
is a permutation of every other row (resp.\ column).  

\subsection{Achievability Bounds: Asymptotic Expansions}
In this section, we provide lower bounds to $\log M^*_{\av}(W^n,\eps)$ and $\log M^*_{\max}(W^n,\eps)$. We focus on the positive $\eps$-dispersion case. For other cases, the reader is referred to \cite[Thm.~47]{Pol10}.

\subsubsection{Independent and Identically Distributed (i.i.d.)  Codes}
The following bounds are achieved using \iid random codes.
\begin{theorem}\label{thm:asymp_ach}
If the DMC satisfies $V_\eps(W)>0$,  
\begin{equation}
\log M^*_{\max}(W^n,\eps) \ge nC + \sqrt{nV_\eps} \Phi^{-1}(\eps)+   O(1). \label{eqn:o1_direct}
\end{equation}
If in addition, the DMC is non-singular,   
\begin{equation}
\log M^*_{\av}(W^n,\eps) \ge nC + \sqrt{nV_\eps} \Phi^{-1}(\eps)+  \frac{1}{2}\log n +  O(1). \label{eqn:logn_direct}
\end{equation}
\end{theorem}

\begin{table}
\centering
    \begin{tabular}{| c | c |  } 
    \hline
    Bound & Third-Order Term    \\ \hline
Feinstein  $+ $    Const.\ Compo.\ (Thm.~\ref{thm:const_comp}) & $-(|\calX| \!-\! \frac{1}{2})\log n  \!+\!  O(1)$ \\  \hline
Feinstein  $+ $ \iid (Rmk.\ \ref{rmk:fein})  & $-\frac{1}{2}\log n + O(1)$   \\   \hline  
Strengthened Feinstein $\!+\! $ \iid (Thm.~\ref{thm:asymp_ach}) & $ O(1)$     \\  \hline  
RCU $+ $ \iid (Thm.~\ref{thm:asymp_ach})  & $ \frac{1}{2}\log n + O(1)$    \\  \hline  
    \end{tabular}
    \caption{Comparison of the third-order terms  achievable by using various achievability bounds (in Section \ref{sec:achiev}) or requirements on the code (such as constant composition). The $\frac{1}{2}\log n+O(1)$ that is achieved by evaluating the RCU bound holds only for the class of non-singular DMCs.  }
    \label{tab:ach}
\end{table}

Theorem~\ref{thm:asymp_ach} says that asymptotically, $\log M^*_{\max}(W^n,\eps) $ is lower bounded by the Gaussian approximation $nC+\sqrt{nV_\eps} \Phi^{-1}(\eps)$ plus a constant term. In addition, under the non-singularity condition, one can say more, namely that $\log M^*_{\av}(W^n,\eps) $ is lower bounded by the Gaussian approximation plus  $\frac{1}{2}\log n+O(1)$, known as the {\em third-order term}. The  proof of the former statement in \eqref{eqn:o1_direct} uses    the strengthened version of Feinstein's theorem in Proposition~\ref{prop:strengthened_fein}, while the proof of the latter statement  in  \eqref{eqn:logn_direct} requires the use of the RCU bound in Proposition~\ref{prop:rcu}. For a comparison of the third-order terms achievable by various achievabilty bounds, the reader is referred to Table \ref{tab:ach}.  

We will only provide the proof of the former statement,   as the proof of latter is similar to the achievability proof for AWGN channels for which we show key steps  in Section~\ref{sec:awgn}. For the proof of the latter statement in~\eqref{eqn:logn_direct}, the reader is referred to \cite[Sec.~3.4.5]{Pol10}.

\begin{proof}[Proof of \eqref{eqn:o1_direct}]
We specialize the strengthened version of Feinstein's result in Proposition~\ref{prop:strengthened_fein}. Choose $P_{X^n}$ to be the $n$-fold product of a CAID $P_X^*$ that achieves $V_\eps$. The first probability in \eqref{eqn:strengthed-fein} can be bounded using the Berry-Esseen theorem as 
\begin{align}
 \Pr\bigg(\log\frac{W^n(Y^n|X^n) }{ (P_{X}^*W)^n (Y^n)} \le\gamma\bigg) &=  \Pr\bigg(\sum_{i=1}^n\log\frac{W(Y_i|X_i) }{P_{X}^*W (Y_i)} \le\gamma\bigg) \\ 
&  \le\Phi\bigg(  \frac{\gamma-nC}{\sqrt{nV_\eps}}\bigg)+ \frac{6\, \tilT}{\sqrt{n V_\eps^3}} \label{eqn:useUV}
 \end{align} 
 where $\tilT$ is the third absolute moment of  $\log W(Y|X)-\log P_{X}^*W (Y)$ and the variance is $U(P_{X}^*,W)$ which is equal to $V_\eps$  by~\eqref{eqn:uequalsv}. To bound the second probability in~\eqref{eqn:strengthed-fein}, we define 
\begin{align}
 V_x & :=V\big( W(\cdot|x) \| P_X^*W\big), \quad\mbox{and}\\*
 T_x & :=\bbE\left[ \bigg|\log\frac{W(Y|x)}{P_X^*W(Y)}- D\big( W(\cdot|x) \| P_X^* W\big) \bigg|^3\right],  
 \end{align} 
 Since the CAOD $P_X^*W$ is positive on $\calY$ \cite[Cor.\ 1 to Thm.\ 4.5.2]{gallagerIT}, $ V_{-} :=\min_{x\in\calX}V_x>0$. It can also be shown similarly to \cite[Lem.~46]{PPV10} that $T^+ :=\max_{x\in\calX}T_x<\infty$. Now,  for all $\bx \in\calX^n$, the second probability in~\eqref{eqn:strengthed-fein} can be bounded as 
 \begin{align}
&  \Pr\bigg(\log\frac{W^n(Y^n|\bx) }{(P_{X}^*W)^n (Y^n)} > \gamma\bigg) \nn\\
 &= \!  \bbE_{(P_{X }^*W)^n } \left[  \bbI\bigg\{ \log\frac{W^n(Y^n|\bx)}{ (P_{X}^*W)^n(Y^n) } > \gamma\bigg\}\right] \\
&   =\!   \bbE_{W^n(\cdot|\bx)} \left[  \exp\bigg( \!  \! -\log\frac{W^n(Y^n|\bx)}{ (P_{X}^*W)^n(Y^n) }\bigg)\bbI\bigg\{ \!  \log\frac{W^n(Y^n|\bx)}{ (P_{X}^*W)^n(Y^n) }\!>\!\gamma\bigg\}\right] \\
&\le\!   2\bigg(\frac{\log 2}{\sqrt{2\pi}} + \frac{12 \, T^+}{ V_{-}} \bigg)\frac{\exp(-\gamma)}{\sqrt{n  V_{-} }},
 \end{align}
where the final inequality is an application of Theorem~\ref{thm:str_ld}.   Now choose 
 \begin{align}
 \gamma&   :=n C+\sqrt{nV_\eps}\Phi^{-1}(\eps'), \quad\mbox{with} \\*
  \eps'  &:=\eps-\frac{1}{\sqrt{n}}  \left( \frac{2\big(\frac{\log 2}{\sqrt{2\pi}} + \frac{12 \, T^+}{ V_{-}} \big)}{\sqrt{V_{-}}} + \frac{6\, \tilT}{\sqrt{  V_\eps^3}}\right).
 \end{align}
 Also choose $M=\lfloor\exp(\gamma)\rfloor$. Substituting these choices into the above bounds  completes the proof of \eqref{eqn:o1_direct}.  
%We could  also have used Feinstein's theorem (Proposition~\ref{prop:fein}) with the choice $\eta:=\frac{1}{\sqrt{n}}$, the same input distribution as above, together with  the asymptotic expansion of $D_\rms^\eps$ (Corollary \ref{cor:ds_same}). This leads to a simpler proof but the third-order term being $-\frac{1}{2}\log n + O(1)$.% which is weaker than \eqref{eqn:o1_direct}. 
 \end{proof}

 \begin{remark} \label{rmk:fein}
 We remark that if we use Feinstein's theorem in Proposition \ref{prop:fein} (instead of its strengthened version in Proposition \ref{prop:strengthened_fein}), and the  codebook is generated in an \iid manner according to $(P_X^*)^n$, the third-order term would be $-\frac{1}{2}\log n+O(1)$.  Indeed, let $\eta$ in Feinstein's theorem   be $\frac{1}{\sqrt{n}}$. Then,  the $(\eps-\eta)$-information spectrum divergence can be expanded as 
\begin{align}
  D_\rms^{\eps-\eta}\big((P_X^*)^n \times  W^n\,\big\|\, (P_X^*)^n \times  (P_X^*  W)^n \big)    
% &\ge n I(P_X^*,W) + \sqrt{n U(P_X^*,W) } \Phi^{-1}\bigg( \eps-\eta-\frac{6 \, \tilT}{\sqrt{nU(P_X^*,W)}}\bigg) \label{eqn:use_berry}\\
   = n C + \sqrt{n V_\eps } \Phi^{-1}( \eps ) + O(1). \label{eqn:use_UV}
\end{align}
This follows the asymptotic expansion of $D_\rms^{\eps-\eta}$ (Corollary~\ref{cor:ds_same}) and the fact that $U(P_X^*,W)=V(P_X^*,W)=V_\eps(W)$  similarly to \eqref{eqn:useUV}. Coupled with the fact that $-\log\frac{1}{\eta}=-\frac{1}{2}\log n$, we see that the third-order term is (at least) $-\frac{1}{2}\log n + O(1)$.  %Note that the difference vis-\`a-vis Theorem \ref{thm:const_comp} is that  we do not incur the loss of $\log|\scP_n(\calX)|$ due to constant composition constraint and the type counting lemma.\end{remark}
 \end{remark}
 
 \subsubsection{Constant Composition Codes and Cost Constraints}
In many applications, it may not be desirable to use \iid codes as we did in the above proof. For example for channels with additive costs, each codeword $\bx(m), m=1,\ldots, M$, must satisfy
\begin{equation}
\frac{1}{n}\sum_{i=1}^n b\big(x_i (m) \big) \le \Gamma \label{eqn:cost_con}
\end{equation}
for some cost function $b:\calX\to [ 0,\infty)$ and some cost constraint $\Gamma>0$. In this case, if the type $P \in\scP_n(\calX)$ of each codeword  $\bx(m)$ is the same for all $m$ and it satisfies
\begin{equation}
\bbE_P[b(X) ]\le\Gamma,\label{eqn:cost_con2}
\end{equation}
then the cost constraint in \eqref{eqn:cost_con} is satisfied.  This class of codes is called {\em constant composition codes} of type $P$.   The Gaussian approximation can be achieved using constant composition codes. Constant composition  coding was used by Hayashi for the DMC with additive cost constraints~\cite[Thm.~3]{Hayashi09}. He then used this result to prove the second-order asymptotics for the AWGN channel~\cite[Thm.~5]{Hayashi09} by discretizing the real line increasingly finely as the blocklength grows. It is more difficult to prove conclusive results on the third-order terms using a constant composition ensemble~\cite{kost14}, nonetheless it is instructive  to understand the technique to demonstrate the achievability of the Gaussian approximation.  Let $M^*_{\max,\mathrm{cc}}(W^n,\eps)$ denote the maximum number of codewords transmissible over $W^n$ with maximum error probability $\eps$ using constant composition codes.

\begin{theorem} \label{thm:const_comp}
If the DMC satisfies $V_\eps(W)>0$,  
\begin{equation}
\log M^*_{\max,\mathrm{cc}}(W^n,\eps) \ge nC + \sqrt{nV_\eps} \Phi^{-1}(\eps) - \bigg(|\calX|-\frac{1}{2}  \bigg)\log n+   O(1). \label{eqn:o1_direct_cc}
\end{equation}
\end{theorem}

\begin{proof}[Proof sketch of Theorem~\ref{thm:const_comp}]
We use Feinstein's theorem (Proposition~\ref{prop:fein}.  Choose a type $P \in \scP_n(\calX)$ that is the closest  in the variational distance sense  to $P_X^*$ achieving $V_\eps$. By \cite[Lem.~2.1.2]{Dembo}, we know that
\begin{equation}
\big\|P-P_X^*\big\|_1\le \frac{|\calX|}{n}. \label{eqn:types_appr}
\end{equation}
Then consider the input distribution  in Feinstein's theorem to be $P_{X^n}(\bx)$,  the  uniform distribution over $\calT_P$, i.e., 
\begin{equation}
P_{X^n}(\bx) = \frac{ \bbI\{\bx\in\calT_P\}}{|\calT_P|}.
\end{equation}
Clearly such a code is constant composition. Now we claim that  \begin{align}
 P_{X^n} W^n (\by) \le |\scP_n(\calX)| (PW)^n(\by)\label{eqn:ineq_y}
\end{align}
for all $\by \in\calY^n$. To see this note that for $\bx\in\calT_P$, 
\begin{equation}
P_{X^n}(\bx) = \frac{1}{|\calT_P|}\le  |\scP_n(\calX)|\exp\big( -nH(P)\big)=|\scP_n(\calX)| P^n(\bx) \label{eqn:inequlity_x}
\end{equation}
where the inequality follows from Lemma~\ref{lem:size_type_class} and the final equality from Lemma~\ref{lem:prob_seq}.  For $\bx\notin\calT_P$,   \eqref{eqn:inequlity_x} also holds as $P_{X^n}(\bx)=0$. Multiplying \eqref{eqn:inequlity_x} by $W^n(\by|\bx)$ and summing over all $\bx$ yields \eqref{eqn:ineq_y}. Let $\tilde{\bx}$ be an arbitrary sequence in $\calT_P$, i.e., $\tilde{\bx}$ is a sequence with type $P$.  The $(\eps-\eta)$-information spectrum divergence in Feinstein's theorem  can be bounded as
\begin{align}
&D_\rms^{\eps-\eta} \big(P_{X^n}\times W^n \, \big\| \, P_{X^n} \times P_{X^n}W^n\big)  \nn\\*
&=  D_\rms^{\eps-\eta} \big(W^n(\cdot|\tilde{\bx}) \, \big\| \,  P_{X^n}W^n\big)\label{eqn:symm}  \\
&\ge  D_\rms^{\eps-\eta} \big(W^n(\cdot|\tilde{\bx}) \, \big\| \,  (PW)^n\big)  -\log|\scP_n(\calX)|\label{eqn:usey} \\
&\ge  n I(P,W)+ \sqrt{n V(P,W)}   \Phi^{-1}\bigg(\eps-\eta-\frac{6\, T(P,W)}{\sqrt{n V(P,W)^3}}\bigg)   \nn\\* 
 &\qquad\qquad\qquad\qquad - \log|\scP_n(\calX)| \label{eqn:use_be_cc}
\end{align}
%\begin{align}
%& \Pr\bigg(\log\frac{W^n(Y^n|X^n) }{ P_{X^n}W^n  (Y^n)} \le\gamma\bigg) \nn\\* 
% &=\sum_{\bx\in\calT_P} \frac{1}{|\calT_P|} \Pr\bigg(\log\frac{W^n(Y^n|X^n) }{ P_{X^n}W^n  (Y^n)} \le\gamma \,\bigg|\, X^n=\bx\bigg) \\
% &\le \sum_{\bx\in\calT_P} \frac{1}{|\calT_P|}\Pr\bigg(\log\frac{W^n(Y^n|X^n) }{( PW)^n  (Y^n)} \!\le\!\gamma \!+\! \log|\scP_n(\calX)| \,\bigg|\, X^n\!=\!\bx\bigg)   \label{eqn:usey}\\*
% &\le \Phi\bigg(  \frac{\gamma +\log|\scP_n(\calX) | -n I(P,W)}{\sqrt{nV(P,W)}}\bigg)+ \frac{6\, T(P,W)}{\sqrt{n V(P,W)^3}},\label{eqn:be_cc}
%\end{align}
where \eqref{eqn:symm} follows from permutation invariance within a type class, and  the change of output measure step in~\eqref{eqn:usey} uses the bound in~\eqref{eqn:ineq_y} as well as the   consequence of the sifting property of $D_\rms^{\eps-\eta}$ in~\eqref{eqn:pmf_dom}. Inequality \eqref{eqn:use_be_cc} uses the lower bound in the  Berry-Esseen bound on $D_\rms^{\eps-\eta}$ in Proposition \ref{prop:ds_be} with $T(P,W) := \sum_{x } P(x) T(W(\cdot|x) \| PW)$. Choose $\eta$ in Feinstein's theorem to be $\frac{1}{\sqrt{n}}$. %The proof is complete by appealing to the type counting lemma in \eqref{eqn:type_coutn2} and the following continuity properties  
In view of \eqref{eqn:types_appr}, the following continuity properties hold for  $c_1,c_2>0$:
\begin{align}
\big|I(P,W) - I(P_X^*,W) \big|  &\le  c_1 n^{-2}    ,\quad\mbox{and} \label{eqn:mi_cont} \\*
 \Big|\sqrt{ V(P,W) }-\sqrt{V(P_X^*,W)}\Big|&\le c_2 n^{-1}. \label{eqn:v_cont}
\end{align}
The bound in \eqref{eqn:mi_cont} follows because $P\mapsto I(P,W)$ behaves as a quadratic function near $P_X^*$ while \eqref{eqn:v_cont} follows from the Lipschitz-ness of $P\mapsto\sqrt{V(P,W)}$ near $P_X^*$ because $V_\eps(W)>0$.  Combining these bounds with the type counting lemma in \eqref{eqn:type_coutn2}  and Taylor expansion of~$\Phi^{-1}(\cdot)$ in~\eqref{eqn:use_be_cc} concludes the proof.
 \end{proof}
We remark that if there are additive cost constraints on the codewords, the above proof goes through almost unchanged. The leading term in the asymptotic expansion in \eqref{eqn:o1_direct_cc} would, of course, be the capacity-cost function \cite[Sec.~3.3]{elgamal}. The analogues of $V_{\min}(W)$ and $V_{\max}(W)$ that define  the $\eps$-dispersion (cf.~\eqref{eqn:eps_disp2}) would involve the maximum and minimum over the set of input distributions $P$ satisfying $\bbE_P[b(X)]\le \Gamma$. The third-order term remains unchanged. For more details, the reader is referred to~\cite{kost14}.

In fact, the Gaussian approximation can   be achieved with   constant composition codes that are also {\em partially universal}. The only statistics of the DMC we need to know are the capacity and the $\eps$-dispersion.  The idea is to compare the empirical mutual information of a codeword and the channel output $\hatI(\bx(m) \wedge \by)$ to a threshold (that depends on capacity and dispersion), similar to maximum mutual information decoding~\cite{CK81,Goppa}.  This technique was delineated in  the proof of Theorem~\ref{thm:src_univ} for   lossless source coding.  Essentially, in  channel coding, it uses the fact that if $X^n$ is uniform over the type class $\calT_P$  and $Y^n$ is the corresponding channel output, the empirical mutual information $\hatI(X^n\wedge Y^n) $ satisfies the central limit relation
\begin{equation}
\sqrt{n}\big(\hatI(X^n\wedge Y^n) - I(P,W) \big)\stackrel{\mathrm{d}}{\longrightarrow}\calN\big(0,V(P,W) \big). \label{eqn:channel_con}
\end{equation}

\subsection{Converse Bounds: Asymptotic Expansions}
The following are the strongest known asymptotic converse bounds. % Again, we focus on the case where $V_\eps(W)>0$. For the zero $\eps$-dispersion case, the reader is referred to \cite{TomTan12}.
\begin{theorem} \label{thm:asymp_conv}
If the DMC $W$ satisfies  $V_\eps(W)>0$,   
\begin{equation}
\log M^*_{\av}(W^n,\eps) \le nC + \sqrt{nV_\eps} \Phi^{-1}(\eps)+\frac{1}{2}\log n + O(1).  \label{eqn:tt}
\end{equation}
If, in addition, the DMC is symmetric and singular,  
\begin{equation}
\log M^*_{\av}(W^n,\eps) \le nC + \sqrt{nV_\eps} \Phi^{-1}(\eps)+O(1). \label{eqn:aw}
\end{equation}
\end{theorem}
The claim in \eqref{eqn:tt} is due to Tomamichel-Tan~\cite{TomTan12}, and proved concurrently by Moulin \cite{Moulin13}, while~\eqref{eqn:aw} is due to Altu\u{g}-Wagner~\cite{altug13}. The case $V_\eps(W)=0$ was also treated in Tomamichel-Tan~\cite{TomTan12} but we focus on channels with $V_\eps(W)>0$. See \cite[Fig.~1]{TomTan12} for a summary of the best known upper bounds on $\log M^*_{\av}(W^n,\eps)$ for  all classes of DMCs (regardless of the positivity of  $V_\eps(W)$).

Theorem~\ref{thm:asymp_conv} implies that $\log M^*_{\av}(W^n,\eps) $ is upper bounded by the Gaussian approximation   $nC + \sqrt{nV_\eps} \Phi^{-1}(\eps)$ plus at most $\frac{1}{2}\log n+O(1)$. In general, this cannot be improved without further assumptions on the channel because it can be shown that third-order term is $\frac{1}{2}\log n+O(1)$ for binary symmetric channels  \cite[Thm.~52]{PPV10}. In fact, for non-singular channels, Theorem~\ref{thm:asymp_ach} shows that $\frac{1}{2}\log n+O(1)$  is achievable in the third-order.  The inequality in \eqref{eqn:tt} improves on the results by Strassen~\cite[Thm.~1.2]{Strassen} and Polyanskiy-Poor-Verd\'u~\cite[Eq.~(279)]{PPV10} who showed that the third-order term is upper bounded by $(|\calX|-\frac{1}{2})\log n + O(1)$. The upper bound presented in \eqref{eqn:tt} is independent of the input alphabet $|\calX|$.

Furthermore, under the stronger condition of symmetry and singularity, the third-order term can be   tightened to  $O(1)$. In view of  the first part of Theorem~\ref{thm:asymp_ach}, the third-order term of these channels {\em is} $O(1)$. % A canonical example of a symmetric, singular channel is the BEC.

As the entire proof of Theorem~\ref{thm:asymp_conv} is rather lengthy, we will only provide a proof sketch of \eqref{eqn:tt} for $V_{\min}(W)>0$,  highlighting the key features, including a  novel construction of a net to approximate all output distributions. The following proof sketch  is still fairly long, and the reader can skip it without any essential loss of any continuity. %The main idea is a novel construction of an appropriately dense  net to approximate all output distributions. %For details, the reader is referred to \cite{altug13} and \cite{TomTan12}.

\begin{proof}[Proof sketch of \eqref{eqn:tt}]
We assume that $V_{\min}(W)>0$.  
For  DMC, the bound in Proposition~\ref{prop:converse} evaluates to
\begin{align}
\log  M^*_\av(W^n, \eps) \leq \min_{Q^{(n)}  }\ \max_{\bx \in \calX^{n}} D_\rms^{\eps+\eta}\big( W^n(\cdot|\bx) \big\|Q^{(n)}\big) + 
    \log \frac{1}{\eta}.\label{eqn:general-n}
\end{align}
In the following, we choose $\eta=\frac{1}{\sqrt{n}}$ so the $\log$ term above gives our $\frac{1}{2}\log n$. It is thus important to find a suitable choice of $Q^{(n)} \in \scP(\calY^{n})$ to further upper bound the above.  Symmetry considerations (see, e.g., \cite[Sec.~V]{Pol13}) allow us to restrict the search to distributions that are invariant under permutations of the $n$ channel uses.

Let $\zeta :=|\calY| (|\calY|- 1)$ and let $\gamma>0$. Consider the following convex combination of product distributions:
\begin{align}
  Q^{(n)}(\by) &:= \frac{1}{2} \sum_{\bk \in \mathcal{K}}
  \frac{\exp \big(- \gamma \|\bk\|_2^2 \big)}{F}
  \, \prod_{i=1}^n Q_{\bk}(y_i)  \nn\\*
   &\qquad + \frac12 \sum_{P_{\bx} \in \scP_n(\calX)} \frac{1}{|\scP_n(\calX)|} \prod_{i=1}^n P_{\bx}W(y_i),  \label{eqn:Qn}
\end{align}
where $F$ is a normalization constant that ensures $\sum_{\by}Q^{(n)}(\by) =1$, 
\begin{align}
  Q_{\bk}(y)  := Q^*(y) + \frac{k_{y}}{\sqrt{n \zeta }}, 
  \end{align}
 and the index set $\calK$ is defined as 
 \begin{align} 
  \mathcal{K} := \bigg\{ \bk=\{k_y\}_{y\in\calY} \in \mathbb{Z}^{|\calY|} \, : \, \sum_{y\in\calY} k_y = 0 ,\, k_y \geq - Q^*(y) \sqrt{n \zeta } \bigg\}.
\end{align}
See Fig.~\ref{eqn:quantize}.  The convex combination of output distributions induced by input types $(P_{\bx}W)^{  n}$ and the optimal  output distribution $(Q^*)^{n}$ (corresponding to $\bk=\bzero$) in $Q^{(n)}$ is  inspired partly  by Hayashi~\cite[Thm.~2]{Hayashi09}.  What we have done in our choice of $Q_{\bk}$ is   to uniformly quantize the simplex $\scP(\calY)$ along axis-parallel directions to form a net.   The constraint  that each $\bk$ belongs to $\mathcal{K}$ ensures that each $Q_{\bk}$ is a valid probability mass function.  It can be shown that $F<\infty$. Furthermore one can verify that  for any $Q\in\scP(\calY)$, there exists a $\bk\in\calK$ such that 
\begin{equation}
\|Q-Q_{\bk}\|_2\le\frac{1}{\sqrt{n}}
\end{equation}
so the net we have constructed is $\frac{1}{\sqrt{n}}$-dense in the $\ell_2$-norm metric.  

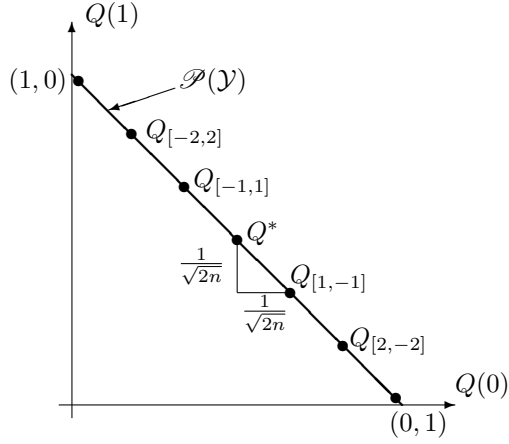
\begin{figure}
\centering
\begin{picture}(150,150)
\put(0,5){\vector(1,0){150}}
\put(5,0){\vector(0,1){150}}

\put(150,10){\mbox{$Q(0)$}}
\put(10,150){\mbox{$Q(1)$}}
\put(67.5,67.5){\circle*{4}}
\put(71,67.5){\mbox{$Q^*$}}
%\put(77.5,57.5){\circle*{4}}
\put(87.5,47.5){\circle*{4}}
%\put(97.5,37.5){\circle*{4}}
\put(107.5,27.5){\circle*{4}}
%\put(117.5,17.5){\circle*{4}}
\put(127.5,7.5){\circle*{4}}

%\put(57.5,77.5){\circle*{4}}
\put(47.5,87.5){\circle*{4}}
%\put(37.5,97.5){\circle*{4}}
\put(27.5,107.5){\circle*{4}}
%\put(17.5,117.5){\circle*{4}}
\put(7.5,127.5){\circle*{4}}

\put(67.5,67.5){\line(0,-1){20}}
\put(67.5,47.5){\line(1,0){20}}

\put(69,38){\mbox{$\frac{1}{\sqrt{2n}}$}}
\put(45,56){\mbox{$\frac{1}{\sqrt{2n}}$}}

\put(51,87.5){\mbox{$Q_{[-1,1]}$}}
\put(87.5,51){\mbox{$Q_{[1,-1]}$}}

\put(110,27.5){\mbox{$Q_{[2,-2]}$}}

\put(33,106){\mbox{$Q_{[-2,2]}$}}

\put(45,125){\mbox{$\scP(\calY)$}}
\put(44,125){\vector(-3,-1){25}}

\put(125,-5){\mbox{$(0,1)$}}
\put(-19,125){\mbox{$(1,0)$}}
\thicklines
\put(5,130){\line(1,-1){125}}
\end{picture}
\caption{Illustration of the choice of $\{Q_{\bk} \}_{\bk\in\calK}$ for $\calY=\{0,1\}$. Note that all probability distributions lie on the line $Q(0)+ Q(1)=1$ and each element of the net is denoted by $Q_{\bk}$ where $\bk$ denotes some vector with integer elements. }
\label{eqn:quantize}
\end{figure}

Let us  provide some intuition for the choice of $Q^{(n)}$. The first part of the convex
combination is used to approximate output distributions induced by input  types that are close to
the set of CAIDs $\Pi$. We choose a weight for each element of the net that drops exponentially with
the distance from the unique CAOD. This ensures that the normalization $F$ does not depend on
$n$ even though the number of elements in the net increases with $n$. The smaller weights for types
far from the CAIDs will later be compensated by the larger deviation of the corresponding mutual
information from the capacity. This is achieved by the second part of the convex combination which
we use to match the input types far from the CAIDs.  This partition of input types into those that are close and far from $\Pi$ was also used by Strassen~\cite{Strassen} in his proof of the second-order asymptotics for DMCs,

Now we just have to evaluate $ D_\rms^{\eps+\eta}\big( W^n(\cdot|\bx) \big\|Q^{(n)}\big)$ for all $\bx\in\calX^n$. The idea is to partition input sequences depending on their distance from the set of CAIDs. For this define
\begin{equation}
\Pi_\mu:=\Big\{  P\in\scP(\calX)\, :\, \min_{P^*\in \Pi} \| P-P^*\|_2\le\mu \Big\}
\end{equation}
for some small $\mu>0$.  The choice of $\mu$ will be made later.

For sequences not in $\Pi_\mu$, we pick $(P_{\bx}W)^n$ from the  convex combination (per Lemma~\ref{lem:Ds1}) giving 
\begin{equation}
D_\rms^{\eps+\eta}\big( W^n(\cdot|\bx) \big\|Q^{(n)}\big)\le D_\rms^{\eps+\eta}\big( W^n(\cdot|\bx) \big\| (P_{\bx}W)^n \big)+ \log \big(2|\scP_n(\calX) |\big).
\end{equation}
Next the Chebyshev type bound in Proposition~\ref{prop:ds_ch} yields
\begin{equation}
D_\rms^{\eps+\eta}\big( W^n(\cdot|\bx) \big\|Q^{(n)}\big)\le nI(P_{\bx},W) + \sqrt{ \frac{nV(P_{\bx},W)}{ 1-\eps-\eta}} + \log \big(2|\scP_n(\calX) |\big).
\end{equation}
Since $I(P_{\bx},W)\le C'<C$ (i.e., the first-order  mutual information term is strictly bounded away from capacity), $V(P_{\bx},W)$ is uniformly bounded \cite[Rmk.~3.1.1]{Han10} and the number of types is polynomial, the right-hand-side of the preceding inequality is upper bounded by $nC'+ O(\sqrt{n})$. This is smaller than the Gaussian approximation for all sufficiently large $n$ as $C'<C$.

Now for sequences   in $\Pi_\mu$, we pick $Q_{\bk(\bx)}$ from the net that is closest to $P_{\bx}W$. Per Lemma~\ref{lem:Ds1},  this gives
\begin{equation}
D_\rms^{\eps+\eta}\big( W^n(\cdot|\bx) \big\|Q^{(n)}\big)\le D_\rms^{\eps+\eta}\big( W^n(\cdot|\bx) \big\|Q_{\bk(\bx)}^n\big)+  \gamma\|\bk(\bx)\|_2^2+ \log \big(2F\big).
\end{equation}
By the Berry-Esseen-type bound in Proposition~\ref{prop:ds_be}, we have 
\begin{align}
 &D_\rms^{\eps+\eta}\big( W^n(\cdot|\bx) \big\|Q^{(n)}\big)  
   \le nD(W\| Q_{\bk(\bx)}|P_{\bx})   \nn\\*
  &+ \sqrt{nV(W\| Q_{\bk(\bx)}|P_{\bx}) } \Phi^{-1}\bigg(\eps+\frac{\kappa}{\sqrt{n}}\bigg) +  \gamma\|\bk(\bx)\|_2^2+\log\big(2F\big) \label{eqn:boundDs}
\end{align}
for some finite $\kappa >0$. By the $\frac{1}{\sqrt{n}}$-denseness of the net, the positivity of the CAOD, and the bound $D(\tilQ \|Q )\le \|\tilQ -Q \|_2^2/\min_z Q (z)$  \cite[Lem.~6.3]{Csi06} we can show that there exists a constant $q>0$ such that 
\begin{equation}
D(W\| Q_{\bk(\bx)}|P_{\bx}) \le I(P_{\bx},W) + \frac{\| P_{\bx}W-Q_{\bk(\bx)}\|_2^2}{q }\le I(P_{\bx},W)+\frac{1}{nq}. \label{eqn:boundD}
\end{equation}
Furthermore by the Lipschitz-ness of $Q\mapsto\sqrt{V(W\|Q|P)}$ which follows from the fact that $Q(y)>0$ for all $y\in\calY$, we have 
\begin{equation}
\Big|\sqrt{nV(W\| Q_{\bk(\bx)}|P_{\bx}) } -\sqrt{V(P_{\bx},W)}  \Big| \le \beta\|P_{\bx} W-Q_{\bk(\bx)}\|_2 \le \frac{\beta}{\sqrt{n}}.\label{eqn:boundV}
\end{equation}
It is known from Strassen's work \cite[Eq.~(4.41)]{Strassen} and continuity considerations that for all $P_{\bx}\in\Pi_\mu$, 
\begin{equation}
I(P_{\bx},W)\le C-\alpha\xi^2 \quad\mbox{and}\quad \Big|\sqrt{V(P_{\bx},W)}-\sqrt{V(P^*,W)}\Big|\le\beta\xi,
\end{equation}
where $P^*$ is the closest element in $\Pi$ to $P_\bx$ and $\xi$ is the corresponding Euclidean distance.  Let $\|W\|_2$ be the spectral norm of $W$. By the construction of the net, 
\begin{align}
\|\bk(\bx)\|_2&\le\sqrt{n\zeta}\|Q_{\bk(\bx)}-Q^*\|_2  \\*
 &\le \sqrt{n\zeta} \big( \|Q_{\bk(\bx)}-P_{\bx}W\|_2 + \|P_{\bx}W-Q^*\|_2\big) \\*
 &\le \sqrt{n\zeta}\bigg( \frac{1}{\sqrt{n}}+\|W\|_2\xi\bigg) \label{eqn:boundK}
\end{align}
Uniting \eqref{eqn:boundDs}, \eqref{eqn:boundD}, \eqref{eqn:boundV} and \eqref{eqn:boundK} and using some simple algebra completes the proof.
\end{proof}

As can be seen from the above proof, the net serves to approximate all possible output distributions so that, together with standard continuity arguments concerning information quantities, the remainder terms resulting from \eqref{eqn:boundD}, \eqref{eqn:boundV} and \eqref{eqn:boundK}  are all $O(1)$. 

If we had chosen the more ``natural'' output distribution 
\begin{equation}
\tilQ^{(n)}(\by)=\sum_{P_{\bx} \in \scP_n(\calX)}   \frac{1}{|\scP_n(\calX)|}   \prod_{i=1}^n P_{\bx}W(y_i) \label{eqn:Qn1}
\end{equation}
in place of $Q^{(n)}$ in \eqref{eqn:Qn}, an application of Lemma~\ref{lem:Ds1}, the type counting lemma in \eqref{eqn:type_coutn2},  and  continuity arguments shows that the third-order term would be $(|\calX|-\frac{1}{2})\log n+O(1)$. This upper bound on the third-order term was shown in the works by Strassen~\cite[Thm.~1.2]{Strassen} and Polyanskiy-Poor-Verd\'u~\cite[Eq.~(279)]{PPV10}. The choice of output distribution in \eqref{eqn:Qn1} is essentially due to Hayashi~\cite{Hayashi09}.
 
\section{Asymptotic Expansions for Gaussian Channels}   \label{sec:awgn}
In this section, we consider discrete-time additive white Gaussian noise (AWGN) channels in which
\begin{equation}
Y_i = X_i + Z_i ,
\end{equation}
for each time $i = 1,\ldots, n$. The noise  $\{Z_i\}_{i=1}^n $ is a memoryless, stationary Gaussian process with zero mean and unit variance  so the channel can be expressed as 
\begin{equation}
W(y|x)=\calN(y;x,1)=\frac{1}{\sqrt{2\pi}}\rme^{ - (y-x)^2/2}.
\end{equation}
This is perhaps the most important  and well-studied channel in communication systems.  In the case of Gaussian channels, we must impose a cost constraint on the codewords, namely for every $m$, 
\begin{equation}
\|f(m)\|_2^2  = \sum_{i=1}^n f_i(m)^2\le n\, \snr
\end{equation}
where $n$ is the blocklength, $\snr$ is the admissible power and $f_i(m)$ is the $i$-th coordinate of the $m$-th codeword. The {\em signal-to-noise  ratio} is thus $\snr$.   We use the notation $M^*_{\av}(W^n,\snr,\eps)$ to mean the maximum number of codewords transmissible over $W^n$ with average error probability and signal-to-noise ratio not exceeding $\eps\in (0,1)$ and $\snr$ respectively. We define $M^*_{\max}(W^n,\snr,\eps)$ in an analogous fashion. 

Define the {\em Gaussian capacity} and {\em Gaussian dispersion} functions as
\begin{align}
\rvC(\snr) := \frac{1}{2}\log(1+\snr),  \quad\mbox{and}\quad 
\rvV(\snr) := \log^2\rme\cdot\frac{\snr(\snr+2)}{2(\snr+1)^2} \label{eqn:gauss_disp}
\end{align}
respectively.
The  direct part of the following theorem was proved in Tan-Tomamichel~\cite{TanTom13} and the converse in Polyanskiy-Poor-Verd\'u~\cite[Thm.~54]{PPV10}. The second-order asymptotics (ignoring the third-order term) was proved  concurrently with \cite{PPV10} by Hayashi~\cite[Thm.~5]{Hayashi09}. Hayashi showed  the direct part   using the second-order asymptotics for DMCs with cost constraints (similar to Theorem~\ref{thm:const_comp}) and a quantization argument (also see~\cite{Tan_ICCS}). The converse part  was shown using the Hayashi-Nagaoka  converse bound   in \eqref{eqn:hay_nag} with the output distribution chosen to be the product CAOD. 
\begin{theorem} \label{thm:awgn_asy}
For every $\snr\in (0,\infty)$, 
\begin{equation}
\log M^*_{\av}(W^n,\snr,\eps) = n\rvC(\snr) + \sqrt{n\rvV(\snr)} \Phi^{-1}(\eps)+ \frac{1}{2}\log n +   O(1). \label{eqn:awgn_asymp}
\end{equation}
\end{theorem}
For the AWGN channel, we see that the asymptotic expansion is known exactly up to the third order under the average error setting.  The converse proof  (upper bound of \eqref{eqn:awgn_asymp}) is simple and uses a specialization of Proposition~\ref{prop:converse} with the product CAOD. 

The achievability proof is, however, more involved and uses the RCU bound  and Laplace's technique for approximating high-dimensional integrals~\cite{Shun95,Tierney86}. The main step establishes that if $X^n$ is uniform  on the   power sphere $\{\bx:  \|\bx\|_2^2=n\, \snr\}$, one has
\begin{equation}
  \Pr\big( \langle X^n,Y^n\rangle\in [b, b + \mu] \, \big| \, Y^n = \by\big)\le \kappa\cdot\frac{\mu}{\sqrt{n}}. \label{eqn:estimate}
\end{equation}
where $\kappa$ does not depend on $b\in\bbR$ and typical $\by$, i.e., $\by$ such that $\|\by\|_2^2\approx n(\snr+1)$. The estimate in \eqref{eqn:estimate} is not obvious as  the inner product $\langle X^n,Y^n\rangle$ is not a sum of independent random variables and so standard limit theorems (like those in Section~\ref{sec:prob}) cannot be employed directly.  The division by $\sqrt{n}$  gives us the $\frac{1}{2}\log n$ beyond the Gaussian approximation. 

If one is content with just the Gaussian approximation with an $O(1)$ third-order term,  one can evaluate the  so-called $\kappa \beta$-bound  \cite[Thm.~25]{PPV10}. See \cite[Thm.~54]{PPV10} for the justification. The reader is also  referred to MolavianJazi-Laneman~\cite{Mol13} for an elegant proof strategy using the central limit theorem for functions (Theorem \ref{thm:func_clt}) to prove the achievability part of Theorem \ref{thm:awgn_asy} under   the average error setting with an $O(1)$ third-order term. %This is also illustrated in Chapter~\ref{ch:ic}.

It remains an open question with regard to  whether  $\frac{1}{2}\log n+O(1)$ is achievable under the maximum error setting, i.e., whether $\log M^*_{\max}(W^n,\snr,\eps)$ is lower bounded by the expansion in~\eqref{eqn:awgn_asymp}.

%$n\to n+1$ argument

\begin{proof}
We start with the converse. By  appending   to a length-$n$ codeword (possibly  power strictly less than $\snr$) an extra $(n+1)^{\mathrm{st}}$ coordinate to equalize powers \cite[Lem.~39]{PPV10} \cite[Sec.~X]{Shannon59} (known as the $n\to n+1$ argument or the Yaglom map trick~\cite[Ch.~9, Thm.~6]{conway}), we have that 
\begin{equation}
M^*_{\av}(W^n,\snr,\eps)\le M^*_{\av,\mathrm{eq}}(W^{n+1},\snr,\eps)
\end{equation}
where $M^*_{\av,\mathrm{eq}}(W^{n },\snr,\eps)$ is similar to $M^*_{\av}(W^n,\snr,\eps)$, except that the codewords must satisfy the cost constraints with {\em equality}, i.e., $\|f(m)\|_2^2 =\| \bx(m)\|_2^2=n\, \snr$. Since increasing the blocklength by $1$ does not affect the asymptotics of $\log M^*_{\av}(W^n,\snr,\eps)$,  we may as well assume   that all codewords satisfy the cost constraints with equality.  By Proposition~\ref{prop:converse} applied to $n$ uses of the AWGN channel, we have
\begin{equation}
\!\log \! M^*_{\av}(W^n,\snr,\eps)\!\le\!\inf_{Q^{(n)}  }\sup_{\|\bx \|_2^2=n\, \snr} D_\rms^{\eps+\eta}\big( W^n(\cdot|\bx) \| Q^{(n)}\big) +\log\frac{1}{\eta}. \!\!
\end{equation}
Take $\eta=\frac{1}{\sqrt{n}}$ so the final $\log$ term gives   $\frac{1}{2}\log n$. It remains to show that the $(\eps+\eta)$-information spectrum divergence term is upper bounded by the Gaussian approximation plus at most a constant term.

For this purpose, we have to choose the output distribution  $Q^{(n)}\in\scP(\bbR^n)$. This choice is easy compared to the DMC case. We choose
\begin{equation}
Q^{(n)}(\by) = \prod_{i=1}^nQ_Y^*(y_i),\quad\mbox{where}\quad Q_Y^*(y)=\calN(y ;0,1+\snr). \label{eqn:output_ga}
\end{equation}
One can then check that  for every $\bx \in\bbR^n$ such that $\|\bx\|_2^2=n\, \snr$,
\begin{align}
\bbE\bigg[\frac{1}{n}\sum_{i=1}^n \log\frac{W(Y_i|x_i)}{Q_{Y}^*(Y_i)}\bigg] &= \rvC(\snr),\,\,\mbox{and} \label{eqn:stats1}\\*
\var\bigg[\frac{1}{n}\sum_{i=1}^n \log\frac{W(Y_i|x_i)}{Q_{Y}^*(Y_i)}\bigg]&= \frac{\rvV(\snr)}{n}.\label{eqn:stats2}
\end{align}
Then, by the Berry-Esseen-type bound in Proposition~\ref{prop:ds_be}, we have 
\begin{equation}
D_\rms^{\eps+\eta}\big( W^n(\cdot|\bx) \| Q^{(n)}\big)\le n  \rvC(\snr) +\sqrt{n \rvV(\snr)}\Phi^{-1} \bigg( \eps+\eta+ \frac{6\, T}{\sqrt{n \rvV(\snr)^3}}\bigg)
\end{equation}
where $T<\infty$ is related to    the third absolute moments of   $\log\frac{W(Y|x_i)}{Q_{Y^*}(Y)}$. A Taylor expansion of  $\Phi^{-1}(\cdot)$ concludes the proof of the converse.

Since the proof of the direct part is long, we only highlight some key ideas in the following steps. Details can be found in \cite{TanTom13}. 

Step 1: (Random coding distribution)  Consider the following input distribution to be applied to the RCU bound:
\begin{equation}
P_{X^n}(\bx)=\frac{\delta\{\|\bx\|_2^2-n\, \snr \} }{ A_n(\sqrt{n\, \snr})} \label{eqn:awgn_input}
\end{equation}
where $\delta\{\cdot\}$ is the Dirac delta and $A_n(r)=\frac{2\pi^{n/2}}{\Gamma(n/2 )}r^{n-1}$ is the surface area of a sphere of radius-$r$ in $\bbR^n$.  The power constraints are automatically satisfied with probability one. Let 
\begin{equation}
q(\bx,\by) := \log\frac{W^n(\by|\bx)}{ P_{X^n}W^n(\by)}
\end{equation}
be the log-likelihood ratio. We will take advantage of  the fact that
\begin{equation}
  q(\bx,\by) = \frac{n}{2} \log \frac{1}{2\pi} + \langle \bx, \by \rangle - n\, \snr - \|\by\|_2^2 - \log P_{X^n}W^n(\by) \label{eqn:inner_product}
\end{equation}
only depends on the codeword through the inner product $\langle \bx, \by \rangle=\sum_{i=1}^n x_i y_i$. In fact, $q(\bx,\by)$ is equal to $\langle \bx, \by \rangle$ up to a shift that only depends on $\|\by\|_2^2$.

Step 2: (RCU bound)
The RCU bound (Proposition \ref{prop:rcu}) states that there exists a blocklength-$n$ code with $M$ codewords and average error probability $\eps'$ such that 
\begin{equation}
 \eps'\le\bbE\Big[\min\big\{1,M\,\Pr\big(q(\bar{X}^n, Y^n)\ge q(X^n, Y^n)\,\big|\,X^n,Y^n\big)   \big\} \Big], \label{eqn:gauss_rcu}
 \end{equation} 
 where $(\barX^n, X^n,Y^n)\sim P_{X^n}(\bar{\bx}) P_{X^n}(\bx)W^n(\by|\bx)$.  Let 
 \begin{equation}
 g(t,\by):=\Pr\big(q(\bar{X}^n, Y^n)\ge t\,\big|\,Y^n=\by\big)   \label{eqn:gty}
 \end{equation}
 so the probability in \eqref{eqn:gauss_rcu} can be written as 
 \begin{equation}
 \Pr\big(q(\bar{X}^n, Y^n)\ge q(X^n, Y^n)\,\big|\,X^n,Y^n\big) = g(q(X^n,Y^n), Y^n). 
 \end{equation}
By using Bayes rule, we see that 
\begin{equation}
 g(t,\by) = \bbE\big[ \exp(-q(X^n,Y^n)) \bbI\{ q(X^n,Y^n) >t \} \, \big|\, Y^n=\by\big].  \label{eqn:integrating}
 \end{equation} 
 
  Step 3: (A high-probability set)
 Now, we define a set of channel outputs with high probability
\begin{equation}
 \calT:=\Big\{ \by: \frac{1}{n}\|\by\|_2^2 \in [\snr+1-\delta_n, \snr+1+\delta_n]\Big\}
 \end{equation} 
 With $\delta_n =n^{-1/3}$, it is easy to show that $P_{X^n}W^n(\calT)\ge 1-\xi_n$ where $\xi_n = \exp(-\Theta(n^{1/3}))$. 
 
Step 4: (Probability of the log-likelihood ratio belonging to an interval)
We would like to upper bound $g(t,\by)$ in \eqref{eqn:gty} to evaluate the RCU bound.  As an intermediate step,   we consider estimating %the problem of upper bounding 
\begin{equation}
\mathfrak{p}(a, \mu\, |\, \by) := \Pr\big( q(X^n,Y^n)\in [a,a+\mu] \, \big| \, Y^n= \by\big) ,
\end{equation}
where $a \in \mathbb{R}$ and $\mu > 0$ are some constants. Because $Y^n$ is fixed to some constant vector $\by$ and $\|X^n\|_2^2$ is also constant, $\mathfrak{p}(a, \mu\, |\, \by)$ can be rewritten using~\eqref{eqn:inner_product} as
\begin{equation}
\mathfrak{p}(a, \mu\, |\, \by) := \Pr\big( \langle X^n,Y^n\rangle\in [b,b + \mu] \, \big| \, Y^n = \by\big) ,  \label{eqn:inner_prod}
\end{equation}
for some other constant $b$ that depends on $a$.  So the crux of the proof boils down to understanding the behavior of the inner product $\langle X^n, Y^n\rangle=\sum_{i=1}^n X_i Y_i$ per the input distribution in \eqref{eqn:awgn_input}.
%It is clear that $h(\by; a, \mu)$ depends on $\by$ through its norm and so we may define (with an abuse of notation),
%\begin{equation}
%h(r;a,\mu) : = h(\by;a,\mu) ,\quad\mbox{if}\quad r = \frac{1}{n}\|\by\|_2^2. \label{eqn:hs}
%\end{equation}
%For now, we assume that $\by\in\calT$ or, equivalently, $r\in [ S+1 -\delta ,  S+1+\delta]$.  
The following important estimate is shown in~\cite{TanTom13} using  Laplace approximation for integrals~\cite{Shun95,Tierney86}.

\begin{lemma} \label{lem:hs}
For all large enough $n$ (depending only on $\snr$),   all $\by\in\calT$ and all $a\in\bbR$,  
\begin{equation}
\mathfrak{p}(a, \mu\, |\, \by)\le \kappa \cdot \frac{\mu}{\sqrt{n}},
\end{equation}
where $\kappa >0$ also only depends only on the power $\snr$. 
\end{lemma}
%See \cite{TanTom13} for the proof of this lemma. 

 Step 5: (Probability that the decoding metric exceeds $t$ for an incorrect codeword)
We now return to bounding $g(t,\by)$ in \eqref{eqn:gty}. Again, we assume $\by\in\calT$. The idea here is to consider the   second form of $g(t,\by)$ in \eqref{eqn:integrating} and to slice the interval $[t,\infty)$ into non-overlapping segments $\{ [t+l\mu, t+(l+1)\mu): l \in \bbN\cup\{0\}\}$ where $\mu>0$ is a constant. Then we apply Lemma~\ref{lem:hs} to each segment. This is modeled on the proof of Theorem~\ref{thm:str_ld}. Carrying out the calculations, we have 
\begin{align}
g(t,\by) &\le \sum_{l=0}^{\infty} \exp(-t -l\mu)\mathfrak{p}(t+l\mu, \mu\, |\, \by)  \label{eqn:slices} \\*
&\le\sum_{l=0}^{\infty} \exp(-t -l\mu) \cdot\kappa \cdot \frac{\mu}{\sqrt{n}}    \\*
&= \frac{\exp(-t )}{1-\exp(-\mu)}   \cdot\frac{\kappa \cdot \mu}{ \sqrt{n}}    . \label{eqn:geom}
\end{align}
Since $\mu>0$ is a free parameter, we may choose it to be $\log 2$ yielding  \begin{equation}
g(t,\by) \le (2\log 2) \, \kappa \cdot\frac{ \exp(-t )}{\sqrt{n}} =:\gamma\cdot \frac{ \exp(-t )}{\sqrt{n}}. \label{eqn:Lambda_bd}
\end{equation}
%where $\gamma = (2\log 2) \, \kappa$. % is a constant (a multiple of $\kappa$).

Step 6: (Evaluation of RCU)
We now have all the necessary ingredients to evaluate the RCU bound in \eqref{eqn:gauss_rcu}. Consider,
 \begin{align}
\eps'&\le \bbE\left[ \min\big\{1,Mg\big(q(X^n,Y^n) ,Y^n\big) \big\}\right]\\*
&\le\Pr(Y^n\in \calT^c) \nn\\*
&  \,\, +  \bbE\left[ \min\big\{1,Mg(q(X^n,Y^n) ,Y^n) \big\}\,\Big|\, Y^n\in\calT \right] \cdot \Pr(Y^n \in \calT) .
% & \le\Pr(Y^n\in \calT^c) + \bbE\left[  \min\left\{1, \frac{M G  \exp(-q(X^n,Y^n) ) }{\sqrt{n}} \right\} \,\bigg|\, Y^n\in\calT \right] \cdot \Pr(Y^n \in \calT)\label{eqn:use_integrating}\\
%  & \le\xi_n  + \bbE\left[  \min\left\{1, \frac{M\gamma\exp(-q(X^n,Y^n) ) }{\sqrt{n}} \right\} \,\bigg|\, Y^n\in\calT\right] \cdot \Pr(Y^n \in \calT)\label{eqn:use_F_bound}
  \end{align}
The first term is bounded above by $\xi_n$ and the second can be bounded above by 
\begin{align}
\bbE\left[  \min\left\{1, \frac{M \gamma  \exp(-q(X^n,Y^n) ) }{\sqrt{n}} \right\} \,\bigg|\, Y^n\in\calT \right] \cdot \Pr(Y^n \in \calT)
\end{align}
 due to~\eqref{eqn:Lambda_bd} with $t = q(X^n,Y^n)$. We split the expectation into two parts depending on whether $q(\bx,\by)> \log (M\gamma/\sqrt{n})$ or otherwise, i.e.,
\begin{align}  
& \bbE\left[  \min\left\{1, \frac{M\gamma\exp(-q(X^n,Y^n) ) }{\sqrt{n}} \right\} \,\bigg|\, Y^n\in\calT\right]  \\
&\le \Pr\left( q(X^n,Y^n) \le \log \frac{M\gamma}{\sqrt{n}} \,\bigg|\,  Y^n\in\calT\right) \nn\\*
& +\!  \frac{M\gamma}{\sqrt{n }} \bbE\left[ \bbI\left\{ q(X^n,Y^n)\!  >\!  \log\frac{M\gamma}{\sqrt{n}} \right\}\exp(-q(X^n,Y^n))  \bigg|  Y^n\!\in\!\calT\right] \label{eqn:expand_min} .
\end{align}
By applying \eqref{eqn:Lambda_bd} with $t = \log (M\gamma/\sqrt{n})$, we know that  the second term can be bounded above by $\gamma/\sqrt{n}$. 

Now let $Q_{Y}^*(y) = \calN(y;0,\snr +1)$ be the CAOD and $Q_{Y^n}^* (\by)=\prod_{i=1}^n Q_{Y}^*(y_i)$ its $n$-fold memoryless extension.   In Step 1 of the proof of Lem.~61 in~\cite{PPV10}, Polyanskiy-Poor-Verd\'u showed that  there exists a finite constant  $\zeta>0$ such that 
\begin{equation}
\sup_{\by\in\calF} \frac{  P_{X^n} W^n(\by)}{  Q_{Y^n}^*(\by)} \le \zeta .\label{eqn:change_meas}
\end{equation}
Thus, the first probability in \eqref{eqn:expand_min} multiplied by  $\Pr( Y^n\in\calT)$ can be upper bounded  using the Berry-Esseen theorem  and the statistics in \eqref{eqn:stats1}--\eqref{eqn:stats2} by
\begin{align}
\Pr\left( \log\frac{W^n(Y^n|X^n)}{Q_{Y^n}^*(Y^n)} \le \log \frac{M \gamma\zeta }{\sqrt{n}} \right)
 \le\Phi\left( \frac{\log \frac{M \gamma \zeta }{\sqrt{n}} -n\rvC(\snr )}{\sqrt{n\rvV(\snr )}} \right) + \frac{\beta}{\sqrt{n}} ,\label{eqn:be}  
\end{align}
where $\beta$ is a finite positive   constant that  depends only on $\snr$.

 Putting all the bounds together, we obtain
\begin{equation}
\eps'\le \Phi\left( \frac{\log \frac{M \gamma\zeta }{\sqrt{n}} -n\rvC(\snr )}{\sqrt{n\rvV(\snr )}} \right) + \frac{\beta}{\sqrt{n}} + \frac{\gamma}{\sqrt{n}}+\xi_n. \label{eqn:coiceM}
\end{equation}
Now choose $M$ to be the largest integer  satisfying
\begin{align}
\log M \le n\rvC(\snr ) + \sqrt{n \rvV(\snr )}& \Phi^{-1}\left( \eps - \frac{\beta+\gamma}{\sqrt{n}}-\xi_n \right)  + \frac{1}{2}\log n-\log(\gamma\zeta). \label{eqn:logM}
\end{align}
This choice ensures that $\eps'\le \eps$.   By a Taylor expansion of $\Phi^{-1}(\cdot)$, this completes the proof of the lower bound in \eqref{eqn:awgn_asymp}.
\end{proof}

\section{A Digression: Third-Order Asymptotics vs Error Exponent Prefactors} \label{sec:digress}
We conclude our discussion on fixed error asymptotics for channel coding with a final remark. We have seen from Theorems \ref{thm:asymp_ach} and \ref{thm:asymp_conv} that the third-order term in the normal approximation for DMCs is given by $\frac{1}{2}\log n + O(1)$ (resp.\ $O(1)$) for non-singular channels (resp.\ singular, symmetric channels). We have also seen from Theorem \ref{thm:awgn_asy} that the third-order term for AWGN channels is $\frac{1}{2}\log n + O(1)$. These results are summarized in Table \ref{tab:prefactor}.

\begin{table}
\centering
    \begin{tabular}{| c | c | c| }
    \hline
    Channel & Third-Order Term &  Prefactor     $\varrho_n$ \\ \hline
Non-singular, Symm.\ DMC   & $\frac{1}{2}\log n + O(1)$  & $\Theta\bigg(\displaystyle \frac{1}{n^{  (1+|E'(R)| )/2}}\bigg)$  \\[1.5ex]  \hline  
Singular, Symm.\  DMC & $O(1)$  & $\Theta\bigg(\displaystyle \frac{1}{n^{  1/2}}\bigg)$    \\[1.5ex]  \hline  
AWGN  & $\frac{1}{2}\log n + O(1)$ & $\Theta\bigg(\displaystyle \frac{1}{n^{  (1+|E'(R)| )/2}} \bigg)$  \\[1.5ex] \hline  
    \end{tabular}
    \caption{Comparison between the third-order term  in the normal approximation and prefactors in the error exponents regime $\varrho_n$ for various classes of channels. The reliability function \cite{Csi97, gallagerIT, Har08} is denoted as $E(R)$ and its derivative (if it exists) is $E'(R)$. For  the first row of the table, symmetry is   not required for the third-order term to  be equal to $\frac{1}{2}\log n + O(1)$ (cf.~\eqref{eqn:logn_direct} and~\eqref{eqn:tt}).   }
    \label{tab:prefactor}
\end{table}

In another line of study, Altu\u{g}-Wagner \cite{altug_refinement1,altug_refinement2} and Scarlett-Martinez-{Guill\'{e}n i F\`{a}bregas} \cite{Sca13} derived   {\em prefactors in the error exponents regime} for DMCs. In a nutshell, the authors were concerned with finding a sequence $\varrho_n$ such that, for high rates (i.e., rates above the critical rate),\footnote{We recall  from Section \ref{sec:nota} that $a_n\sim b_n$ iff $a_n/b_n\to 1$ as $n\to\infty$. }
\begin{equation}
 \eps^*\left(W^n,\lfloor\exp(nR) \rfloor\right) \sim  \varrho_n\cdot \exp\big(-nE(R) \big) , \label{eqn:exact_as}
  \end{equation}  
 where $\eps^*(W^n,M)$ is the smallest average error probability of a code for the channel $W^n$ with $M$ codewords, and    $E(R)$
 is the reliability function (or error exponent) of the channel \cite{Csi97,gallagerIT, Har08}. The results are also summarized in Table \ref{tab:prefactor}. For the AWGN channel, it can be verified from Shannon's work on the error exponents for the AWGN channel~\cite{Shannon59} that the prefactor is the same as that for non-singular, symmetric DMCs.  Also see the work by Wiechman and Sason~\cite{Wei08}. Table \ref{tab:prefactor} suggests that there is a correspondence between third-order terms and prefactors. A precise relation between these two fundamental quantities is an interesting avenue for future research.

%\section{Decoding with the Erasure Option and List Decoding} \cite{TanMou14}
\section{Joint Source-Channel Coding} 
We conclude our discussion on channel coding by putting together the results and techniques presented in this and the previous chapter on (lossy and lossless) source coding. We consider the fundamental problem of transmitting a memoryless source over a memoryless channel as shown in Fig.~\ref{fig:jscc}. Shannon showed~\cite{Shannon48, Shannon59b} that as long as  
\begin{equation}
\limsup_{n\to\infty}\frac{k_n}{n}< \frac{C(W)}{R(P,\Delta)}, \label{eqn:shannon_jscc}
\end{equation}
where $k_n$ is the number of  independent source symbols from $P$ and  $n$ is the number of channel uses, the probability of excess distortion   can be arbitrarily small in the limit of large blocklengths. The ratio $k_n/n$ is also known as the {\em bandwidth expansion ratio}. We summarize known fixed error probability-type results on source-channel transmission  in this section.    %The material in this section is based on the the work by Kostina-Verd\'u~\cite{kost13} and Wang-Ingber-Kochman~\cite{wang11}.  

 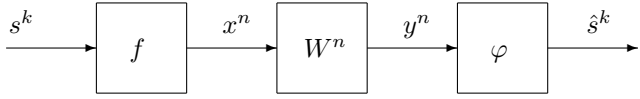
\begin{figure}[t]
\centering
\setlength{\unitlength}{.4mm}
\begin{picture}(200, 35)
%\thicklines
\put(0, 15){\vector(1, 0){30}}
\put(60, 15){\vector(1,0){30}}
\put(120, 15){\vector(1,0){30}}
\put(180, 15){\vector(1,0){30}}
\put(30, 0){\line(1, 0){30}}
\put(30, 0){\line(0,1){30}}
\put(60, 0){\line(0,1){30}}
\put(30, 30){\line(1,0){30}}

\put(90, 0){\line(1, 0){30}}
\put(90, 0){\line(0,1){30}}
\put(120, 0){\line(0,1){30}}
\put(90, 30){\line(1,0){30}}

\put(-2, 20){  $s^k$}
\put(69, 20){  $x^n$}
%\put(51, -10){  $\bbE[\rvg(X)]\le\Gamma$}
%\put(65, 8){  $[2^{nR}]$}
\put(129, 20){  $y^n$} 
\put(41, 12){$f$ } 
\put(99, 12){$W^n$} 

\put(150, 0){\line(1, 0){30}}
\put(150, 0){\line(0,1){30}}
\put(180, 0){\line(0,1){30}}
\put(150, 30){\line(1,0){30}}
\put(161, 12){$\varphi$} 
\put(190, 20){  $\hats^k$} 
%\put(186, 1){  $\Pr(\hatM \ne M)$} 
  \end{picture}
  \caption{Illustration of the joint source-channel coding problem.   }
  \label{fig:jscc}
\end{figure}
 
The  source-channel transmission problem is formally defined as follows:  A {\em $(d,\Delta, \eps)$-code} for source $S$ with distribution $P \in\scP(\calS)$ over the  channel $W \in\scP(\calY|\calX)$ is a pair of maps including an {\em encoder} $f: \calS\to\calX$ and a {\em decoder} $\varphi:\calY\to\calS$ such that the probability of excess distortion
\begin{equation}
\sum_{s\in\calS} P(s) W\big(  \{ y: d(s,\varphi(y) ) > \Delta \}\,\big|\, f(s)\big) \le\eps. 
\end{equation}
Again we assume there are no cost constraints on the channel inputs to simplify the exposition.  If there are cost constraints, a natural coding strategy would involve constant compostion codes as discussed in  Theorem \ref{thm:const_comp}.

In the conventional fixed-to-fixed length setting in which $\calX$ and $\calY$ are $n$-fold Cartesian products of the input and output alphabets respectively  and $\calS$ is the $k$-fold Cartesian product of  the source alphabet respectively, we may define the following: A {\em $(k,n,d^{(k)}, \Delta,\eps)$-code} is simply a $(d^{(k)},\Delta, \eps)$-code for the source $S^k$ with distribution $P^k \in\scP(\calS^k)$ and  over the channel $W^n \in\scP(\calY^n|\calX^n)$  such that the probability of excess distortion measure according to $d^{(k)}$ is no greater than $\eps$. 

The  source-channel non-asymptotic fundamental limit we are interested in is defined as follows:
\begin{align}
k^*(n, d^{(k)}, \Delta,\eps)\!:=\!\max\big\{k\in\bbN : \exists \mbox{ a }  (k,n,d^{(k)}, \Delta,\eps)\mbox{-code for } (P^k,W^n)\big\}.
\end{align}
This represents the maximum number of source symbols transmissible over the channel $W^n$ such that the probability of excess distortion (at distortion level $\Delta$) does not exceed $\eps$.  One is also interested in the maximum joint source-channel coding rate which is   ratio between the number of source symbols and the number of channel uses, i.e., 
\begin{equation}
R^*(n,  d^{(k)}, \Delta,\eps ):=\frac{k^*(n, d^{(k)}, \Delta,\eps)}{n}.
\end{equation}

\subsection{Asymptotic Expansion}
The main result of this section was proved independently by Kostina-Verd\'u~\cite{kost13} and  Wang-Ingber-Kochman~\cite{wang11} (for the special case of transmitting DMSes over DMCs). 
\begin{theorem} \label{thm:disp_jscc}
Assume the regularity conditions on the source  and distortion as in Theorem~\ref{thm:disp_lossy}. Assume that $W$ is a DMC with dispersion $V(W)=V_{\min}(W)=V_{\max}(W)>0$. Then, there exists a sequence of  $(k,n,d^{(k)}, \Delta,\eps)$-codes for $P^k$ and $W^n$ if and only if
\begin{align}
k R(P,\Delta) \!-\! n C(W)  \!=\!\sqrt{k V(P,\Delta) \!+\! nV(W) } \Phi^{-1}(\eps)\!+\! O(\log n). \label{eqn:disp_jscc}
\end{align}
Accordingly, by a simple rearrangement, one easily sees that
\begin{equation}
R^*(n,  d^{(k)}, \Delta,\eps )\! = \!\frac{C(W)}{ R(P,\Delta)}  \!+ \!\sqrt{ \frac{V(W,P, \Delta)}{n} }  \Phi^{-1}(\eps) \!+\! O\bigg( \frac{\log n}{n}\bigg)
\end{equation}
where the {\em rate-dispersion function} is 
\begin{equation}
V(W,P, \Delta) := \frac{R(P,\Delta)V(W) + C(W)V(P,\Delta) }{R(P,\Delta)^3}.
\end{equation}
\end{theorem}
We will not prove this theorem here, as the main ideas, based on new non-asymptotic bounds,  have been detailed in previous asymptotic expansions.%  The details can be found in \cite{kost13}. We also refer the reader to \cite{wang11} for a different proof based on the method of types.

The {\em intuition} behind the result in Theorem \ref{thm:disp_jscc} is perhaps more important. The non-asymptotic bounds that are evaluated very roughly say that a  joint source-channel coding scheme with probability of excess distortion no larger than $\eps$ exists if and only if 
\begin{equation}
\Pr \big( I_n  < J_{k,n}    \big) \le \eps \label{eqn:jscc_bd}
\end{equation}
where the random variables $I_n$ and $J_{k,n}$ are defined as
\begin{equation}
I_n := \frac{1}{n}\log\frac{W^n(Y^n|\bx)}{(P_X^*W)^n(Y^n)},\quad\mbox{and} \quad J_{k,n} := \frac{1}{n}\jmath(S^k;P^k,\Delta)
\end{equation}
and $\bx$ has type  $P \in\scP_n(\calX)$ close to $\Pi \subset\scP(\calX)$, the set of CAIDs. The bound in \eqref{eqn:jscc_bd} provides the intuition that erroneous transmission of the source occurs if and only if the information density random variable   $I_n$ of the channel is not large enough to support the information content of the source,  represented by the $\Delta$-tilted information   $J_{k,n}$.   We can now estimate the probability in \eqref{eqn:jscc_bd} by using the central limit theorem for $k+n$ independent random variables, and the fact that    $I_n-J_{k,n}$ has  first- and second-order statistics
\begin{align}
\bbE[I_n-J_{k,n}] &= C(W) -\frac{k}{n}R(P,\Delta),\quad\mbox{and}  \\*
\var[I_n-J_{k,n}] &= \frac{1}{n}V(W)+\frac{k}{n^2}V(P,\Delta).
\end{align}
This essentially explains the asymptotic expansions in Theorem~\ref{thm:disp_jscc}. 
\begin{figure}[t]
\centering
\setlength{\unitlength}{.35mm}
\begin{picture}(330, 35)
%\thicklines
\put(0, 15){\vector(1, 0){30}}
\put(60, 15){\vector(1,0){30}}
\put(120, 15){\vector(1,0){30}}
\put(180, 15){\vector(1,0){30}}
\put(30, 0){\line(1, 0){30}}
\put(30, 0){\line(0,1){30}}
\put(60, 0){\line(0,1){30}}
\put(30, 30){\line(1,0){30}}

\put(90, 0){\line(1, 0){30}}
\put(90, 0){\line(0,1){30}}
\put(120, 0){\line(0,1){30}}
\put(90, 30){\line(1,0){30}}

\put(-2, 20){  $s^k$}
\put(68, 20){  $m$}
%\put(51, -10){  $\bbE[\rvg(X)]\le\Gamma$}
%\put(65, 8){  $[2^{nR}]$}
\put(127, 20){  $x^n$} 
\put(41, 12){$f_\rms$ } 
\put(99, 12){$f_\rmc$} 

\put(150, 0){\line(1, 0){30}}
\put(150, 0){\line(0,1){30}}
\put(180, 0){\line(0,1){30}}
\put(150, 30){\line(1,0){30}}
\put(156, 12){$W^n$} 
\put(188, 20){  $y^n$} 
%\put(186, 1){  $\Pr(\hatM \ne M)$} 

\put(219, 12){$\varphi_\rmc$} 
\put(279, 12){$\varphi_\rms$} 

\put(248, 20){  $\hatm$} 

\put(310, 20){  $\hats^k$} 

\put(210, 0){\line(1, 0){30}}
\put(210, 0){\line(0,1){30}}
\put(240, 0){\line(0,1){30}}
\put(210, 30){\line(1,0){30}}
\put(240, 15){\vector(1,0){30}}

\put(270, 0){\line(1, 0){30}}
\put(270, 0){\line(0,1){30}}
\put(300, 0){\line(0,1){30}}
\put(270, 30){\line(1,0){30}}
\put(300, 15){\vector(1,0){30}}
  \end{picture}
  \caption{Illustration of the separation scheme for    source-channel transmission    }
  \label{fig:sep_jscc}
\end{figure}
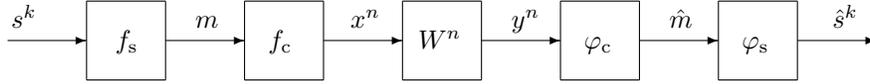

%\subsection{A Separation Architecture}
\subsection{What is the   Cost of Separation?} \label{sec:sep}
In showing the seminal result in~\eqref{eqn:shannon_jscc}, Shannon used a {\em separation} scheme. That is, he first considers   source compression to distortion level $\Delta$ using a source encoder $f_\rms$ and subsequently,   information transmission over channel $W^n$ using a channel encoder $f_\rmc$. To decode, simply reverse the process by using a channel decoder $\varphi_\rmd$ and a source decoder $\varphi_\rms$.  See Fig.~\ref{fig:sep_jscc} where $m$ denotes  the {\em digital interface}. While this idea of separation has guided the design of communication systems for decades and is  first-order optimal  in the limit of large blocklengths, it turns out that such is scheme is neither optimal from the error exponents\footnote{To be more precise, the suboptimality of separation in the error exponents  regime occurs only when $k R(P,\Delta)< n C(W)$. In the other case, by analyzing the probability of {\em no excess distortion}, Wang-Ingber-Kochman~\cite{wang12} showed, somewhat surprisingly, that separation is optimal.} perspective~\cite{Csi80} nor the  fixed error  setting.  What is the cost of separation in when the error probability is allowed to be non-vanishing?   By combining Theorem~\ref{thm:disp_lossy} (for rate distorion), Theorems~\ref{thm:asymp_ach}---\ref{thm:asymp_conv} (for channel coding), one sees that there exists a sequence of $(k,n,d^{(k)}, \Delta,\eps)$-codes for $P^k$ and $W^n$ satisfying
\begin{align} 
&k R(P,\Delta) - n C(W) +  O(\log n )  \nn\\*
& \quad \ge \max_{\eps_\rms + \eps_\rmc \le \eps}   \Big\{   \sqrt{kV (P,\Delta)}\Phi^{-1}(\eps_\rms)  +\sqrt{nV (W)}\Phi^{-1}(\eps_\rmc)   \Big\}    . \label{eqn:sep}
\end{align}
Inequality    \eqref{eqn:sep} suggests  that we first compress the source up to distortion level $\Delta$ with excess distortion probability $\eps_\rms$, then we transmit the resultant bit string  over the channel $W^n$ with average error probability $\eps_\rmc$. In order to have the  end-to-end excess distortion probability be no larger than $\eps$, one has to design the source and channel codes so that $\eps_\rms+\eps_\rmc\le\eps$. 
 
Because the maximum in \eqref{eqn:sep} is no larger than the  square root term in~\eqref{eqn:disp_jscc}, separation is strictly sub-optimal in the second-order asymptotic sense (unless either $V(W)$ or $V(P,\Delta)$ vanishes). This is unsurprising  because for the separation scheme, the source and channel error events are treated separately, while the (approximate) non-asymptotic bound in \eqref{eqn:jscc_bd} suggests that treating the system {\em jointly} results in better performances in terms of both error and rate.

%%%%%%%%%%%%%%%%%%%%%%%%

\part{Network Information Theory}

\chapter{Channels with Random State} \label{ch:state}
This chapter departs from a key assumption in   usual channel coding   (Chapter~\ref{ch:cc}) in which the channel statistics do not change with time.  In many practical communication settings, one may encounter situations where there is uncertain knowledge of  the medium of transmission, or  where the medium is changing over time, such as a wireless channel with fading or   memory with stuck-at faults.  This situation may be modeled using a channel   whose conditional output probability distribution depends on a state process. Other prominent applications include digital watermarking and information hiding~\cite{Mou03}. A thorough review of the (first-order) results in channels with state (or side information) is available  in the excellent books by Keshet, Steinberg and Merhav \cite{Keshet} and   El Gamal and Kim~\cite[Ch.~7]{elgamal}.

The state may be known at the encoder only, the decoder only, or at both the encoder and decoder. The capacity is known in these cases when the state follows an \iid process and the channel is stationary and memoryless given the state.  In this chapter, we review   known fixed error probability results for channels with   random state known only at the decoder, channels with random state known  at both the encoder and decoder, Costa's dirty-paper coding (DPC) problem~\cite{costa},  mixed channels~\cite[Sec.~3.3]{Han10} and quasi-static  single-input-multiple-output (SIMO)  fading channels.   Asymptotic expansions of the logarithm of the maximum code size  are derived for each problem.

%Conspicuously missing in our treatment in this chapter are the following: First, we do not discuss the fixed error asymptotics   of the discrete memoryless Gel'fand-Pinsker~\cite{GP80} problem where the random state is known only at the encoder. The fixed error probability asymptotics for the Gel'fand-Pinsker  problem is still unsolved but good non-asymptotic achievability bounds~\cite{WKT13,YAG13} and second-order achievability bounds \cite{Sca14} are available. Second, we also do not discuss the situation where the state is known {\em causally} at the encoder. This was  first considered by Shannon~\cite{Sha58} and the coding scheme is known as the {\em Shannon strategy}. The fixed error asymptotics  for the causal case remains unsolved. We also do not discuss other channels with state, for example, the {\em  Gilbert-Elliott channel} for which the second-order asymptotics is known~\cite{PPV11}.  Finally, we do not discuss the model in which the channel depends on a non-random state. This is also known as  the {\em compound channel} and was analyzed for the fixed error setting in~\cite{Pol13b}.

We briefly mention some problems we do not treat in this chapter. The second-order asymptotics for the discrete memoryless Gel'fand-Pinsker \cite{GP80} problem (where the state is known noncausally at the encoder only) has not been completely solved~\cite{WKT13, YAG13b} so we do not discuss this beyond the Gaussian case (the DPC problem). We also do not discuss the case where the state is known {\em causally} at the encoder. Second-order asymptotic   analysis has also not been performed for this problem   first considered by Shannon \cite{Sha58} (i.e.,  Shannon strategies).  We leave out channels with non-memoryless state, for example, the {\em  Gilbert-Elliott channel}~\cite{Elliott, Gilbert, Mush} for which the second-order asymptotics (dispersion) are known~\cite{PPV11} under various scenarios.  Finally, our focus here is on channels with a {\em random} state. We do not explore channels that depend on a non-random (but unknown) state. This is also known as  the {\em compound channel}, and the asymptotic expansion was derived by Polyanskiy~\cite{Pol13b}.  
\section{Random State at  the Decoder} 

We warm up with the simple model shown in Fig.~\ref{fig:state_d}. Here there is a  state distribution $P_S \in\scP(\calS)$  on a finite alphabet $\calS$ which generates an \iid random state $S$, i.e.,  a discrete memoryless source (DMS). The channel $W$   is a conditional probability distribution from $\calX\times\calS$ to $\calY$.   If the state process is \iid and the channel is discrete, stationary and memoryless given the state, it is easy to see that the capacity is 
\begin{equation}
C_{\mathrm{SI-D}}(W,P_S) = \max_{P\in\scP(\calX)} I(X;Y|S) = \max_{P\in\scP(\calX)} I(X;Y S). \label{eqn:state_dec}
\end{equation}
The idea is to regard $(Y, S)$ as the output of a new channel $\tilW(y,s|x) = P_S(s) W(y|x,s)$, and then to use Shannon's result for the capacity of a DMC in \eqref{eqn:cap}. Analogously to the problems we treated previously, we  define $M^*_{\mathrm{SI-D}}(W^n,P_{S^n},\eps)$ to be the maximum number of messages transmissible over the DMC $W^n$ with \iid state $S^n\sim P_{S^n}$  known at the decoder and with average error probability not exceeding $\eps\in (0,1)$. We also let $W_s(y|x) := W(y| x,s)$ denote the channel indexed by $s\in\calS$. 

\begin{figure}
\centering
\setlength{\unitlength}{.4mm}
\begin{picture}(200, 90)
%\thicklines
\put(0, 15){\vector(1, 0){30}}
\put(60, 15){\vector(1,0){30}}
\put(120, 15){\vector(1,0){30}}
\put(180, 15){\vector(1,0){30}}
\put(30, 0){\line(1, 0){30}}
\put(30, 0){\line(0,1){30}}
\put(60, 0){\line(0,1){30}}
\put(30, 30){\line(1,0){30}}

\put(90, 0){\line(1, 0){30}}
\put(90, 0){\line(0,1){30}}
\put(120, 0){\line(0,1){30}}
\put(90, 30){\line(1,0){30}}

\put(10, 20){  $m$}
\put(68, 20){  $x$}
%\put(51, -10){  $\bbE[\rvg(X)]\le\Gamma$}
%\put(65, 8){  $[2^{nR}]$}
\put(128, 20){  $y$} 
\put(41, 12){$f$ } 
\put(99, 12){$W$} 

\put(150, 0){\line(1, 0){30}}
\put(150, 0){\line(0,1){30}}
\put(180, 0){\line(0,1){30}}
\put(150, 30){\line(1,0){30}}
\put(161, 12){$\varphi$} 
\put(190, 20){  $\hatm $} 
%\put(186, 1){  $\Pr(\hatM \ne M)$} 

\put(90, 60){\line(1, 0){30}}
\put(90, 60){\line(0,1){30}}
\put(120, 60){\line(0,1){30}}
\put(90, 90){\line(1,0){30}}
%\put(120, 75){\vector(1,0){30}}
%\put(165, 60){\vector(0,-1){30}}
\put(105, 60){\vector(0,-1){30}}
\put(120, 75){\line(1,0){45}}
\put(165, 75){\vector(0,-1){45}}

%\put(150, 60){\line(1, 0){30}}
%\put(150, 60){\line(0,1){30}}
%\put(180, 60){\line(0,1){30}}
%\put(150, 90){\line(1,0){30}}

\put(105, 45){  $s$} 
\put(166, 45){  $s$} 
%\put(165, 45){  $M_{\mathrm{d}}$} 
%\put(128, 81){  $S$} 
\put(99, 71){$P_{S}$} 
%\put(157, 76){State} 
%\put(157, 67){Enc} 
  \end{picture}
  \caption{Illustration of the state at decoder problem }
  \label{fig:state_d}
\end{figure}
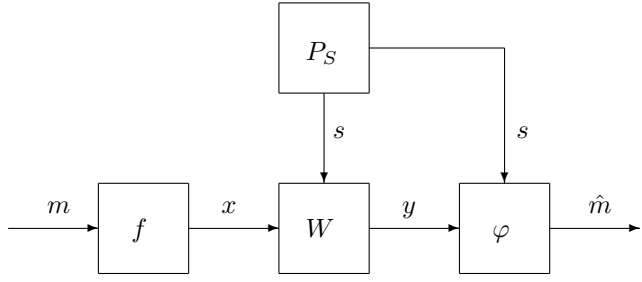

The following is due to Ingber and Feder  \cite{ingber10}.
\begin{theorem} \label{prop:state_d}
Assume that $V_\eps(W_s)>0$ for all $s\in\calS$ and $V_\eps(W_s)$ does not depend\footnote{If the CAIDs of each $W_s$ is unique, $V_\eps(W_s)$ does not depend on $\eps \in (0,1)$.} on $\eps\in (0,1)$. Then, 
\begin{align}
&\log M^*_{\mathrm{SI-D}}(W^n,P_{S^n},\eps) \nn\\*
&\quad = n C_{\mathrm{SI-D}}(W,P_S)  +\sqrt{n V_{\mathrm{SI-D}}(W,P_S) } \Phi^{-1}(\eps)   +O(\log n) ,
\end{align}
where the dispersion $V_{\mathrm{SI-D}}(W,P_S) $ is 
\begin{equation}
V_{\mathrm{SI-D}}(W,P_S)=\bbE_S[V(W_S)] + \var_S[ C(W_S) ]  \label{eqn:disp_d}
\end{equation}
and where $C(W_s)$ is the   capacity of channel $W_s \in\scP(\calY|\calX)$. 
\end{theorem}
The proof is   based on the fact that we can define a new channel $\tilW$ from  $\calX $ to $\calY\times \calS$ and so letting $X$  be a random variable whose distribution $P \in \scP(\calX)$ is a CAID, we have
\begin{align}
&V_{\mathrm{SI-D}}(W,P_S)  \nn\\*
&=\var\bigg[ \log\frac{\tilW(Y,S|X)}{P_X \tilW(Y,S)} \bigg]=\var\bigg[ \log\frac{W(Y |X,S)}{P_X W(Y|S)} \bigg] \\
&=\bbE\bigg[\var\Big[ \log\frac{W(Y |X,S)}{P_X W(Y|S)} \,\Big|\, S \Big] \bigg] + \var\bigg[\bbE \Big[ \log\frac{W(Y |X,S)}{P_X W(Y|S)} \,\Big|\, S \Big] \bigg] \label{eqn:law_var}\\
&=\bbE_S[V(W_S)] + \var[ C(W_S) ] \label{eqn:law_var2} 
\end{align}
where \eqref{eqn:law_var} follows from the law of total variance with the conditional distribution $P_X W(y|s):=\sum_x P_X(x) W(y|x,s)$, and \eqref{eqn:law_var2} follows from the definition of the capacity and dispersion of $W_s$. 

The dispersion in \eqref{eqn:disp_d} is intuitively pleasing: The term $\bbE_S[V(W_S)] $ represents the randomness  of  the channels $\{W_s: s\in\calS\}$ given the state; the term $\var_S[ C(W_S) ]$ represents the    randomness  of the state. 

\section{Random State at the Encoder and Decoder}   \label{sec:state_ed} 
The next model we will study is similar to that in the previous section. However, here the \iid state is known  noncausally at {\em both} the encoder and the decoder. See Fig.~\ref{fig:state_ed}. Again, let $W\in\scP(\calY|\calX\times\calS)$ be a state-dependent discrete memoryless channel, stationary and memoryless given the  state and let $P_S\in\scP(\calS)$ be a DMS. It is known~\cite[Sec.~7.4.1]{elgamal} that the capacity of this channel   is 
\begin{equation}
C_{\mathrm{SI-ED}}(W,P_S) = \max_{P_{X|S}\in\scP(\calX|\calS)} I(X;Y|S)  . \label{eqn:cap_ed}
\end{equation}
Goldsmith and Varaiya~\cite{Gold97} used time sharing of the state sequence to prove  the achievability part of \eqref{eqn:cap_ed}. Essentially, their  idea is to divide the message  into $|\calS|$ sub-messages (rate-splitting). Each of these sub-messages can be sent reliably if and only if its rate is smaller than $I(X;Y|S=s)$  for some $P_{X|S}(\cdot|s)$ assuming that the state sequence $S^n$ is strongly typical. Averaging $I(X;Y|S=s)$ over $P_S(s)$ proves the direct part of~\eqref{eqn:cap_ed}.  Clearly, if  there exists an optimizing distribution $P_{X|S}^*$ in \eqref{eqn:cap_ed} such that $P_{X|S}^*(\cdot|s)$ does not depend on $s$, then $C_{\mathrm{SI-ED}}(W,P_S) =C_{\mathrm{SI-D}}(W,P_S)$. For example, if the set of channels $\{W_s:s\in\calS\}$ consists of binary symmetric channels  with different crossover probabilities,  $P_{X|S}^*(\cdot|s)$ is uniform for all $s\in\calS$.

\begin{figure}
\centering
\setlength{\unitlength}{.4mm}
\begin{picture}(200, 90)
%\thicklines
\put(0, 15){\vector(1, 0){30}}
\put(60, 15){\vector(1,0){30}}
\put(120, 15){\vector(1,0){30}}
\put(180, 15){\vector(1,0){30}}
\put(30, 0){\line(1, 0){30}}
\put(30, 0){\line(0,1){30}}
\put(60, 0){\line(0,1){30}}
\put(30, 30){\line(1,0){30}}

\put(90, 0){\line(1, 0){30}}
\put(90, 0){\line(0,1){30}}
\put(120, 0){\line(0,1){30}}
\put(90, 30){\line(1,0){30}}

\put(10, 20){  $m$}
\put(68, 20){  $x$}
%\put(51, -10){  $\bbE[\rvg(X)]\le\Gamma$}
%\put(65, 8){  $[2^{nR}]$}
\put(128, 20){  $y$} 
\put(41, 12){$f$ } 
\put(99, 12){$W$} 

\put(150, 0){\line(1, 0){30}}
\put(150, 0){\line(0,1){30}}
\put(180, 0){\line(0,1){30}}
\put(150, 30){\line(1,0){30}}
\put(161, 12){$\varphi$} 
\put(190, 20){  $\hatm $} 
%\put(186, 1){  $\Pr(\hatM \ne M)$} 

\put(90, 60){\line(1, 0){30}}
\put(90, 60){\line(0,1){30}}
\put(120, 60){\line(0,1){30}}
\put(90, 90){\line(1,0){30}}
%\put(120, 75){\vector(1,0){30}}
%\put(165, 60){\vector(0,-1){30}}
\put(105, 60){\vector(0,-1){30}}
\put(120, 75){\line(1,0){45}}
\put(165, 75){\vector(0,-1){45}}

\put(90, 75){\line(-1,0){45}}
\put(45, 75){\vector(0,-1){45}}

%\put(150, 60){\line(1, 0){30}}
%\put(150, 60){\line(0,1){30}}
%\put(180, 60){\line(0,1){30}}
%\put(150, 90){\line(1,0){30}}

\put(105, 45){  $s$} 
\put(169, 45){  $s$} 
\put(45, 45){  $s$} 
%\put(165, 45){  $M_{\mathrm{d}}$} 
%\put(128, 81){  $S$} 
\put(99, 71){$P_{S}$} 
%\put(157, 76){State} 
%\put(157, 67){Enc} 
  \end{picture}
  \caption{Illustration of the state at encoder and decoder problem }
  \label{fig:state_ed}
\end{figure}
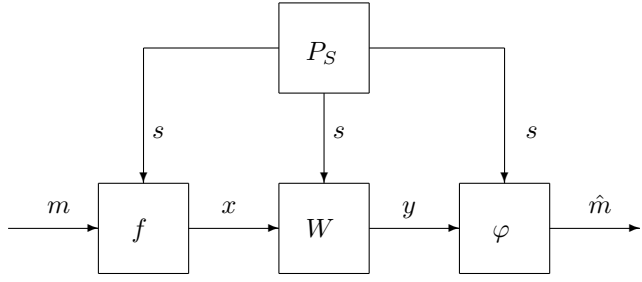

In the spirit of this monograph, let $M^*_{\mathrm{SI-ED}}(W^n,P_{S^n},\eps)$   be the maximum number of messages transmissible over the channel $W^n$ with \iid random state $S^n\sim P_{S^n}$  known at  both encoder and decoder and with average error probability not exceeding $\eps\in (0,1)$.  

The following is due to Tomamichel and Tan~\cite{TomTan13a}.

\begin{theorem} \label{thm:state_ed}
Let $W$ satisfy the assumptions in  Theorem~\ref{prop:state_d}. Then,
\begin{align}
&\log M^*_{\mathrm{SI-ED}}(W^n,P_{S^n},\eps) \nn\\*
&\quad = n C_{\mathrm{SI-ED}}(W,P_S)  +\sqrt{n V_{\mathrm{SI-ED}}(W,P_S) } \Phi^{-1}(\eps)   +O(\log n) ,
\end{align}
where the dispersion $V_{\mathrm{SI-ED}}(W,P_S)$ is the expression   given in~\eqref{eqn:disp_d}.
\end{theorem}

While the appearance of Theorem~\ref{thm:state_ed} is remarkably similar to that of Theorem~\ref{prop:state_d}, its justification is significantly more involved.  We will not provide the whole proof here as it is long but only highlight the key steps in the sketch below.  Before we do so, for  a sequence $\bs\in\calS^n$,  denote $P_{\bs} \in\scP_n(\calS)$ as its type and define
\begin{align}
\chi({\bs} ) &:= \sum_{s \in\calS} P_{\bs}(s)C(W_s) =  \frac{1}{n}\sum_{i=1}^n C(W_{s_i}) ,\quad\mbox{and} \label{eqn:emp_cap} \\ 
\nu({\bs} ) &:=   \sum_{s \in\calS} P_{\bs}(s)V(W_s) =\frac{1}{n}\sum_{i=1}^n V(W_{s_i}) 
\end{align}
to be the {\em empirical capacity} and the {\em empirical dispersion} respectively. 

\begin{proof}[Proof sketch of Theorem~\ref{thm:state_ed}]
Suppose first that the state is known to be some deterministic  sequence $\bs\in\calS^n$ of type $P_{\bs}$.  Denote the optimum error probability for a length-$n$ block code with $M$ codewords as $\eps^*(W^n,M,\bs)$. We know by a slight extension of the channel coding result  (Theorems~\ref{thm:asymp_ach} and~\ref{thm:asymp_conv}) to  memoryless but non-stationary channels that 
\begin{equation}
\eps^*(W^n,M,\bs)= \Phi\bigg( \frac{\log M -n\chi(\bs ) }{ \sqrt{n\nu(\bs )}}\bigg) + O\bigg( \frac{1}{\sqrt{n} }\bigg), \label{eqn:prob_err_state}
\end{equation}
where the implied constant in the $O(\cdot)$-notation above is uniform over all strongly typical state types $P_{\bs}$. 
The optimum error probability when the state is random and \iid is  denoted as $\eps^*(W^n,M)$ and  it can be written as the following expectation:
\begin{equation}
\eps^*(W^n,M ) = \bbE_{S^n}\big[ \eps^*(W^n,M,S^n)\big].
\end{equation}
Therefore, the analysis of the following expectation is crucial:
\begin{equation}
\bbE_{S^n} \bigg[  \Phi\bigg( \frac{\log M -n\chi({S^n} ) }{ \sqrt{n\nu( {S^n} )}}\bigg)  \bigg]. \label{eqn:expect_type}
\end{equation}
%Note that $P_{S^n}$ is the (random) type of $S^n$. 
The analysis  of \eqref{eqn:expect_type} is facilitated  by following lemmas whose proofs can be found in \cite{TomTan13a}.

\begin{lemma} \label{lem:approx_V}
The following holds uniformly in $\alpha\in\bbR$:
\begin{equation}
 \bbE \left[\Phi\left(\sqrt{n} \cdot \frac{\alpha-\chi({S^n} )}{\sqrt{\nu({S^n} )}} \right)\right]  -\bbE \left[\Phi\left(\sqrt{n}\cdot\frac{\alpha-\chi({S^n})}{\sqrt{ \bbE_S[V(W_S)] }} \right)\right] = O\left(\frac{\log n}{n}\right).
\end{equation}
\end{lemma}
This lemma says that we can essentially replace the random quantity $\nu({S^n})$ in \eqref{eqn:expect_type} with the deterministic quantity $\bbE_S [V(W_S)]$.  The next step involves approximating $\chi({S^n})$ in \eqref{eqn:expect_type} with the true capacity $C_{\mathrm{SI-ED}}(W,P_S)$.

\begin{lemma}  \label{lem:approx_sum_Vs}
The following holds uniformly in $\alpha\in\mathbb{R}$:
\begin{align}
&\bbE\left[\Phi\left(\sqrt{n}\cdot\frac{\alpha-\chi({S^n})}{\sqrt{ \bbE_S[V(W_S)] }} \right)\right]  \nn\\*
 &\quad = \Phi\left(\sqrt{n}\cdot\frac{\alpha-C_{\mathrm{SI-ED}}(W,P_S) }{\sqrt{  \bbE_S[V(W_S)] + \var_S[ C(W_S) ] }} \right) + O\left( \frac{1}{\sqrt{n}} \right). \label{eqn:approx_sum_Vs}
\end{align}
\end{lemma}
The idea behind the proof of this lemma is as follows: From \eqref{eqn:emp_cap}, one  can write $\chi({S^n})$ as an average of \iid random variables $C(W_{S_i})$. The expectation  in \eqref{eqn:approx_sum_Vs} can then be written as 
\begin{equation}
\bbE\left[\Phi\left( \sqrt{n} \cdot \frac{\alpha-  C_{\mathrm{SI-ED}}(W,P_S) }{\sqrt{\bbE_S[V(W_S)]} } + \sqrt{ \frac{\var_S[ C(W_S) ]}{\bbE_S[V(W_S)]} }  \cdot  J_n \right)\right]
\end{equation}
where
\begin{equation}
J_n := \frac{1}{\sqrt{n}} \sum_{i=1}^n E_i,\,\,\,\,\mbox{and}\,\,\,\, E_i :=  \frac{C(W_{S_i}) - C_{\mathrm{SI-ED}}(W,P_S)}{\sqrt{\var_S[ C(W_S) ] }}  
\end{equation}
Clearly, $E_i$ are zero-mean,  unit-variance, \iid random variables and thus $J_n$ converges in distribution to a standard Gaussian. Now, \eqref{eqn:approx_sum_Vs} can be established by using the fact that the convolution of two independent Gaussians is a Gaussian, where the mean and variance are the sums of the constituent means and variances.  Combining  Lemmas~\ref{lem:approx_V} and \ref{lem:approx_sum_Vs} with \eqref{eqn:prob_err_state}--\eqref{eqn:expect_type} completes the proof.
\end{proof}
Finally, we remark that by appropriate modifications to Lemmas~\ref{lem:approx_V} and~\ref{lem:approx_sum_Vs}, Theorem~\ref{thm:state_ed} can be generalized to the case where the distribution of the state sequence follows a time-homogeneous and ergodic Markov chain~\cite[Thm.~8]{TomTan13a}. 
\section{Writing on Dirty Paper} 
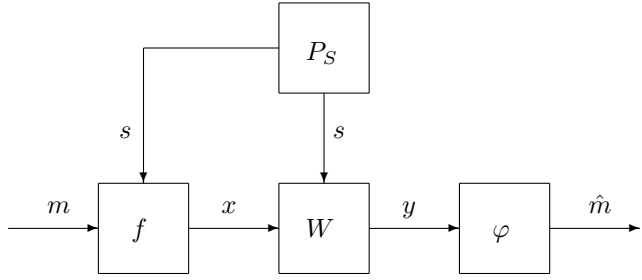
\begin{figure}
\centering
\setlength{\unitlength}{.4mm}
\begin{picture}(200, 90)
%\thicklines
\put(0, 15){\vector(1, 0){30}}
\put(60, 15){\vector(1,0){30}}
\put(120, 15){\vector(1,0){30}}
\put(180, 15){\vector(1,0){30}}
\put(30, 0){\line(1, 0){30}}
\put(30, 0){\line(0,1){30}}
\put(60, 0){\line(0,1){30}}
\put(30, 30){\line(1,0){30}}

\put(90, 0){\line(1, 0){30}}
\put(90, 0){\line(0,1){30}}
\put(120, 0){\line(0,1){30}}
\put(90, 30){\line(1,0){30}}

\put(10, 20){  $m$}
\put(68, 20){  $x$}
%\put(51, -10){  $\bbE[\rvg(X)]\le\Gamma$}
%\put(65, 8){  $[2^{nR}]$}
\put(128, 20){  $y$} 
\put(41, 12){$f$ } 
\put(99, 12){$W$} 

\put(150, 0){\line(1, 0){30}}
\put(150, 0){\line(0,1){30}}
\put(180, 0){\line(0,1){30}}
\put(150, 30){\line(1,0){30}}
\put(161, 12){$\varphi$} 
\put(190, 20){  $\hatm $} 
%\put(186, 1){  $\Pr(\hatM \ne M)$} 

\put(90, 60){\line(1, 0){30}}
\put(90, 60){\line(0,1){30}}
\put(120, 60){\line(0,1){30}}
\put(90, 90){\line(1,0){30}}
%\put(120, 75){\vector(1,0){30}}
%\put(165, 60){\vector(0,-1){30}}
\put(105, 60){\vector(0,-1){30}}
\put(90, 75){\line(-1,0){45}}
\put(45, 75){\vector(0,-1){45}}

%\put(150, 60){\line(1, 0){30}}
%\put(150, 60){\line(0,1){30}}
%\put(180, 60){\line(0,1){30}}
%\put(150, 90){\line(1,0){30}}

\put(105, 45){  $s$} 
\put(34, 45){  $s$} 
%\put(165, 45){  $M_{\mathrm{d}}$} 
%\put(128, 81){  $S$} 
\put(99, 71){$P_{S}$} 
%\put(157, 76){State} 
%\put(157, 67){Enc} 
  \end{picture}
  \caption{Illustration of the Gel'fand-Pinsker problem }
  \label{fig:state_e}
\end{figure}
Costa's ``writing on dirty paper'' result is probably one of the most surprising   in network information theory. It is a special instance of the {\em Gel'fand-Pinsker} problem \cite{GP80} whose setup is shown in Fig.~\ref{fig:state_e}.  In contrast to the previous two sections, here the state (usually assumed to be i.i.d.) is    known noncausally at the encoder.  %This is a good model for memory with stuck-at-faults. See  \cite[Example~7.3]{elgamal} and~the motivating example in Heegard-El Gamal~\cite{heegard}. 
The capacity of the Gel'fand-Pinsker channel   is
\begin{equation}
C_{\mathrm{SI-E}}(W,P_S)=\max_{P_{U|S}, f:\calU\times\calS\to\calX} I(U;Y)-I(U;S) \label{eqn:gp}
\end{equation}
where the auxiliary random variable $U$ can be constrained to have cardinality $|\calU| \le \min\{ |\calX| |\calS | , |\calY| + |\calS| +1 \}$. A strong converse was proved by Tyagi and Narayan~\cite{tyagi}.  

The Gaussian version of the problem, studied by Costa~\cite{costa}, and called {\em writing on dirty paper}, is as follows. The output of the channel $Y$ is the sum of   the channel input $X$, a Gaussian state $S\sim \calN(0,\inr)$ and independent noise $Z\sim\calN(0,1)$, i.e., 
\begin{equation}
Y_i=X_i+S_i+Z_i,\qquad\forall\, i = 1,\ldots, n.
\end{equation}
  As usual, we assume that the codeword  power is constrained to not exceed $\snr$, i.e., \begin{equation}
  \frac{1}{n}\sum_{i=1}^n X_i^2 \le \snr\label{eqn:dpc_con}
  \end{equation}
  with probability one. 
   If the state is not known at either terminal, then the capacity of the channel is 
\begin{equation}
C_{\mathrm{no-SI}}(W,P_S)=\rvC\bigg( \frac{\snr}{1+\inr}\bigg) .
\end{equation}
If the state is known at both terminals, the decoder can simply subtract it off and the channel behaves like an AWGN channel with signal-to-noise ratio $\snr$. Thus, the capacity is 
\begin{equation}
C_{\mathrm{SI-ED}}(W,P_S)=\rvC(\snr).
\end{equation}
Costa's showed the surprising result~\cite{costa} that  knowledge of the state is not required at the decoder for the capacity to be $\rvC(\snr)$! In other words,
\begin{equation}
C_{\mathrm{SI-E}}(W,P_S)=\rvC(\snr). \label{eqn:costa}
\end{equation}
%Thus, knowledge of the state is not required at the decoder for there to be no degradation to the capacity! 

The natural question, in the spirit of this monograph, is whether there is a degradation to higher-order terms in the asymptotic expansion of logarithm of the  maximum code size of the channel for a fixed average error probability (cf.~the AWGN case in Theorem~\ref{thm:awgn_asy}).  Scarlett~\cite{Scarlett14} and Jiang-Liu~\cite{jiang11} showed   the surprising result that   there is no degradation up to the second-order dispersion term! Furthermore, Scarlett \cite{Scarlett14} showed that the state sequence only has to satisfy a very  mild condition. In particular, it neither has to be   Gaussian nor ergodic.  The approach by Jiang-Liu~\cite{jiang11}  is via  lattice coding~\cite{Zamir}. The proof sketch below follows Scarlett's approach in~\cite{Scarlett14}. % The intuition behind the proof of achievability is based on principles in minimum mean squared error (MMSE) estimation.  %Set $U=X+\alpha S$ to be a linear function of $X$ and $S$. The best $\alpha$ to choose can either be done analytically by maximizing  the Gel'fand-Pinsker formula in \eqref{eqn:gp} over $\alpha$ or  by  simply observing that if one sets $\alpha^* = \frac{\snr}{\snr+\inr}$, then $\alpha^*(X+Z)$  is the MMSE estimate of $X$ given $X+Z$ and therefore the error  $X-\alpha^*(X+Z)$ is independent of $X+Z$ and also of $Y$.  

\begin{theorem} \label{thm:dpc}
Assume that there exists some  finite $\Gamma>0$ such that 
\begin{equation}
\Pr\bigg(\frac{1}{n} \|S^n\|_2^2  >  \Gamma \bigg) = O\bigg( \frac{\log n}{\sqrt{n}}\bigg). \label{eqn:state_cond} 
\end{equation}
For any $\snr \in (0,\infty)$, the maximum code size  for average error probability no larger than $\eps$ satisfies
\begin{equation}
\log M^*_{\mathrm{SI-E}}(W^n,P_{S^n},\eps)=n\rvC(\snr) + \sqrt{n\rvV(\snr)}\Phi^{-1}(\eps) + O(\log n).  \label{eqn:dpc}
\end{equation}
\end{theorem}
The condition in \eqref{eqn:state_cond} is mild. For example if $S^n$ is a zero-mean, \iid process and $\Gamma$ is chosen to be larger than $\bbE[S_1^2]$, under the condition that $\bbE[S_1^4]<\infty$, the probability decays at least as fast as $O(\frac{1}{n})$ by Chebyshev's inequality, thus satisfying \eqref{eqn:state_cond}. 

Before we sketch the proof of Theorem~\ref{thm:dpc}, let us recap Costa's proof of the DPC capacity in \eqref{eqn:costa}. He assumes $S$ is Gaussian with some variance $\inr$ and  chooses  $U=X+\alpha S$, where $X\sim\calN(0, \snr)$ and $S$ are independent. He then performs calculations which yield
\begin{align}
I(U;Y)  &= \frac{1}{2}\log\bigg( \frac{ (\snr + \inr  + 1) (\snr + \alpha^2\inr  ) }{\snr\cdot\inr  (1-\alpha)^2 + (\snr + \alpha^2 \inr )}\bigg),\\
I(U;S) &=\frac{1}{2}\log \bigg( \frac{\snr +\alpha^2\inr }{\snr}\bigg) , \quad \mbox{and}\\
I(U;Y) -I(U;S) &=\frac{1}{2}\log\bigg( \frac{\snr(\snr+\inr +1)}{\snr \cdot \inr  (1-\alpha)^2 + (\snr+\alpha^2 \inr )}\bigg).\label{eqn:cos_diff}
\end{align}
Differentiating the final expression \eqref{eqn:cos_diff} with respect to $\alpha$ and setting it to zero shows that $\alpha^* = \frac{\snr}{\snr+1}$ independent of $\inr$. Furthermore the expression \eqref{eqn:cos_diff} evaluated at $\alpha^*$ yields $\rvC(\snr)$ which is, of course, also independent of $\inr$. So the important thing to note here is that $I(U;Y)-I(U;S)$ is independent of $\inr$ at the optimum $\alpha$, which is the weight of the minimum mean squared error estimate of $X$ given $X+Z$. 

\begin{proof}[Proof sketch of Theorem~\ref{thm:dpc}]
The main ideas of the proof are sketched here. The converse follows from  Theorem~\ref{thm:awgn_asy} so we only have to prove achievability.  We start with some preliminary definitions. 

The analogue of types of states which take values in Euclidean space and  which we find helpful here is the notion of {\em power types}~\cite{merhav93}. Fix $\xi>0$ and consider the {\em power type class}
\begin{equation}
\calT_n (\tau) := \bigg\{  \bs \in\bbR^n :  \tau\le \frac{1}{n}\|\bs\|_2^2 <  \tau + \frac{\xi}{n}\bigg\}
\end{equation}
where $\tau = \frac{k \xi}{n}$. Intuitively, what we are doing is partitioning $[0,\infty)$ into small intervals, each  of length $\frac{\xi}{n}$. For any sequence $\bs\in \calT_n(\tau)$, we say that its power type is $\tau$, i.e.,  $n\tau\le \|\bs\|_2^2  \le  n\tau+\xi$. Thus, the normalized square of the $\ell_2$-norm, quantized to the left endpoint of the interval $[\tau,\tau+\frac{\xi}{n} )$, is the power type of $\bs$.   The set of all power types is denoted as $\calP_n \subset [0,\infty)$.

Consider the following {\em typical} set of power types (also called {\em typical types}) 
\begin{equation}
\tilde{\calP}_n:=\calP_n\cap  [0,\Gamma].
\end{equation}
Thus, we are simply truncating those power types $\tau$ that are larger than $\Gamma$, the threshold in the statement of the theorem. Clearly, the size of the typical  set of power types $| \tilde{\calP}_n |= \lfloor\Gamma n / \xi\rfloor= \Theta(n)$, which is polynomial in $n$. This is  similar  to the discrete case \cite[Ch.~2]{Csi97}. Furthermore,    by the assumption in \eqref{eqn:state_cond},
\begin{equation}
\Pr \big( P_{S^n} \notin \tilde{\calP}_n\big)=O\bigg( \frac{\log n}{\sqrt{n}}\bigg). \label{eqn:state_typ}
\end{equation}
%It can then be shown by regarding  the type of $S^n$ (namely $P_{S^n}$), which is also assumed to be typical in the sense of \eqref{eqn:state_typ}. Because the number of power types is only polynomial, the rate loss is $O(\frac{\log n}{n})$ which can be subsumed into the third-order term in \eqref{eqn:dpc}.
We use the first $\Theta(\log n)$ symbols to transmit $P_{S^n}$, the state type.  The rest of the $n-\Theta(\log n)$ symbols are used to transmit the message.   By using the theory of error exponents for the Gel'fand-Pinsker problem~\cite{Mou07} and the fact that the number of  state types is polynomial, one can show that $P_{S^n}$  can be decoded with error probability $O(\frac{1}{\sqrt{n}})$.  The $\Theta(\log n)$ symbols used to transmit the state type does not affect the dispersion term. In the following, with a slight abuse of notation, $n$ refers to the {\em remaining}  channel uses.

The decoder uses information density thresholding with respect to the joint distribution
\begin{equation}
P_{SUY}^{(\tau)} (s,u,y) := P_{S}^{(\tau)} (s) P_{U|S}(y|s) P_{Y|SU}(y|s,u) \label{eqn:joint_dpc}
\end{equation}
where $\tau$ indexes a power type, the state distribution is $P_{S}^{(\tau)}= \calN(0,\tau)$, the conditional distributions $P_{U|S}(\cdot|s) = \calN(-\alpha s, \snr)$ and $P_{Y|SU}(\cdot|s,u ) = \calN(u+(1-\alpha)s, 1)$.  The corresponding mutual informations induced by  the joint distribution in~\eqref{eqn:joint_dpc} are denoted as $I^{(\tau)}(U;S)$ and  $I^{(\tau)}(U;Y)$. The constant $\alpha>0$ is arbitrary for now.

With these preparations, we are ready to prove Theorem~\ref{thm:dpc} and we divide the proof into several steps.  

Step 1 (Codebook Generation): The number of auxiliary codewords for each type $\tau\in \tilde{ \calP}_n$ is denoted as $L^{(\tau)}$. For each state type $\tau\in\tilde{ \calP}_n$ and each message $m \in \{1,\ldots, M\}$, generate a type-dependent codebook $\calC^{(\tau)}$ consisting of codewords $\{ U^n (m,l) : m \in \{1,\ldots, M\}, l \in \{1,\ldots, L^{(\tau)}\}\}$ where each codeword is drawn independently from  
\begin{equation}
P_{U^n}^{(\tau)}(\bu) := \frac{\delta \{ \|\bu\|_2^2 - n (\snr + \alpha^2 \tau) \}}{A_n\big(\sqrt{ n (\snr + \alpha^2 \tau)} \big)}.
\end{equation}
That is, similar to the proof of Theorem~\ref{thm:awgn_asy}, we uniformly generate  codewords $ U^n (m,l)$ from a  sphere in $\bbR^n$ with radius depending on the type, namely $\sqrt{n (\snr + \alpha^2 \tau)}$. 

Step 2 (Encoding): Given the state sequence $S^n$ and message $m$, the encoder first calculates the type of $S^n$, denoted as $\tau$. If $\tau$ is not typical in the sense of \eqref{eqn:state_typ} declare an error. The contribution to the overall error probability is given in~\eqref{eqn:state_typ} which is easily seen to not affect the second-order term in \eqref{eqn:dpc}.  If $\tau$ is typical, the encoder then proceeds to find  an index $\hatl\in \{1,\ldots, L^{(\tau)}\}$ such that $U^n (m,\hatl)$ is typical in the sense that 
\begin{equation}
\big\| U^n (m,\hatl) -\alpha S^n \big\|_2^2  \in \big[n \, \snr-\eta , n \, \snr\big],
\end{equation}
where $\eta>0$ is chosen to be a small constant.
If there are multiple such $\hatl$, choose one with the smallest index. If there is none, declare an encoding error. The encoder transmits $X^n :=  U^n (m,\hatl)-\alpha S^n$. Clearly the power constraint  on $X^n$ in \eqref{eqn:dpc_con} is satisfied with probability one. 

Step 3 (Decoding): Given the channel output $\by$ and the state type $\tau$, the decoder looks for a codeword $\bu (\tilm,\till)\in\calC^{(\tau)}$ such that 
\begin{equation}
q^{(\tau)}(\bu (\tilm,\till),\by):= \sum_{i=1}^n\log\frac{ P_{Y|U}^{(\tau)}(y_i | u_i(\tilm,\till)) }{P_{Y}^{(\tau)}(y_i)}\ge\gamma^{(\tau)}
\end{equation}
where $\gamma^{(\tau)}$ is a power type-dependent threshold to be chosen in the following.  The distribution $P_{UY}^{(\tau)}$ is defined according to \eqref{eqn:joint_dpc} and $q^{(\tau)}$ is simply an information density indexed by the power type $\tau$.

Step 4 (Analysis of Error Probability): Assume $m=1$. Let $\tau$ be the power type of the state   $S^n$.  Let $\hatl$ be the chosen index in the encoder step.  Clearly, the error event is the union of the following two events:
\begin{align}
\calE_\rmc &  \!:=\! \left\{ \forall\,  U^n (1,l ) \in\calC^{(\tau)} : \big\| U^n (1,l) -\alpha S^n \big\|_2^2    \notin [n \, \snr \!-\!\eta , n \, \snr] \right\}\\*
\calE_{\rmp} &\! :=\! \Big\{ \mbox{Decoder estimates an } \tilm\ne 1 \Big\}
\end{align}
If we set the number of auxiliary codewords for type class indexed by $\tau$ to be 
\begin{equation}
\log L^{(\tau)} := n I^{(\tau)} (U;S) + \kappa_1\log n,
\end{equation}
for some $\kappa_1>0$,  then by techniques similar to the covering lemma~\cite{elgamal}, we can show that 
\begin{equation}
\Pr\big( \calE_\rmc \, \big|\,  P_{S^n}  = \tau\big) \le\exp(-\psi n)
\end{equation}
for some $\psi>0$ and all typical types $\tau \in \tilde{\calP}_n$. The event $\calE_{\rmp}$ can be analyzed per Feinstein-style~\cite{Feinstein} threshold decoding as follows:
\begin{align}
\Pr\big(  \calE_\rmp  &  \, \big|\, \calE_1^c , P_{S^n}  = \tau \big)  \le \Pr\Big( q^{(\tau)}(U^n  ,Y^n) 
  \le \gamma^{(\tau)}  \big|\, \calE_1^c , P_{S^n}  = \tau  \Big)  \nn\\*
  & + ML^{(\tau)}  \Pr\Big( q^{(\tau)}(\barU^n  ,Y^n) > \gamma^{(\tau)}\big|\, \calE_1^c , P_{S^n}  = \tau   \Big) \label{eqn:f_thres}
\end{align}
where $\barU^n \sim P_{U^n}^{(\tau)}$ is independent of $Y^n$. By a change-of-measure argument similar to \eqref{eqn:change_meas}--\eqref{eqn:be} for the AWGN case, one can show that if $\gamma^{(\tau)}$ is chosen to be 
\begin{equation}
\gamma^{(\tau)}=\log M + n I^{(\tau)}(U;S) + \kappa_2\log n
\end{equation}
where  $\kappa_2 :=\kappa_1+1$, then the second term in \eqref{eqn:f_thres}  decays as $O(\frac{1}{n})$. So it remains to analyze the first-term. We do so using the Berry-Esseen theorem and  the fact that with $\alpha^* = \frac{\snr}{\snr+1}$, for any $\tau\in\tilde{\calP}_n$ and any $\bs$ and $\bu$ in the support of $P_{S^n, U^n,Y^n}$ conditioned on $\calE_1^c$ and $P_{S^n}=\tau$, 
\begin{align}
\bbE\big[q^{(\tau)}(U^n,Y^n) \big|  S^n \!=\!\bs,U^n = \bu \big] &= n I^{(\tau)} (U;Y) + O(1),\,\,\mbox{and} \\
\var\big[q^{(\tau)}(U^n,Y^n) \big|  S^n \!=\!\bs,U^n = \bu \big] & = n \rvV(\snr) +O(1)  .
\end{align}
The proof is completed by noting that for $\alpha^* = \frac{\snr}{\snr+1}$, the difference of mutual informations $I^{(\tau)}(U;Y)-I^{(\tau)}(U;S)$ equals $\rvC(\snr)$ for every  power type $\tau$ (in fact every variance) as we discussed prior to the start of this proof.
\end{proof}

\section{Mixed Channels} \label{sec:mixed}
In this section, we consider state-dependent DMCs  $W\in\scP(\calY|\calX\times\calS)$ where the state sequence is random but fixed throughout the entire transmission block once it is determined at the start. This class of channels is known as {\em mixed channels} \cite[Sec.~3.3]{Han10}.  The precise setup is as follows. Let $S$ be a state random variable with a binary  alphabet $\calS = \{0,1\}$ and let $P_S$ be its distribution. We consider two DMCs, each indexed by a state $s \in \calS$. These DMCs are denoted as $W_0:=W(\cdot | \cdot ,0)$ and $W_1:=W(\cdot|\cdot,1)$ and have  capacities  $C(W_0)$ and $C(W_1)$ respectively. Without loss of generality, we assume that $C(W_0)\le C(W_1)$.   We also assume that each of these channels has a unique CAID and the CAIDs coincide.\footnote{An example of this would be two binary symmetric channels. Both the CAIDs are uniform distributions on $\{0,1\}$ and they are clearly unique.} Their $\eps$-channel dispersions (cf.~\eqref{eqn:eps_disp1}--\eqref{eqn:eps_disp2}) are  denoted by $V(W_0)$ and $V(W_1)$ respectively. The $\eps$-dispersions are assumed to be positive and  are independent of $\eps$ because the CAIDs are unique.  %Let their capacities be denoted as $C(W_0)$ and $C(W_1)$. Without loss of generality, we assume that $C(W_0)\le C(W_1)$. 

Before transmission begins, the entire state sequence $S^n = (S, \ldots, S) \in \{0,1\}^n$ is determined. Note that the probability that the DMC is $W_s$ is $\pi_s:=P_S(s)$. The realization of the state is   known to neither   the encoder nor the decoder.  The  probability of observing the sequence $\by\in\calY^n$ given an input sequence $\bx\in\calX^n$ is
\begin{equation}
\Pr(Y^n=\by|X^n=\bx)=\sum_{s\in\calS} \pi_s\prod_{i=1}^n W_s(y_i|x_i) =: W_{\mathrm{mix}}^{(n)}(\by|\bx), \label{eqn:mixed_eq}
\end{equation}
explaining the term {\em mixed} channels. 
We let $M^*_{\mathrm{mix}}(W^n,P_S,\eps)$ denote the maximum number of messages that can be transmitted through the channel $W^n$ when the state distribution is $P_S$ and if the tolerable average error probability is $\eps\in (0,1)$.

%This class of channels is the prototypical class of It is known that the strong converse property \cite[Sec.~3.5]{Han10} for this class of channels does not hold. 
The class of mixed channels is the prototypical one in which the strong converse property \cite[Sec.~3.5]{Han10}  does not hold in general.
This means that the {\em $\eps$-capacity}
\begin{equation}
C_\eps(W,P_S):=\liminf_{n\to\infty}\frac{1}{n}\log M^*_{\mathrm{mixed}}(W^n,P_S,\eps)
\end{equation}
depends on $\eps$ in general. To state $C_\eps(W,P_S)$ for binary state distributions, we consider three different cases:\\
\noindent Case (i): $C(W_0)=C(W_1)$ and relative magnitudes of $\eps$ and $\pi_0$ are arbitrary \\
\noindent  Case (ii): $C(W_0)<C(W_1)$ and $\eps<\pi_0$\\
\noindent Case (iii): $C(W_0)<C(W_1)$ and $\eps\ge\pi_0$

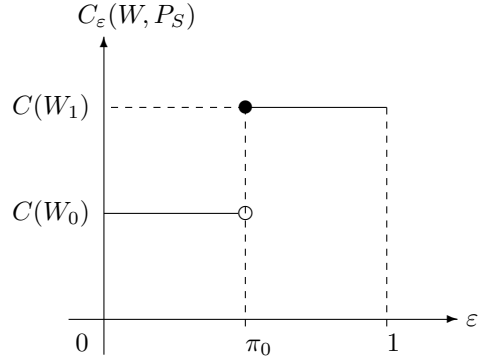
\begin{figure}
\centering
\begin{picture}(115, 115)
\setlength{\unitlength}{.47mm}
\put(0, 10){\vector(1, 0){110}}
\put(10, 0){\vector(0,1){90}}
%\put(50, 30){\line(1, 0){55}}
%\put(30, 50){\line(0,1){55}}
%\put(50, 30){\line(-1, 1){20}}
\put(110,8){ $\eps$}
\put(0, 94){ $C_\eps(W,P_S)$} 
%\put(10, 70){\line(1, 0){80}}
%\put(90, 10){\line(0, 1){60}}

\put(10, 40){\line(1, 0){38}}
\put(50, 70){\line(1, 0){40}}

\multiput(50, 70)(0, -4){16}{\line(0,-1){2}}
\multiput(90, 70)(0, -4){16}{\line(0,-1){2}}

\multiput(50, 70)(-4, 0){10}{\line(-1,0){2}}

\put(-16, 68){$C(W_1)$}
\put(-16, 38){$C(W_0)$}
\put(90, 1){$1$}

\put(50, 1){$\pi_0$}
\put(2, 1){0}

\put(50, 70){\circle*{4}}
\put(50, 40){\circle{4}}
\end{picture}
\caption{Plot of the $\eps$-capacity against $\eps$ for the case $C(W_0)<C(W_1)$. The strong converse property \cite[Sec.~3.5]{Han10}  holds iff $C(W_0)=C(W_1)$ in which case $C_\eps(W,P_S)$ does not depend on $\eps$. }
\label{fig:eps_ca}
\end{figure}

It is known that \cite[Sec.~3.3]{Han10}   that 
\begin{equation}
C_\eps(W,P_S) =\left\{ \begin{array}{cc}
C(W_0) =C(W_1)& \mbox{Case (i)} \\
C(W_0) & \mbox{Case (ii)} \\
C(W_1) & \mbox{Case (iii)} 
\end{array} \right. \label{eqn:eps_c}
\end{equation}
A plot of the $\eps$-capacity is provided in Fig.~\ref{fig:eps_ca}.

The following theorem  was proved for the special case of Gilbert-Elliott channels~\cite{Elliott, Gilbert, Mush}  by Polyanskiy-Poor-Verd\'u~\cite[Thm.~7]{PPV11} where  $W_0$ and $W_1$ are binary symmetric channels so their CAIDs are uniform on $\calX$. The  coefficient $L(\eps;W,P_S) \in\bbR$ in the asymptotic expansion
\begin{equation}
\log M^*_{\mathrm{mix}}(W^n,P_S,\eps)=nC_\eps(W,P_S) + \sqrt{n}L(\eps;W,P_S)  + o(\sqrt{n}) ,
\end{equation}
was sought. This coefficient is termed the {\em second-order coding rate}. In the following theorem, we state and prove a more general version of the result by Polyanskiy-Poor-Verd\'u~\cite[Thm.~7]{PPV11}. For a result imposing even less restrictive assumptions, we refer the reader to the work by Yagi and Nomura~\cite{yagi}.
%is characterized. 
%Let $V(W_0)$ and $V(W_1)$ be the dispersions  of the DMCs $W_0$ and $W_1$ respectively.  
\begin{theorem}\label{thm:mixed}
Assume that each channel $W_s,s\in\calS$ has a unique  CAID and the CAIDs  coincide.    In the various cases above, the second-order coding rate is given as follows:

\noindent Case (i): $L(\eps;W,P_S) $ is the solution $l$ to the following equation:
\begin{equation}
\pi_0\,\Phi\bigg( \frac{l}{\sqrt{V(W_0)}}\bigg)+\pi_1\,\Phi\bigg( \frac{l}{\sqrt{V(W_1)}}\bigg)= \eps. \label{eqn:casei_mixed}
\end{equation}
\noindent Case (ii): 
\begin{equation}
L(\eps;W,P_S) =\sqrt{V(W_0)}\, \Phi^{-1}\bigg( \frac{\eps}{\pi_0}\bigg).\label{eqn:caseii_mixed}
\end{equation}
\noindent  Case (iii): 
\begin{equation}
L(\eps;W,P_S) =\sqrt{V(W_1)}\,\Phi^{-1}\bigg( \frac{\eps-\pi_0}{ \pi_1}\bigg).\label{eqn:caseiii_mixed}
\end{equation}
If $\eps = \pi_0$, then $L(\eps;W,P_S) = -\infty$.  
\end{theorem}
%The proof of this theorem is omitted and can be found in~\cite[Thm.~7]{PPV11} and~\cite[Thm.~5]{TomTan13a}.

We observe that in Case (i) where the capacities $C(W_0)$ and $C(W_1)$ coincide (but not necessarily the dispersions), the second-order coding rate is a function of both the dispersions $V(W_0)$ and $V(W_1)$, together with $\pi_0$ and $\eps$. This function also involves two Gaussian cdfs, suggesting, in the proof, that we apply the central limit theorem twice. In the case where one capacity is strictly smaller than another (Cases (ii) and (iii)), there is only one Gaussian cdf, which means that one of the two channels dominates the overall system behavior. Intuitively for Case (ii), the first order term is $C(W_0)<C(W_1)$ and $\eps<\pi_0$, so the channel with the smaller capacity dominates the asymptotic behavior of the channel, resulting in the second-order term being solely dependent on $V(W_0)$.  In Case (iii), since $\eps\ge\pi_0$, we can tolerate a higher error probability  so the channel with the larger capacity dominates the asymptotic behavior. Hence, $L(\eps;W,P_S)$ depends only on $V(W_1)$.

 The corresponding result for source coding, random number generation  and Slepian-Wolf coding were derived by Nomura-Han~\cite{Nom13b,Nom13}. We only provide a proof sketch  of Case (i) in Theorem \ref{thm:mixed} here. % This proof sketch can be omitted at a first reading without loss of continuity.

\begin{proof}[Proof sketch of Case (i) in  Theorem \ref{thm:mixed}]
For the direct part of Case (i), we specialize Feinstein's theorem (Proposition~\ref{prop:fein}) with the input distribution chosen to be the $n$-fold product of the common CAID of $W_0$ and $W_1$, denoted as $P\in\scP(\calX)$.  Recall the definition of $W_{\mathrm{mix}}^{(n)}(\by|\bx)$ in  \eqref{eqn:mixed_eq}. By the law of total probability,  the probability defining   the $(\eps-\eta)$-information spectrum divergence simplifies as follows: 
\begin{align}
\mathfrak{p}& :=\Pr\left(\log \frac{W_{\mathrm{mix}}^{(n)} ( Y^n|X^n )}{P^n W_{\mathrm{mix}}^{(n)} (Y^n)} \le R \right) \\* 
&= \sum_{s \in\calS}\pi_s\Pr\left( \log\frac{W_s^n( Y_s^n|X^n )}{ P^nW_{\mathrm{mix}}^{(n)}  (Y_s^n)}  \le  R\right) ,\label{eqn:split_mixed}  %+  (1  -  \alpha) \Pr\left( \log\frac{W_1^n( Y_1^n|X^n )}{(PV)^n(Y_1^n)}  \le  R\right),  \label{eqn:split_mixed}
\end{align}
where $Y_s^n , s\in\calS$ denotes the output of  $W_s^n$ when the input is $X^n$. Fix $\gamma>0$. Consider the    probability indexed by $s=0$ in \eqref{eqn:split_mixed}:
\begin{align}
 \mathfrak{p}_0 &   :=  \Pr\left( \log\frac{W_0^n( Y_0^n|X^n )}{(PW_0)^n(Y_0^n)}  + \log\frac{(PW_0)^n( Y_0^n )}{ P^nW_{\mathrm{mix}}^{(n)}  (Y_0^n)} \le R\right) \\
&\le \Pr\left(   \log\frac{W_0^n( Y_0^n|X^n )}{(PW_0)^n(Y_0^n)}   +  \log\frac{(PW_0)^n( Y_0^n )}{ P^nW_{\mathrm{mix}}^{(n)}  (Y_0^n)}   \le  R \,\bigg|\, Y_0^n\in\calA_\gamma \right) \nn\\*
&\hspace{1in}+   \Pr \big(Y_0^n\in\calA_\gamma^c\big)  \label{eqn:change_out}
\end{align}
where  the set  
\begin{equation}
\calA_\gamma:=\bigg\{\by\in\calY^n:\log\frac{(PW_0)^n(\by )}{P^nW_{\mathrm{mix}}^{(n)} ( \by)}    \ge -\gamma \bigg\}.
\end{equation}
Because $Y_0^n \sim (PW_0)^n$,  we have $\Pr(\calA_\gamma^c)\le\exp(-\gamma)$. This, together with the definition of $\calA_\gamma$, implies that  
\begin{align}
\mathfrak{p}_0& \le \Pr\left( \log\frac{W_0^n( Y_0^n|X^n )}{(PW_0)^n(Y_0^n)} \le R+\gamma\right) + \exp( -\gamma)\\*
&\le \Phi\bigg( \frac{R+\gamma-nC(W_0)}{\sqrt{nV(W_0)}}\bigg) + O\bigg(  \frac{1}{\sqrt{n}}\bigg)+ \exp(- \gamma),
\end{align}
where the final step follows from the \iid version of the Berry-Esseen theorem (Theorem \ref{thm:berry_iid}).  The same technique can be used to upper bound the second probability in  \eqref{eqn:split_mixed}.  Choosing $\eta=\frac{1}{\sqrt{n}}$ and $\gamma= \frac12\log n$ results in  
\begin{align}
\mathfrak{p} \le \sum_{s \in\calS}\pi_s \Phi\bigg( \frac{R+\frac12\log n-nC_s}{\sqrt{nV_s}}\bigg)  +O\bigg(  \frac{1}{\sqrt{n}}\bigg). \label{eqn:mixed_direct}
\end{align}
Now we substitute this bound on $\mathfrak{p}$ into the definition of $(\eps-\eta)$-information spectrum divergence in Feinstein's theorem. We note that $C(W_0)=C(W_1)$ and thus may solve for a lower bound of $R$. This then completes the direct part of Case (i) in \eqref{eqn:casei_mixed}. Notice that for Case (ii),  all the derivations up to \eqref{eqn:mixed_direct} hold verbatim. However, note that since $C_\eps(W,P_S)=C(W_0)$, we have that  $R= nC(W_0) + l\sqrt{n} + o(\sqrt{n})$ for some $l\in\bbR$. By virtue of the fact that $C(W_0)<C(W_1)$, the second term in \eqref{eqn:mixed_direct} vanishes asymptotically and we recover~\eqref{eqn:caseii_mixed} which involves only one Gaussian cdf.

For the converse part of Case (i), we appeal to the symbol-wise converse bound (Proposition \ref{prop:converse}). For a fixed $\bx\in\calX^n$ and arbitrary output distribution $Q^{(n)} \in\scP(\calY^n)$,  the probability that defines the  $(\eps+\eta)$-information spectrum divergence can be written as 
\begin{align}
\mathfrak{q}   &:= \Pr \left(  \log\frac{W_{\mathrm{mix}}^{(n)}(Y^n|\bx )}{Q^{(n)}(Y^n) }\le R\right) \\*
&=\sum_{s\in\calS}\pi_s\Pr\left( \log\frac{W_s^n( Y_s^n|\bx )}{ Q^{(n)}(Y_s^n)}  \le  R\right)   \label{eqn:mixed_conv}
%  +  (1 - \alpha) \Pr\left( \log\frac{W_1^n( Y_1^n|\bx )}{  Q^{(n)}(Y_1^n)}  \le   R\right) \label{eqn:mixed_conv}
\end{align}
where $Y_s^n,s\in\calS$ is the output of  $W_s^n$ given input $\bx$. Now choose the output distribution to be 
\begin{equation}
Q^{(n)}(\by):= \frac{1}{2} \left( Q^{(n)}_0(\by) + Q^{(n)}_1(\by) \right) 
\end{equation}
where for each $s\in\calS$,
\begin{equation}
Q_s^{(n)} (\by) : =\sum_{P_{\bx} \in \scP_n(\calX)}   \frac{1}{|\scP_n(\calX)|}   \prod_{i=1}^n P_{\bx}W_s(y_i) \label{eqn:Qn1_mixed}
\end{equation}  
Now note that 
\begin{equation}
Q^{(n)}(\by) \ge \frac{1}{2 |\scP_n(\calX) |}\prod_{i=1}^n P_{\bx}W_s(y_i) 
\end{equation}
for any $s\in\calS$ and type $P_\bx\in\scP(\calX)$. 
By sifting out the type corresponding to $\bx$ for channel $W_0$, the probability in  \eqref{eqn:mixed_conv} corresponding to $s=0$ can be lower bounded as 
\begin{align}
\mathfrak{q}_0\ge\Pr\left( \log\frac{W_0^n( Y_0^n|\bx )}{ (P_{\bx}W_0)^n(Y_0^n)  } \le R - \log (2|\scP_n(\calX)|) \right). \label{eqn:mixed_lb}
\end{align}
By separately considering types close to (Berry-Esseen) and far away (Chebyshev) from the CAID similarly to the  proof of Theorem~\ref{thm:asymp_conv} (or~\cite[Thm.~3]{Hayashi09}), we can  show that   \eqref{eqn:mixed_lb} simplifies to 
\begin{equation}
\mathfrak{q}_0\ge\Phi  \bigg( \frac{R - |\calX|\log (2(n+1)) -n C(W_0)}{\sqrt{nV(W_0)}}\bigg)  - O\bigg(  \frac{1}{\sqrt{n}}\bigg)
\end{equation}
uniformly for all $\bx\in\calX^n$.
The same calculation holds for the second probability in \eqref{eqn:mixed_conv}. By choosing $\eta=\frac{1}{\sqrt{n}}$, we can upper bound $R$ using Proposition~\ref{prop:converse} and   the converse proof of Case (i) can be completed. 
% This completes the proof.
%At the same time, 
%\begin{equation}
%P^nV^n(\by) = \alpha P^n W_0^n(\by)+ (1-\alpha ) P^n W_1^n(\by) .
%\end{equation}
%Thus, we have 
%\begin{equation}
%\mathfrak{p}\le \alpha\Pr\left( \log\frac{W_0^n( Y^n|X^n )}{P^nW_0^n(Y^n)} \le R\right)+
%\end{equation}
\end{proof}
We observe that the crux of the above proof is to use the law of total probability to write the probabilities in the information spectrum divergences as   convex combination of constituent  probabilities involving non-mixed channels.  For the direct part, a change-of-output-measure by conditioning on the event $Y_0^n \in\calA_\gamma$ in \eqref{eqn:change_out} is required. For the converse part, the proof proceeds in a manner similar to the converse proof for the second-order asymptotics for DMCs, upon choosing the auxiliary output measure $Q^{(n)}$ appropriately.

\newcommand{\pow}{\snr}

\section{Quasi-Static  Fading Channels}  \label{sec:quasi}
The final channel with state we consider in this chapter is the {\em quasi-static single-input-multiple-output (SIMO)  channel} with $r$ receive antennas. The term {\em quasi-static} means that the channel statistics (fading coefficients) remain constant during the transmission of each codeword, similarly to  mixed channels.  Yang-Durisi-Koch-Polyanskiy~\cite{WangDurisi13} derived  asymptotic expansions for this channel model which is described precisely as follows: For time   $i =1,\ldots, n$, the channel law is given as
\begin{equation}
\begin{bmatrix}
Y_{i1} \\ \vdots \\ Y_{ir}
\end{bmatrix} = \begin{bmatrix}
H_1 \\ \vdots \\ H_{r}
\end{bmatrix}  X_i + \begin{bmatrix}
Z_{i1} \\ \vdots \\ Z_{ir}
\end{bmatrix} \label{eqn:kmimp}
%Y_{ij} =  X_{i} H_j + Z_{ij}  , \label{eqn:kmimp}
\end{equation}
where    $H^r := (H_1,\ldots, H_r)'$ is  the  vector of (real-valued) \iid fading coefficients, which are random but remain constant for all channel uses, and $\{Z_{ij}\}$ are \iid    noises distributed as $\calN(0,1)$. In the theory of fading channels~\cite{Biglieri}, the channel inputs and outputs are usually complex-valued, but to illustrate the key ideas, it is sufficient  to consider real-valued channels  and fading coefficients. In this section, we restrict our attention to the real-valued SIMO model in \eqref{eqn:kmimp}. The channel input $X^n$ must satisfy
\begin{equation}
\| X^n \|_2^2 = \sum_{i=1}^n X_{ i}^2\le n\, \pow \label{eqn:rho_constraint}
\end{equation}
with probability one   for some permissible power $\pow  >0$. 

Two different setups are considered. First, both the encoder and decoder do not have information about the realization of $H^r$. Second, both the encoder and decoder have this information.   

For a given distribution on the fading coefficients $P_{H^r}$ (this plays the role of the state or side information), define $M^*_{\mathrm{no-SI}}(W^n,P_{H^r},\pow,\eps)$ and $M^*_{\mathrm{SI-ED}}(W^n,P_{H^r},\pow,\eps)$ to be the maximum number of codewords transmissible over $n$  independent uses of the channel under constraint~\eqref{eqn:rho_constraint}, with fading distribution $P_{H^r}$, and with average error probability not exceeding $\eps$ under the no side information and complete knowledge of side information settings respectively. It is known using the theory of general channels \cite[Thm.~6]{VH94}   that for every $\eps\in (0,1)$,  the following limits exist and are equal
\begin{align}
\lim_{n\to\infty}\frac{1}{n}\log& M^*_{\mathrm{no-SI}}(W^n,P_{H^r}, \pow,\eps) \nn\\* 
& =\lim_{n\to\infty}\frac{1}{n}\log M^*_{\mathrm{SI-ED}}(W^n,P_{H^r},\pow,\eps).
\end{align}
Their common value is the {\em $\eps$-capacity} \cite{Biglieri}, defined as 
\begin{equation}
C_\eps(W,P_{H^r}) :=\sup \Big\{ \xi \in \bbR: F(\xi; \pow,P_{H^r}) \le\eps  \Big\}, \label{eqn:Ceps}
\end{equation}
where the {\em outage function} is defined as 
\begin{equation}
F(\xi; \pow, P_{H^r}):=\Pr\Big( \rvC\big(  \pow \|H^r\|_2^2 \big)\le\xi\Big) \label{eqn:outage}
\end{equation}
Observe that for a fixed value of $H^r = \bh$ (i.e., the channel state is not random), the expression $ \rvC\big(  \pow \|\bh\|_2^2 \big)$ is simply the Shannon capacity of the channel. Beyond the first-order characterization, what  are the refined asymptotics of  $\log M^*_{\mathrm{no-SI}}(W^n,P_{H^r},\pow,\eps)$ and $\log M^*_{\mathrm{SI-ED}}(W^n,P_{H^r},\pow,\eps)$? The following surprising result was proved by Yang-Durisi-Koch-Polyanskiy~\cite{WangDurisi13}.
\begin{theorem} \label{thm:quasi}
Assume that the random variable  $G=\|H^r\|_2^2$ has a pdf that is twice continuously differentiable and that $C_\eps(W,P_{H^r}) $ in~\eqref{eqn:Ceps} is a point of growth of the outage function defined in \eqref{eqn:outage}, i.e., $F'(C_\eps(W,P_{H^r}); \pow, P_{H^r})>0$. Then 
\begin{align}
\log M^*_{\mathrm{no-SI}}(W^n,P_{H^r},\pow,\eps)  &= n C_\eps(W,P_{H^r}) + O(\log n),\quad \mbox{and}  \label{eqn:qs1}\\*
\log M^*_{\mathrm{SI-ED}}(W^n,P_{H^r},\pow,\eps)  &= n C_\eps(W,P_{H^r}) + O(\log n) .\label{eqn:qs2}
\end{align}
\end{theorem}
The condition on the channel gain $G$  is satisfied by many fading models of interest, including Rayleigh, Rician and Nakagami.

Theorem~\ref{thm:quasi} says  interestingly that, in the quasi-static setting, the $\Theta(\sqrt{n})$ dispersion terms that we usually see  in asymptotic expansions are absent. This means that the $\eps$-capacity is  good benchmark for the finite blocklength fundamental limits $\log M^*_{\mathrm{no-SI}}(W^n,P_{H^r},\pow,\eps)$ and $\log M^*_{\mathrm{SI-ED}}(W^n,P_{H^r},\pow,\eps)$ since the backoff from the $\eps$-capacity is of the order $\Theta(\frac{\log n}{n})$ and not the larger $\Theta(\frac{1}{\sqrt{n}})$.% without the $\Theta(\sqrt{n})$ dispersion term!

We will not detail the proof of Theorem  \ref{thm:quasi} here,  as it is rather involved. See~\cite{WangDurisi13} for the details. However, we will provide a plausibility argument as to why the  $\Theta(\sqrt{n})$ term is absent in the expansions in \eqref{eqn:qs1}--\eqref{eqn:qs2}. Since the quasi-static fading channel is conditionally ergodic (meaning that given $H^r = \bh$, it is ergodic), one has that 
\begin{equation}
\eps^*(   W^n,\bh , \pow,M) \approx \Pr \Big( n  \rvC\big(  \pow \|\bh\|_2^2 \big) + \sqrt{n  \rvV\big(  \pow \|\bh\|_2^2 \big)}\, Z\le\log M \Big) \label{eqn:eh}
\end{equation}
where $\eps^*(   W^n,\bh, \pow,M)$ the smallest  error probability with $M$ codewords and channel gains $H^r=\bh$, and $Z$ is the standard normal   random variable.   Note that  $ \rvC\big(  \pow \|\bh\|_2^2 \big)$ and $ \rvV\big(  \pow \|\bh\|_2^2 \big)$ are respectively the capacity and dispersion of the   channels  conditioned on $H^r=\bh$. If $Z$ is independent of $H^r$, the above probability is close to one in the ``outage case'', i.e., when $n \rvC\big(  \pow \|\bh\|_2^2 \big)< \log M$. Hence, taking the expectation over $H^r$,  
\begin{equation}
\eps^*(   W^n,P_{H^r} , \pow,M) \approx \Pr\left( n  \rvC\big(  \pow \|H^r\|_2^2 \big)  \le \log M\right), \label{eqn:drop_dis}
\end{equation}
where $\eps^*(   W^n,P_{H^r} , \pow,M)$ the  smallest error probability with $M$ codewords and random channel gains. In fact, the above argument can be formalized using the following lemma whose proof can be found in~\cite{WangDurisi13}.
\begin{lemma}
Let $A$ be a random variable with zero mean, unit variance and finite third moment. Let $B$ be independent of $A$ with twice continuously differentiable pdf. Then,
\begin{equation}
\Pr\big(A \le\sqrt{n}B\big) = \Pr\big(B\ge 0\big) + O\bigg( \frac{1}{n}\bigg).
\end{equation}
\end{lemma}
The approximation in \eqref{eqn:drop_dis} is then justified by taking
\begin{equation}
A= \sqrt{   \rvV\big(  \pow \|H^r\|_2^2 \big)}\,Z,\quad\mbox{and}\quad B=\log M - n\rvC\big(  \pow \|H^r\|_2^2 \big) .
\end{equation}
Finally, we remark that this quasi-static SIMO model is different from that in Section~\ref{sec:state_ed} in two significant ways: First, the state here is a continuous random variable  and second, according to \eqref{eqn:kmimp}, the quasi-static scenario here implies that the state $H^r$ is constant throughout transmission and does not vary across time $i = 1,\ldots, n$. This explains the difference in second-order behavior vis-\`a-vis the result in Theorem~\ref{thm:state_ed}. The distinction between this model and that in   Section \ref{sec:mixed} on mixed channels with finitely many  states is that the fading coefficients contained in  $H^r$ are continuous random variables.

%\section{The Compound Channel} \cite{Pol13b}
 
\chapter{Distributed Lossless Source Coding}\label{ch:sw}
It is not an exaggeration to say that  one of  the most surprising results  in network information theory is the theorem by Slepian and Wolf~\cite{sw73} concerning  distributed lossless source coding.   For the lossless source coding problem as discussed  extensively in Chapter~\ref{ch:src}, it can be easily seen that if we would like to   losslessly and reliably reconstruct $X^n$ from its compressed version and correlated side-information $Y^n$ that is available to both encoder and decoder, then the minimum rate of compression is $H(X|Y)$.  What happens if the side information is only available to the decoder but {\em not} the encoder? Surprisingly, the minimum rate of compression is still $H(X|Y)$! It hints at the encoder being able to perform some form of universal encoding regardless of the nature of whatever side-information is available to the decoder. 

A more general version of this problem is shown in Fig.~\ref{fig:sw}. Here, two correlated sources are to be losslessly reconstructed in a distributed fashion. That is, encoder $1$   sees $X_1$  and not $X_2$, and vice versa. Slepian and Wolf showed in \cite{sw73} that if $X_1^n$ and $X_2^n$ are generated from a {\em discrete memoryless multiple source} (DMMS) $P_{X_1^n X_2^n}$, then the set of achievable rate pairs $(R_1, R_2)$  belongs to the set
\begin{align}
R_1 \ge H(X_1|X_2) ,\quad R_2  \ge H(X_2|X_1), \quad  
R_1+R_2  \ge H(X_1, X_2) . \label{eqn:orr_sw}
\end{align}
In this chapter, we analyze refinements to Slepian and Wolf's seminal result. Essentially, we fix a point $(R_1^*, R_2^*)$ on the boundary of the region in \eqref{eqn:orr_sw}. We then find all possible {\em second-order coding rate pairs} $(L_1, L_2) \in\bbR^2$ such that there exists length-$n$ block codes of sizes $M_{jn}, j=1,2$ and error probabilities $\eps_n$ such that 
\begin{equation}
\log M_{jn}\le n R_j^* + \sqrt{n} L_j + o\big(\sqrt{n}\big),\quad\mbox{and}\quad \eps_n\le\eps+ o(1).
\end{equation}
The latter condition means that the sequence of codes is {\em $\eps$-reliable}. We will see that if $(R_1^*, R_2^*)$ is a corner point, the set of all such $(L_1, L_2)$ is characterized in terms of a {\em multivariate} Gaussian cdf.  This is the distinguishing feature compared to  results in the previous chapters.

The material in this chapter is based on the work by  Nomura and Han~\cite{Nom13} and Tan and Kosut~\cite{TK14}.
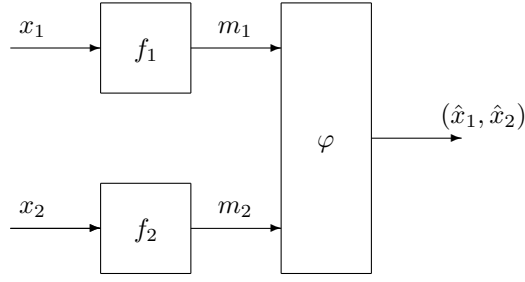
\begin{figure}[t]
\centering
\setlength{\unitlength}{.4mm}
\begin{picture}(150, 90)
%\thicklines
\put(0, 15){\vector(1, 0){30}}
\put(60, 15){\vector(1,0){30}}
%\put(120, 15){\vector(1,0){30}}
%\put(180, 15){\vector(1,0){30}}
\put(30, 0){\line(1, 0){30}}
\put(30, 0){\line(0,1){30}}
\put(60, 0){\line(0,1){30}}
\put(30, 30){\line(1,0){30}}

\put(0, 75){\vector(1, 0){30}}
\put(60, 75){\vector(1, 0){30}}
\put(0, 15){\vector(1, 0){30}}
%\put(60, 15){\vector(1,0){30}}
\put(120, 45){\vector(1,0){30}}
%\put(180, 15){\vector(1,0){30}}
\put(30, 60){\line(1, 0){30}}
\put(30, 60){\line(0,1){30}}
\put(60, 60){\line(0,1){30}}
\put(30, 90){\line(1,0){30}}

\put(90, 0){\line(1, 0){30}}
\put(90, 0){\line(0,1){90}}
\put(120, 0){\line(0,1){90}}
\put(90, 90){\line(1,0){30}}

\put(0, 20){  $x_2$}
\put(0, 80){  $x_1$}
\put(66, 20){  $m_2$}
\put(66, 80){  $m_1$}
%\put(51, -10){  $\bbE[\rvg(X)]\le\Gamma$}
%\put(65, 8){  $[2^{nR}]$}
\put(140, 50){  $(\hatx_1,\hatx_2)$} 
\put(41, 12){$f_2$ } 
\put(41, 72){$f_1$ } 
\put(102, 42){$\varphi$} 
%
%\put(150, 0){\line(1, 0){30}}
%\put(150, 0){\line(0,1){30}}
%\put(180, 0){\line(0,1){30}}
%\put(150, 30){\line(1,0){30}}
%\put(161, 12){$\varphi$} 
%\put(190, 20){  $\hatm$} 
%\put(186, 1){  $\Pr(\hatM \ne M)$} 
  \end{picture}
  \caption{Illustration of the Slepian-Wolf~\cite{sw73} problem.   }
  \label{fig:sw}
\end{figure}
\section{Definitions and  Non-Asymptotic Bounds}
In this section, we set up the distributed lossless source coding problem formally and mention some known non-asymptotic bounds.  Let $P_{X_1 X_2} \in \scP(\calX_1\times\calX_2)$ be a correlated source. See Fig.~\ref{fig:sw}. 

An {\em $(M_1, M_2, \eps)$-code} for the correlated source  $P_{X_1 X_2}\in \scP(\calX_1\times\calX_2)$ consists of a triplet of maps that includes two encoders $f_j : \calX_j\to\{1,\ldots, M_j\}$ for $j=1,2$ and a decoder $\varphi: \{1,\ldots, M_1\}\times \{1,\ldots, M_2\}\to\calX_1\times\calX_2$ such that the {\em error probability}
\begin{equation}
P_{X_1 X_2} \big( \{ (x_1, x_2) \in\calX_1\times\calX_2 : \varphi\big(f_1(x_1), f_2(x_2)\big) \ne (x_1, x_2)   \}\big)\le\eps.
\end{equation}

The  numbers $M_1$ and $M_2$  are called the  {\em sizes} of the code. 

We now state known achievability and converse bounds due to Miyake and Kanaya~\cite{miyake}.  See Theorems 7.2.1 and 7.2.2 in Han's book~\cite{Han10} for the proofs of these results. The achievability bound is based on Cover's random binning~\cite{cover75} idea. 

\begin{proposition}[Achievability Bound for Slepian-Wolf problem]\label{prop:ach_sw}
For every $\gamma>0$, there  exists an $(M_1, M_2, \eps)$-code   satisfying
\begin{align}
\eps\le\Pr\bigg(\log\frac{1}{P_{X_1|X_2}(X_1|X_2)} &\ge\log M_1 - \gamma \quad\mbox{or} \nn\\*
\log\frac{1}{P_{X_2|X_1}(X_2|X_1)}&\ge\log M_2 - \gamma \quad\mbox{or}  \nn\\*
\log\frac{1}{P_{X_1 X_2}(X_1,  X_2)}&\ge\log (M_1 M_2) -\gamma  \bigg)  + 3\exp(-\gamma). \label{eqn:ach_sw}
\end{align}
\end{proposition}

The converse bound is based on standard  techniques in information spectrum \cite[Ch.~7]{Han10} analysis.
\begin{proposition}[Converse Bound for Slepian-Wolf problem] \label{prop:conv_sw}
For any $\gamma>0$, every $(M_1, M_2, \eps)$-code  must satisfy
\begin{align}
\eps\ge\Pr\bigg(\log\frac{1}{P_{X_1|X_2}(X_1|X_2)} &\ge\log M_1 +\gamma \quad\mbox{or} \nn\\*
\log\frac{1}{P_{X_2|X_1}(X_2|X_1)}&\ge\log M_2+\gamma \quad\mbox{or}  \nn\\*
\log\frac{1}{P_{X_1 X_2}(X_1,  X_2)}&\ge\log (M_1 M_2)+ \gamma  \bigg)- 3\exp(-\gamma). \label{eqn:conv_sw}
\end{align}
\end{proposition}

Notice that the {\em entropy density vector}
\begin{equation}
\bh_{X_1 X_2}(x_1, x_2)  :=  \begin{bmatrix}
\log\frac{1}{P_{X_1|X_2}(x_1|x_2)} & \log\frac{1}{P_{X_2|X_1}(x_2|x_1)}& \log\frac{1}{P_{X_1 X_2}(x_1, x_2)} 
\end{bmatrix}' \label{eqn:ent_dens}
\end{equation}
plays a  prominent role in both the direct and converse bounds.

\section{Second-Order Asymptotics}

We would like to  make concrete statements about performance of optimal codes with asymptotic error probabilities not exceeding $\eps$ and blocklength  $n$ tending to infinity. For this purpose, we assume that the source $P_{X_1 X_2}$ is a DMMS, i.e., 
\begin{equation}
P_{X_1^n X_2^n}(\bx_1, \bx_2) =  \prod_{i=1}^n P_{X_1 X_2}(x_{1i}, x_{2i} ),\quad\forall\, (\bx_1,\bx_2)\in\calX_1^n \times \calX_2^n.
\end{equation}
As such, the alphabets $\calX_j , j = 1,2$ in the definition of an $(M_1, M_2, \eps)$-code are replaced by their $n$-fold Cartesian products.

\subsection{Definition of the Second-Order Rate Region and Remarks}
Unlike the point-to-point problems where the first-order fundamental limit is a single number (e.g., capacity for channel coding, rate-distortion function for lossy compression), for multi-terminal problems like the Slepian-Wolf problem there is a {\em continuum} of first-order fundamental limits. Hence, to define second-order quantities, we must ``center'' the analysis at a   point  $(R_1^*, R_2^*)$ on the boundary of the optimal rate region (in source coding scenarios) or capacity region (in channel coding settings). Subsequently, we can ask what is the {\em local} second-order behavior of the system in the vicinity of $(R_1^*, R_2^*)$. This is the essence of second-order asymptotics for multi-terminal problems.  Note that for multi-terminal problems, we exclusively study second-order asymptotics, and we do not go beyond this to study third-order asymptotics.

Fix a rate pair $(R_1^*, R_2^*)$ on the boundary of the optimal rate region given by  \eqref{eqn:orr_sw}.  Let $(L_1, L_2) \in\bbR^2$ be called an  {\em achievable $(\eps, R_1^*, R_2^*)$-second-order coding rate pair} if there exists a sequence of $(M_{1n}, M_{2n} ,\eps_n)$-codes for the correlated source $P_{X_1^n X_2^n}$ such that the sequence of error probabilities does not exceed $\eps$ asymptotically, i.e.,
\begin{equation}
\limsup_{n\to\infty}\eps_n\le\eps \label{eqn:sw_err}
\end{equation}
and furthermore, the size of the codes satisfy
\begin{equation}
\limsup_{n\to\infty}\frac{1}{\sqrt{n}} \big( \log M_{jn}-n R_j^* \big) \le L_j ,\quad j = 1,2. \label{eqn:sw_second}
\end{equation}
The set of all achievable $(\eps, R_1^*, R_2^*)$-second-order coding rate pairs is denoted as $\calL(\eps; R_1^*, R_2^*) \subset\bbR^2$, the {\em second-order coding rate region}. Note that even though we term the elements of $\calL(\eps; R_1^*, R_2^*)$ as ``rates'', they could be negative. This convention follows that in Hayashi's works~\cite{Hayashi08, Hayashi09}. The number $L_j$ has units is bits per square-root source symbols. 

Let us pause for a moment to understand the above definition as it is a recurring theme in subsequent chapters on second-order asymptotics in network information theory. Slepian-Wolf~\cite{sw73} showed that there exists  a sequence of codes for the (stationary, memoryless) correlated source $(X_1, X_2)$ whose error probabilities vanish asymptotically (i.e., $\eps_n=o(1)$) and whose sizes $M_{jn}$ satisfy 
\begin{equation}
\limsup_{n\to\infty}\frac{1}{n} \log M_{jn}\le R_j,\quad j = 1,2
\end{equation}
where the rates $R_1$ and $R_2$ satisfy the bounds in \eqref{eqn:orr_sw}. Hence, the definition of a second-order coding rate pair  in \eqref{eqn:sw_second} is a refinement of the scaling of the code sizes in Slepian-Wolf's setting,  centering the rate analysis at $(R_1^*, R_2^*)$, and analyzing deviations of order $\Theta(\frac{1}{\sqrt{n}})$ from this first-order fundamental limit. In doing so, we allow the error probability to be non-vanishing per \eqref{eqn:sw_err}. This requirement   is subtly  different from that in the chapters on source and channel coding where we are interested in approximating non-asymptotic fundamental limits like $\log M^*(P,\eps)$ or $\log M^*_{\av}(W,\eps)$ and  therein, the error probabilities  are constrained to be no larger than a non-vanishing $\eps \in (0,1)$ {\em for all blocklengths}. Here we allow some slack (i.e., $\eps_n\le\eps+o(1)$). This turns out to be  immaterial from the perspective of second-order asymptotics, as we are  seeking to characterize a region of second-order rates $\calL(\eps;R_1^*, R_2^*)$ and we are not attempting to characterize higher-order (i.e., third-order) terms in an asymptotic expansion. The $o(1)$ slack affects the  third-order asymptotics but since we are not interested in this study for network information theory problems, we find it convenient to define  $\calL(\eps;R_1^*, R_2^*)$ using~\eqref{eqn:sw_err}--\eqref{eqn:sw_second}, analogous to information spectrum analysis \cite{Han10}.

\begin{figure}[t]
\centering
\begin{picture}(115, 135)
\setlength{\unitlength}{.47mm}
\input{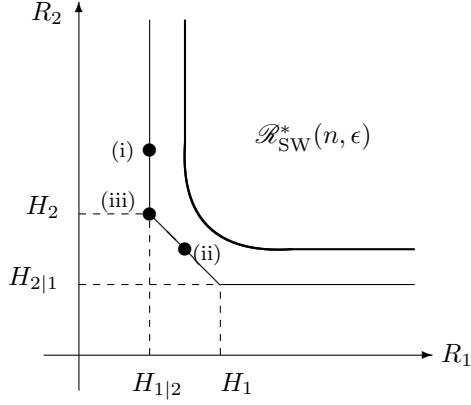}
 
%\put(65, 60){\vector(-1,-2){15}}
%\put(66,61 ){\footnotesize $\vec{\bd}$}
%\qbezier(53,36)(56,34)(56,30)
%\put(58,32){\footnotesize $\theta$}
%\put(56,41){\circle*{4}}
%%\put(61,45){\footnotesize $(   R_1(n,\epsilon),   R_2(n,\epsilon)   )$}
\put(30,50){\circle*{4}}
%\put(52,20){\footnotesize $(   R_1^*,    R_2^*   )$}

\put(16,52){\footnotesize (iii)}
\put(19,66){\footnotesize (i)}

\put(30,68){\circle*{4}}

\put(40,40){\circle*{4}}

\put(42,37){\footnotesize (ii)}
\end{picture}   
\caption{Illustration of the different cases in Theorem~\ref{thm:disp_sw} where  $H_1=H(X_1)$ and $H_{2|1}=H(X_2|X_1)$ etc. The curve   is a schematic    of the boundary of the set of  rate pairs $(R_1, R_2)$ achievable at   blocklength $n$ with error probability no more than $\eps<\frac{1}{2}$. The set is  denoted by $\scR^*_{\mathrm{SW}}(n,\epsilon)$.}
\label{fig:slices}
\end{figure}
Before we state the main result of this chapter, let us consider the following bivariate generalization of the cdf of a Gaussian:
\begin{equation}
\Psi( t_1, t_2 ; \bmu,\bSigma) := \int_{-\infty}^{t_1 }\int_{-\infty}^{t_2 } \calN(\bx; \bmu;\bSigma)\, \rmd \bx,  \label{eqn:biv}
\end{equation}
where $\calN(\bx; \bmu;\bSigma)$ is the pdf of a bivariate Gaussian,   defined in \eqref{eqn:pdf_multi_gauss}. Also define the {\em source dispersion matrix}
\begin{align}
\bV &  = \bV(P_{X_1 X_2}):=\cov\big[\bh(X_1, X_2) \big] \\*
& = \begin{bmatrix}
V_{1|2} & \rho_{1,2} \sqrt{V_{1|2} V_{2|1}} & \rho_{1,12}  \sqrt{V_{1|2}  V_{1,2} }  \\
 \rho_{1,2} \sqrt{V_{1|2} V_{2|1}} & V_{2|1} & \rho_{2,12} \sqrt{V_{2|1} V_{1,2} }  \\
 \rho_{1,12}  \sqrt{V_{1|2}  V_{1,2} } & \rho_{2,12} \sqrt{V_{2|1}  V_{1,2} }  &   V_{1,2} 
\end{bmatrix} .
\end{align}
We also denote the diagonal entries as  $ V(X_1|X_2)=V_{1|2} , V(X_2|X_1)=V_{2|1}$ and $V(X_1, X_2)=V_{1,2}$. 
Define $\bV_{1,12} $ (resp.\ $\bV_{2,12}$) as the $2\times 2$ submatrix indexed by the  $1^{\mathrm{st}}$ (resp.\ $2^{\mathrm{nd}}$) and $3^{\mathrm{rd}}$ entries of $\bV$, i.e.,
\begin{align}
\bV_{1,12}&:= \begin{bmatrix}
V_{1|2}   & \rho_{1,12}  \sqrt{V_{1|2}  V_{1,2} }  \\
 \rho_{1,12}  \sqrt{V_{1|2}  V_{1,2} }   &  V_{1,2} 
\end{bmatrix} \label{eqn:V112} , 
%\bV_{2,12} &:= \begin{bmatrix}
%V_{2|1}   & \rho_{2,12}  \sqrt{V_{2|1}  V_{1,2} }  \\
% \rho_{2,12}  \sqrt{V_{ 2|1}  V_{1,2} }   &  V_{1,2} 
%\end{bmatrix} .
\end{align}
and $\bV_{2,12}$ is defined similarly. 

\subsection{Main Result: Second-Order Coding Rate Region}

The set $\calL( \eps; R_1^*, R_2^*)$ is characterized in the following result. This result was proved by Nomura and Han~\cite{Nom13}. A slightly different form of this result was proved earlier by Tan and Kosut~\cite{TK14}.
\begin{theorem} \label{thm:disp_sw}
Assume   $\bV$ is positive definite.  Depending on $(R_1^*, R_2^*)$ (see Fig.~\ref{fig:slices}), there are $5$   cases of which we state~$3$ explicitly: \\
Case (i):    $R_1^*=H(X_1|X_2)$ and $R_2^*  >  H(X_2)$ (vertical boundary) 
\begin{equation}
\calL(\eps; R_1^*,R_2^* ) = \Big\{ (L_1, L_2) : L_1\ge \sqrt{V(X_1|X_2)}\Phi^{-1}(1- \eps)\Big\}.  \label{eqn:vert_bd}
\end{equation}
Case (ii):    $R_1^*+R_2^*=H(X_1,X_2)$  and $H(X_1|X_2)<R_1^* <H(X_1)$  (diagonal face) 
\begin{equation}
\calL(\eps; R_1^*,R_2^* ) = \Big\{ (L_1, L_2) : L_1 + L_2\ge \sqrt{V(X_1, X_2)}\Phi^{-1}(1- \eps)\Big\}.  \label{eqn:diag_bd}
\end{equation}
Case (iii):   $R_1^* = H(X_1|X_2)$ and $R_2^*=H(X_2)$ (top-left corner point) 
\begin{equation}
\calL(\eps; R_1^*,R_2^* ) = \Big\{ (L_1, L_2) :   \Psi(L_1, L_1+L_2; \bzero, \bV_{1,12}) \ge 1-\eps \Big\}.  \label{eqn:corner_bd}
\end{equation}
\end{theorem}

\begin{figure}
\centering
\includegraphics[width = .75\columnwidth]{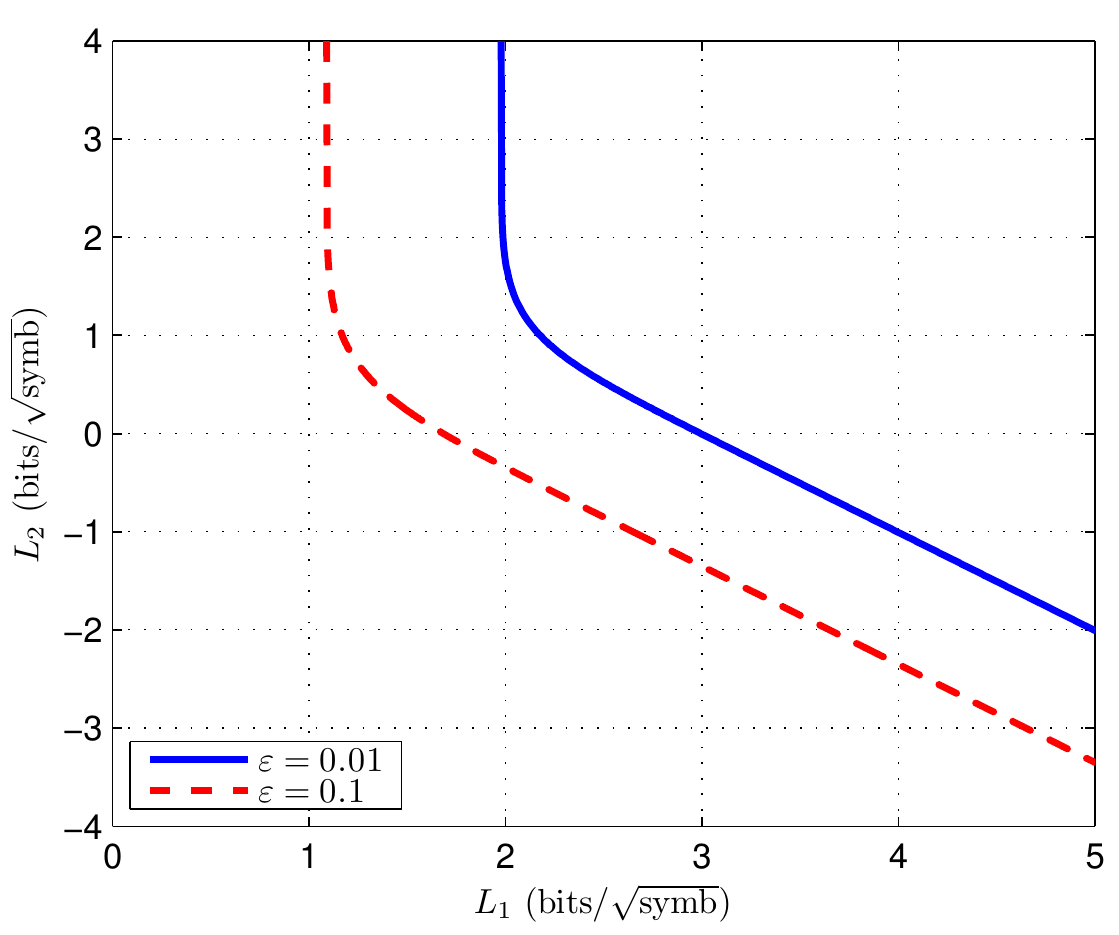}
\caption{Illustration of   $\calL(\eps; R_1^*,R_2^* )$  in \eqref{eqn:corner_bd}  for the source  $P_{X_1 X_2}$ in \eqref{eqn:sw_src} with $\eps=0.01,0.1$. The regions are to the top right of the boundaries indicated. }
\label{fig:sw_reg}
\end{figure}

The region $\calL(\eps; R_1^*,R_2^* )$ for Case (iii) is illustrated in Fig.~\ref{fig:sw_reg} for a binary source $(X_1, X_2)$ with distribution 
\begin{equation} \label{eqn:sw_src}
P_{X_1 X_2}(x_1, x_2) =\begin{bmatrix}
0.7 & 0.1 \\ 0.1 & 0.1
\end{bmatrix}.
\end{equation}
Note that $\calL( \eps; R_1^*, R_2^*)$ for other points on the boundary can be found by symmetry. For example for the horizontal boundary, simply interchange the indices $1$ and $2$ in~\eqref{eqn:vert_bd}. The case in which $\bV$ is not positive definite was dealt with in detail in~\cite{TK14}. 

\subsection{Proof of Main Result and Remarks}
\begin{proof}
The proof of the direct part specializes the non-asymptotic bound in Proposition~\ref{prop:ach_sw} with the choice $\gamma = n^{1/4}$. Choose code sizes $M_{1n}$ and $M_{2n}$ to be the smallest integers satisfying 
\begin{align}
\log M_{jn} &\ge n R_j^*+\sqrt{n} L_j + 2n^{1/4}  ,\quad j = 1,2 , \label{eqn:sizeM1}
\end{align}
for some $(L_1, L_2)\in\bbR^2$.  Substitute these choices into  the probability in~\eqref{eqn:ach_sw}, denoted as $\mathfrak{p}$. The complementary  probability $1-\mathfrak{p}$ is
\begin{align}
1\!-\!\mathfrak{p}\!=\!\Pr\left( \bh_{X_1^n X_2^n}(X_1^n, X_2^n) \!<  \!\begin{bmatrix}
n R_1^*+\sqrt{n} L_1 +  n^{1/4} \\
n R_2^*+\sqrt{n} L_2 +  n^{1/4} \\
n (R_1^*\!+\!R_2^*)\!+\!\sqrt{n}(L_1\!+\! L_2) \!+\! 3n^{1/4} 
\end{bmatrix} \right). \label{eqn:one-minus-p}
\end{align}
Recall that $\bh_{X_1^n X_2^n}(\bx_1,\bx_2)$ is the entropy density in \eqref{eqn:ent_dens} and   that  inequalities (like $<$)  are applied element-wise.   The  three events in the probability above  are 
\begin{align}
\calA_1 &  \!:=\! \bigg\{ \frac{1}{n}\log\frac{1}{P_{X_1^n|X_2^n}(X_1^n|X_2^n) } \!<\! R_1^* \!+ \!\frac{L_1}{\sqrt{n}} + n^{-3/4}\bigg\},\\
\calA_2 & \!:=\! \bigg\{ \frac{1}{n}\log\frac{1}{P_{X_2^n|X_1^n}(X_2^n|X_1^n) } \!< \! R_2^* \!+\! \frac{L_2}{\sqrt{n}} + n^{-3/4}\bigg\},\quad\mbox{and}\\
\calA_{12} & \!:=\! \bigg\{ \frac{1}{n}\log\frac{1}{P_{X_1^n X_2^n}(X_1^n , X_2^n) } \!<\! R_1^*  \!+\! R_2^*\! +\! \frac{L_1 \!+\! L_2}{\sqrt{n}} +3 n^{-3/4}\bigg\}.
\end{align}
As such, the probability in \eqref{eqn:one-minus-p} is $\Pr(\calA_1\cap\calA_2\cap\calA_{12})$. 

Let us consider Case (i) in Theorem~\ref{thm:disp_sw}. In this case,  $R_2^*>H(X_2)$ and $R_1^*+R_2^*> H(X_1 , X_2)$.  By the weak law of large numbers, $\Pr(\calA_2)\to 1$ and $\Pr(\calA_{12})\to 1$ as $n$ grows. In fact,   these probabilities converge to one  exponentially fast.  Thus, 
\begin{equation}
1-\mathfrak{p}\ge \Pr(\calA_1) +  \exp(-n\xi) \label{eqn:apply_ld}
\end{equation}
for some $\xi>0$. Furthermore, because $R_1^*=H(X_1|X_2)$, $\Pr(\calA_1)$ can be estimated using the Berry-Esseen theorem as 
\begin{equation}
\Pr(\calA_1)\ge \Phi\bigg( \frac{L_1}{\sqrt{V( X_1|X_2)} } \bigg) + O(n^{-1/4}).
\end{equation}
Hence, one has
\begin{equation}
\mathfrak{p}\le 1-\Phi\bigg( \frac{L_1}{\sqrt{V( X_1|X_2)} }  \bigg) + O(n^{-1/4}).
\end{equation}
Coupled with the fact that $\exp(-\gamma)=\exp(-n^{1/4})$, the proof of the direct part of~\eqref{eqn:vert_bd} is complete. The converse employs essentially the same technique. Case (ii) is also similar with the exception that now $\Pr(\calA_1)\to 1$ and $\Pr(\calA_2)\to 1$, while $\Pr(\calA_{12})$ is estimated using the Berry-Esseen theorem. 

We are left with Case (iii). In this case, only $\Pr(\calA_2)\to 1$. Thus,   just as in \eqref{eqn:apply_ld},  \eqref{eqn:one-minus-p} can be estimated as 
\begin{equation}
1-\mathfrak{p}\ge \Pr(\calA_1\cap\calA_{12}) + \exp(-n\xi') \label{eqn:drop1}
\end{equation}
for some $\xi'>0$. The probability can now be estimated using the multivariate Berry-Esseen theorem (Corollary~\ref{corollary:multidimensional-berry-esseen}) as 
\begin{equation}
\Pr(\calA_1\cap\calA_{12})\ge\Psi \big( L_1, L_1+ L_2 ; \bzero, \bV_{1, 12} \big) + O(n^{-1/4}).\label{eqn:drop2}
\end{equation}
We complete the proof of \eqref{eqn:corner_bd} similarly to Case (i). The converse is completely analogous.
\end{proof}

A couple of take-home messages  are in order:
  
 First, consider Case (i). In this case, we are operating ``far away'' from the constraint concerning the second rate  and the sum rate constraint. This corresponds to the events $\calA_2^c$ and $\calA_{12}^c$. Thus, by the theory of large deviations, $\Pr(\calA_2^c)$ and $\Pr(\calA_{12}^c)$ both tend to zero exponentially fast. Essentially for these two error events, we are in the error exponents regime.\footnote{Of course, the error exponents for the Slepian-Wolf problem are known~\cite{Csi82,Gal76} but any exponential bound suffices for our purposes here.} The same holds true for Case~(ii).
 
 Second, consider Case (iii). This is the most interesting case for the second-order asymptotics for the Slepian-Wolf problem.  We are operating at a corner point and are far away from the second rate constraint, i.e., in the error exponents regime for $\calA_2^c$. The remaining two events $\calA_1$ and $\calA_{12}$ are, however, still in the central limit regime and hence their {\em joint probability} must be estimated using the {\em multivariate} Berry-Esseen theorem. Instead of the result being expressible  in terms of a univariate Gaussian cdf $\Phi$ (which is the case for single-terminal problems in Part II of this monograph), the multivariate version of the Gaussian cdf $\Psi$, parameterized by the (in general, full)  {\em covariance   matrix} $\bV_{1,12}$ in \eqref{eqn:V112}, must be employed. Compared to the {\em cooperative} case where $X_2^n$ (resp.\ $X_1^n$) is available to encoder $1$ (resp.\ encoder $2$), we see from the result in Case (iii)   that Slepian-Wolf coding, in general, incurs a rate-loss over the case where side-information is available to all terminals. Indeed, when side-information is available at all terminals, the matrix $\bV_{1,12}$ that characterizes $\calL(\eps;R_1^*, R_2^*)$ in Case (iii) would be diagonal~\cite{TK14}, since the source coding problems involving $X_1$ and $X_2$ are now independent of each other. In other words, in this case, there exists a sequence of   codes  with error probabilities $\eps_n$ satisfying  \eqref{eqn:sw_err} and sizes $(M_{1n}, M_{2n})$  satisfying 
 \begin{align}
 \log M_{1n}  &\le nH(X_1|X_2) -  \sqrt{n V(X_1 |X_2) } \Phi^{-1}(\eps) + o(\sqrt{n}),  \label{eqn:src_X1}\\
  \log M_{2n}  &\le nH(X_2|X_1) -  \sqrt{n V(X_2 |X_1) } \Phi^{-1}(\eps) +o(\sqrt{n}),\label{eqn:src_X2}\\
    \log (M_{1n}M_{2n} )&\le nH(X_1,X_2) -  \sqrt{n V(X_1 , X_2) } \Phi^{-1}(\eps) + o(\sqrt{n})  .\label{eqn:src_X12}
 \end{align}
Inequality  \eqref{eqn:src_X1} corresponds to the problem of source coding $X_1$ with $X_2$ available as full (non-coded) side information at the decoder. Inequality \eqref{eqn:src_X2} swaps the role of $X_1$ and $X_2$. Finally, inequality \eqref{eqn:src_X12} corresponds to lossless source coding of the vector source $(X_1, X_2)$, similarly to the result  on lossless source coding without side information in Section~\ref{sec:asymp_lossless}.
%and \eqref{eqn:sw_err} is satisfied.

%The above observations were made by Tan-Kosut~\cite{TK14}, Nomura-Han~\cite{Nom13} as well as Haim-Erez-Kochman~\cite{Haim12}.

\section{Second-Order Asymptotics of Slepian-Wolf Coding via the Method of  Types}  \label{sec:sw_uni}
Just as in Section~\ref{sec:universal_lossless} (second-order asymptotics of lossless data compression via the method of types), we can show that codes that do not necessarily have to have full knowledge of the source statistics (i.e., partially universal source codes) can achieve the second-order coding rate region $\calL(\eps; R_1^*,R_2^*)$.  However, the coding scheme does require the knowledge of the entropies together with the pair of second-order rates $(L_1,L_2)$ we would like to achieve.   We illustrate the  achievability proof technique   for Case (iii) of Theorem~\ref{thm:disp_sw}, in which $R_1^* = H(X_1|X_2)$ and $R_2^*=H(X_2)$. %Our proof technique uses  the method of types. 

The code construction is based on  Cover's {\em random binning} idea~\cite{cover75} and the decoding strategy is similar to  {\em minimum empirical entropy decoding}~\cite{CK81,Csi97}.   Fix $(L_1, L_2)\in\calL(\eps; R_1^*,R_2^*)$ where $\calL(\eps; R_1^*,R_2^*)$ is given in \eqref{eqn:corner_bd}. Also  fix code sizes $M_{1n}$ and $M_{2n}$ satisfying \eqref{eqn:sizeM1}.  For each $j = 1,2$, uniformly  and independently assign each sequence $\bx_j\in\calX_j^n$ into one of $M_{jn}$ bins labeled as $\calB_j(m_j), m_j \in \{1,\ldots, M_{jn}\}$.  The bin assignments are revealed to all parties. To send $\bx_j \in\calX_j^n$, encoder $j$ transmits its bin index $m_j$.  

The decoder, upon receipt of the bin indices $(m_1, m_2 ) \in\{1,\ldots, M_{1n}\} \times \{1,\ldots, M_{2n}\}$, finds a pair of sequences  $(\hat{\bx}_1, \hat{\bx}_2) \in\calB_1(m_1)\times\calB_2(m_2)$ satisfying
\begin{align}
\hat{\bH}(\bx_1,\bx_2):=\begin{bmatrix}
\hatH(\bx_1 |\bx_2 )\\ \hatH(\bx_2 |\bx_1 ) \\ \hatH(\bx_1  , \bx_2 ) 
\end{bmatrix} \le \begin{bmatrix}
\gamma_1 \\ \gamma_2\\ \gamma_{12}
\end{bmatrix} =: \bgamma \label{eqn:vector_thres}
\end{align}
for some thresholds $\gamma_1,\gamma_2,\gamma_{12}$ defined as 
\begin{align}
\gamma_1  &:= H(X_1|X_2) + \frac{L_1}{\sqrt{n}}  + n^{-1/4} \\
\gamma_2  &:= H(X_2 ) + \frac{L_2}{\sqrt{n}}  + n^{-1/4} \\
\gamma_{ 12}  &:= H(X_1 , X_2 ) + \frac{L_1+L_2}{\sqrt{n}} + n^{-1/4}  
\end{align}
If there is no sequence pair   $(\hat{\bx}_1, \hat{\bx}_2) \in\calB_1(m_1)\times\calB_2(m_2)$  satisfying \eqref{eqn:vector_thres} or if there is more than one, declare an error.   Note that the thresholds depend on the entropies and $(L_1, L_2)$, hence these values need to be known to the decoder. 

  Let the generated source sequences be $X_1^n$ and $X_2^n$ and their associated bin indices be $M_1 = M_1(X_1^n)$ and $M_2 = M_2(X_2^n)$ respectively. By symmetry, we may assume that $M_1=M_2=1$.  The error events are  as follows:
\begin{align}
\calE_0 &:= \big\{  \hat{\bH}(X_1^n, X_2^n)  \not\le\bgamma\big\},\\*
\calE_1 &:= \big\{\exists\, \tilde{\bx}_1\in\calB_1(1) :\tilde{\bx}_1\ne X_1^n,  \hat{\bH}(\tilde{\bx}_1, X_2^n)  \le\bgamma\big\},\\
\calE_2 &:= \big\{\exists\, \tilde{\bx}_2\in\calB_2(1):  \tilde{\bx}_2 \ne X_2^n ,  \hat{\bH}(X_1^n, \tilde{\bx}_2)  \le\bgamma\big\},\quad\mbox{and}\\
\calE_{12} &:= \big\{\exists\, (\tilde{\bx}_1, \tilde{\bx}_2)\in\calB_1(1)\times\calB_2(1): \tilde{\bx}_1\ne X_1^n, \tilde{\bx}_2 \ne X_2^n , \nn\\*
&  \qquad\qquad  \hat{\bH}(\tilde{\bx}_1,\tilde{ \bx}_2)  \le\bgamma\big\}.
\end{align}
Let $\bH(X_1, X_2)=[H(X_1|X_2), H(X_2|X_1), H(X_1, X_2)]'$. It can be verified that the following central limit relation holds \cite{TK14}:
\begin{equation}
\sqrt{n}\big(\hat{\bH}(X_1^n, X_2^n)  - \bH(X_1,X_2) \big)\stackrel{\mathrm{d}}{\longrightarrow}\calN\big(\bzero ,\bV \big).
\end{equation}
This is the multi-dimensional analogue of \eqref{eqn:clt_lossles} for almost lossless source coding.
Thus, by the same argument as that in \eqref{eqn:drop1}--\eqref{eqn:drop2} (ignoring the second entry in $\hat{\bH}(X_1^n, X_2^n) $ because $R_2^*=H(X_2) > H(X_2|X_1)$), one has 
\begin{equation}
\Pr(\calE_0)\le \eps +O(n^{-1/4}).
\end{equation}
Furthermore, by using the method of types, we may verify that 
\begin{align}
\Pr(\calE_1) & \!\le\! \sum_{\bx_1,\bx_2} P_{X_1^n X_2^n} (\bx_1,\bx_2) \sum_{\tilde{\bx}_1\ne\bx_1: \hat{\bH}(\tilde{\bx}_1,\bx_2)\le\bgamma} \Pr \big(\tilde{\bx}_1\in\calB_1(1)\big)\\
&\! \le \!\sum_{\bx_1,\bx_2} P_{X_1^n X_2^n} (\bx_1,\bx_2) \sum_{\tilde{\bx}_1\ne\bx_1: \hatH(\tilde{\bx}_1 | \bx_2)\le\gamma_1} \Pr\big(\tilde{\bx}_1\in\calB_1(1)\big)\\
&\! = \!\sum_{\bx_1,\bx_2} P_{X_1^n X_2^n} (\bx_1,\bx_2) \sum_{\tilde{\bx}_1\ne\bx_1: \hatH(\tilde{\bx}_1 | \bx_2)\le\gamma_1} \frac{1}{ M_{1n} }\label{eqn:uniformity}\\
&\! \le \!\sum_{\bx_1,\bx_2} P_{X_1^n X_2^n} (\bx_1,\bx_2) \sum_{ \substack{V\in\scV_n(\calX_2;P_{\bx_2}) :\\ H(V|P_{\bx_2}) \le\gamma_1 }}\sum_{\tilde{\bx}_1\in\calT_V(\bx_2) }   \frac{1}{M_{1n}}\label{eqn:partition_types} \\
&\! \le \!\sum_{\bx_1,\bx_2} P_{X_1^n X_2^n} (\bx_1,\bx_2) \sum_{ \substack{V\in\scV_n(\calX_2;P_{\bx_2}) :\\ H(V|P_{\bx_2}) \le\gamma_1 }}  \frac{\exp\big(nH(V|P_{\bx_2})\big)}{M_{1n}} \label{eqn:size_cond_tc}\\
&\! \le \!\sum_{\bx_1,\bx_2} P_{X_1^n X_2^n} (\bx_1,\bx_2) \sum_{ \substack{V\in\scV_n(\calX_2;P_{\bx_2}) :\\ H(V|P_{\bx_2}) \le\gamma_1 }}  \frac{\exp\big(n \gamma_1\big)}{M_{1n}} \\*
&\!\le\! (n+1)^{|\calX_1| |\calX_2|} \exp( - n^{ 1/4}),\label{eqn:use_type_count}
\end{align}
where in \eqref{eqn:uniformity} we used the uniformity of the binning, in \eqref{eqn:partition_types}  we partitioned the set of sequences $\tilde{\bx}_1$ into conditional types given $\bx_2$ and in \eqref{eqn:size_cond_tc}, we used the fact that $|\calT_V(\bx_2)|\le \exp\big( nH(V|P_{\bx_2})\big)$ (cf.~Lemma~\ref{lem:size_type_class}).  Finally, the type counting lemma  and the choices of $\gamma_1$ and $M_{1n}$ were used in \eqref{eqn:use_type_count}. The same calculation can be performed for $\Pr(\calE_2)$ and $\Pr(\calE_{12})$. Thus, asymptotically, the error probability is no larger than $\eps$, as desired. 
\section{Other    Fixed Error Asymptotic  Notions} \label{sec:other_fixed}
In the preceding sections, we were solely concerned with the deviations of order $\Theta(\frac{1}{\sqrt{n}})$ away from the first-order fundamental limit $(R_1^*, R_2^*)$.   However, one may also be interested in other metrics that quantify backoffs from particular first-order fundamental limits. Here we mention three other quantities that have appeared in the literature.

\subsection{Weighted Sum-Rate Dispersion}
For constants $\alpha, \beta\ge 0$, the minimum value of  $\alpha R_1 + \beta R_2$
for asymptotically achievable $(R_1, R_2)$ is called the  {\em optimal weighted sum-rate}. Of particular interest is the case
$\alpha =\beta=1$, corresponding to the standard sum-rate $R_1+R_2$, but
other cases may be important as well, e.g.,  if transmitting
from encoder $1$ is more costly than transmitting from encoder
$2$. Because of the polygonal shape of the optimal region
described in the Slepian-Wolf region in~\eqref{eqn:orr_sw}, the optimal weighted sum-rate is always
achieved at (at least) one of the two corner points, and the optimal rate
is given by
\begin{equation}
R_{\mathrm{sum}}^*(\alpha,\beta) :=  \left\{ \begin{array}{cc}
 \alpha H(X_1|X_2) + \beta H(X_2)& \alpha\ge \beta \\
  \alpha H(X_1) + \beta H(X_2|X_1)& \alpha < \beta 
\end{array} \right. .
\end{equation}
One can then define $J \in\bbR$ to be  an  {\em achievable $(\eps,\alpha,\beta)$-weighted second-order coding rate} if there exists a sequence of $(M_{1n}, M_{2n}, \eps_n)$-codes for the correlated source $P_{X_1^n X_2^n}$ such that  the error probability condition in~\eqref{eqn:sw_err} holds and 
\begin{equation}
\limsup_{n\to\infty}  \frac{1}{\sqrt{n}} \big( \alpha \log M_{1n}  + \beta \log M_{2n}  -n R_{\mathrm{sum}}^*(\alpha,\beta) \big)\le J.
\end{equation}
In \cite{TK14}, the smallest such $J$, denoted as $J^*(\eps;\alpha,\beta)$, was found using a proof technique similar to that for Theorem~\ref{thm:disp_sw}.

\subsection{Dispersion-Angle Pairs}
One can also imagine approaching a point on the boundary $(R_1^*, R_2^*)$ {\em fixing an angle of approach $\theta \in [0,2\pi)$}. Let $(F,\theta)$ be called an {\em achievable $(\eps,R_1^*,R_2^*)$-dispersion-angle pair} if there exists a sequence of $(M_{1n}, M_{2n}, \eps_n)$-codes for the correlated source $P_{X_1^n X_2^n}$ such that    the error probability condition in~\eqref{eqn:sw_err} holds and
\begin{align}
\limsup_{n\to\infty}  \frac{1}{\sqrt{n}}  \big( \log M_{1n} - n R_1^* \big) &\le \sqrt{F}  \,\cos\theta  ,\quad\mbox{and} \\*
\limsup_{n\to\infty}  \frac{1}{\sqrt{n}}  \big( \log M_{2n} - n R_2^* \big) & \le \sqrt{F} \,  \sin\theta   .
\end{align}
Clearly,   dispersion-angle pairs $(F,\theta)$  are in one-to-one correspondence with second-order coding rate pairs $(L_1, L_2)$. The minimum such $F$ for a given $\theta$, denoted as $F^*(\theta , \eps; R_1^*, R_2^*)$, measures the speed of approach to $(R_1^*, R_2^*)$   at an angle $\theta$.  This fundamental quantity $F^*(\theta  ,\eps; R_1^*, R_2^*)$ was also characterized in \cite{TK14}.

\subsection{Global Approaches} \label{sec:glob}
Authors of early works on second-order  asymptotics in multi-terminal systems~\cite{huang12,Mol12,TK12a} considered {\em global rate regions}, meaning that they  were concerned with quantifying the sizes $(M_{1n}, M_{2n})$ of length-$n$ block codes with error probability not exceeding $\eps$. These sizes  are  called {\em $(n,\eps)$-achievable}. In the Slepian-Wolf context, a result by Tan-Kosut~\cite{TK12a} states that $(M_{1n}, M_{2n})$ are $(n,\eps)$-achievable  iff
\begin{align}
 \begin{bmatrix}\log M_{1n} \\ \log M_{2n} \\ \log (M_{1n}  M_{2n})
\end{bmatrix} \!\in \! \begin{bmatrix}
nH(X_1|X_2) \\ nH(X_2|X_1)  \\ nH(X_1 , X_2) 
\end{bmatrix}  \!- \! \sqrt{n} \Psi^{-1}(\bV,\eps) \! + \! O\left(\log n\right)\bone  \label{eqn:glo_sw}
\end{align}
where ${\Psi^{-1}}(\bV,\eps)$ is an appropriate generalization of the $\Phi^{-1}$ function and $\bone$ is the vector of all ones.  The precise definition of $\Psi^{-1}(\bV,\eps)$, given in \eqref{eqn:psi_inv1} and illustrated Fig.~\ref{fig:psi_inv}, will not be of concern here. 

While statements such as \eqref{eqn:glo_sw} are mathematically correct and are reminiscent of asymptotic expansions in the point-to-point case (cf.~that for lossless source coding in~\eqref{eqn:nonzero_VP}), they do not provide the complete picture with regard to the  convergence of rate pairs to a fundamental limit, e.g.,  a corner point of the Slepian-Wolf region. Indeed, an {\em achievability}  statement similar  to \eqref{eqn:glo_sw} holds for the  DM-MAC {\em for each input distribution} \cite{huang12,Mol12,Scarlett13b,TK12a}
%, i.e., 
%\begin{equation}
% \begin{bmatrix}\log M_{1n} \\ \log M_{2n} \\ \log (M_{1n}  M_{2n})
%\end{bmatrix}  \in \scR_n(P_Q, P_{X_1 | Q }, P_{X_2 | Q } )
%\end{equation}
%where 
%\begin{align}
%& \scR_n(P_Q, P_{X_1 | Q }, P_{X_2 | Q } )   \nn\\*
%& = \begin{bmatrix}
%nI(X_1 ;Y|X_2,Q)\\ nI(X_2 ;Y|X_1,Q) \\ nI(X_1,X_2 ;Y|Q)
%\end{bmatrix}  + \sqrt{n} \Psi^{-1}(\bV(P_Q, P_{X_1 | Q }, P_{X_2 | Q } ) ,\eps) + O\left(\log n\right)\bone 
%\end{align}
 and hence the {\em union} over all input distributions. However,   one of the major deficiencies of such a statement  is that the $O(\log n)$ third-order term is not uniform in the input distributions; this poses serious challenges in the interpretation of the result if we consider  random coding using a sequence of input distributions that varies  with the  blocklength (cf.\ Chapter \ref{ch:mac}). Thus, as pointed out by Haim-Erez-Kochman~\cite{Haim12}, for multi-user problems, the value of {\em global  expansions} such as that in \eqref{eqn:glo_sw} is   limited, and can only be regarded as stepping stones to obtain {\em local} results (if possible).  Indeed, we do this for the Gaussian MAC with degraded message sets in Chapter~\ref{ch:mac}.  

The {\em main takeaway} of this section is that one should adopt  the local, weighted sum-rate, or dispersion-angle problem setups to analyze the second-order asymptotics for multi-terminal problems. These setups are  information-theoretic in nature. In particular, operational quantities  (such as the set $\calL(\eps;R_1^*, R_2^*)$ or the number $F^*(\theta, \eps; R_1^*, R_2^*)$)  are \emph{defined} then    \emph{equated} to information quantities. 
\chapter{A Special Class of Gaussian Interference Channels} \label{ch:ic}

This chapter presents results on second-order asymptotics for a channel-type network information theory problem. The problem we consider here is a special case of the two-sender, two-receiver interference channel (IC) shown in Fig.~\ref{fig:int}. This model is a basic building block in many modern wireless systems, so theoretical results and insights are   of tremendous practical relevance. The IC   was first studied  by Ahlswede~\cite{Ahls74} who established basic bounds on the capacity region. However, the capacity region for the discrete memoryless and Gaussian memoryless cases have remained as open problems for over $40$ years except for some very  special cases. The best known inner bound is due to Han and Kobayashi~\cite{Han81}. A simplified form of the Han-Kobayashi inner bound was presented by Chong-Motani-Garg-El Gamal~\cite{Chong08}. %  Approximations to the capacity region for the Gaussian case were studied recently by Avestimehr-Diggavi-Tse~\cite{ADT}. 

\begin{figure}[t]
\centering
\setlength{\unitlength}{.4mm}
\begin{picture}(190, 90)
%\thicklines
\put(0, 15){\vector(1, 0){30}}
\put(60, 15){\vector(1,0){30}}
%\put(120, 15){\vector(1,0){30}}
%\put(180, 15){\vector(1,0){30}}
\put(30, 0){\line(1, 0){30}}
\put(30, 0){\line(0,1){30}}
\put(60, 0){\line(0,1){30}}
\put(30, 30){\line(1,0){30}}

\put(0, 75){\vector(1, 0){30}}
\put(60, 75){\vector(1, 0){30}}
\put(0, 15){\vector(1, 0){30}}
%\put(60, 15){\vector(1,0){30}}
%\put(120, 45){\vector(1,0){30}}
%\put(180, 15){\vector(1,0){30}}
\put(30, 60){\line(1, 0){30}}
\put(30, 60){\line(0,1){30}}
\put(60, 60){\line(0,1){30}}
\put(30, 90){\line(1,0){30}}

\put(90, 0){\line(1, 0){30}}
\put(90, 0){\line(0,1){90}}
\put(120, 0){\line(0,1){90}}
\put(90, 90){\line(1,0){30}}

\put(0, 20){  $m_2$}
\put(0, 80){  $m_1$}
\put(66, 20){  $x_2$}
\put(66, 80){  $x_1$}
%\put(51, -10){  $\bbE[\rvg(X)]\le\Gamma$}
%\put(65, 8){  $[2^{nR}]$}
%\put(140, 50){  $(\hatx_1,\hatx_2)$} 
\put(41, 12){$f_2$ } 
\put(41, 72){$f_1$ } 
\put(101, 42){$W$} 

\put(126, 20){  $y_2$}
\put(126, 80){  $y_1$}

\put(186, 20){  $\hatm_2$}
\put(186, 80){  $\hatm_1$}

\put(120, 15){\vector(1, 0){30}}
\put(120, 75){\vector(1, 0){30}}

\put(150, 0){\line(1, 0){30}}
\put(150, 0){\line(0,1){30}}
\put(180, 0){\line(0,1){30}}
\put(150, 30){\line(1,0){30}}

\put(150, 60){\line(1, 0){30}}
\put(150, 60){\line(0,1){30}}
\put(180, 60){\line(0,1){30}}
\put(150, 90){\line(1,0){30}}

\put(180, 15){\vector(1, 0){30}}
\put(180, 75){\vector(1, 0){30}}

\put(161, 12){$\varphi_2$ } 
\put(161, 72){$\varphi_1$ } 
  \end{picture}
  \caption{Illustration of the interference channel problem.   }
  \label{fig:int}
\end{figure}
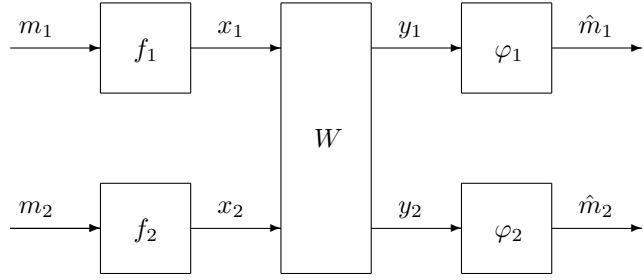

Since the determination of the capacity region is formidable,  the derivation of conclusive results for the second-order asymptotics of general  memoryless ICs is also beyond us at this point in time. One very special case in which the capacity region is known is the  IC with {\em very strong interference} (VSI). In this case, the intuition is that each receiver can reliably decode the non-intended message which then aids in decoding the intended message. The capacity region for the discrete memoryless IC with VSI consists of the set of rate pairs $(R_1, R_2)$ satisfying
\begin{equation}
R_1 \le I(X_1;Y_1 | X_2, Q),\quad\mbox{and}\quad  R_2 \le I(X_2;Y_2 | X_1, Q)\label{eqn:vsi_disc}
\end{equation}
for some $P_{Q} , P_{X_1|Q}$ and $P_{X_2|Q}$, where  $Q$ is known as the {\em time-sharing random variable}. In the Gaussian case in which Carleial~\cite{Carleial75} studied, the above region can be written more  explicitly as 
\begin{equation}
R_1 \le \rvC(\snr_1) \quad\mbox{and}\quad R_2 \le \rvC(\snr_2), \label{eqn:carl}
\end{equation}
where $\snr_j$ is the signal-to-noise ratio  of the direct channel from sender $j$ to receiver $j$ and the Gaussian capacity function is  defined as  $\rvC(\snr) := \frac{1}{2}\log (1+\snr)$.   See Fig.~\ref{fig:vs} for an illustration of the capacity region and~the monograph by Shang and Chen~\cite{Shang} for further discussions on Gaussian interference channels. Carleial's result    is   surprising because it appears that interference does not reduce the capacity of the constituent channels since $\rvC(\snr_j)$ {\em is}  the capacity of the $j^{\mathrm{th}}$ channel. In Carleial's own words \cite{Carleial75},
\begin{quote}
\emph{``Very strong interference  is as innocuous as no interference at all.''}  
\end{quote}
 Similarly to the discrete case in \eqref{eqn:vsi_disc}, the (first-order optimal) achievability proof strategy for the Gaussian case  involves first decoding the interference, subtracting it off from the received channel output, and finally, reliably decoding the intended message. The VSI condition ensures that the rate constraints in \eqref{eqn:carl}, representing requirement for the second decoding steps to succeed, dominate.

In this chapter, we make a slightly stronger  assumption compared to that made by Carleial~\cite{Carleial75}. We assume that the inequalities that define the VSI condition are strict; we call this the {\em strictly VSI} (SVSI) assumption/regime. With this assumption, we are able to derive the second-order asymptotics of this class of Gaussian ICs. 

Although the main result in this chapter appears to be  similar to the Slepian-Wolf case (in Chapter~\ref{ch:sw}), there are several take-home messages that differ from the simpler Slepian-Wolf problem. 
\begin{enumerate}
\item First, similar to Carleial's observation  that for Gaussian ICs with VSI the capacity is not reduced, we show that the {\em dispersions are   not affected} under the SVSI assumption. More precisely,  the second-order coding rate region (a set similar to that for the Slepian-Wolf problem in Chapter \ref{ch:sw}), is characterized entirely in terms of the dispersions $\rvV(\snr_j)$ of the two {\em direct}  AWGN channels from encoder $j$ to decoder $j$;
\item Second, the main result in this chapter suggests that under the SVSI assumption, and in the second-order asymptotic setting, the two error events (of incorrectly decoding messages $1$ and $2$) are almost independent;
\item Third,  for the direct part, we demonstrate  the  utility of an achievability  proof  technique by MolavianJazi-Laneman~\cite{Mol13} that is also applicable to our problem of  Gaussian ICs with SVSI. This technique   is,  in general, applicable to multi-terminal Gaussian channels. In the  asymptotic evaluation of the information spectrum bound (Feinstein bound~\cite{Feinstein}), the problem is ``lifted''  to higher dimensions to  facilitate the application of limit theorems for {\em independent}  random vectors;
%\item  Finally, it is shown that the evaluation of a Verd\'u-Han-type information spectrum converse bound \cite[Lem.~4]{VH94} suffices to ensure that the codewords of users $1$ and $2$ are independent, thus more sophisticated strong converse techniques (cf.~Ahlswede's strong converse proof of the discrete memoryless multiple-access channel~\cite{Ahl82}) are not required. 
\end{enumerate}
This chapter is based on work by Le, Tan and Motani~\cite{Quoc14}.

\section{Definitions and Non-Asymptotic Bounds}
Let us now state  the Gaussian IC problem. The Gaussian IC is defined by the following input-output relation:
\begin{align}
Y_{1i}  &= g_{11} X_{1i}  + g_{12} X_{2i} + Z_{1i} , \label{eqn:ic1} \\*
Y_{2i}  &= g_{21} X_{1i}  + g_{22} X_{2i} + Z_{2i}  , 
\end{align}
where  $i =1,\ldots, n$ and $g_{jk}$ are the  channel gains from sender $k$ to receiver $j$ and $Z_{1i}\sim\calN(0,1)$ and $Z_{2i}\sim\calN(0,1)$  are {\em independent} noise components.\footnote{The  independence  assumption between $Z_{1i}$ and $Z_{2i}$ was not made in Carleial's work~\cite{Carleial75}  (i.e., $Z_{1i}$ and $Z_{2i}$  may be correlated) but we need this assumption for the analyses here. It is well known that the capacity region of any general  IC depends only on the marginals   \cite[Ch.~6]{elgamal} but it is, in general, not true that  the set of achievable second-order rates    $\calL(\eps; R_1^*, R_2^*)$, defined in \eqref{eqn:2nd_ic}, has the same property. This will become clear in the proof of Theorem~\ref{thm:disp_vs} in the text following~\eqref{eqn:ic_stats}.  }     Thus, the channel from $(x_{1},x_2)$ to $(y_1 , y_2)$ is 
\begin{equation}
W(y_1, y_2 | x_1, x_2) = \frac{1}{2\pi}\exp\Bigg( -\frac{1}{2}\bigg\| \begin{bmatrix}
y_1 \\ y_2
\end{bmatrix}- \begin{bmatrix} g_{11} & g_{12} \\ g_{21} & g_{22} \end{bmatrix} \begin{bmatrix}
x_1  \\ x_2
\end{bmatrix}\bigg\|_2^2\Bigg).
\end{equation}
Let $W_1$ and $W_2$ denote the marginals of $W$.  The channel also acts in a stationary, memoryless way so
\begin{equation}
W^n(\by_1,\by_2|\bx_1,\bx_2) =\prod_{i=1}^n W(y_{1i}, y_{2i} | x_{1i}, x_{2i}). 
\end{equation}

It will be convenient  to make the dependence of the code on the blocklength explicit right away. We define an  {\em $(n,M_1, M_2,S_1, S_2,\eps)$-code} for the Gaussian IC as four maps that consists of two encoders $f_j : \{1,\ldots, M_j\}\to\bbR^n , j = 1,2$ and two decoders $\varphi_j : \bbR^n\to \{1,\ldots, M_j\}$ such that the following {\em power constraints}\footnote{The notation $f_{ji}(m_j)$ denotes the $i^{\mathrm{th}}$ coordinate of the codeword  $f_j(m_j ) \in \bbR^n$.} are satisfied 
\begin{equation}
\big\|  f_j(m_j ) \big\|_2^2 =\sum_{i=1}^n f_{ji}(m_j)^2\le nS_j  \label{eqn:snr_ic}
\end{equation}
and, denoting $\calD_{m_1,m_2}:= \{ (\by_1, \by_2) : \varphi_1(\by_1)= m_1\mbox{ and } \varphi_2(\by_2)= m_2\}$ as the decoding region for $(m_1, m_2)$,   the {\em average error probability} 
\begin{equation}\label{eqn:error_ic}
\frac{1}{M_1 M_2}\sum_{m_1=1}^{M_1}\sum_{m_2=1}^{M_2}W^n\big( \bbR^n\times\bbR^n\setminus\calD_{m_1, m_2}\big| f_1(m_1) , f_2(m_2) \big)\le\eps.
\end{equation}
In \eqref{eqn:snr_ic}, $S_1$ and $S_2$ are the admissible powers on the codewords $f_1(m_1)$ and $f_2(m_2)$. The {\em signal-to-noise  ratios} of the direct channels are 
\begin{equation}
\snr_1 :=g_{11}^2S_1,\quad\mbox{and}\quad \snr_2 := g_{22}^2 S_2. 
\end{equation}
The {\em interference-to-noise ratios} are 
\begin{equation}
\inr_1 := g_{12}^2 S_2,\quad\mbox{and}\quad \inr_2 := g_{21}^2  S_1.
\end{equation}

We say that the Gaussian IC $W$, together with the transmit  powers $(S_1, S_2)$,  is in the {\em VSI regime} 
%if 
%\begin{equation}
%g_{11}^2 \le \frac{g_{21}^2 }{1+  g_{22}^2 S_2},\quad\mbox{and} \quad g_{22}^2 \le \frac{g_{12}^2 }{1+  g_{11}^2 S_1} \label{eqn:svsi1}
%\end{equation}
%or in terms of the signal- and interference-to-noise ratios,
 if the signal- and interference-to-noise ratios satisfy
\begin{equation}
\snr_1 \le \frac{\inr_2}{1+\snr_2},\quad\mbox{and} \quad  \snr_2 \le \frac{\inr_1}{1+\snr_1},\label{eqn:svsi2}
\end{equation}
or equivalently, in terms of capacities, 
\begin{equation}
\rvC(\snr_1 )+\rvC(\snr_2)\le\min\{ \rvC(\snr_1 + \inr_1),  \rvC(\snr_2 + \inr_2) \} . \label{eqn:svsi3}
\end{equation}
The Gaussian IC is in the {\em SVSI regime} if the inequalities in  \eqref{eqn:svsi2}--\eqref{eqn:svsi3} are strict. 
% We assume that the Gaussian IC  is in the SVSI regime. 
Intuitively, the VSI (or SVSI) assumption means that the cross channel gains $g_{12}$ and $g_{21}$ are much stronger (larger) than the direct gains $g_{11}$ and $g_{22}$ for given transmit powers $S_1$ and $S_2$.

We now state non-asymptotic bounds that are   evaluated asymptotically later. The proofs of these bounds  are standard. See~\cite{Han98} or~\cite{Quoc14}.

\begin{proposition}[Achievability bound for   IC] \label{prop:ach_ic}
Fix any  input distributions  $P_{X_1^n}$ and $P_{X_2^n}$ whose support satisfies the power constraints in \eqref{eqn:snr_ic},  i.e., $\|X_j^n\|_2\le nS_j$ with probability one.    For every $n\in\bbN$, every $\gamma>0$ and for any choice of (conditional) output distributions $Q_{Y_1^n | X_2^n}$, $Q_{Y_2^n | X_1^n}$, $Q_{Y_1^n}$ and $Q_{Y_2}^n$   there exists an $(n,M_{1}, M_{2}, S_1, S_2, \eps)$-code for the   IC such that 
\begin{align}
\eps  \le  \Pr \bigg(  \log\frac{W_1^n(Y_1^n | X_1^n , X_2^n) }{Q_{Y_1^n|X_2^n} (Y_1^n|X_2^n) } & \le\log M_{1} + n\gamma   \quad\mbox{or } \nn\\ 
\log\frac{W_2^n(Y_2^n | X_1^n , X_2^n) }{Q_{Y_2^n|X_1^n} (Y_2^n|X_1^n) } & \le\log M_{2} + n\gamma   \quad\mbox{or } \nn\\
\log\frac{W_1^n(Y_1^n | X_1^n , X_2^n) }{Q_{Y_1^n } (Y_1^n ) } & \le\log (M_{1} M_{2}) \!+\! n\gamma   \quad\mbox{or } \nn\\
\log\frac{W_2^n(Y_2^n | X_1^n , X_2^n) }{Q_{Y_2^n } (Y_2^n ) } & \le\log (M_{1} M_{2}) \!+ \!n\gamma    \bigg)   +\zeta\exp(-n\gamma), \label{eqn:direct_ic}
\end{align}
where $\zeta:=\sum_{k=1}^2\sum_{j =1}^2 \zeta_{jk} $ and 
\begin{align}
\zeta_{11} & := \sup_{\bx_2,\by_1} \frac{ P_{X_1^n} W_1^n (\by_1 | \bx_2) }{ Q_{Y_1^n|X_2^n}(\by_1|\bx_2) },\quad \zeta_{12} :  = \sup_{ \by_1} \frac{ P_{X_1^n}P_{X_2^n} W_1^n (\by_1 ) }{ Q_{Y_1^n }(\by_1 ) }  \label{eqn:K1}\\*
\zeta_{21} & := \sup_{\bx_1,\by_2} \frac{ P_{X_2^n} W_2^n (\by_2 | \bx_1) }{ Q_{Y_2^n|X_1^n}(\by_2|\bx_1) },\quad 
\zeta_{22}  := \sup_{ \by_2} \frac{ P_{X_1^n}P_{X_2^n} W_2^n (\by_2 ) }{ Q_{Y_2^n }(\by_2 ) }  \label{eqn:K2}
\end{align}
\end{proposition}
This is a generalization of the average error version of Feinstein's lemma~\cite{Feinstein} (Proposition~\ref{prop:fein}). Notice that we have the freedom to choose the  output distributions at the cost of having to control the ratios $\zeta_{jk}$ of the induced output distributions and our choice  of output distributions.

\begin{proposition}[Converse bound for   IC]  \label{prop:conv_ic}
For every $n\in\bbN$, every $\gamma>0$ and for any choice of (conditional)  output distributions $Q_{Y_1^n | X_2^n}$ and $Q_{Y_2^n | X_1^n}$, every $(n,M_{1}, M_{2}, S_1, S_2, \eps)$-code for the   IC must satisfy
\begin{align}
\eps \ge  \Pr \bigg(  \log\frac{W_1^n(Y_1^n | X_1^n , X_2^n) }{Q_{Y_1^n|X_2^n} (Y_1^n|X_2^n) } & \le\log M_{1} - n\gamma   \quad\mbox{or } \nn\\ 
\log\frac{W_2^n(Y_2^n | X_1^n , X_2^n) }{Q_{Y_2^n|X_1^n} (Y_2^n|X_1^n) } & \le\log M_{2}- n\gamma    \bigg) -2\exp(-n\gamma)\label{eqn:converse_ic}
\end{align}
for some   input distributions  $P_{X_1^n}$ and $P_{X_2^n}$ whose support satisfies the power constraints in \eqref{eqn:snr_ic}. 
\end{proposition}
Observe the following features of the non-asymptotic converse, which is a generalization of the ideas of Verd\'u-Han~\cite[Lem.~4]{VH94}  and Hayashi-Nagaoka~\cite[Lem.~4]{Hayashi03}: First, there are only two error events compared to the four in the achievability bound. The SVSI assumption allows us to eliminate two error events in the direct bound so the two bounds match in the second-order sense. Second, we are free to choose   output distributions without any penalty (cf.~the achievability bound in Proposition~\ref{prop:ach_ic}). Third, the intuition behind this bound is in line with the SVSI assumption--namely that decoder~$1$ knows the codeword $X_2^n$ and vice versa. Indeed, the proof of Proposition~\ref{prop:conv_ic} uses this genie-aided idea.

\section{Second-Order Asymptotics}
Similar to the study of the second-order asymptotics for the Slepian-Wolf problem, we are interested in deviations from the boundary of the capacity region of order $O(\frac{1}{\sqrt{n}})$ for the Gaussian IC under the SVSI assumption. This motivates the following definition.

Let $(R_1^*, R_2^*)$ be a point on the boundary of the capacity region in~\eqref{eqn:carl}. Let $(L_1, L_2)\in\bbR^2$ be called an {\em achievable $(\eps, R_1^*, R_2^*)$-second-order coding rate pair} if there exists a sequence of $(n, M_{1n}, M_{2n}, S_1, S_2 ,\eps_n)$-codes for the Gaussian IC such that 
\begin{align}
\limsup_{n\to\infty}\eps_n\le\eps ,\quad\mbox{and}\quad \liminf_{n\to\infty}\frac{1}{\sqrt{n}} \big( \log M_{jn}-nR_j^*\big)\ge L_j,\label{eqn:2nd_ic}
\end{align}
for $j=1,2$. 
The set of all achievable $(\eps, R_1^*, R_2^*)$-second-order coding rate pairs is denoted as $\calL(\eps; R_1^*, R_2^*) \subset\bbR^2$.   The intuition behind this definition is exactly analogous to the Slepian-Wolf case. 

Define $V_j :=\rvV(\snr_j)$ where, recall from \eqref{eqn:gauss_disp}  that,
\begin{equation} 
\rvV(\snr) =\log^2\rme\cdot \frac{\snr (\snr+2 )}{2(\snr+1)^2}
\end{equation}
 is the Gaussian dispersion function.

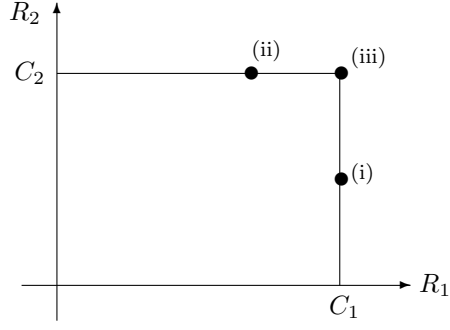
\begin{figure}[t]
\centering
\begin{picture}(115, 115)
\setlength{\unitlength}{.47mm}
\put(0, 10){\vector(1, 0){110}}
\put(10, 0){\vector(0,1){90}}
%\put(50, 30){\line(1, 0){55}}
%\put(30, 50){\line(0,1){55}}
%\put(50, 30){\line(-1, 1){20}}
\put(110,8){ $R_1$}
\put(-6, 85){ $R_2$}
%\multiput(10,50)(4,0){5}{\line(1,0){2}}
%\put(-5, 50){$H_2$}
%\put(-10, 30){$H_{2|1}$}
%\multiput(10,30)(4,0){10}{\line(1,0){2}}
%\multiput(50,10)(0,4){5}{\line(0,1){2}}
%\multiput(30,10)(0,4){10}{\line(0,1){2}}
%\put(50, 0){$H_1$}
%\put(25, 0){$H_{1|2}$}
 
%\put(65, 60){\vector(-1,-2){15}}
%\put(66,61 ){\footnotesize $\vec{\bd}$}
%\qbezier(53,36)(56,34)(56,30)
%\put(58,32){\footnotesize $\theta$}
%\put(56,41){\circle*{4}}
%%\put(61,45){\footnotesize $(   R_1(n,\epsilon),   R_2(n,\epsilon)   )$}
%\put(30,50){\circle*{4}}
%%\put(52,20){\footnotesize $(   R_1^*,    R_2^*   )$}
%
\put(93,73){\footnotesize (iii)}
\put(93,40){\footnotesize (i)}
%
%\put(30,68){\circle*{4}}
%
%\put(40,40){\circle*{4}}
%
\put(65,75){\footnotesize (ii)}

\put(10, 70){\line(1, 0){80}}
\put(90, 10){\line(0, 1){60}}

\put(65, 70){\circle*{4}}
\put(90.5, 70){\circle*{4}}
\put(90.5, 40){\circle*{4}}

\put(-2, 68){$C_2$}
\put(87, 2){$C_1$}
\end{picture}
\caption{Illustration of the different cases in Theorem~\ref{thm:disp_vs}. For brevity, we write $C_j = \rvC(\snr_j)$ for $j = 1,2$.  }
\label{fig:vs}
\end{figure}

\begin{theorem} \label{thm:disp_vs}
Let the Gaussian IC $W$, together with the transmit powers $(S_1, S_2)$, be in the SVSI regime. Depending on $(R_1^*, R_2^*)$ (see Fig.~\ref{fig:vs}), there are $3$ different cases:   \\
Case (i): $R_1^* = \rvC(\snr_1)$ and $R_2^* < \rvC(\snr_2)$ (vertical boundary)
\begin{equation}
\calL(\eps; R_1^*,R_2^* ) = \Big\{ (L_1, L_2) : L_1\le \sqrt{V_1}\Phi^{-1}(  \eps)\Big\}.  \label{eqn:vert_bd_ic}
\end{equation}
Case (ii): $R_1^* < \rvC(\snr_1)$ and $R_2^* = \rvC(\snr_2)$ (horizontal boundary)
\begin{equation}
\calL(\eps; R_1^*,R_2^* ) = \Big\{ (L_1, L_2) : L_2\le \sqrt{V_2}\Phi^{-1}(  \eps)\Big\}.  \label{eqn:hor_bd_ic}
\end{equation}
Case (iii): $R_1^* =\rvC(\snr_1)$ and $R_2^* = \rvC(\snr_2)$ (corner point)
\begin{equation}
\calL(\eps; R_1^*,R_2^* ) = \bigg\{ (L_1, L_2) :\Phi\Big(- \frac{L_1}{\sqrt{V_1}}\Big)\Phi\Big( -\frac{L_2}{\sqrt{V_2}}\Big) \ge 1-\eps\bigg\}.  \label{eqn:cor_pt_ic}
\end{equation}
\end{theorem}
A proof sketch of this result is provided in Section \ref{sec:prf_vs}.  The region $\calL(\eps; R_1^*,R_2^* ) $ for Case (iii) is sketched in   Fig.~\ref{fig:region_ic} for the symmetric case in which $V_1=V_2$.

A few remarks are in order:  First, for Case (i), $\calL(\eps; R_1^*, R_2^*)$  depends only on $\eps $ and $V_1$.   Note that $\sqrt{V_1}\Phi^{-1}(\eps )$ is   the optimum (maximum) second-order coding rate of the AWGN channel  (Theorem \ref{thm:awgn_asy}) from $X_1$ to $Y_1$  when there is no interference, i.e., $g_{12}=0$ in~\eqref{eqn:ic1}. The fact that user $2$'s parameters do not feature in \eqref{eqn:vert_bd_ic}   is because $R_2^*< \rvC(\snr_2)$. This implies that the channel $2$ operates     in large deviations (error exponents) regime so the second constraint in \eqref{eqn:carl} does not feature  in the second-order analysis, since the error probability of decoding message~$2$ is exponentially small.  An analogous observation was also made for the Slepian-Wolf problem in Chapter~\ref{ch:sw}. 

\begin{figure}[t]
\centering
\includegraphics[width = .7\columnwidth]{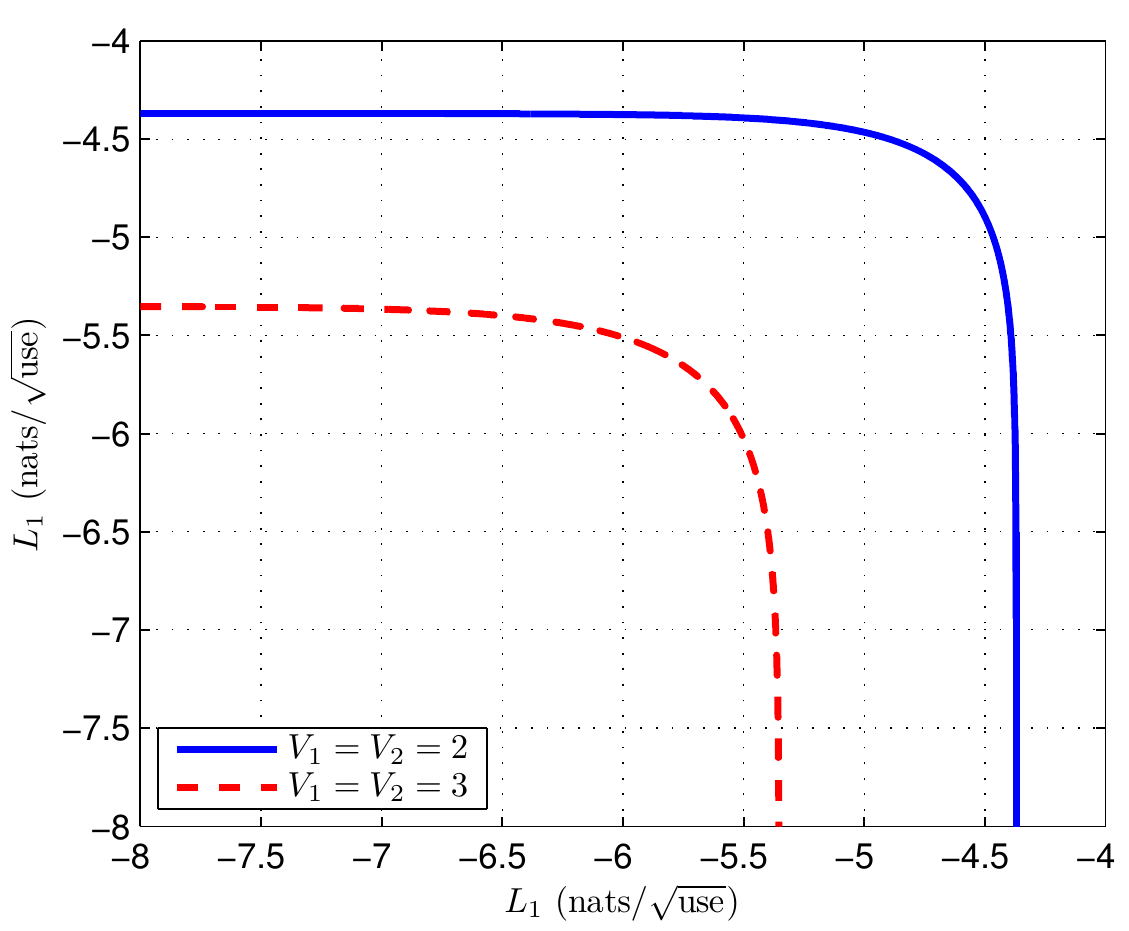}
\caption{Illustration of the region  $\calL(\eps; R_1^*,R_2^* )$ in Case (iii) with $\eps=10^{-3}$. The regions are to the bottom left of the boundaries indicated. }
\label{fig:region_ic}
\end{figure}

Second, notice that   for Case (iii), $\calL(\eps; R_1^*, R_2^*)$ is a function of $\eps$ and  {\em both} $V_1$ and $V_2$ as we are operating at rates near the {\em corner point} of the capacity region. Both constraints  in the capacity region in~\eqref{eqn:carl} are active. We provide an intuitive reasoning for the result in \eqref{eqn:cor_pt_ic}. 
%Roughly speaking $\Phi(-L_j/\sqrt{V_j})$ is the probability that the $j^{\mathrm{th}}$-decoder decodes correctly if the number of codewords in the $j^{\mathrm{th}}$-codebook is
%\begin{equation}
%M_{jn} = \big\lfloor\exp \big( n \kappa_j  + \sqrt{n}L_j + o(\sqrt{n } )  \big) \big\rfloor. \label{eqn:size_code}
%\end{equation}
%Thus, the product $\Phi(-L_1/\sqrt{V_1}) \Phi(-L_2/\sqrt{V_2})$, which is constrained to be larger than $1-\eps $ in \eqref{eqn:cor_pt_ic}, is the probability that {\em both} messages  $\rvM_1,\rvM_2$ are decoded correctly assuming that channels operate {\em independently}. More explicitly,  
Let $\calG_j$ denote the event that message $j=1,2$ is decoded correctly. The error probability criterion in~\eqref{eqn:error_ic} can be rewritten as 
\begin{align}
 \Pr \big( \calG_1  \cap\calG_2 \big)  \geq 1 - \eps  . \label{eqn:1minuseps}
 \end{align}
% where $\rvM_1$ and $\rvM_2$ are messages uniformly distributed on $\{1,\ldots, M_{1n}\}$ and  $\{1,\ldots, M_{2n}\}$  respectively. 
Assuming {\em independence} of the events $\calG_1$ and $\calG_2$, which is generally not true   in an IC because of   interfering signals,  
 \begin{align}
   \Pr \big(\calG_1 \big) \Pr \big(\calG_2\big)  \geq 1 - \eps  . \label{eqn:indep}
   \end{align}
Given  that the  number of messages for codebook $j$ satisfies 
\begin{equation}
M_{jn} = \big\lfloor\exp \big( n R^*_j  + \sqrt{n}L_j + o(\sqrt{n } )  \big) \big\rfloor, \label{eqn:size_code}
\end{equation}
the optimum probability of correct detection satisfies  (cf.~Theorem~\ref{thm:awgn_asy})
\begin{equation}
   \Pr\big(\calG_j \big) =\Phi \bigg(-\frac{L_j}{\sqrt{V_j}}\bigg) + o(1) \label{eqn:opt_prob}
   \end{equation}   
   which then (heuristically) justifies \eqref{eqn:cor_pt_ic}.  The proof makes the steps from~\eqref{eqn:1minuseps}--\eqref{eqn:opt_prob} rigorous.    Since $V_ 1$ and $V_2$ are the dispersions of the Gaussian channels without interference, this is the second-order analogue of   Carleial's result for Gaussian ICs in the VSI regime~\cite{Carleial75} because the {\em dispersions are not affected}. Note that no cross dispersion terms  are present in~\eqref{eqn:cor_pt_ic} unlike the Slepian-Wolf problem, where the correlation of two different entropy densities appears in the characterization of $\calL(\eps; R_1^*, R_2^*)$ for corner points $(R_1^*,R_2^*)$. 

 Finally, it is   somewhat surprising that in the converse, even though we must ensure that the codewords  $X_1^n$ and $X_2^n$ are independent,  we do not  need to leverage the wringing technique invented by Ahlswede~\cite{Ahl82}, which was used to prove that the discrete memoryless MAC admits a strong converse. This is thanks to Gaussianity which allows us to show that the first- and second-order  statistics of a certain set of information densities  in \eqref{eqn:dens_ic0}--\eqref{eqn:dens_ic} are independent of $\bx_1$ and $\bx_2$ belonging to their respective    power spheres.

\section{Proof Sketch of the Main Result}\label{sec:prf_vs}
The proof of Theorem~\ref{thm:disp_vs} is somewhat long and tedious so we only sketch the  key steps and refer the reader to \cite{Quoc14} for the detailed calculations.

\begin{proof} 
We begin with the converse. We may assume, using the same argument as that for the point-to-point AWGN channel (cf.~the Yaglom map trick~\cite[Ch.~9, Thm.~6]{conway} in the proof of Theorem~\ref{thm:awgn_asy}) that all the codewords $\bx_j(m_j)$  satisfy $\|\bx_j(m_j ) \|_2^2 = nS_j,j=1,2$. Choose the  auxiliary output distributions in Proposition~\ref{prop:conv_ic} to be the $n$-fold products of 
\begin{align}
Q_{Y_1|X_2} (y_1|x_2)  &:= \calN ( y_1 ; g_{12}x_2 , g_{11}^2 S_1 + 1) ,\quad\mbox{and} \label{eqn:ic_out1}\\
Q_{Y_2|X_1} (y_2|x_1)  &:= \calN ( y_2; g_{21}x_1, g_{22}^2 S_2 + 1) .\label{eqn:ic_out2}
\end{align}
These are the output distributions induced if the input distributions $\tilP_{X_1^n}$ and $\tilP_{X_2^n}$ are $n$-fold products of  $\calN(0, S_1)$ and $\calN(0, S_2)$ respectively.  Fix any achievable $(\eps, R_1^*, R_2^*)$-second-order coding rate pair $(L_1, L_2)$, i.e., $(L_1, L_2) \in\calL(\eps;R_1^*, R_2^*)$. Then, for every $\xi>0$, every  sequence of $(n, M_{1n}, M_{2n}, S_1, S_2,\eps_n)$-codes satisfies
\begin{equation}
\log M_{jn}\ge n  R_j^*  + \sqrt{n} (L_j -\xi) ,\quad j =1,2,  \label{eqn:large_en}
\end{equation}
for $n$ large enough.  To keep our notation succinct, define  the information densities
\begin{align}
j_1(\bx_1,\bx_2, Y_1^n)&\!:=\!\log\frac{W_1^n (Y_1^n | \bx_1, \bx_2) }{Q_{Y_1^n | X_2^n } (Y_1^n |\bx_2)} \!=\!\sum_{i=1}^n \log\frac{W_1(Y_{1i}|x_{1i},x_{2i}) }{Q_{Y_1|X_2}(Y_{1i}|x_{2i}) } ,\quad\mbox{and}\label{eqn:dens_ic0}\\
j_2(\bx_1,\bx_2, Y_2^n)&\!:=\! \log\frac{W_2^n (Y_2^n | \bx_1, \bx_2) }{Q_{Y_2^n | X_1^n } (Y_2^n |\bx_1)} \!=\!\sum_{i=1}^n \log\frac{W_2(Y_{2i}|x_{1i},x_{2i}) }{Q_{Y_2|X_1}(Y_{2i}|x_{1i}) }  . \label{eqn:dens_ic}
\end{align}
Let $C_j:=\rvC(\snr_j)$ for $j = 1,2$.  For any pair of vectors $(\bx_1,\bx_2)$ satisfying $\|\bx_j\|_2^2 =nS_j$,
\begin{align}
\bbE  \begin{bmatrix}
j_1(\bx_1,\bx_2, Y_1^n)\\ j_2(\bx_1,\bx_2, Y_2^n)
\end{bmatrix}  &=  n \begin{bmatrix}
C_1 \\ C_2
\end{bmatrix}, \quad\mbox{and}\label{eqn:ic_stats0} \\*
\cov\begin{bmatrix}
j_1(\bx_1,\bx_2, Y_1^n)\\ j_2(\bx_1,\bx_2, Y_2^n)
\end{bmatrix}   &=  n \begin{bmatrix}
V_1 & 0 \\ 0 & V_2
\end{bmatrix}. \label{eqn:ic_stats}
\end{align}
Importantly, notice that the covariance matrix in \eqref{eqn:ic_stats} is {\em diagonal}. This is due to the independence of the noises $Z_{1i}$ and $Z_{2i}$ and is the crux of the converse proof for the corner point case in~\eqref{eqn:cor_pt_ic}. 

Now let $\gamma :=n^{-3/4}$ in the probability in the non-asymptotic converse bound in~\eqref{eqn:converse_ic}. We denote  this probability as $\mathfrak{p}$. By the law of total probability, the complementary probability $1-\mathfrak{p}$ can be written as  
\begin{equation}
1\! -\!\mathfrak{p}\!=\!\int\!  \Pr\Bigg( \begin{bmatrix}
j_1(\bx_1,\bx_2, Y_1^n) \\ j_2(\bx_1,\bx_2, Y_2^n)
\end{bmatrix} \!>\!\begin{bmatrix}
\log M_{1n} \!-\! n^{1/4}\\ \log M_{2n}\!-\! n^{1/4}
\end{bmatrix}  \Bigg) \, \rmd P_{X_1^n}(\bx_1)\, \rmd P_{X_2^n}(\bx_2 ).
\end{equation}
By \eqref{eqn:large_en},  for large enough $n$,  the inner  probability  evaluates to 
\begin{align}
&\Pr\Bigg( \begin{bmatrix}
j_1(\bx_1,\bx_2, Y_1^n) \\ j_2(\bx_1,\bx_2, Y_2^n)
\end{bmatrix} > \begin{bmatrix}
\log M_{1n} -n^{1/4}\\ \log M_{2n}-n^{1/4}
\end{bmatrix}  \Bigg) \nn\\*
&\le \Pr \Bigg( \!\begin{bmatrix}
j_1(\bx_1,\bx_2, Y_1^n) \\ j_2(\bx_1,\bx_2, Y_2^n)
\end{bmatrix} >  \begin{bmatrix}
nR_1^* - \sqrt{n} (L_1 -2\xi) \\ nR_2^* - \sqrt{n} (L_2 -2\xi) 
\end{bmatrix}   \! \Bigg) \\
&\le \Psi \Bigg( \begin{bmatrix}
\sqrt{n} (C_1 - R_1^* )  - L_1  +   2\xi \\ \sqrt{n} (C_2  - R_2^* )  - L_2  + 2\xi 
\end{bmatrix} ; \bzero,   \begin{bmatrix}
V_1 & 0 \\ 0 & V_2
\end{bmatrix}\Bigg) +\frac{\kappa}{\sqrt{n}}   \label{eqn:apply_multi_be}\\
&= \prod_{j=1}^2\Phi\bigg(  \frac{\sqrt{n} (C_j - R_j^* ) - L_j  + 2\xi }{\sqrt{V_j}}\bigg)+\frac{\kappa}{\sqrt{n}} , \label{eqn:diag_mat}
\end{align}
where \eqref{eqn:apply_multi_be} is an application of the multivariate Berry-Esseen theorem (Corollary \ref{corollary:multidimensional-berry-esseen}) and $\kappa$ is a finite constant. Note that $\Psi$ denotes the bivariate generalization of the Gaussian cdf, defined in~\eqref{eqn:biv}. Equality~\eqref{eqn:diag_mat} holds because the covariance matrix  in \eqref{eqn:apply_multi_be} is diagonal by the calculation in~\eqref{eqn:ic_stats}. Since the bound in \eqref{eqn:diag_mat} does not depend on  $\bx_1,\bx_2$ as long as $\|\bx_j\|_2^2=nS_j$,  we have 
\begin{equation}
1-\mathfrak{p}\le\prod_{j=1}^2\Phi\bigg(  \frac{\sqrt{n} (C_j - R_j^* ) - L_j  + 2\xi }{\sqrt{V_j}}\bigg)+\frac{\kappa}{\sqrt{n}} , \label{eqn:diag_mat2}
\end{equation}
In Case (i), $R_1^*=C_1$ and $R_2^*< C_2$ so the term corresponding to $j=2$ in the above product converges to one and we have 
\begin{equation}
1-\mathfrak{p} \le \Phi\bigg(  \frac{ - L_1  + 2\xi }{\sqrt{V_j}}\bigg)+ \delta_n,
\end{equation}
where $\delta_n\to 0$ as $n \to\infty$. Thus   Proposition~\ref{prop:conv_ic} yields
\begin{equation}
\eps_n\ge \Phi\bigg(  \frac{ L_1  - 2\xi }{\sqrt{V_j}}\bigg)+\delta_n.
\end{equation}
Taking $\limsup$ on both sides yields
\begin{equation}
\limsup_{n\to\infty}\eps_n\ge \Phi\bigg(  \frac{ L_1  -2\xi }{\sqrt{V_j}}\bigg).
\end{equation}
Since $\limsup_{n\to\infty}\eps_n\le\eps$, we can write
\begin{equation}
L_1\le\sqrt{V_1}\Phi^{-1}(\eps) + 2\xi. 
\end{equation}
Since $\xi>0$ is arbitrarily small, we may take $\xi\downarrow 0$ to complete the proof of the converse part for Case (i). For Case (ii), swap the indices $1$ and $2$ in the above calculation. For Case (iii), the analysis until \eqref{eqn:diag_mat} applies. However, now $R_j^*=C_j$ for both $j = 1,2$ so both $\Phi(\cdot)$ functions in \eqref{eqn:diag_mat} are numbers strictly between $0$ and $1$. Consequently, we have 
\begin{equation}
1-\mathfrak{p} \le \Phi\bigg(  \frac{ - L_1  + 2\xi }{\sqrt{V_j}}\bigg)\Phi\bigg(  \frac{ - L_2  + 2\xi }{\sqrt{V_2}}\bigg)+ \frac{\kappa}{\sqrt{n}}.
\end{equation}
The rest of the arguments are  similar  to those for  Case (i).

For the direct part, similarly to the single-user case in \eqref{eqn:awgn_input}, we choose the input distributions
\begin{equation}
P_{X_j^n}(\bx_j) = \frac{\delta\{\|\bx_j\|_2^2-nS_j \} }{ A_n(\sqrt{nS_j} )},\qquad j =1,2, \label{eqn:inputs_ic}
\end{equation}
where $\delta\{\cdot\}$ is the Dirac $\delta$-function and $A_n(r)$ is the area of a sphere in $\bbR^n$ with radius $r$.   Clearly, the power constraints are satisfied with probability one. We choose the conditional output distributions  $Q_{Y_1^n|X_2^n}$ and $Q_{Y_2^n|X_1^n}$ as in  \eqref{eqn:ic_out1} and \eqref{eqn:ic_out2} and the output distributions  $Q_{Y_1^n}$ and $Q_{Y_2^n}$ to be the $n$-fold products of 
\begin{align}
Q_{Y_1}(y_1) &:= \calN(y_1; 0, g_{11}^2 S_1 + g_{12}^2 S_2 + 1),  \label{eqn:QY1}\quad\mbox{and}\\*
Q_{Y_2}(y_2) &:= \calN(y_2; 0, g_{21}^2 S_1 + g_{22}^2 S_2 + 1). \label{eqn:QY2}
\end{align}
With these choices of auxiliary output distributions, one can show the following technical lemma  concerning the ratios of the induced (conditional) output distributions and the chosen (conditional)  output distributions in Proposition~\ref{prop:ach_ic}. This is the multi-terminal analogue of  \eqref{eqn:change_meas} for the point-to-point AWGN channel and it allows us to  replace the inconvenient induced output distributions $P_{X_1^n  }W_1^n$ and $P_{X_1^n}P_{X_2^n}W_1^n$  (which is present in standard Feinstein-type achievability bounds, for example \cite{Han98}) with the convenient $Q_{Y_1^n|X_2^n}$ and $Q_{Y_1^n}$ without too much degradation in error probability. 
\begin{lemma}
Let $Q_{Y_1^n},Q_{Y_2^n}, Q_{Y_1^n|X_2^n}$ and $Q_{Y_2^n|X_1^n}$ be defined as the $n$-fold products of those in \eqref{eqn:QY1}, \eqref{eqn:QY2}, \eqref{eqn:ic_out1} and \eqref{eqn:ic_out2} respectively. Then, there exists a  finite constant $\bar{\zeta}$ such that  the ratios $\zeta_{jk}$ in \eqref{eqn:K1}--\eqref{eqn:K2} are uniformly bounded  by $\bar{\zeta}$ as $n$ grows.  Hence, their sum $\zeta=\sum_{j,k=1}^2 \zeta_{jk}$ is also uniformly bounded. 
\end{lemma}
The proof of this lemma can be found in \cite{Quoc14} and \cite{Mol13}.

Because $X_1^n$ and $X_2^n$ are uniform on their respective power spheres, it is not straightforward to analyze the behavior of   random vector
\begin{equation}
\bB  =\begin{bmatrix}
B_{11} \\ B_{21} \\ B_{12} \\ B_{22}
\end{bmatrix}:= \begin{bmatrix}
 \log\frac{W_1^n(Y_1^n | X_1^n , X_2^n) }{Q_{Y_1^n|X_2^n} (Y_1^n|X_2^n) } \\ \log\frac{W_2^n(Y_2^n | X_1^n , X_2^n) }{Q_{Y_2^n|X_1^n} (Y_2^n|X_1^n) }  \\
\log\frac{W_1^n(Y_1^n | X_1^n , X_2^n) }{Q_{Y_1^n } (Y_1^n ) }  \\
\log\frac{W_2^n(Y_2^n | X_1^n , X_2^n) }{Q_{Y_2^n } (Y_2^n ) } 
\end{bmatrix},
\end{equation}
which is present in \eqref{eqn:direct_ic}.  Note that $\bB$ can be written as a sum of {\em dependent} random variables due to the product structure of  the chosen  output distributions.  To analyze the probabilistic behavior of $\bB$ for large $n$, we leverage   a technique by MolavianJazi and Laneman~\cite{Mol13}. The basic ideas are      as follows: Let $T_j^n \sim \calN(\bzero_n, \bI_{n\times n})$ for $j = 1,2$ be standard Gaussian random vectors that are independent of each other and of the noises $Z^n_j$.  Note that the input distributions in \eqref{eqn:inputs_ic} allow  us to write  $X_{ji}$ as
\begin{equation}
X_{ji} = \sqrt{nS_j } \frac{ T_{ji} }{ \| T_j^n \|_2 },\qquad i =  1,\ldots, n. \label{eqn:XJi}
\end{equation}
Indeed, $\|X_j^n\|_2^2 = nS_j$ with probability one from the random code construction and \eqref{eqn:XJi}. Now consider the  length-$10$  random vector $\mathbf{U}_i := ( \{U_{j1i} \}_{j=1}^4,\{U_{j2i} \}_{j=1}^4,  U_{9i}, U_{10i})$, where  
\begin{alignat}{2}
&U_{11i}  :=  1 - Z_{1i}^2,\quad   & &  U_{21i}   :=  g_{11}\sqrt{S_1 }T_{1i}Z_{1i},\,   \\*
&U_{31i}  :=  g_{12}\sqrt{S_2 }T_{2i}Z_{1i}, \quad  && U_{41i}  :=  g_{11}g_{12} \sqrt{S_1  S_2 }T_{1i} T_{2i}, \,   \\*
&U_{12i}  :=  1 - Z_{2i}^2,\quad  &&U_{22i}  :=  g_{22}\sqrt{S_2 }T_{2i}Z_{2i}, \,   \\
&U_{32i}  :=  g_{21}\sqrt{S_1 } T_{1i}Z_{2i},\quad &&U_{42i}  :=  g_{21}g_{22} \sqrt{S_1  S_2 }T_{1i} T_{2i}, \,   \\
&U_{9i}    :=  T_{1i}^2 -1,\quad && U_{10i}  :=  T_{2i}^2 -1.
\end{alignat}
Clearly, $\bU_i$ is \iid across channel uses. Furthermore, $\bbE[\bU_1]=\bzero$ and $\bbE[ \| \bU_1\|_2^3]$ is finite.  The covariance matrix of $\bU_1$  can also be computed. Define   the functions $\tau_{11},\tau_{12}:\mathbb{R}^{10}\to \mathbb{R}$  
\begin{align}
\tau_{11}(\mathbf{u}) &:=   \snr_1 \,   u_{11} + \frac{2 u_{21}}{\sqrt{1 + u_9}},\quad\mbox{and} \\
\tau_{12}(\mathbf{u}) &:= ( \snr_1 + \inr_1) u_{11} + \frac{2 u_{21}}{\sqrt{1 + u_9}} + \frac{2 u_{31}}{\sqrt{1 + u_{10}}}  \nn\\* 
                          &\qquad           + \frac{2 u_{41}}{\sqrt{1 + u_9} \sqrt{1 + u_{10}}},
\end{align}
for user $1$, and  analogously for user $2$. Then, through some algebra, one sees that   $B_{11}$ and $B_{12}$ can be written as 
\begin{align}
B_{11}  &= n \rvC( \snr_1) + \frac{n}{2 (1+\snr_1  ) }  \, \tau_{11} \bigg( \frac{1}{n}\sum_{i=1}^n\bU_i\bigg) ,\quad\mbox{and} \label{eqn:defB1} \\
B_{12} &= n \rvC( \snr_1+\inr_1) + \frac{n}{2 (1+\snr_1+\inr_1   ) }  \, \tau_{12} \bigg( \frac{1}{n}\sum_{i=1}^n\bU_i\bigg) . \label{eqn:defB3}
\end{align}
The  other random variables in  the  $\bB$ vector can be expressed similarly. 
%We only need to note that 
%\begin{equation}
%\bbE[B_3] =n \rvC( \snr_1 + \inr_1 ) ,\quad\mbox{and}\quad \bbE[B_4] =n \rvC( \snr_2 + \inr_2 )  .\label{eqn:meansB}
%\end{equation}

From \eqref{eqn:defB1}--\eqref{eqn:defB3}, we are able to see the essence of the MolavianJazi-Laneman~\cite{Mol13} technique. The information densities $B_{jk}, j,k=1,2$ were initially difficult to analyze because the input random vectors $X_j^n$ in \eqref{eqn:inputs_ic} are uniform on power spheres. This choice of input distributions  results in codewords $X_j^n$ whose coordinates are {\em dependent} so standard limit theorems do not readily apply. By defining   higher-dimensional  random vectors $\bU_i$ and appropriate functions $\tau_{jk}$, one then sees that $\bB$ can be expressed as a {\em  function   of  a sum of \iid random vectors}.  Now, one may consider  a Taylor expansion of the differentiable functions $\tau_{jk}$ around the mean $\bzero$ to approximate   $\bB$ with a  sum of \iid random vectors. Through this analysis, one can rigorously show that 
\begin{align}
\frac{1}{\sqrt{n}}\left( \bB- n \begin{bmatrix}
\rvC( \snr_1)  \\ \rvC( \snr_2) \\ \rvC( \snr_1+\inr_1) \\ \rvC( \snr_2+\inr_2    ) 
\end{bmatrix} \right)\stackrel{\mathrm{d}}{\longrightarrow}\calN\left( \bzero, \begin{bmatrix}
V_1 & 0 & * & * \\ 0 & V_2 & *  &  * \\ * & * & * & *  \\ * & * & * & * 
\end{bmatrix} \right), \label{eqn:ic_conver}
\end{align}
where the entries marked as $*$ are finite and inconsequential for the subsequent analyses. Recall also that $V_j = \rvV(\snr_j)$ for $j = 1,2$.  In fact, the rate of convergence to Gaussianity in \eqref{eqn:ic_conver} can be quantified by means of Theorem~\ref{thm:func_clt}. 

With these preparations, we are ready to evaluate the probability in the direct bound in~\eqref{eqn:direct_ic}, which we denote as $\mathfrak{p}$.  We consider all three cases in tandem. Fix $(L_1,L_2 ) \in\calL(\eps; R_1^*, R_2^*)$.  Let the number of codewords in the $j^{\mathrm{th}}$ codebook be 
\begin{align}
M_{jn}  = \lfloor \exp\big( nR_j^* + \sqrt{n} L_j - 2 n^{1/4}   \big)  \rfloor
\end{align}
for $j=1,2$. It is clear that 
\begin{align}
\liminf_{n \to \infty} \frac {1}{\sqrt{n}} \big(\log M_{jn} - n R_j^*\big) \geq L_j.
\end{align}
Also let $\gamma  := n^{-3/4}$. With these choices, the complementary probability $1-\mathfrak{p}$ can be expressed as 
\begin{align}
1-\mathfrak{p} =\Pr\left( \bB> \begin{bmatrix}
n R_1^* + \sqrt{n} L_1 -   n^{1/4}  \\ n R_2^* + \sqrt{n} L_2-   n^{1/4}  \\ n (R_1^*+R_2^*) + \sqrt{n} (L_1+L_2) -  3 n^{1/4}   \\ n (R_1^*+R_2^*) + \sqrt{n} (L_1+L_2) -  3 n^{1/4}   
\end{bmatrix} \right). \label{eqn:comp_p}
\end{align}
Now by the SVSI assumption in \eqref{eqn:svsi2}--\eqref{eqn:svsi3},
\begin{equation}
 R_1^*+R_2^*  \le \rvC(\snr_1  ) +  \rvC(\snr_2 )  < \min\{\rvC(\snr_1 + \inr_1 ),\rvC(\snr_2 + \inr_2 )\}. \label{eqn:assump_svsi}
\end{equation}
The convergence in~\eqref{eqn:ic_conver} implies that 
\begin{equation}
\bbE[B_{12}] =n\rvC(\snr_1+\inr_1),\quad\mbox{and}\quad \bbE[B_{22}]=n\rvC(\snr_2+\inr_2).
\end{equation}
Since the expectations of $B_{12}$ and $B_{22}$ are strictly larger than $R_1^* +R_2^*$ (cf.~\eqref{eqn:assump_svsi}), by  standard Chernoff bounding techniques,
\begin{align}
\Pr\big(B_{12}\!\le \! n (R_1^*\!+\! R_2^*) \!+\! \sqrt{n} (L_1\!+\!L_2) \! - \! 3 n^{1/4} \big)&\!\le\!\exp(-n\xi),\,\, \mbox{and} \\* 
\Pr\big(B_{22}\! \le \! n (R_1^*\!+\! R_2^*) \!+ \!\sqrt{n} (L_1\!+\!L_2)\! - \! 3 n^{1/4} \big) &\!\le\!\exp(-n\xi) ,
\end{align}
for some $\xi>0$. 
Consequently, by the union bound, \eqref{eqn:comp_p} reduces to 
\begin{equation}
1-\mathfrak{p}\ge \Pr\left( \begin{bmatrix}
B_{11}\\ B_{21}
\end{bmatrix}  > \begin{bmatrix}
n R_1^* + \sqrt{n} L_1 -   n^{1/4}  \\ n R_2^* + \sqrt{n} L_2-   n^{1/4}   
\end{bmatrix} \right) -2\exp(-n\xi) .\label{eqn:ic_final}
\end{equation}
 Just as in the converse, one can then analyze this probability for the various cases using the convergence to Gaussianity in~\eqref{eqn:ic_conver}. This completes the proof of the direct part.
\end{proof} 
\newcommand{\underg}{\underline{g}} 
\newcommand{\overg}{\overline{g}}

\chapter{A Special Class of  Gaussian Multiple Access Channels} \label{ch:mac}

The multiple access channel (MAC) is a communication model in which many parties would like to simultaneously send independent messages over a common medium to a sole destination. Together with the broadcast, interference and relay channels, the MAC is a  fundamental building block of more complicated communication networks. For example, the MAC is an appropriate   model  for the uplink of cellular systems where multiple mobile phone users would like to communicate to a distant base station over a wireless medium.  The capacity region of the MAC is, by now, well known and goes back to the work by  Ahlswede~\cite{ahl71} and Liao~\cite{liao} in the early 1970s. The strong converse was established by Dueck~\cite{dueck81} and Ahlswede~\cite{Ahl82}. 

%\footnote{We recognize that the use of the acronym DMS for  {\em degraded message sets}  is unfortunately in conflict with that for {\em discrete memoryless source} in previous chapters. However, the use of DMS for degraded message sets is local to this chapter so  it is the author's hope that the reader will not be confused by this abuse of acronyms.  } 

A yet simpler model, which we consider in this chapter, is the  {\em asymmetric  MAC} (A-MAC) as  shown in Fig.~\ref{fig:g_mac}. This channel model, also known as the MAC with {\em degraded message sets}   \cite[Ex.~5.18(b)]{elgamal} or the {\em cognitive}~\cite{dev06} MAC,  was first studied by Haroutunian~\cite{Har75}, Prelov \cite{Pre84} and van der Meulen~\cite{vdM85}. Here,  encoder $1$ has knowledge of {\em both} messages $m_1 $ and $m_2$, while encoder $2$ only has its own message $m_2$.  For the Gaussian case,  the channel law is $Y=X_1 + X_2 + Z$, where $Z$ is standard Gaussian noise. The capacity region \cite[Ex.~5.18(b)]{elgamal} is  the set of all $(R_1, R_2)$ satisfying
\begin{align}
R_1   \le \rvC\big( (1-\rho^2) S_1 \big),\quad \mbox{and} \quad   R_1+ R_2   \le \rvC\big( S_1 + S_2 + 2\rho\sqrt{  S_1S_2} \big) \label{eqn:gmac_region}
\end{align}
for some $\rho \in [0,1]$ where $S_1$ and $S_2$ are the admissible transmit powers. Rate pairs in \eqref{eqn:gmac_region} are achieved using superposition coding~\cite{cover72}. This region for $S_1=S_2=1$ is shown in Fig.~\ref{fig:cr}. Observe that $\rho \in [0,1]$ parametrizes points on the boundary. Each point on the  curved part of the boundary is achieved by a unique   bivariate Gaussian distribution.% which we will describe in detail in the following. 

 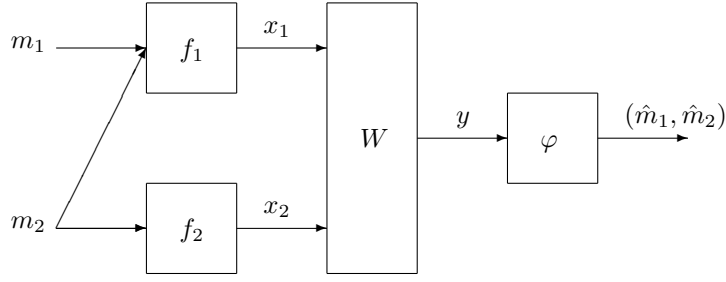
\begin{figure}[t]
\centering
\setlength{\unitlength}{.4mm}
\begin{picture}(190, 90)
%\thicklines
\put(0, 15){\vector(1, 0){30}}
\put(60, 15){\vector(1,0){30}}
%\put(120, 15){\vector(1,0){30}}
%\put(180, 15){\vector(1,0){30}}
\put(30, 0){\line(1, 0){30}}
\put(30, 0){\line(0,1){30}}
\put(60, 0){\line(0,1){30}}
\put(30, 30){\line(1,0){30}}

\put(0, 15){\vector(1,2){30}}

\put(0, 75){\vector(1, 0){30}}
\put(60, 75){\vector(1, 0){30}}
\put(0, 15){\vector(1, 0){30}}
%\put(60, 15){\vector(1,0){30}}
%\put(120, 45){\vector(1,0){30}}
%\put(180, 15){\vector(1,0){30}}
\put(30, 60){\line(1, 0){30}}
\put(30, 60){\line(0,1){30}}
\put(60, 60){\line(0,1){30}}
\put(30, 90){\line(1,0){30}}

\put(90, 0){\line(1, 0){30}}
\put(90, 0){\line(0,1){90}}
\put(120, 0){\line(0,1){90}}
\put(90, 90){\line(1,0){30}}

\put(-18, 15){  $m_2$}
\put(-18, 75){  $m_1$}
\put(66, 20){  $x_2$}
\put(66, 80){  $x_1$}
%\put(51, -10){  $\bbE[\rvg(X)]\le\Gamma$}
%\put(65, 8){  $[2^{nR}]$}
%\put(140, 50){  $(\hatx_1,\hatx_2)$} 
\put(41, 12){$f_2$ } 
\put(41, 72){$f_1$ } 
\put(101, 42){$W$} 

\put(130,50){  $y$}
%\put(126, 80){  $y_1$}

%\put(186, 20){  $\hatm_2$}
%\put(186, 80){  $\hatm_1$}

\put(186, 50){  $(\hatm_1,\hatm_2)$}

\put(120, 45){\vector(1, 0){30}}
%\put(120, 75){\vector(1, 0){30}}

\put(150, 30){\line(1, 0){30}}
\put(150, 30){\line(0,1){30}}
\put(180, 30){\line(0,1){30}}
\put(150, 60){\line(1,0){30}}

%\put(150, 60){\line(1, 0){30}}
%\put(150, 60){\line(0,1){30}}
%\put(180, 60){\line(0,1){30}}
%\put(150, 90){\line(1,0){30}}

\put(180, 45){\vector(1, 0){30}}
%\put(180, 75){\vector(1, 0){30}}

\put(161, 42){$\varphi$ } 
%\put(161, 72){$\varphi_1$ } 
  \end{picture}
  \caption{Illustration of the asymmetric  MAC or A-MAC    }
  \label{fig:g_mac}
\end{figure}

\begin{figure} 
\centering
\begin{overpic}[width = .725\columnwidth]{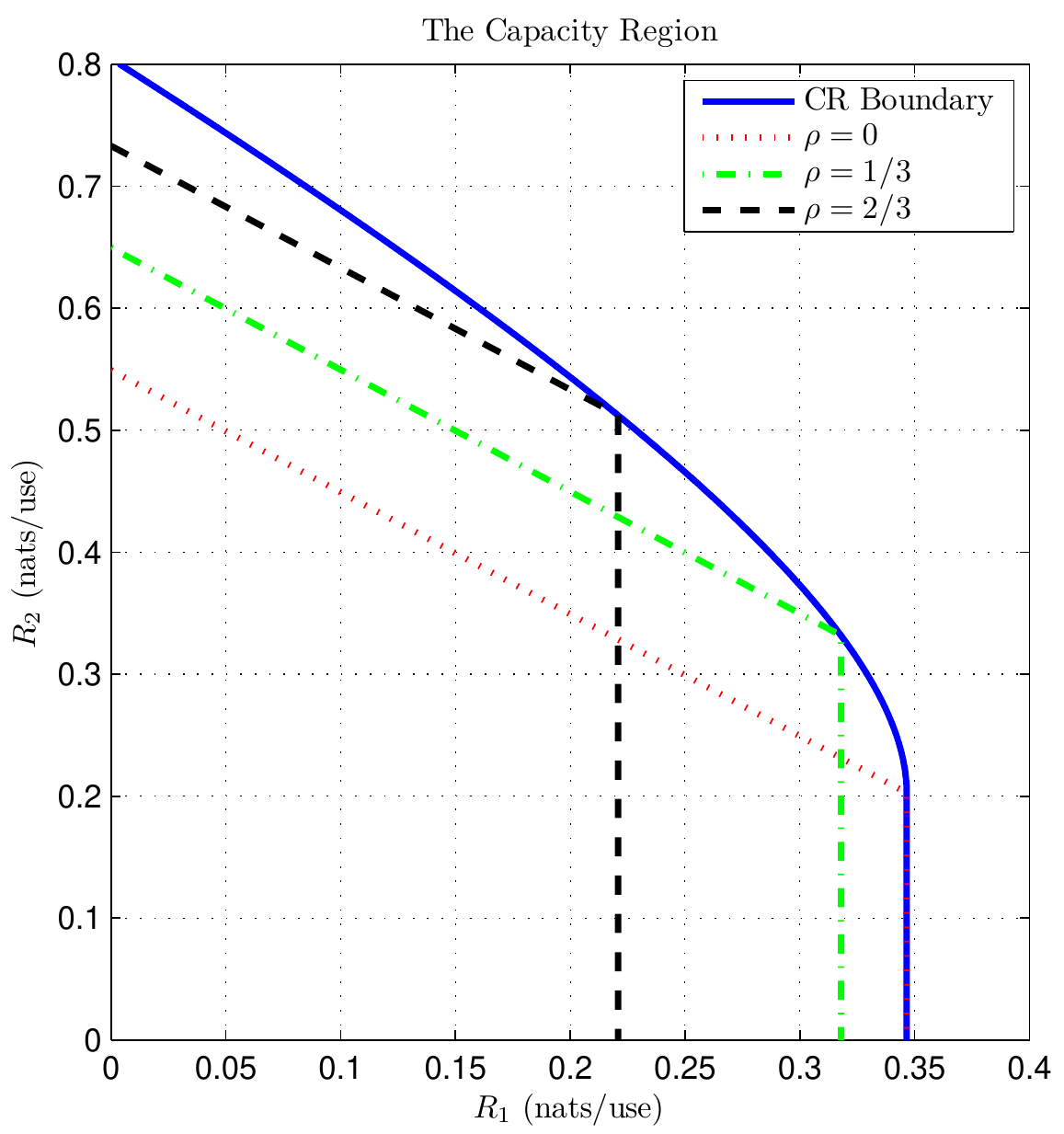}
\put(80,23){\circle*{2}}
\put(82,23){(i)}
\put(9,96){(iii)}
\put(55,66){(ii)}
\put(55,63){\circle*{2}}
\put(10,94){\circle*{2}}
\thicklines 
\put(40, 48){\vector(1, 1){15}}
\put(65, 42.5){\vector(-1, 2){10}}  
\put(38, 45){\mbox{$\bv$}}
\put(65, 39.5){\mbox{$\bv'$}}
\end{overpic}
\caption{Capacity region   of a Gaussian   A-MAC where $S_1 =  S_2 = 1$. The three cases of Theorem~\ref{thm:local} are   illustrated.  Each $\rho \in (0,1]$ corresponds to a trapezoid of rate pairs   achievable by a unique input distribution   $\calN(\bzero,\bSigma(\rho))$. However, coding with a fixed input distribution 
is insufficient to achieve all angles of approach to a   boundary point 
 as there are  regions within $\calC$   not in the 
trapezoid parametrized by $\rho$. Suppose $\rho=\frac{2}{3}$, one can 
approach the corner point in the direction indicated by the vector $\bv$ using the fixed input distribution 
    $\calN(\bzero,\bSigma(\frac{2}{3}))$, but the same is not true of
the direction indicated by $\bv'$, since the approach is from  {\em outside} the trapezoid.  }
\label{fig:cr}
\end{figure}

In this chapter, we show that the assumptions concerning  Gaussianity and asymmetry of the messages sets (i.e., partial cooperation) allow us to determine the second-order asymptotics of this model. The main result here is  of a somewhat different flavor compared to results in previous chapters on multi-terminal information theory problems because the second-order rate region $\calL(\eps;R_1^*, R_2^*)$ is characterized not only in terms of covariances of vectors of information densities or {\em dispersions}.   Indeed, we will see that there is a subtle interaction between the  {\em derivatives  of the first-order capacity terms}  in~\eqref{eqn:gmac_region} with respect to $\rho$, and the dispersions in the description of $\calL(\eps;R_1^*, R_2^*)$. The fact that the derivatives appear in the answer to an information-theoretic question appears to be novel.\footnote{In fact, the dispersion of the compound channel \cite{Pol13b} is a function of the dispersions of the constitudent channels and the derivatives of the  capacity terms.}  The difference  in the characterization of $\calL(\eps;R_1^*, R_2^*)$ compared to second-order regions in previous chapters  is because, with the union over $\rho\in [0,1]$,  the boundary of the capacity region  in \eqref{eqn:gmac_region} is curved in contrast to the polygonal capacity regions in previous chapters. We will see that the curvature of the boundary results in the second-order region $\calL(\eps;R_1^*, R_2^*)$ being a {\em half-space} in $\bbR^2$. This half-space  is characterized by a slope and intercept, both of which are expressible in terms of the  dispersions, together with the derivatives of the capacities.

Intuitively, the extra derivative term  arises because we need to account for all possible angles of approach to a boundary point $(R_1^*,R_2^*)$. Using a  sequence of input distributions parametrized by a {\em single} correlation parameter $\rho$ not depending on the blocklength 
%{\em single} multivariate Gaussian input distribution with correlation $\rho$ for all blocklengths 
turns out to be  suboptimal in the second-order sense, as we can only achieve the angles of approach within the specific trapezoid parametrized by  $\rho$ (see Fig.~\ref{fig:cr} and its caption). Thus, our coding strategy is to let the sequence of input distributions vary with the blocklength. In particular, they are parametrized by a sequence $\{\rho_n\}_{n\in\bbN}$ that converges to $\rho$ with speed $\Theta(\frac{1}{\sqrt{n}})$. A Taylor expansion of the first-order capacity   vector then yields the derivative term.

Similarly to the Gaussian IC with SVSI, the achievability proof uses the coding on spheres strategy in which pairs of codewords are drawn uniformly at random from high-dimensional spheres. However, because the underlying coding strategy involves superposition coding, the analysis is more subtle. In particular, we are required to bound the ratios of certain induced output densities  and product output densities. % and  consists of a new (relative to the other material in this monograph) idea called {\em cost-constrained random coding}. This idea was suggested by Gallager (e.g., in \cite[Ch.~7]{gallagerIT}) to deal with system cost constraints (such as power constraints). Later building on initial work by Ganti-Lapidoth-Telatar~\cite{ganti}, Shamai and Sason~\cite{shamai2002} used this technique with auxiliary cost functions to deal with channels with continuous alphabets.  Scarlett-Martinez-{Guill\'{e}n i F\`{a}bregas}~\cite{Scarlett13} extended this technique to multiple auxiliary cost functions which turns out to yield improved rates and second-order  asymptotics especially in network settings.  This technique   is also amenable to second-order asymptotic analyses of  problems that involve superposition coding~\cite{scarlett2013}. 
The proof of the converse part   involves several new ideas including (i) reduction to  {\em almost  constant correlation type subcodes}; (ii) evaluation of a  {\em global} outer bound and (iii) specialization of the global outer bound to obtain  {\em local} second-order asymptotic results.

The material in this chapter is   based on work by Scarlett and Tan~\cite{ScarlettTan}.

\section{Definitions and Non-Asymptotic Bounds}
The model we consider is as follows:
\begin{align}
Y_{i}  = X_{1i} + X_{2i} + Z_i
\end{align}
where $i = 1,\ldots, n$ and $Z_i \sim\calN(0,1)$ is white Gaussian noise.  The channel gains are set to unity  without loss of generality.  Thus, the channel transition  law is 
\begin{equation}
W(y|x_1, x_2) = \frac{1}{\sqrt{2\pi} } \exp\left( -\frac{1}{2}(y- x_1 - x_2)^2\right). 
\end{equation}
The channel operates in a stationary and memoryless manner. 

We define an {\em $(n, M_{1}, M_{2}, S_1, S_2,\eps)$-code}  for the Gaussian A-MAC  which includes two encoders $f_1 : \{ 1,\ldots, M_{1}\} \times   \{1,\ldots, M_2\}\to\bbR^n$, $f_2:\{1,\ldots, M_2\}\to\bbR^n$ and a decoder $\varphi:\bbR^n \to  \{ 1,\ldots, M_{1} \} \times   \{1,\ldots, M_2\}$ such that the following {\em power constraints} are satisfied
\begin{align}
\big\|  f_1(m_1, m_2 ) \big\|_2^2 =\sum_{i=1}^n f_{1i}(m_1)^2 &\le nS_1,\quad\mbox{and} \label{eqn:snr_mac1}\\* 
\big\|  f_2( m_2 ) \big\|_2^2 = \sum_{i=1}^n f_{2i}(m_2)^2 &\le nS_2 ,\label{eqn:snr_mac2}
\end{align}
and the {\em average error probability} 
\begin{equation}\label{eqn:error_mac}
\frac{1}{M_1 M_2}\sum_{m_1=1}^{M_1}\sum_{m_2=1}^{M_2}W^n\big( \bbR^n\!\times\!\bbR^n\setminus\calD_{m_1, m_2}\big| f_1(m_1, m_2) , f_2(m_2) \big)\!\le\!\eps.
\end{equation}
As with the Gaussian IC discussed in the previous chapter, $\calD_{m_1, m_2}$ denotes the decoding region for messages $(m_1, m_2)$ and $S_j$ represents the admissible power for the $j^{\mathrm{th}}$ user.

The following non-asymptotic bounds are easily derived. They are analogues of the bounds by Feinstein~\cite{Feinstein} and Verd\'u-Han~\cite{VH94} (or Hayashi-Nagaoka \cite{Hayashi03}). See  Boucheron-Salamatian~\cite{bouch00} for the proofs of   similar results.

\begin{proposition}[Achievability bound for  the A-MAC] \label{prop:ach_mac}
Fix any  input joint distribution   $P_{X_1^n   X_2^n}$ whose support satisfies the power constraints in \eqref{eqn:snr_mac1}--\eqref{eqn:snr_mac2},  i.e., $\|X_j^n\|_2\le nS_j$ with probability one.    For every $n\in\bbN$, every $\gamma>0$, any choice of output  distributions $Q_{Y^n | X_2^n}$ and  $Q_{Y^n}$, and any two sets $\calA_1 \subset\calX_2^n\times\calY^n$ and $\calA_{12}\subset \calY^n$,      there exists an $(n,M_{1}, M_{2}, S_1, S_2, \eps)$-code for the   A-MAC such that 
\begin{align}
\eps & \le  \Pr \bigg(  \log\frac{W^n(Y^n | X_1^n , X_2^n) }{Q_{Y^n|X_2^n} (Y^n|X_2^n) } \! \le\!\log M_{1} \!+\! n\gamma   \quad\mbox{or } \nn\\*
 & \hspace{.5in}\log\frac{W^n(Y^n | X_1^n , X_2^n) }{Q_{Y^n } (Y^n ) }\!\le\!\log (M_{1} M_{2}) \!+\! n\gamma   \bigg)   \nn\\*
&\hspace{.5in}+\Pr\big( (X_2^n,Y^n) \notin\calA_1 \big)+\Pr\big( Y^n  \notin\calA_{12}\big) +\zeta \exp(-n\gamma), \label{eqn:direct_mac}
\end{align}
where $\zeta=\zeta_1+\zeta_{12}$ and 
\begin{align}
\zeta_{1} & := \sup_{(\bx_2,\by ) \in\calA_1} \frac{ P_{X_1^n|X_2^n} W^n (\by | \bx_2) }{ Q_{Y^n|X_2^n}(\by | \bx_2) },\quad \zeta_{12} :  = \sup_{ \by \in\calA_{12} } \frac{ P_{X_1^n X_2^n }  W^n (\by  ) }{ Q_{Y^n }(\by ) }.   \label{eqn:K1_mac}
\end{align}
\end{proposition}
Again notice that our freedom to choose $Q_{Y^n|X_2^n}$ and $Q_{Y^n}$ results in the need to control $\zeta_1$ and $\zeta_{12}$, which are the maximum values of the ratios of the densities induced by the code with respect to the chosen output densities. The maximum values are restricted to those typical values of $(\bx_2,\by)$ and $\by$ indicated by the chosen sets $\calA_{1}$ and $\calA_{12}$.

\begin{proposition}[Converse bound for  the  A-MAC] \label{prop:con_mac}
%Fix any  input joint distributions  $P_{X_1^n   X_2^n}$ whose support satisfies the power constraints in \eqref{eqn:snr_mac1}--\eqref{eqn:snr_mac2},  i.e., $\|X_j^n\|_2\le nS_j$ with probability one.    
For every $n\in\bbN$, every $\gamma>0$ and for any choice of output distributions $Q_{Y^n | X_2^n}$ and  $Q_{Y^n}$,     every $(n,M_{1}, M_{2}, S_1, S_2, \eps)$-code for the A-MAC must satisfy 
\begin{align}
\eps  \ge  \Pr \bigg(  \log\frac{W^n(Y^n | X_1^n , X_2^n) }{Q_{Y^n|X_2^n} (Y^n|X_2^n) } &\! \le\!\log M_{1} \!-\! n\gamma   \quad\mbox{or } \nn\\ 
\log\frac{W^n(Y^n | X_1^n , X_2^n) }{Q_{Y^n } (Y^n ) } & \!\le\!\log (M_{1} M_{2}) \!-\! n\gamma   \bigg)  \!-\!2  \exp(-n\gamma), \label{eqn:con_mac}
\end{align}
for some input joint distribution   $P_{X_1^n X_2^n}$  whose support satisfies the power constraints in \eqref{eqn:snr_mac1}--\eqref{eqn:snr_mac2}.
\end{proposition}

\section{Second-Order Asymptotics}
As in the previous chapters on multi-terminal problems, given a point on the boundary of the capacity region $(R_1^* , R_2^*)$, we are interested in characterizing the set of all $(L_1,L_2)$ pairs for which there exists a sequence of $(n, M_{1n}, M_{2n}, S_1, S_2,\eps_n)$-codes such that 
\begin{align}
\liminf_{n\to\infty}\frac{1}{\sqrt{n}} \big( \log M_{jn}-nR_j^*\big)\ge L_j, \,\,\mbox{and}\,\,\limsup_{n\to\infty}\eps_n\le\eps .  \label{eqn:2nd_mac}
\end{align}
We denote this set as $\calL(\eps; R_1^*, R_2^*) \subset\bbR^2$. 

\subsection{Preliminary Definitions}

Before we can state the main results,  we need to define a few more fundamental quantities. For a pair of rates $(R_1, R_2)$, the {\em rate vector}  is
\begin{align}
\bR := \begin{bmatrix}
R_1 \\ R_1 + R_2
\end{bmatrix} . \label{eqn:rate_vec}
\end{align}
The input distribution to achieve a point on the boundary 
characterized by some $\rho \in [0,1]$ is a $2$-dimensional 
Gaussian distribution with zero mean and covariance matrix 
\begin{align} \label{eqn:sigmamatrix}
    \bSigma(\rho) :=\begin{bmatrix}
    S_1 & \rho \sqrt{S_1S_2}\\
    \rho \sqrt{S_1S_2} & S_2
    \end{bmatrix} .
\end{align}
The corresponding {\em mutual information vector} is given by 
\begin{align}
\bI(\rho)=\begin{bmatrix}
I_1(\rho) \\ I_{12}(\rho)
\end{bmatrix} := \begin{bmatrix}
\rvC\big( S_1(1-\rho^2)  \big)  \\ \rvC\big( S_1+S_2 + 2\rho\sqrt{S_1S_2}   \big) 
\end{bmatrix} . \label{eqn:mi_vec}
\end{align}
Let 
\begin{equation}
\rvV(x,y) := \log^2\rme\cdot \frac{x(y+2)}{2(x+1)(y+1)} 
\end{equation}
 be the  {\em Gaussian cross-dispersion function} and note that  $\rvV(x) := \rvV(x,x)$ is the 
{\em Gaussian dispersion function} defined previously in \eqref{eqn:gauss_disp}.   For fixed $0\le \rho\le 1$, define the {\em information-dispersion matrix}
\begin{align}
\bV(\rho) :=\begin{bmatrix}
V_1(\rho)  & V_{1,12}(\rho)\\ V_{1,12}(\rho)  & V_{12}(\rho) \end{bmatrix}, \label{eqn:inf_disp_matr}
\end{align}
where the elements of the matrix are  
\begin{align}
V_1(\rho) &:= \rvV\big( S_1 (1-\rho^2)\big) ,   \\* %\frac{S_{1}(1-\rho^{2})\big(S_{1}(1-\rho^{2})+2\big)}{2\big(S_{1}(1-\rho^{2})+1\big)^{2}} ,\\ 
V_{1,12}(\rho) &:= \rvV\big( S_1 (1-\rho^2), S_1 + S_2 + 2\rho\sqrt{S_1 S_2} \big)  ,  \\*
V_{12}(\rho) &:= \rvV\big(S_{1}+S_{2}+2\rho\sqrt{S_{1}S_{2}} \big)  .  
\end{align}

Let  $(X_1,X_2)\sim P_{X_1  X_2} = \calN(\bzero,\bSigma(\rho))$  and define $Q_{Y|X_2}$ and $Q_Y$ to be Gaussian distributions   induced by $P_{X_1  X_2}$ and  $W$, namely
\begin{align}
Q_{Y|X_2}(y|x_2) & : = \calN\big(y; x_2 (1+ \rho \sqrt{S_1/S_2}) , 1+ S_1(1- \rho ^2) \big), \quad\mbox{and} \label{eqn:out1}\\*
Q_{Y}(y) & : = \calN\big(y;0, 1+S_1 + S_2 + 2\rho \sqrt{S_1 S_2}\big). \label{eqn:out2}
\end{align}  
It should be noted that the random variables  $(X_1, X_2)$ and the densities 
$Q_{Y|X_2}$ and $Q_Y$  all depend on $\rho$; this dependence is suppressed 
throughout the chapter.  The mutual information vector  $\bI(\rho)$ and  
information-dispersion matrix $\bV(\rho)$ are  the mean vector and conditional 
covariance matrix of the information density vector 
\begin{align} 
    \bj(x_1, x_2  , y)  := \begin{bmatrix}
    j_1(x_1, x_2  , y) \\ j_{12}(x_1, x_2  , y)
    \end{bmatrix}   
     = \begin{bmatrix}
    \log\frac{W(y|x_1,x_2)}{Q_{Y|X_2}(y|x_2)}  \\ \log\frac{W(y|x_1,x_2)}{Q_{Y}(y)}  
    \end{bmatrix}.
    \label{eqn:info_dens}
\end{align} 
That is, we can write $\bI(\rho)$ and $\bV(\rho)$ as
\begin{align}
\bI(\rho)   &=\bbE\big[\, \bj(X_1,X_2,Y)\big], \quad\mbox{and} \label{eqn:mean_v}\\*
  \bV(\rho)   &=\bbE\big[\cov \big(\bj(X_1,X_2,Y) \, \big|\, X_1,X_2 \big)\big], \label{eqn:cov_v}
\end{align}
with $(X_1, X_2, Y)\sim P_{X_1 X_2} \times W$.  We also need a generalization of the $\Phi^{-1}(\cdot)$ function. Define the ``inverse image'' of $\Psi(z_1,z_2;\bzero,\bSigma)$ as 
\begin{equation}
\Psi^{-1}(\bSigma,\eps) := \big\{ (t_1, t_2) \in\bbR^2: \Psi( -z_1, -z_2 ; \bzero ,\bSigma)  \ge 1-\eps\big\}. \label{eqn:psi_inv1}
\end{equation}
An illustration of this set is provided in Fig.~\ref{fig:psi_inv}. Observe that for $\eps<\frac{1}{2}$, the set lies entirely within the third quadrant of the $\bbR^2$ plane. This represents ``backoffs'' from the first-order fundamental limits.

\begin{figure}
\centering
\includegraphics[width = .97\columnwidth]{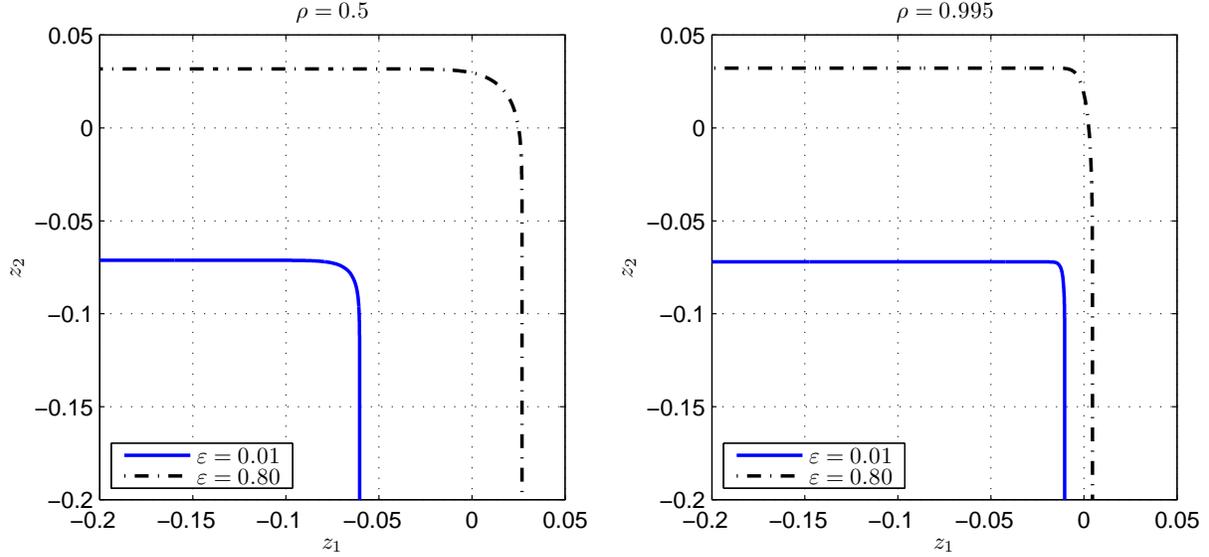}
\caption{Illustration of the set $\Psi^{-1}(\bV(\rho),\eps)$ where $\bV(\rho)$ is defined in \eqref{eqn:inf_disp_matr}. The regions are to the bottom left of the boundaries indicated.}
%\caption{Illustration of the region $\calL(\eps; R_1^*, R_2^*)$ and the   in (6.16). The region is to the
%top right of the boundary indicated.
\label{fig:psi_inv}
\end{figure}

\subsection{Global Second-Order Asymptotics}

Here we provide inner and outer bounds on $\calC(n,\eps)$, defined to be
the set of $(R_1,R_2)$ pairs such that there exist codebooks of
length $n$ and rates at least $R_1$ and $R_2$ yielding an average error
probability not exceeding $\eps$.  Let $\underg(\rho,\eps,n)$ 
and $\overg(\rho,\eps,n)$ be arbitrary functions of $\rho$, 
$\eps$ and $n$ for now, and define the inner and outer regions
\begin{align}
\underline{\calR}(n,\eps;\rho)  & :=  \bigg\{ (R_1,R_2)  :   \bR  \in  \bI(\rho)  +  \frac{\Psi^{-1}(\bV(\rho),\eps)}{\sqrt{n}}  +   \underg(\rho,\eps,n) \bone\bigg\}, \label{eqn:Rin}   \\*
\overline{\calR}(n,\eps;\rho)&  :=  \bigg\{ (R_1,R_2) :  \bR  \in  \bI(\rho)  +  \frac{\Psi^{-1}(\bV(\rho),\eps)}{\sqrt{n}}  +   \overg(\rho,\eps,n) \bone\bigg\} \label{eqn:Rout} .
\end{align}

\begin{lemma}[Global Bounds on the $(n,\eps)$-Capacity Region] \label{lem:global}
    There exist functions $\underg(\rho,\eps,n)$ and $\overg(\rho,\eps,n)$ such that 
    \begin{align}
    \bigcup_{0\le\rho\le 1}\underline{\calR}(n,\eps;\rho) \subset\calC(n,\eps) \subset\bigcup_{-1\le\rho\le 1}\overline{\calR}(n,\eps;\rho), \label{eqn:unions}
    \end{align}
    and $\underg$ and $\overg$ satisfy the following properties:  \\ 
    (i)  For    any  sequence $\{\rho_n\}_{n\in\bbN}$ with $\rho_n\to\rho \in (-1 ,1)$, we have % sequence $\{\rho_n\}_{n\in\bbN}$ with $\rho_n\to\rho \in (-1,1)$, we have
            \begin{equation} 
          \underg(\rho_n,\eps,n) = O\left(\frac{\log n}{n}\right),\quad\mbox{and}\quad \overg(\rho_n,\eps,n) = O\left(\frac{\log n}{n}\right). \label{eqn:thirdorder1}
            \end{equation}
(ii) Else,  for any  sequence $\{\rho_n\}_{n\in\bbN}$ with $\rho_n\to\pm1$, we have
            \begin{equation}
         \underg(\rho_n,\eps,n) = o\left(\frac{1}{\sqrt{n}}\right),\quad\mbox{and}\quad \overg(\rho_n,\eps,n) = o\left(\frac{1}{\sqrt{n}}\right). \label{eqn:thirdorder2}
            \end{equation}
\end{lemma}
Lemma~\ref{lem:global} serves as a stepping stone to establish the local behavior of first-order optimal codes near a boundary point. A proof  sketch of the lemma is provided in Section~\ref{sec:prof_sk_gl}.  

We remark that even though the union for the outer bound in \eqref{eqn:unions} is taken over $\rho\in [-1,1]$, only the values $\rho\in[0,1]$ will play a role in establishing the local asymptotics in Section~\ref{sec:local}, since negative values of $\rho$ are not even first-order optimal, i.e., they fail to achieve a point on the boundary of the capacity region.

We do not claim that the remainder terms in \eqref{eqn:thirdorder1}--\eqref{eqn:thirdorder2} are uniform
in the limiting value $\rho$ of $\{\rho_n\}_{n\in\bbN}$; such uniformity will not be required in establishing our main
local result below.  On the other hand, it is crucial that values of $\rho$ varying with $n$ are 
handled.% (in contrast, most existing global results in other settings consider fixed input distributions).

%We remark that the situation in which there exists  a sequence $\{\rho_n\}_{n\in\bbN}$  satisfying $\rho_n\to\rho\in (-1,1)$  essentially corresponds to  case (i). The terms  $\underg(\rho_n,\eps,n)$ and $\overg(\rho_n,\eps,n)$  behave as in \eqref{eqn:thirdorder1}.

%can be handled similarly to \eqref{eqn:thirdorder1} with $\underg(\rho_n,\eps,n)$ and $\overg(\rho_n,\eps,n)$ satisfying \eqref{eqn:t. This follows from continuity of the information quantities.  The precise asymptotic behavior of  $\underg(\rho,\eps,n)$ and $\overg(\rho,\eps,n)$ depend on the rate of convergence of $\rho_n$ to $\rho$. 

%Lemma~\ref{lem:global} serves as a stepping stone for us to establish the local behavior of optimal codes near a boundary point.  Note that even though the union for the outer bound in \eqref{eqn:unions} is
%taken over $\rho\in [-1, 1]$, only the values $\rho  \in [0, 1]$ will play a role in establishing the local asymptotics in the following subsection,
%since negative values of $\rho$ are not even first-order optimal, i.e., they fail to achieve a point on the boundary of the
%capacity region.
 
\subsection{Local Second-Order Asymptotics} \label{sec:local}
To characterize $\calL(\eps; R_1^*, R_2^*)$, we need yet another definition, which is a feature we have not encountered thus far in this monograph. Define 
\begin{align}
    \bD(\rho)=\begin{bmatrix}
    D_1(\rho) \\ D_{12}(\rho) 
    \end{bmatrix} :=\frac{\partial}{\partial\rho}\begin{bmatrix}
    I_1(\rho) \\ I_{12}(\rho) 
    \end{bmatrix},  \label{eqn:derI} 
\end{align}
to be the derivative of the mutual information vector with respect to $\rho$  where the individual derivatives are given by
\begin{align}
\frac{\partial I_1(\rho)}{\partial\rho}  & =  \frac{-S_1\rho}{1+S_1(1-\rho^2)}, \quad\mbox{and}\label{eqn:dvalues0} \\* 
 \frac{\partial I_{12}(\rho)}{\partial\rho}  & = \frac{\sqrt{S_1 S_2}}{1+S_1 + S_2 + 2\rho\sqrt{S_1S_2}} . \label{eqn:dvalues}
\end{align}
Note that $\rho\in (0,1]$ represents the strictly concave part of the boundary (the part of the boundary where $R_2>0.2$ in Fig.~\ref{fig:cr}), and in this interval we have
$D_1(\rho) < 0$ and $D_{12}(\rho)  >  0$.

\begin{figure}
\centering
    \includegraphics[width=0.87\columnwidth]{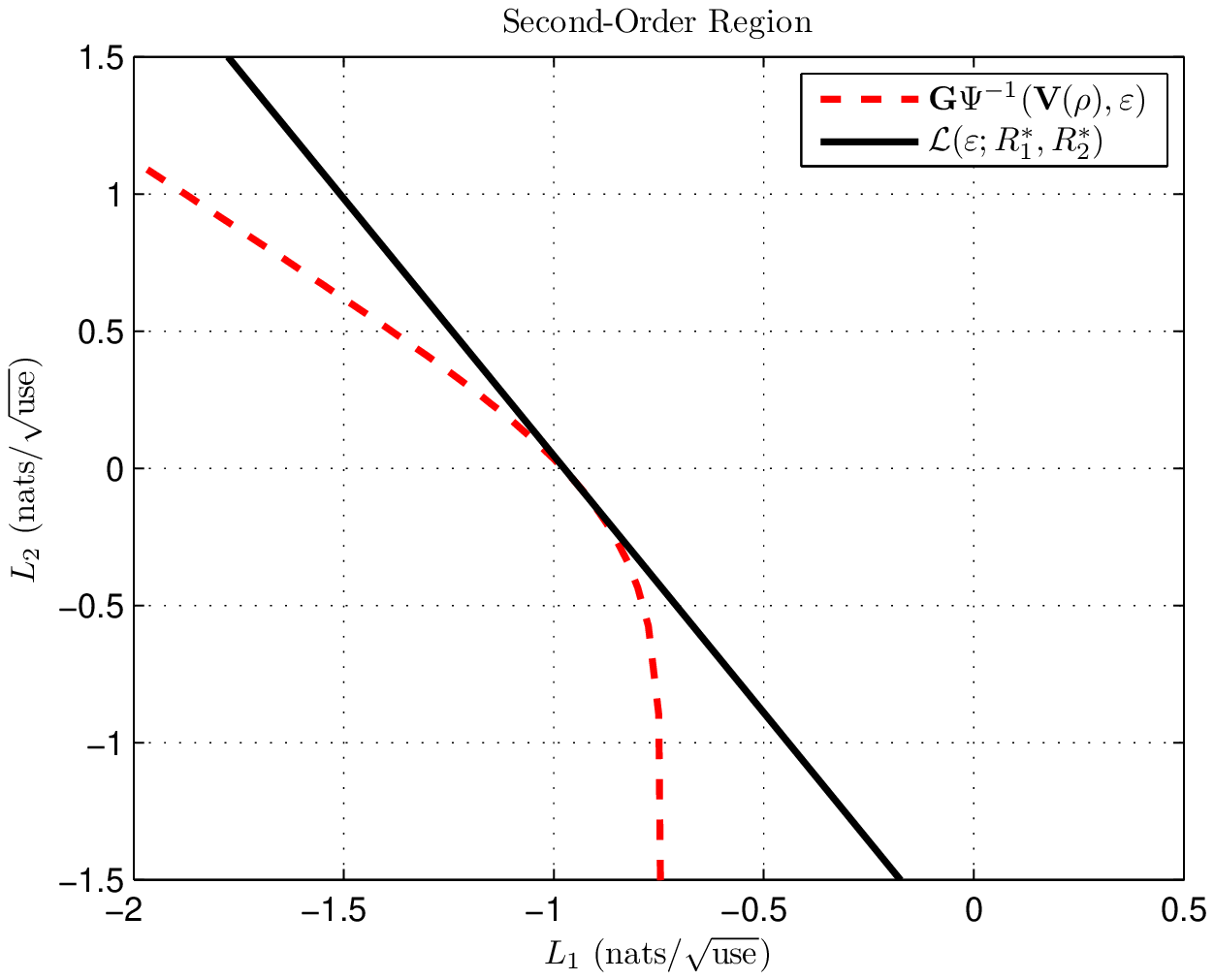}
    \caption{Illustration of the set  $\calL(\eps;R_1^*,R_2^*)$ in \eqref{eqn:second2}
         with $S_1=S_2=1$, $\rho=\frac{1}{2}$ and $\eps = 0.1$.  The set corresponding to $\beta=0$ in \eqref{eqn:second2} is denoted as $\bG\Psi^{-1}(\bV(\rho),\eps)$.  Regions are to the bottom-left of the boundaries. \label{fig:local_region}}
\end{figure}

Furthermore, for a vector $\bv=(v_1, v_2)\in\bbR^2$, we define the {\em down-set} of $\bv$ as
\begin{equation}
\bv^- := \{(w_1,  w_2) \in \bbR^2 : w_1 \le v_1, w_2\le v_2\}  \label{eqn:minus_notation} .
\end{equation}
We are now in a position to state our main result whose proof is sketched in Section \ref{sec:prf_sk_l}.

\begin{theorem}[Local Second-Order Rates] \label{thm:local}
Depending on $(R_1^*, R_2^*)$ (see Fig.~\ref{fig:cr}), we have the following three cases:  \\
Case (i):     $R_1^*=I_1(0)$ and  $R_1^*+R_2^* \le I_{12}(0)$ (vertical segment of the boundary corresponding to  $\rho=0$), 
\begin{align}
\calL(\eps;R_1^*,R_2^*) = \left\{ (L_1, L_2)  : L_1 \le  \sqrt{V_1(0)} \Phi^{-1}(\eps)\right\} \label{eqn:second1} .
\end{align}
Case (ii):    $R_1^*=I_1(\rho)$ and $R_1^*+R_2^*=I_{12}(\rho)$ (curved segment of the boundary corresponding to  $0 <\rho <1$),  
\begin{align}
\calL(\eps;R_1^*,R_2^*)\! =\!    \Bigg\{ (L_1, L_2)  : \begin{bmatrix}
L_1 \\ L_1\!   +\! L_2 
\end{bmatrix} \!\in\!  \bigcup_{\beta\in \bbR}\Big\{  \beta\, \bD(\rho)\! +\!\Psi^{-1}(\bV(\rho),\eps)  \Big\}\Bigg\}. \label{eqn:second2}
\end{align}
Case (iii):       $R_1^*=0$ and  $R_1^*+R_2^* = I_{12}(1)$ (point on the vertical axis corresponding to  $\rho=1$),  
\begin{align}
    &\calL(\eps;R_1^*,   R_2^*) \nn\\*
    &=   \Bigg\{(L_1,L_2) :
     \begin{bmatrix}
L_1 \\ L_1 + L_2 
\end{bmatrix} \!\in\!  \bigcup_{\beta\le0}   \bigg\{ \beta\, \bD(1)  +\begin{bmatrix}
0 \\ \sqrt{ V_{12}(1)}\Phi^{-1}(\eps)
\end{bmatrix}^-\bigg\}    \Bigg\} .\label{eqn:second3}
\end{align}
\end{theorem}
See Fig.~\ref{fig:local_region} for an illustration of  $\calL(\eps;R_1^*,R_2^*)$ in \eqref{eqn:second2} and the set of $(L_1, L_2)$ such that $(L_1,L_1+L_2)$ belongs to $\Psi^{-1}(\bV(\rho),\eps)$, i.e., $\bG\Psi^{-1}(\bV(\rho),\eps)$, where  $\bG = [1,0;-1 ,1]$ is the invertible matrix that
transforms the coordinate system from $[L_1, L_1+L_2]'$ to $[L_1, L_2]'$. In other words,  $\bG\Psi^{-1}(\bV(\rho),\eps)$ is the same set as that in \eqref{eqn:second2} neglecting the union  and setting    $\beta=0$.  It can be seen that $\bG\Psi^{-1}(\bV(\rho),\eps)$ is a strict subset of $\calL(\eps;R_1^*,R_2^*)$. In fact, $\calL(\eps;R_1^*,R_2^*)$ is a half-space  in $\bbR^2$ for any $(R_1^*, R_2^*)$ on the boundary of the capacity region corresponding to $\rho<1$. So  $\calL(\eps;R_1^*,R_2^*)$ in \eqref{eqn:second2} can be alternatively written     as 
\begin{equation}
\calL(\eps;R_1^*,R_2^*) = \big\{ (L_1,L_2) : L_2\le a_\rho L_1 + b_{\rho,\eps}\big\}
\end{equation}
where the     slope  and intercept are respectively defined as
\begin{align}
a_\rho &:= \frac{D_{12}(\rho) - D_1(\rho)}{D_1(\rho)},\quad\mbox{and}\\ 
b_{\rho,\eps} & :=  \inf\big\{ b  \in\bbR:  \exists\,  L_1\in\bbR \mbox{ s.t. }  \nn \\*
&\qquad\qquad (L_1, (a_\rho + 1) L_1 + b)  \in \bG\Psi^{-1}(\bV(\rho),\eps)   \big\}.
\end{align}
\subsection{Discussion of the Main Result}
Observe that in Case (i), the second-order region is simply characterized by a scalar dispersion term $V_1(0)$ and the inverse of the Gaussian cdf $\Phi^{-1}$.  In this part of the boundary, there is effectively only a single rate constraint in terms of $R_1$, since we are operating ``far away'' from the sum rate constraint. This results in a large deviations-type event for the sum rate constraint which has no bearing on second-order asymptotics. This is similar to observations made in Chapters~\ref{ch:sw} and \ref{ch:ic}.

Cases (ii)--(iii) are more interesting, and their proofs do not follow from standard techniques. As in Case (iii) for Theorem~\ref{thm:disp_sw}, the second-order asymptotics for Case (ii) depend on the dispersion matrix $\bV(\rho)$ and the  bivariate Gaussian cdf,  since {\em both} rate constraints are active at a point on the boundary parametrized by $\rho \in (0,1)$. However,   the expression containing $\Psi^{-1}$ alone (i.e., the expression obtained by setting $\beta=0$ in \eqref{eqn:second2}) corresponds to only considering the unique   input distribution $\calN(\bzero,\bSigma(\rho))$ achieving the point $(R_1^*,R_2^*)=(I_1(\rho), I_{12}(\rho)-I_1(\rho))$. From Fig.~\ref{fig:cr}, this is  {\em not} sufficient to achieve all second-order coding rates, since there are non-empty regions within the capacity region that are not contained in the trapezoid  of rate pairs achievable using a single Gaussian $\calN(\bzero,\bSigma(\rho))$. 

Thus, to achieve all $(L_1, L_2)$ pairs in $\calL(\eps;R_1^*, R_2^*)$,  we must allow the sequence of input distributions to vary with the blocklength $n$. This is manifested in the $\beta\, \bD(\rho)$ term. Roughly speaking, our proof strategy of the direct part involves random coding with a sequence of  input distributions that are uniform on two  spheres with correlation coefficient $\rho_n = \rho+O\big(\frac{1}{\sqrt{n}}\big)$ between them.   By a Taylor expansion, the resulting mutual information vector 
\begin{equation}
 \bI(\rho_n)\approx\bI(\rho) + (\rho_n-\rho) \bD(\rho).
\end{equation}
  Since $\rho_n - \rho= O\big(\frac{1}{\sqrt{n}}\big)$, the gradient term $(\rho_n-\rho) \bD(\rho)$ also contributes to the second-order behavior, together with the traditional Gaussian approximation term $\Psi^{-1}(\bV(\rho),\eps)$. 

For the converse, we consider an arbitrary sequence of codes with rate pairs $\{(R_{1n}, R_{2n})\}_{n\in\bbN}$ converging to $(R_1^*,R_2^*)=(I_1(\rho), I_{12}(\rho)-I_1(\rho))$ with second-order behavior given by \eqref{eqn:2nd_mac}. From the global result, we know  $[R_{1n}, R_{1n}+R_{2n}]^T \in \overline{\calR } (n,\eps;\rho_n)$ for some sequence $\{\rho_n\}_{n\in\bbN}$. We then establish, using the definition of the second-order coding rates in~\eqref{eqn:2nd_mac}, that   $\rho_n=\rho+O\big(\frac{1}{\sqrt{n}}\big)$. Finally, by the Bolzano-Weierstrass theorem, we may  pass to a  subsequence of $\rho_n$ (if necessary), thus establishing the converse.

A similar discussion holds true for Case (iii); the main differences are that the covariance matrix is singular, and that the union in \eqref{eqn:second3} is taken over $\beta\le0$ only, since $\rho_n$ can only approach one from below.

\section{Proof Sketches of the Main Results}  \label{sec:prof_sk}

\subsection{Proof Sketch of the Global Bound (Lemma \ref{lem:global})} \label{sec:prof_sk_gl}

\begin{proof} Because the proof is rather lengthy, we only focus on the case where $\rho_n\to\rho \in (-1,1)$. %Furthermore, we assume $\rho_n=\rho$ for all $n$ to simplify the exposition. The general case in which  $\rho_n\to\rho$ follows by continuity.  
The main ideas are already present here. The case where $\rho_n \to \pm 1$ is omitted, and the reader is referred to \cite{ScarlettTan} for the details.  

The converse proof is split into several steps for clarity. In the  first three steps, we perform a series of reductions to simplify the problem. We do so to simplify the evaluation of the probability in the non-asymptotic converse bound  in Proposition~\ref{prop:con_mac}.

Step 1: (Reduction from Maximal to Equal Power Constraints) As usual, by the Yaglom map trick~\cite[Ch.~9, Thm.~6]{conway},  it suffices to consider codes such that the inequalities in 
\eqref{eqn:2nd_mac} hold with equality.  See the argument for the proof of the converse for the asymptotic expansion of the AWGN channel (Theorem~\ref{thm:awgn_asy}). 
 
Step 2: (Reduction from Average to Maximal Error Probability)  Using similar arguments to \cite[Sec.~3.4.4]{Pol10}, it suffices
to prove the converse for maximal (rather than average) error
probability.\footnote{This argument is not valid for the standard MAC, but is possible here due to the partial cooperation (i.e., user 1 knowing both messages).  It is well known that the capacity regions for the MAC under the average and maximum error probability criteria are different, an observation first made by Dueck~\cite{Dueck}. } 
This is shown by starting with an average-error code, and then 
constructing a maximal-error code as follows: (i) Keep only
the fraction $\frac{1}{\sqrt n}$ of user 2's messages with
the smallest error probabilities (averaged over user 1's message);
(ii) For each of user 2's messages, keep only the fraction
$\frac{1}{\sqrt n}$ of user 1's messages with the smallest
error probabilities. 

Step 3: (Correlation Type Classes) Define $\calI_0 := \{0\}$ and  $\calI_k := (\frac{k-1}{n},\frac{k}{n}], k = 1,\ldots, n$, 
and let $\calI_{-k} := -\calI_k$. Consider the {\em correlation 
type classes} (or simply {\em type classes})
\begin{align} 
\calT_n(k)  := \left\{ (\bx_1, \bx_2)  : \frac{\langle\bx_1,\bx_2\rangle}{\|\bx_1\|_2 \|\bx_2\|_2} \in\calI_k \right\}   \label{eqn:corr_type}
\end{align}
where $k = -n,\ldots, n$.  The total number of type classes is 
$2n+1$, which is polynomial in $n$ analogously to the finite alphabet case (cf.~the type counting lemma).  
Using a similar argument to that for the asymmetric broadcast channel in~\cite[Lem.~16.2]{Csi97}, and the fact that
we are considering the maximal error probability so all message pairs $(m_1, m_2)$ have error probabilities not exceeding $\eps$ (cf.~Step 2), 
it suffices to consider codes for which all pairs $(\bx_1,\bx_2)$  that 
are in a {\em single} type class, say indexed by $k$.  This results in a rate loss of $R_1$ and $R_2$ of only $O(\frac{\log n}{n})$. 
We define $\hrho := \frac{k}{n}$ according to the type class indexed by $k$ in \eqref{eqn:corr_type}.

Step 4: (Approximation of Empirical Moments with True Moments) The value of $\rho$ used in the single-letter information densities in \eqref{eqn:info_dens}   is arbitrary, and is chosen to be $\hrho$.% Let the vector $\bj (x_1, x_2, y): = [j_1 (x_1, x_2, y), j_{12}  (x_1, x_2, y)]'$.  

Using the definition of $\calT_n(k)$ and the information densities in \eqref{eqn:info_dens},
we can show that the first and second moments of $\sum_{i=1}^{n}\bj(x_{1i},x_{2i},Y_i)$
are approximately given by $\bI(\hrho)$ and $\bV(\hrho)$ respectively, i.e.,
\begin{align}
\left\|\bbE\left[\frac{1}{n}\sum_{i=1}^n \bj(x_{1i},x_{2i}, Y_i)\right] - \bI(\hrho )\right\|_\infty &\le \frac{\xi_1}{n} ,\quad \mbox{and}   \label{eqn:expectation_j} \\*
\left\|\cov\left[ \frac{1}{\sqrt{n}}\sum_{i=1}^n \bj(x_{1i},x_{2i}, Y_i)\right]  -  \bV(\hrho )\right\|_{\infty} &\le \frac{\xi_2}{n} \label{eqn:Vcalc}
\end{align}
for some $\xi_1>0$ and $\xi_2>0$ not depending on $\hrho$.  The expectation and covariance above are taken with respect to  $W^n(\cdot|\bx_1,\bx_2)$. Roughly speaking, the reason for \eqref{eqn:expectation_j} and \eqref{eqn:Vcalc} is because all pairs of vectors in $\calT_n(k)$ have approximately the same empirical correlation coefficient so the expectation and covariance of appropriately normalized information density vectors are also close to   a representative mutual information vector and dispersion matrix respectively. 

Step 5: (Evaluation of the Non-Asymptotic Converse Bound  in Proposition~\ref{prop:con_mac}) 
Let $\bR_{n} := [R_{1 n}  , R_{1 n} + R_{2  n}  ]'$ (where $R_{jn}=\frac{1}{n}\log M_{jn}$) be the rate vector consisting of the non-asymptotic rates $(R_{1n}, R_{2n})$. Additionally,   let 
\begin{align}
\calF &:=\bigg\{  \log\frac{W^n(Y^n | X_1^n , X_2^n) }{Q_{Y^n|X_2^n} (Y^n|X_2^n) }  \le \log M_{1n}  -  n\gamma \bigg\} \\*
\calG &:=\bigg\{  \log\frac{W^n(Y^n | X_1^n , X_2^n) }{Q_{Y^n } (Y^n ) }  \le \log (M_{1n}M_{2n})   -  n\gamma \bigg\} 
\end{align}
be the two ``error'' events within the probability in~\eqref{eqn:con_mac}.  We then have 
\begin{align}
\Pr(\calF\cup\calG)&=1-\Pr(\calF^c\cap\calG^c) \\
&=1-\bbE_{ X_1^n, X_2^n} \big[ \Pr(\calF^c\cap\calG^c | X_1^n,X_2^n) \big] .\label{eqn:complement}
\end{align}
In particular, using the definition of $\bj(x_1,x_2,y)$ in \eqref{eqn:info_dens} and the fact that $Q_{Y^n|X_2^n}$ and $Q_{Y^n}$ are product distributions, the conditional probability in \eqref{eqn:complement} can be bounded as
\begin{align}
&\Pr(\calF^c\cap\calG^c | X_1^n=\bx_1,X_2^n=\bx_2)\nn\\*
& = \Pr\left( \frac{1}{n}\sum_{i=1}^n \bj(x_{1i},x_{2i}, Y_i) > \bR_{n}  -\gamma\bone \right) \label{eqn:AB}\\
&\le \Pr\Bigg(  \frac{1}{n}\sum_{i=1}^n \Big( \bj(x_{1i},x_{2i}, Y_i) -  \bbE[\bj(x_{1i},x_{2i}, Y_i)] \Big) \nn\\*
 &\hspace{1in}> \bR_{n} - \bI(\hrho )- \gamma\bone - \frac{\xi_1}{n}\bone  \Bigg), \label{eqn:use_norm_less} 
\end{align}
where \eqref{eqn:use_norm_less} follows from the approximation of the empirical expectation in~\eqref{eqn:expectation_j}. In the rest of this global converse proof, $\gamma$ is set to $\frac{\log n}{2n}$ so $\exp(-n\gamma)=\frac{1}{\sqrt{n}}$ in the non-asymptotic converse bound in~\eqref{eqn:con_mac}.

%Now we may write the probability in~\eqref{eqn:con_mac} as a convex combination of  Substituting
%these empirical moments into \eqref{eqn:con_mac}, 

Applying the multivariate Berry-Esseen theorem (Corollary~\ref{corollary:multidimensional-berry-esseen}) to~\eqref{eqn:use_norm_less} yields
\begin{align}
&\Pr(\calF^c\cap\calG^c | X_1^n=\bx_1,X_2^n=\bx_2)  \nn\\*
&\le    \Psi \Bigg(\begin{bmatrix}
\sqrt{n}  \big( I_1(\hrho )  + \gamma + {\xi_1}/{n} -  R_{1 n}  \big) \\   \sqrt{n} \big(  I_{12}(\hrho ) + \gamma + {\xi_1}/{n}-  (R_{1 n}   + R_{2 n} ) \big) 
\end{bmatrix} ;  \nn\\*
&\hspace{1in}  \bzero,\cov\left[ \frac{1}{\sqrt{n}}\sum_{i=1}^n \bj(x_{1i},x_{2i}, Y_i)\right]   \Bigg)  +\frac{\psi(\hrho ) }{\sqrt{n}}, \label{eqn:berry}
\end{align}
where $\psi(\hrho)$ is  a constant.  By Taylor expanding the continuously differentiable function $(z_1,z_2,\bV)\mapsto\Psi(z_1,z_2;\bzero,\bV)$, and  using the approximation of the empirical covariance in~\eqref{eqn:Vcalc} together with the fact that $\det(\bV(\hrho))>0$ for $\hrho\in(-1,1)$, we obtain
\begin{align}
&\Pr(\calF^c\cap\calG^c | X_1^n=\bx_1,X_2^n=\bx_2)  \nn\\* 
&\le    \Psi \Bigg(\begin{bmatrix}
\sqrt{n}  \big( I_1(\hrho )  + \gamma + {\xi_1}/{n} -  R_{1 n}  \big) \\   \sqrt{n} \big(  I_{12}(\hrho ) + \gamma + {\xi_1}/{n}-  (R_{1 n}   + R_{2 n} ) \big) 
\end{bmatrix} ;  \bzero, \bV(\hrho) \Bigg)
%&\hspace{1in}  \bzero,\cov\left[ \frac{1}{\sqrt{n}}\sum_{i=1}^n \bj(x_{1i},x_{2i}, Y_i)\right]   \Bigg) \nn\\*
 + \frac{\eta(\hrho )\log n}{\sqrt{n}} \label{eqn:taylor}
\end{align}
where $\eta(\hrho)$ is a constant.  It should be noted that  $\psi(\hrho ),\eta(\hrho )\to \infty$ as $\hrho \to  \pm 1$, since $\bV(\hrho )$ becomes singular as $\hrho \to \pm 1$.  Despite this non-uniformity, we conclude from \eqref{eqn:con_mac}, \eqref{eqn:complement} and \eqref{eqn:taylor} that any $(n,\eps)$-code with codewords $(\bx_1,\bx_2)$  all belonging to   $\calT_n(k)$ must have rates $(R_{1n}, R_{2n})$ that satisfy 
\begin{align}
\begin{bmatrix}
R_{1 n}  \\ R_{1 n}  + R_{2 n} 
\end{bmatrix}& \in \bI(\hrho ) + \frac{\Psi^{-1}\Big(\bV(\hrho ),\eps + \frac{2}{\sqrt{n}} \!+ \!\frac{\eta(\hrho )\log n}{\sqrt{n}}\Big)}{\sqrt n} \label{eqn:constant_type_bd0}.
%&\subset \bigcup_{\hrho \in\calI_k}\left\{ \bI(\hrho ) + \frac{\Psi^{-1}\big(\bV(\hrho ),\eps\big)}{\sqrt{n}}  + \frac{h(\hrho ,\eps)\log n}{n}\bone\right\}. \label{eqn:constant_type_bd}
\end{align}
We immediately obtain the global converse bound on the $(n,\eps)$-capacity region  (outer bound in~\eqref{eqn:unions} of Lemma \ref{lem:global}) by employing the approximation
\begin{equation}
\Psi^{-1}\bigg( \bV(\hrho),\eps+\frac{c\log n}{\sqrt{n}}\bigg)\subset\Psi^{-1}\big( \bV(\hrho),\eps )+\frac{h(   \bV(\hrho),\eps,c) \,\log n}{\sqrt{n}}\, \bone, \label{eqn:expand_Psi}
\end{equation}
where $c>0$ is an arbitrary finite constant and $h(   \bV(\hrho),\eps,c)$ is finite for $\hrho\ne \pm 1$.   The details of the approximation in \eqref{eqn:expand_Psi} are omitted, and can be found in \cite{ScarlettTan}. \\

% the empirical covariance matrix using~\eqref{eqn:Vcalc}, using standard steps based on the multivariate
%Berry-Esseen theorem  (Corollary~\ref{corollary:multidimensional-berry-esseen}), and finally inverting the relationship between the rates and the error probability.  In particular, note that
%$\bV(\rho)$ is non-singular for $\rho\in(-1,1)$, and the required
%third moment can be uniformly bounded in terms of the powers $S_1$ and $S_2$
%similarly to the single-user case.  \\

We now provide a proof sketch of the achievability part of Lemma~\ref{lem:global} (inner bound in~\eqref{eqn:unions}). At a high level, we will adopt the strategy of drawing random codewords on   appropriate power spheres, similar to the coding strategy for AWGN channels (Section~\ref{sec:awgn}) and the Gaussian IC with SVSI (Chapter~\ref{ch:ic}). We then analyze the ensemble behavior of this random code.

Step 1: (Random-Coding Ensemble) Let $\rho\in[0,1]$ be a fixed correlation coefficient. The ensemble will be defined  
in such a way that, with probability one, each codeword pair falls into the set
\begin{equation}
    \calD_n(\rho) :=  \Big\{\big(\bx_1,\bx_2\big) :  \|\bx_{1}\|_2^2 = nS_1,  \|\bx_{2}\|_2^2  =  nS_2,  \langle\bx_1,\bx_2\rangle =  n\rho\sqrt{S_1S_2}  \Big\}.\label{eqn:setD}
\end{equation}
This means that the power constraints  in \eqref{eqn:snr_mac1}--\eqref{eqn:snr_mac2} are satisfied with equality and the 
empirical correlation between each codeword pair is also exactly $\rho$.    We use superposition coding,
in which the codewords are generated according to
\begin{align}
    &\bigg\{\Big(X_2^n (m_2) ,\{X_1^n(m_1,m_2)\}_{m_1=1}^{M_{1 }}\Big)\bigg\}_{m_2=1}^{M_{2 }} \nn\\*
    &\sim\prod_{m_2=1}^{M_{2 }}\bigg( P_{X_2^n}(\bx_2(m_2))\prod_{m_1=1}^{M_{1 }}P_{X_1^n|X_2^n}(\bx_1(m_1,m_2)|\bx_2(m_2))  \bigg) \label{eqn:super}
\end{align}
for codeword distributions $P_{X_2^n}$ and $P_{X_1^n|X_2^n}$. We choose the codeword distributions to be 
\begin{align}
P_{X_2^n}(\bx_2) & \propto \delta\big\{\|\bx_2\|_2^2 = nS_2\big\},\qquad\mbox{and} \label{sc_px2}\\*
P_{X_1^n|X_2^n}(\bx_1|\bx_2) & \propto \delta\big\{\|\bx_1\|_2^2 = nS_1, \langle\bx_1,\bx_2\rangle=n\rho\sqrt{S_1S_2}\big\} ,\label{sc_px1} 
\end{align}
where $\delta\{\cdot\}$ is the Dirac $\delta$-function, and 
$P_{X^n}(\bx)\propto\delta\{\bx\in\calA\}$ means that $P_{X^n}(\bx)=\frac{\delta\{\bx\in\calA\}}{c}$, 
with the normalization constant $c>0$ chosen such that $\int_{\calA} P_{X^n}(\bx)\, \rmd \bx = 1$.  
In other words, each $\bx_2(m_2)$ is drawn uniformly from an $(n-1)$-sphere  with radius $\sqrt{nS_2}$ and for each $m_2$, each $\bx_1(m_1, m_2)$ is drawn 
uniformly from the set of all $\bx_1$ satisfying the power and correlation coefficient constraints  with equality.  
  These distributions clearly
ensure that the codeword pairs belong to $\calD_n(\rho)$ with probability one. 

Step 2: (Evaluation of the Non-Asymptotic Achievability Bound in Proposition~\ref{prop:ach_mac}) We now need to identify typical sets of $(\bx_2,\by)$ and $\by$ such that the maximum values of the ratios of the densities $\zeta_1$ and $\zeta_{12}$, defined in \eqref{eqn:K1_mac}, are uniformly  bounded on these sets. For this purpose, we leverage the following lemma.  %For concreteness, we make the dependence of certain
%quantities appearing in Proposition \ref{prop:ach_mac} explicit, e.g. $\zeta_1(n,\rho)$.

\begin{lemma} \label{lem:rn_bd}
    Consider the setup of Proposition~\ref{prop:ach_mac}, where the output distributions 
    are chosen to be $Q_{Y^n|X_2^n}:=(P_{X_1|X_2} W)^n$ and $Q_{Y^n }:=(P_{X_1  X_2}  W)^n$ with
    $P_{X_1 X_2} :=\calN(\bzero,\bSigma(\rho))$, and the
    input joint distribution $P_{X_1^n X_2^n}$ is described by~\eqref{sc_px2}--\eqref{sc_px1}.  
    There exist sets $\calA_{1} \subset\calX_2^n\times\calY^n$ and $\calA_{12}\subset \calY^n$ (depending on $n$ and $\rho$) such that the following
%    yielding the following for sufficiently large $n$:
    \begin{align}
        \max_{\rho\in[0,1]} \max\{ \zeta_{1} ,\zeta_{12}  \}& \le \Lambda \label{eqn:Lambda_plus} \\*
        \max_{\rho\in[0,1]} \max \big\{ \Pr\big((X_2^n,Y^n) \notin \calA_1 \big), \Pr\big(Y^n \notin \calA_{12} \big)\big\} &\le \exp(- n\xi ), \label{eqn:atypical}
    \end{align}
    hold for all $n>n_0$, where 
    where $\Lambda <\infty$, $\xi >0$ and $n_0 \in\bbN$ are constants not depending on $\rho$.
\end{lemma}
The proof of this technical lemma is omitted and can be found in~\cite{ScarlettTan}. It extends and refines ideas in Polyanskiy-Poor-Verd\'u's proof of the dispersion of AWGN channels \cite[Thm.~54 \& Lem.~61]{PPV10}. 

Note that the uniformity of the bounds $\Lambda$ and $\exp(-n\xi )$ in~\eqref{eqn:Lambda_plus}--\eqref{eqn:atypical} in $\rho$ is
crucial for handling $\rho$ varying with $n$, as is required in Lemma~\ref{lem:global}.

Equipped with Lemma \ref{lem:rn_bd}, we now apply the multivariate Berry-Esseen theorem  (Corollary~\ref{corollary:multidimensional-berry-esseen}) to estimate the probability in the non-asymptotic achievability bound in Proposition~\ref{prop:ach_mac}. This computation is similar  to that sketched in the converse proof with $\xi_1=\xi_2=0$. This  concludes the achievability proof of Lemma~\ref{lem:global}. 
\iffalse
This follows similar steps to the converse
upon suitable choices of the auxiliary costs.  In particular, let 
us define \begin{align}
    \bj(x_1,x_2) &:= \bbE[\bj(x_1,x_2,Y)],\,\,\mbox{and}\,\,
    \bv(x_1,x_2)  := \cov[\bj(x_1,x_2,Y)], \label{eqn:funcV}
\end{align}
where $Y |\{ X_1=x_1,X_2=x_2\} \sim W(\cdot|x_1,x_2)$.  For both $\rho\in (0,1)$,
we let the first $5$ auxiliary costs equal the entries of the
vector and matrix in \eqref{eqn:funcV} (2 
for $\bj(x_1,x_2)$ and 3 for the symmetric matrix
$\bv(x_1,x_2)$).  It follows from \eqref{eqn:setD} that 
\begin{align}
    \left\|\bbE\left[ \frac{1}{n}\sum_{i=1}^n\bj (x_{1},x_2,Y_i) \right] -  \bI(\rho)\right\|_{\infty} &\le \frac{\delta }{n}, \quad\mbox{and} \label{eqn:costmoment1} \\* 
    \left\|\cov\left[\frac{1}{\sqrt n}\sum_{i=1}^n\bj (x_{1},x_2,Y_i)\right] -  \bV(\rho) \right\|_{\infty} &\le \frac{\delta }{n} \label{eqn:costmoment2}
\end{align}
for all $(\bx_1,\bx_2)\in\calD_n$, where the expectations and covariance
are taken with respect to $W^n(\cdot|\bx_1,\bx_2)$.  The third moment(s)
can be bounded similarly by a suitable choice of the remaining auxiliary
cost(s).
This allows us to apply the multivariate Berry-Esseen theorem  (Corollary~\ref{corollary:multidimensional-berry-esseen}) to conclude the achievability proof of Lemma~\ref{lem:global}. \fi
\end{proof}

\subsection{Proof Sketch of the Local Result (Theorem~\ref{thm:local})} \label{sec:prf_sk_l}
\begin{proof} 
We begin with the converse proof. We only prove the result in Case (ii), because Case (i) is standard (follows from the single-user case in Theorem~\ref{thm:awgn_asy}) and Case (iii) similar to Case (ii).

%\subsubsection{Passage to a Convergent Subsequence}
Step 1: (Passage to a Convergent Subsequence) 
%As mentioned above, we only consider cases (ii) and (iii) as case (i) follows directly from the single-user converse (Theorem \ref{thm:awgn_asy}).
Fix  a correlation coefficient $\rho\in(0,1]$, and consider any sequence of 
$(n,M_{1n},M_{2 n},S_1, S_2,\eps_n)$-codes satisfying~\eqref{eqn:2nd_mac}. 
Let us consider the associated rates $\{ (R_{1 n}, R_{2 n} ) \}_{n\in\bbN}$. As required by the definition of second-order rate pairs $(L_1,L_2)\in\calL(\eps;R_1^*, R_2^*)$, these codes must satisfy 
\begin{align}
\liminf_{n\to\infty}R_{j n} &\ge R_j^*,    \label{eqn:first_order_opt} \\
\liminf_{n\to\infty} {\sqrt{n}}\big( R_{j n} -  R_j^*\big )&\ge L_j,\quad j = 1,2,  \label{eqn:sec_order_opt} \\*
\limsup_{n\to\infty}\eps_n &\le\eps \label{eqn:error_prob}
\end{align}
for some $(R_1^*,R_2^*)$ on the boundary parametrized by $\rho$, i.e., $R_1^*=I_1(\rho)$ and 
$R_1^*+R_2^*=I_{12}(\rho)$. The first-order optimality condition in \eqref{eqn:first_order_opt} 
is not  explicitly required by \eqref{eqn:2nd_mac}, but it is implied by \eqref{eqn:sec_order_opt}.
Letting $\bR_n := [ R_{1 n},  R_{1 n}+ R_{2 n}]'$ be the non-asymptotic rate vector, we have, from the global converse bound in~\eqref{eqn:unions},  
that there exists   a (possibly non-unique) sequence $\{\rho_n\}_{n\in\bbN}\subset [-1,1]$ such that 
\begin{equation}
    \bR_n \in \bI(\rho_n)+\frac{\Psi^{-1}(\bV(\rho_n),\eps)}{\sqrt{n}} +  \overg(\rho_n,\eps,n) \bone. \label{eqn:useoverg}
\end{equation}
Since we used the $\liminf$ for the rates and $\limsup$ for the error probability in the conditions in~\eqref{eqn:first_order_opt}--\eqref{eqn:error_prob}, we may pass to a \emph{convergent} (but otherwise arbitrary) subsequence 
of $\{\rho_n\}_{n\in\bbN}$, say indexed by $\{n_l\}_{l\in\bbN }$.  Recalling that the $\liminf$ 
(resp.\ $\limsup$)  is the infimum (resp.\ supremum) of all subsequential limits, any 
converse result associated with this subsequence also applies to the original sequence.
Note that at least one convergent subsequence is guaranteed to exist, since $[-1,1]$ is compact.  

For the sake of clarity, we avoid explicitly writing the subscript $l$. However, 
it should be understood that asymptotic notations such as $O(\cdot)$ and $(\cdot)_n \to (\cdot)$ 
are taken with respect to the convergent subsequence.

Step 2: (Establishing The Convergence of $\rho_n$ to $\rho$)  
%Fix  a correlation coefficient $\rho\in(0,1)$, and consider any sequence of 
%codes satisfying \eqref{eqn:2nd_mac} for some $(R_1^*,R_2^*)$ on the 
%boundary parametrized by $\rho$, i.e., $R_1^*=I_1(\rho)$ and 
%$R_1^*+R_2^*=I_{12}(\rho)$.  
%%Again letting $\bR_n := [ R_{1n},  R_{1n}+ R_{2n}]'$, 
%It follows
%from the global converse result that
%\begin{equation}
%    \bR_n \in \bI(\rho_n)+\frac{\Psi^{-1}(\bV(\rho_n),\eps)}{\sqrt{n}} +  o\left(\frac{1}{\sqrt n}\right) \bone\label{eqn:uniformconverse}
%\end{equation}
%for some (possibly non-unique) sequence $\{\rho_n\}_{n\in\bbN}$. 
Although $\overg(\rho_n,\eps,n)$ in \eqref{eqn:useoverg} depends on $\rho_n$, we know from the global bounds on the $(n,\eps)$-capacity region (Lemma~\ref{lem:global}) that it is $o\big(\frac{1}{\sqrt n}\big)$ for both
$\rho_n\to\pm1$ and $\rho_n\to\rho\in(-1,1)$.  Hence,
\begin{equation}
    \bR_n \in \bI(\rho_n)+\frac{\Psi^{-1}(\bV(\rho_n),\eps)}{\sqrt{n}} +  o\left(\frac{1}{\sqrt n}\right)\bone. \label{eqn:uniformconverse}
\end{equation}
We claim that the result in \eqref{eqn:uniformconverse} allows us to conclude that every   sequence $\{\rho_n\}_{n\in\bbN}$  that serves to parametrize an outer bound of the non-asymptotic rates in  \eqref{eqn:useoverg}  converges to $\rho$. Indeed, since the boundary of the capacity region is curved and uniquely 
parametrized by $\rho$ for $\rho\in(0,1]$, $\rho_n\not\to\rho$ implies for some $\eta>0$ 
and for all sufficiently large $n$ that either $I_1(\rho_n) \le I_1(\rho)-\eta$ 
or $I_{12}(\rho_n) \le I_{12}(\rho) - \eta$.
Combining this with~\eqref{eqn:uniformconverse}, we deduce that
\begin{equation}
 R_{1n}\le I_1(\rho)-\frac{\eta}{2},\quad\mbox{or}\quad R_{1n}+R_{2n}\le I_{12}(\rho )-\frac{\eta}{2} 
\end{equation}
for 
sufficiently large $n$. This, in turn, contradicts the convergence  
of $(R_{1n}, R_{2n})$ to $(R_1^*,R_2^*)$ implied by \eqref{eqn:2nd_mac}.

Step 3: (Establishing The Convergence Rate of $\rho_n$ to $\rho$) Because each entry of  $\bI(\rho)$ is twice continuously differentiable, a Taylor expansion yields 
\begin{align}
\bI(\rho_n) = \bI(\rho) + \bD(\rho) (\rho_n-\rho) + O\big( (\rho_n - \rho)^2 \big) \bone. \label{eqn:taylor_expand_I}
\end{align}
In the case that $\rho_n-\rho=\omega\big(\frac{1}{\sqrt n}\big)$, it is 
not difficult to show that \eqref{eqn:uniformconverse} and \eqref{eqn:taylor_expand_I} imply
\begin{align}
\bR_n  \le\bI(\rho) + \bD(\rho) (\rho_n-\rho) + o( \rho_n-\rho )\bone. \label{eqn:omega_term}
\end{align}
Since the first  entry of $\bD(\rho)$ is negative and the second entry is 
positive, \eqref{eqn:omega_term}  states that $L_1=+\infty$ (i.e., a large 
addition to $R_1^*$) only  if $L_1+L_2=-\infty$ (i.e., a large backoff 
from $R_1^* + R_2^*$),  and $L_1+L_2=+\infty$ only if $L_1=-\infty$. This is due the fact that we only consider second-order deviations from the boundary of the capacity region of the order $\Theta\big(\frac{1}{\sqrt n}\big)$. Neglecting these degenerate cases as they are already captured by Theorem \ref{thm:local} (cf.~Fig.~\ref{fig:local_region}), in the remainder, we focus on 
%Thus, this case does not play a role in the characterization of
%$\calL(\eps;R_1^*,R_2^*)$, and we thus focus on the 
case where $\rho_n-\rho=O\big(\frac{1}{\sqrt n}\big)$.
%\footnote{To be more precise, $\rho_n-\rho=O\big(\frac{1}{\sqrt n}\big)$ and $\rho_n-\rho=\omega\big(\frac{1}{\sqrt n}\big)$ are not exhaustive. However, since we are using $\liminf$ for  the rate constraints and $\limsup$ for the error probability constraint (cf.~\eqref{eqn:2nd_mac}),  for the converse part, we may pass to a subsequence $\{n_l\}$ if necessary. Subsequently, we  analyze the behavior of the rates along $\{n_l\}$  in the following. To keep the notation simple in the following, we avoid explicitly writing the subscript $l$, with the understanding that all asymptotic notation are taken with respect to $\{n_l\}$.  }

Step 4: (Completion of the Proof)
Assuming now that $\rho_n-\rho=O\big(\frac{1}{\sqrt n}\big)$, we
can use the Bolzano-Weierstrass theorem 
to conclude that there exists a (further) subsequence indexed by  $\{n_k\}_{k\in\bbN}$  (say) 
such that $\sqrt{n_k}(\rho_{n_k}-\rho)\to\beta$
for some $\beta\in\bbR$.  Then, for the blocklengths indexed by $n_k$,   by combining \eqref{eqn:uniformconverse} and \eqref{eqn:taylor_expand_I}, we have 
\begin{align}
\sqrt{n_k}\big(\bR_{n_k}-\bI(\rho)\big)\in\beta\, \bD(\rho) + \Psi^{-1}(\bV(\rho ) ,\eps)  +o(1) \, \bone \label{eqn:subseq} .
\end{align}
%where the $o(1)$ term   combines the $o\big(\frac{1}{\sqrt{n}}\big)$ term  in \eqref{eqn:after_taylor} and the deviation $(\tau_{n_k} - \beta)\max \{ -D_1(\rho) ,  D_{12}(\rho) \}$.    
Here we have also used the fact that the set-valued function $\rho\mapsto \Psi^{-1}(\bV(\rho ) ,\eps) $ is ``continuous'' to approximate $ \Psi^{-1}(\bV(\rho_n ) ,\eps) $ with $ \Psi^{-1}(\bV(\rho ) ,\eps) $. The details of this technical step are omitted, and can be found in \cite{ScarlettTan}. 

By referring to the second-order optimality condition in \eqref{eqn:sec_order_opt}, and applying the definition of the limit inferior, we know that  {\em every} convergent subsequence of $\{R_{jn}\}_{n\in\bbN}$ has a subsequential limit that satisfies $\lim_{k\to\infty} \sqrt{n_k} \big(R_{jn_k}-R_j^*)\ge L_j$ for $j = 1,2$. In other words,  for all $\gamma>0$, there exists an integer $K_j$ such that $\sqrt{n_k} \big(R_{j n_k}-R_j^*)\ge L_j-\gamma$ for all $k\ge K_j$.  Thus, for all $k\ge \max\{K_1, K_2\}$, we may lower bound  each component in the vector on the left of  \eqref{eqn:subseq} with $L_1-\gamma$ and $L_1 + L_2-2\gamma$. There also exists an integer $K_3$ such that the $o(1)$  terms are upper bounded by $\gamma$ for all $k\ge K_3$. We conclude that any pair of $(\eps,R_1^*, R_2^*)$-achievable second-order coding rates $(L_1, L_2)$ must satisfy
\begin{align}
\begin{bmatrix}
L_1-2\gamma \\ L_1 + L_2 -3\gamma
\end{bmatrix} \in \bigcup_{\beta\in\bbR}\left\{  \beta\, \bD(\rho)+ \Psi^{-1}(\bV(\rho ) ,\eps) \right\}.
\end{align}
Finally, since $\gamma>0$ is arbitrary, we can take $\gamma \to 0$, thus completing the converse proof for Case (ii).  \\

% Combining this observation
%with \eqref{eqn:uniformconverse} and \eqref{eqn:taylor_expand_I} yields \eqref{eqn:second2}. \\
The achievability part is similar to the converse part, yet simpler.  
Specifically, we can simply choose 
\begin{equation}
\rho_n := \rho + \frac{\beta}{\sqrt{n}},
\end{equation}
and apply the above arguments based on Taylor expansions.
\end{proof}

\section{Difficulties in the Fixed Error   Analysis for the MAC}\label{sec:mac_gen}

We conclude our discussion  by discussing the difficulties in performing fixed error  probability analysis for the discrete memoryless or Gaussian MACs (with non-degraded message sets). 

First, it is known that the capacity region of the MAC depends on whether one is adopting the average or maximal error probability criterion. The capacity regions are, in general, different \cite{Dueck}. In Step 2 of the converse proof, we performed an important reduction from  the average to the maximal error probability criterion. This is one obstacle for any (global or local) converse proof for fixed error analysis of the MAC. 

Second, in the characterization of the capacity region of the  discrete memoryless MAC, one needs to involve an auxiliary time-sharing random variable $Q$ \cite[Sec.~4.5]{elgamal}. At the time of writing, there does not appear to be a principled and unified way to introduce such a variable in  strong  converse proofs (unlike weak converse proofs \cite{elgamal}). 

Finally, for the discrete memoryless MAC, one needs to take the convex closure of the union over input distributions $P_{X_1|Q}, P_{X_2|Q}$ for a given time-sharing distribution  $P_Q$ \cite[Sec.~4.5]{elgamal}.  In the absence of the degraded message sets (or asymmetry)  assumption, one needs to develop a converse technique, possibly related to the wringing technique of Ahlswede \cite{Ahl82},  to assert that the given codewords pairs are almost
independent (or almost orthogonal for the Gaussian case). By leveraging the degraded message sets assumption, we circumvented this requirement  in this chapter but for the MAC, it is not clear whether the wringing technique yields a redundancy term that matches the best known inner bound to the second-order region \cite{Mol13,Scarlett13b}. 
 
\chapter{Summary, Other Results, Open Problems}  \label{ch:con}
\section{Summary and Other Results}
In this monograph, we compiled a list of conclusive  fixed error   results in information  theory. We began our discussion with a    review of  binary hypothesis testing and used the asymptotic expansions of the $\eps$-information spectrum divergence and $\eps$-hypothesis testing divergence for product measures to derive similar asymptotic expansions for the minimum code size in lossless data  compression. Lossy data compression and channel coding were discussed in detail  next. These subjects  culminated in our derivation of an asymptotic expansion for the source-channel coding rate. We then analyzed various channel models whose behaviors are governed by random states.

In this monograph, we also discussed  a small collection of problems in multi-user information theory~\cite{elgamal}, where  we were interested in quantifying the optimum speed of rate pairs converging towards  a  fixed point on the boundary of the (first-order) capacity region in the channel coding case, or optimal rate region in the source coding case. We saw three examples of problems in network information theory where conclusive results can be obtained in the second-order sense.  These included the   distributed lossless source coding (Slepian-Wolf) problem,  as well as some special classes of Gaussian multiple-access and interference channels. 

We conclude our treatment of  fixed error asymptotics in information theory  by mentioning   related works in the literature. 
\subsection{Channel Coding}
Early works on  fixed error asymptotics  in channel coding by Dobrushin~\cite{Dobrushin}, Kemperman~\cite{Kemperman} and Strassen~\cite{Strassen} were discussed in Chapter~\ref{ch:cc}. The interest in asymptotic expansions  was revived in recent years  by the works of Hayashi~\cite{Hayashi09} and Polyanskiy-Poor-Verd\'u~\cite{PPV10}. Before these prominent works  came to the fore,  Baron-Khojastepour-Baraniuk~\cite{Baron04b}   considered the rate of convergence to channel capacity for simple channel models such as the binary symmetric channel. %The important relation between hypothesis testing and channel coding (Proposition \ref{prop:converse}) can be traced back to the early works by    Shannon-Gallager-Berlekamp~\cite{sgb} and Dobrushin~\cite{Dobrushin}.

In this monograph, we  did not discuss channels with feedback or variable-length terminations, both of which are important in practical communication systems. Polyanskiy-Poor-Verd\'u~\cite{PPV11b} studied various incremental redundancy schemes and derived  several asymptotic expansions. Generally, the $\Theta(\sqrt{n})$ dispersion term is not present, showing that channels with feedback  perform much better than without the feedback, an observation that is also corroborated by a more traditional error exponent analysis~\cite{Burnashev, Har77}.   Williamson-Chen-Wesel~\cite{Will13}    showed   that their proposed reliability-based decoding schemes for variable-length coding with feedback  can achieve higher rates than~\cite{PPV11b}. Altu\u{g}-Wagner~\cite{AW14} showed that full output feedback improves the second-order term in the asymptotic expansion for channel coding if $V_{\min}<V_{\max}$.    Tan-Moulin~\cite{TanMou14}   considered the second-order  asymptotics of erasure and list decoding. This analysis is the fixed error probability  analogue of Forney's analysis of erasure and list decoding from the error exponents perspective~\cite{Forney68}. Erasure decoding is intimately connected to decision feedback or automatic retransmission request (ARQ) schemes as the declaration of an erasure event at the decoder can inform the encoder   to resend the erased information bits.

Shkel-Tan-Draper~\cite{Shkel13,   Shkel14b} considered the unequal error protection of message classes and related the  asymptotic expansions for this problem to lossless joint source-channel coding \cite{Shkel}. Moulin \cite{Moulin13} studied the   asymptotics for the channel coding problem up to the fourth-order term using strong large deviation techniques~\cite[Thm.~3.7.4]{Dembo}. Matthews~\cite{Matt13} made an interesting observation  concerning the relation of the non-asymptotic channel coding converse (Proposition~\ref{prop:converse}) to so-called non-signaling codes in quantum information.  He demonstrated efficient linear programming-based algorithms to evaluate the converse for DMCs.

Other (rather more unconventional) works on  fixed error asymptotics  for point-to-point communication include Riedl-Coleman-Singer's   analysis of queuing channels~\cite{riedl2011b}, Polyanskiy-Poor-Verd\'u's analysis of the minimum energy for sending $k$ bits  for Gaussian channels with and without feedback~\cite{PPV11c}, and    Ingber-Zamir-Feder's analysis of   the infinite constellations problem \cite{ingber11b}.%,  first considered by Poltyrev~\cite{Poltyrev}.  

\subsection{Random Number Generation, Intrinsic Randomness and Channel Resolvability}
The problem of \emph{intrinsic  randomness} is to approximate an arbitrary source with uniform bits while {\em random number generation} is the dual, i.e., that of generating uniform bits from a given source~\cite[Ch.~2]{Han10}~\cite{HV93}. These problems   were treated from the fixed {\em approximation error} (in terms of the variational distance) perspective by Hayashi~\cite{Hayashi08} and Nomura-Han~\cite{Nom13b}. An interesting observation made by Hayashi in~\cite{Hayashi08} is that the folklore theorem\footnote{The {\em folklore theorem} \cite{Han_folklore} of Han states that   ``the output from any source encoder
working at the optimal coding rate with asymptotically vanishing probability
of error looks   almost completely random.''}  posed by Han~\cite{Han_folklore}   does not hold for the variational distance criterion. This is  interesting, because the first-order fundamental limit for source coding and intrinsic randomness is the same, i.e., the entropy rate~\cite[Ch.~2]{Han10} (at least for sources that satisfy the Shannon-McMillan-Breiman  theorem). Thus, the violation of the folklore theorem for variational distance appears to be distillable only from the study of second- and not first-order asymptotics, demonstrating additional insight one can glean from studying higher-order  terms in asymptotic expansions. 

The {\em channel resolvability} problem consists in approximating the output statistics of an arbitrary channel given uniform bits at the input~\cite[Ch.~6]{Han10}~\cite{HV93}. It is particularly useful for the strong converse of the identification problem~\cite{AD89}. Watanabe and Hayashi~\cite{WH14} considered the channel resolvability problem, proving a second-order coding theorem  under an ``information radius'' condition not dissimilar to what is known for channel coding~\cite[Thm.~4.5.1]{gallagerIT}. 
\subsection{Channels with State}
For channels with random state, Watanabe-Kuzuoka-Tan~\cite{WKT13} and Yassaee-Aref-Gohari~\cite{YAG13b} derived the best non-asymptotic bounds for the Gel'fand-Pinsker problem, improving on those by Verd\'u~\cite{Ver12}. With these bounds, one can easily derive achievable second-order coding rates by appealing to various Berry-Esseen theorems.  The technique in \cite{WKT13} is based on channel resolvability \cite{HV93} and channel simulation \cite{Cuff12} while that in \cite{YAG13b} is based on an elegant coding scheme known as the stochastic likelihood decoder (also known as the ``pretty good measurement'' in quantum information), which is also applicable to other multi-terminal problems such as multiple-description coding and the Berger-Tung problem~\cite{elgamal}.  Scarlett~\cite{Scarlett14} also considered the second-order asymptotics for the discrete memoryless Gel'fand-Pinsker problem and used ideas in Section~\ref{sec:state_ed} to evaluate the best known achievable second-order coding rates based on constant composition codes.  

Polyanskiy~\cite{Pol13b} derived the second-order asymptotics  for the compound channel where the channel state is non-random  in contrast to the models studied in Chapter~\ref{ch:state}. Similar to the Gaussian MAC with degraded message sets, the second-order term depends on the variance of the channel information density {\em and} the derivatives of the mutual informations. Finally, Hoydis {\em et al.}~\cite{hoydis13,Hoydis}   considered block-fading MIMO channels. In contrast to Section~\ref{sec:quasi}, here the channel matrix is not quasi-static and so the analysis is somewhat more involved and requires the use of random matrix theory. 

\subsection{Multi-Terminal  Information Theory}
The advances in the second-order asymptotics for multi-terminal problems have  been modest. Early works include those by Sarvotham-Baron-Baranuik~\cite{Sar05, Sar05b}   and He  {\em et al.}~\cite{He09}  for the single-encoder Slepian-Wolf problem. However, unlike our treatment   in Chapter~\ref{ch:sw}, there is only one source to be compressed, and {\em full} side-information is available at the decoder.

Other authors~\cite{huang12,Mol12, Mol13,Scarlett13b} also considered inner bounds to the $(n,\eps$)-rate regions (also called {\em global} achievability regions) for the discrete memoryless and Gaussian MACs, but it appears  that conclusive results are much harder to derive without any further assumptions on the channel model. These are  multi-terminal channel coding analogues of the corresponding discussion for Slepian-Wolf coding in Section~\ref{sec:glob}. See further discussions in  Section~\ref{sec:open_multi}.

\subsection{Moderate Deviations, Exact Asymptotics and Saddlepoint Approximations}
The study of second-order coding rates is intimately related to moderate deviations analysis. In the former, the error probability is bounded above by a non-zero constant and   optimal rates converge  to the first-order fundamental limit with speed $\Theta(\frac{1}{\sqrt{n}})$. In the latter, the error probability  decays to zero sub-exponentially  and the optimal rates converge to the first-order fundamental limit slower than $\Theta(\frac{1}{\sqrt{n}})$. The dispersion   also appears in the solution of the moderate deviations analysis because the second derivative of the error exponent (reliability function) is inversely proportional to the dispersion. The study of moderate deviations in information theory started with the work by Altu\u{g}-Wagner~\cite{altug10} and Polyanskiy-Verd\'u~\cite{pol10e} on channel coding. Sason~\cite{Sas11},  Tan~\cite{Tan12}  and Tan-Watanabe-Hayashi~\cite{TWH14} considered the binary hypothesis testing, rate-distortion and lossless  joint source-channel coding counterparts respectively. 

In  Section \ref{sec:digress}, we mentioned efforts from Altu\u{g}-Wagner \cite{altug_refinement1,altug_refinement2} and Scarlett-Martinez-{Guill\'{e}n i F\`{a}bregas} \cite{Sca13} in deriving the {\em exact asymptotics} for channel coding. The authors were motivated to find the prefactors in the error exponents regime for various classes of DMCs. Scarlett-Martinez-{Guill\'{e}n i F\`{a}bregas}  \cite{Scarlett14a} recently demonstrated that results concerning second-order coding rates, moderate deviations, large deviations, and even exact asymptotics may be unified through the use of so-called {\em saddlepoint approximations}. %The potential connection between third-order terms and prefactors (Table \ref{tab:prefactor}) is an intrui

\section{Open Problems and Challenges Ahead}
Clearly, there are many avenues of further research, some of which we mention here. We also highlight some challenges we foresee.
\subsection{Universal Codes}
In Section \ref{sec:universal_lossless}, we analyzed a partially universal source code that achieves the source dispersion (varentropy). The source code only requires the knowledge of the entropy and the varentropy.  The channel dispersion can also be achieved using partially universal channel codes as discussed in the paragraph above~\eqref{eqn:channel_con}. However, the third-order terms are much more difficult to quantify. It would be interesting to pursue research in the third-order asymptotics of source and channel coding for {\em fully} universal codes to understand the loss of performance due to universality.  Work along these lines for fixed-to-variable length lossless source coding has been carried out by Kosut and Sankar~\cite{KosutSankar,KosutSankar14}.
\subsection{Side-Information Problems}
Watanabe-Kuzuoka-Tan~\cite{WKT13} and Yassaee-Aref-Gohari~\cite{YAG13b}  derived the best known achievability bounds for side-information problems including the  Wyner-Ahlswede-K\"orner (WAK)  problem~\cite{Ahl75,Wyner75} and the Wyner-Ziv~\cite{wynerziv} problem. However, non-asymptotic converses are difficult to derive for such problems which involve auxiliary random variables. Even when they can be derived, the evaluation of such converses asymptotically  appears to be formidable. 

Because a second-order converse implies the strong converse, it is useful to first understand the techniques involved in obtaining a strong converse. To the best of the author's  knowledge, there are only three approaches that may be used to obtain strong
converses for network problems whose first-order (capacity) characterization involves auxiliary random
variables. The first is the information spectrum method \cite{Han10}. For example, Boucheron and Salamatian~\cite[Lem.~2]{bouch00}  provide a non-asymptotic
converse bound for the asymmetric broadcast channel. However, the bound is neither computable nor amenable to good approximations for large or even moderate blocklengths $n$ as one has to perform an exhaustive search over the space of all $n$-letter auxiliary random
variables. The second is the entropy and image size characterization
technique \cite{Ahls76} based on the blowing-up lemma  \cite{Ahls76,Marton86}.  (Also see the monograph \cite{RagSason} for a thorough description of this technique.) This has been used to prove the strong converse for the
WAK problem~\cite{Ahls76}, the asymmetric broadcast channel~\cite{Ahls76},   the Gel'fand-Pinsker problem \cite{tyagi} and the Gray-Wyner problem  \cite{GuEffros}. However, the use of the blowing-up approach to obtain second-order converse bounds is not straightforward. The third method involves a  change-of-measure argument, and
was used in the work of Kelly and Wagner \cite[Thm.~2]{Kel12} to prove an upper bound on the error exponent for WAK
coding. Again, it does not appear, at first glance, that this argument is amenable to second-order analysis. 

A problem similar  to side-information problems such as Gel'fand-Pinsker, Wyner-Ziv and WAK is the multi-terminal statistical inference problem studied by Han and Amari~\cite{HanAmari} among others. Asymptotic expansions with non-vanishing type-II error probability may be derivable  under some settings (using established techniques),  if the first-order  characterization is conclusively known, and there  are no auxiliary random variables, e.g., the problem of multiterminal detection with zero-rate compression \cite{Shalaby}.
\subsection{Multi-Terminal  Information Theory}\label{sec:open_multi}
The study of second-order asymptotics for  multi-terminal problems is at its infancy and the problems described in this monograph form only the tip of a large iceberg. The primary difficulty is our inability to deal, in a systematic and principled  way, with auxiliary random variables for the (strong)   converse part. Thus, genuinely new non-asymptotic converses need to  be developed, and these converses have to be amenable to asymptotic evaluations in the presence of auxiliary random variables. As an example, for the degraded broadcast channel, the usual non-intuitive identification of the auxiliary random variable by Gallager~\cite{Gal74} (see \cite[Thm.~5.2]{elgamal}) for proving the weak converse   does not suffice as the strong converse is implied by a second-order converse. Other possible techniques,   such as  information spectrum analysis \cite{bouch00} or the blowing-up lemma~\cite{Ahls76}, were highlighted in the previous section. Their limitations were also discussed. For the discrete memoryless MAC, a strong converse was proved by Ahlswede \cite{Ahl82} but his  wringing technique does not seem to be amenable to second-order refinements as discusseed in Section \ref{sec:mac_gen}. 

In contrast to the single-user setting, constant composition codes may be beneficial even in the absence of cost constraints for discrete memoryless multi-user problems. This is because there does not exist an analogue of the relation in \eqref{eqn:uequalsv}, where the unconditional information is equal to the conditional information variance for all CAIDs. Scarlett-Martinez-{Guill\'{e}n i F\`{a}bregas}~\cite{Scarlett13b} provided the best known inner bounds to the $(n,\eps)$-rate region for the discrete memoryless MAC. Tan-Kosut~\cite{TK14} also showed that conditionally constant composition codes also outperforms \iid codes for the asymmetric broadcast channel when the error probability is non-vanishing. It would be fruitful to continue pursuing research in the direction of constant composition codes for multi-user problems (e.g., the interference channel) to exploit their full potential.

\subsection{Information-Theoretic Security}
Finally, we mention that within the realm of information-theoretic security~\cite{Bloch_book, liang_book}, there are several partial results in the fixed error and leakage setting. Yassaee-Aref-Gohari~\cite[Thm.~4]{YAG13} used a general random binning procedure, called {\em output statistics of random binning}, to derive a second-order achievability bound for the wiretap channel~\cite{Wyn75}, improving on earlier work by Tan~\cite{Tan_ICCS}. The constraints pertain to the error probability of the legitimate receiver in decoding the message and the leakage rate to the eavesdropper  measured in terms of the variational distance. However, in \cite{YAG13}, there were no  converse results  even for the  less noisy (or even degraded) case where there are no auxiliary random variables.

The most conclusive work in  thus far in information-theoretic security  pertains to the secret key agreement model \cite{AC93},  where the second-order asymptotics were derived by Hayashi-Tyagi-Watanabe~\cite{HTW14}. Interestingly, the non-asymptotic converse bound relates the size of the key to the $\eps$-hypothesis testing divergence, similar to some point-to-point problems as discussed in this monograph.  The non-asymptotic direct bound is derived based on the {\em information spectrum slicing} technique (e.g., \cite[Thm.~1.9.1]{Han10}). The author believes that the fixed error and  fixed leakage analysis for the wiretap channel, leveraging   the secret key result, may lead to new insights into the design of secure communication systems at the physical layer. For converse theorems, the development of novel strong converse techniques for the wiretap channel appears to be necessary; there are recent results on this for degraded wiretap channels using the information spectrum method~\cite{TanBloch} and active hypothesis testing~\cite{HTW14b}. 

%\subsection{Quantum Shannon Theory}

\chapter*{\centering Acknowledgements}
\addcontentsline{toc}{chapter}{Acknowledgements}  
%My journey in information theory has been an atypical one. My Ph.D.\ was 

%Many parts of this book I would like to thank 

Even though this monograph  bears only my name, many parts of it are works of other information theorists and the remaining parts germinated from my  collaborations with my   co-authors. I sincerely thank my co-authors  for educating me on information theory, ensuring I was  productive and, most importantly, making research fun. My first work on fixed error asymptotics was in collaboration with Oliver Kosut while we were both at MIT. We had a wonderful collaboration on the second-order asymptotics of the Slepian-Wolf problem, discussed in Chapter~\ref{ch:sw}. Upon my return to Singapore, I   had the tremendous pleasure of working with Marco Tomamichel on several projects that led to some of the key results in Chapters~\ref{ch:cc} and~\ref{ch:state}.  I thank Sy-Quoc  Le and Mehul Motani for the collaboration that led to results concerning    Gaussian interference channels in Chapter~\ref{ch:ic}.   Jonathan Scarlett and I had numerous   discussions on  various topics in information theory, including a thread that led to the results on   Gaussian MAC with degraded message sets  in Chapter~\ref{ch:mac}.  I have also had the distinguished honor of collaborating on the topic of fixed error asymptotics   with Stark Draper,   Masahito Hayashi, Shigeaki Kuzuoka, Pierre Moulin, Yanina Shkel, and Shun Watanabe.

In addition to my collaborators, I have had many   interactions with other colleagues on this exciting topic,  including Y\"ucel Altu\u{g}, Yuval Kochman,  Shaowei Lin, Alfonso Martinez,  Ebrahim MolavianJazi, Lalitha Sankar, and Da Wang. I thank them tremendously for sharing their insights on various problems.% in this domain.

I would like to express my deepest gratitude to  Jonathan Scarlett for reading through the first  draft  of this monograph, providing me with  constructive comments,   spotting typos, and helping to fix egregious errors. Special thanks also goes out to Stark Draper, Silas Fong, Ebrahim MolavianJazi, Mehul Motani, Mark Wilde and Lav Varshney for proofreading parts of  later versions of the monograph.% I am appreciative to Stark Draper and Mehul Motani for suggesting the    unambiguous and succinct  title of this  monograph.

I am deeply indebted to my academic mentors   Professor  Alan Willsky, Professor  Stark Draper and Dr.\ C\'edric F\'evotte for teaching me how to write in a clear, concise and yet precise manner. Any parts of this monograph that violate these ideals are, of course, down to my personal  inadequacies.  

I am especially  grateful to  the National University of Singapore (NUS) for providing me with the ideal environment to pursue my academic dreams. This work is supported by    NUS startup grants   R-263-000-A98-750 (FoE) and    R-263-000-A98-133 (ODPRT).

%I sincerely thank    the Editor-In-Chief of the  {\em Foundations and Trends in Communications and Information Theory}  Professor Sergio Verd\'u as well as  Mike Casey from Now Publishers  for encouraging me to write this article.  
I am very grateful to Editor Professor Yury Polyanskiy as well as the two anonymous reviewers for their extensive and constructive suggestions during the revision process. One reviewer, in particular,   suggested the  unambiguous and succinct  title of this  monograph.

Finally, this monograph, and my research that led to it, would not have been possible without the constant love and support of my family, especially my  wife Huili, and my son Oliver.

%\appendix
%\chapter{Appendix: Universal Attainability of Second-Order Asymptotics for Channel Coding} \label{app:universal}

%\backmatter  % references

\bibliographystyle{plain}
\bibliography{isitbib}

\end{document}